\documentclass[11pt]{scrartcl}

\usepackage{comment}

\usepackage{graphicx}
\usepackage{bbm}
\usepackage{bm}  
\usepackage{amsmath}
\usepackage{amssymb}
\usepackage{amsthm}
\usepackage{thmtools}
\usepackage[font=small,labelfont=bf,format=plain]{caption} 
\usepackage{microtype}     

\usepackage{multicol}
\usepackage[font=footnotesize]{subcaption}
\usepackage{tabularray}
\usepackage{pdflscape}

\usepackage{xcolor}
\definecolor{darkgreen}{rgb}{0,.5,0}
\definecolor{darkblue}{rgb}{0,0,.6}
\definecolor{darkred}{rgb}{.6,0,0}
\definecolor{lightgray}{rgb}{.8,.8,.8}
\definecolor{jred}{rgb}{0.8,0,0}
\definecolor{jgreen}{rgb}{0,0.7,0}
\definecolor{jblue}{rgb}{0,0,0.8}

\usepackage{tikz}
\usetikzlibrary{arrows.meta}
\usetikzlibrary{calc}
\usetikzlibrary{shapes, patterns, positioning}

\usetikzlibrary{decorations.pathreplacing,
	calligraphy}
\tikzstyle{c} =	[coordinate]
\tikzstyle{v} = [circle, draw=black, line width=.2pt, fill=black, inner sep=0pt, minimum size=1.5mm]
\tikzstyle{vb} = [circle, inner sep=0.1pt, fill=jred, minimum size=1mm]
\tikzstyle{e} =	[draw=jred,line width=1pt]
\tikzstyle{eb}= [draw=jgreen,line width=1pt]

\tikzstyle{edge} = [black,line width=.4mm]
\tikzstyle{vertex}=[circle,minimum size=2.5mm, draw=black, line width=.35mm,fill=black, inner sep=0mm]
\tikzstyle{smallvertex}=[circle,minimum size=1.8mm, draw=black, fill=black, inner sep=0mm]
\tikzstyle{tikzbrace}= [decorate,
		decoration={calligraphic brace, amplitude=5pt,
			raise=1pt, mirror},
		very thick,
		pen colour=black]
        
\usepackage[
pdfauthor={Paul-Hermann Balduf, Johannes Thürigen},			 
pdfcreator={LaTeX with hyperref},
colorlinks=true,	 
linkcolor=darkred,	 
citecolor=darkgreen,	 
filecolor=black,	 
urlcolor=darkblue,	 
bookmarks=true,
plainpages=false,
pdfpagelabels=true
]{hyperref}

\numberwithin{equation}{section}
\usepackage[nameinlink]{cleveref}	
\usepackage{placeins} 
\usepackage{csquotes} 

\newtheorem{theorem}{Theorem}
\newtheorem{example}[theorem]{Example}
\newtheorem{definition}[theorem]{Definition}
\newtheorem{lemma}[theorem]{Lemma} 
\newtheorem{conjecture}[theorem]{Conjecture} 
\newtheorem{proposition}[theorem]{Proposition}

\usepackage[
	backend=biber,
	style=phys,
	date=year,
	eprint=true,
	url=true,
	isbn=false
]{biblatex}
\AtEveryBibitem{
	\clearfield{urlyear}
	\clearfield{urlmonth}
}

\renewcommand\[{\begin{equation}}
\renewcommand\]{\end{equation}}

\renewcommand{\O}{\textrm{O}}

\newcommand{\R}{ {\mathbb{R}} }

\newcommand{\dual}{{G^\star}}
\newcommand{\abs}[1]{\left\lvert #1 \right\rvert}
\renewcommand{\d}{\textnormal{d}}
\newcommand{\ren} {{ \mathcal R}}
\newcommand{\period}{\mathcal P}
\newcommand{\Aut}{\operatorname{Aut}}
\newcommand{\prim}{\operatorname{prim}}

\newcommand\scalemath[2]{\scalebox{#1}{\mbox{\ensuremath{\displaystyle #2}}}}

\newcommand{\melon}{\begin{tikzpicture}[x=4ex,y=4ex,baseline={([yshift=-.6ex]current bounding box.center)}] 
  \node[v](v1) at (-.8,0){};
  \node[v](v2) at (0,0){};
  \draw[edge,bend angle=25,bend left] (v1) to (v2);
  \draw[edge,bend angle=60,bend left] (v1) to (v2);
  \draw[edge,bend angle=25,bend right] (v1) to (v2);
  \draw[edge,bend angle=60,bend right] (v1) to (v2);
\end{tikzpicture}}

\newcommand{\pmelon}{\begin{tikzpicture}[x=4ex,y=4ex,baseline={([yshift=-.6ex]current bounding box.center)}] 
  \node[v](v1) at (-.8,0){};
  \node[v](v2) at (0,0){};
  \node at (-.4,.1){$.$};
  \node at (-.4,-.05){$.$};
  \node at (-.4,-.2){$.$};
  \draw[edge,bend angle=40,bend left] (v1) to (v2);
  \draw[edge,bend angle=75,bend left] (v1) to (v2);
	\draw[edge,bend angle=75,bend right] (v1) to (v2);
\end{tikzpicture}}

\newcommand{\fish}{\begin{tikzpicture}[x=2ex,y=2ex,baseline={([yshift=-.6ex]current bounding box.center)}] 
  \coordinate(x0);
	\node[v](v1) at ($(x0) +(-1.2,0)$){};
	\node[v](v2) at ($(x0) +(0,0)$){};
	\draw[edge,bend angle=55,bend left] (v1) to (v2);
	\draw[edge,bend angle=55,bend right] (v1) to (v2);
	\draw[edge] (v1) -- +(135:.7);
	\draw[edge] (v1) -- +(225:.7);
	\draw[edge] (v2) -- +(45:.7);
	\draw[edge] (v2) -- +(-45:.7);
\end{tikzpicture}}

\newcommand{\tadpole}{
\begin{tikzpicture}[x=4ex,y=4ex,baseline={([yshift=-.6ex]current bounding box.center)}] 
    \node[v] (v1) {};   
    \draw[edge] (v1) -- +(135:.4);
	\draw[edge] (v1) -- +(225:.4);
    \coordinate [right=.5 of v1] (c);
    \path
    (c) edge [edge, bend right=75] (v1)
    (c) edge [edge, bend left=75] (v1);
\end{tikzpicture}}

\newcommand{\propbubble}{
\begin{tikzpicture}[x=4ex,y=4ex,baseline={([yshift=-.6ex]current bounding box.center)}] 
    \node[v] (v1) {};   
    \draw[edge] (v1) -- +(135:.4);
    \draw[edge] (v1) -- +(225:.4);
    \node [v, right=2.1 of v1] (v4){};
    \foreach \i/\j in {1/2,2/3,4/5}{
      \node [v, right=.5 of v\i] (v\j){} ;
      \draw[edge,bend angle=55,bend left] (v\i) to (v\j);
      \draw[edge,bend angle=55,bend right] (v\i) to (v\j);
    }
    \draw[edge, dotted] (v3) -- (v4);
    \coordinate [right=.5 of v5] (c);
    \path
    (c) edge [edge, bend right=75] (v5)
    (c) edge [edge, bend left=75] (v5);
\end{tikzpicture}}

\newcommand{\vacbubble}{\,
\begin{tikzpicture}[x=4ex,y=4ex,baseline={([yshift=-.6ex]current bounding box.center)}] 
    \coordinate (v);
    \node[v, right=.5 of v] (v1) {};   
    \node [v, right=2.3 of v1, label=left: ...] (v4){};
    \foreach \i/\j in {1/2,2/3,4/5}{
      \node [v, right=.5 of v\i] (v\j){} ;
      \draw[edge,bend angle=55,bend left] (v\i) to (v\j);
      \draw[edge,bend angle=55,bend right] (v\i) to (v\j);
    }
    \coordinate [right=.5 of v5] (c);
    \path
    (v) edge [edge, bend right=75] (v1)
    (v) edge [edge, bend left=75] (v1)
    (c) edge [edge, bend right=75] (v5)
    (c) edge [edge, bend left=75] (v5);
\end{tikzpicture}}

\newcommand{\necklace}{\,
\begin{tikzpicture}[x=4ex,y=4ex,baseline={([yshift=-.6ex]current bounding box.center)}] 
    \node [v] (v1) at (.8,0) {};  
    \foreach \i/\j in {1/2,2/3,3/4,4/5,5/6,6/7}{
      \begin{scope}[rotate=\i *49]
      \node [v] (v\j) at (.8,0) {};
      \draw[edge,bend angle=45,bend left] (v\i) to (v\j);
      \draw[edge,bend angle=45,bend right] (v\i) to (v\j);
      \end{scope}
    }
    \path (v1) edge [edge, dotted, bend left=30] (v7);
\end{tikzpicture}}

\newcommand{\vtxbubble}{
\begin{tikzpicture}[x=4ex,y=4ex,baseline={([yshift=-.6ex]current bounding box.center)}] 
    \node[v] (v1) {};   
    \draw[edge] (v1) -- +(135:.4);
	\draw[edge] (v1) -- +(225:.4);
    \node [v, right=2.3 of v1, label=left: ...] (v4){};
    \foreach \i/\j in {1/2,2/3,4/5}{
      \node [v, right=.5 of v\i] (v\j){} ;
      \draw[edge,bend angle=55,bend left] (v\i) to (v\j);
      \draw[edge,bend angle=55,bend right] (v\i) to (v\j);
    }  
    \draw[edge] (v5) -- +(45:.4);
	\draw[edge] (v5) -- +(-45:.4);
\end{tikzpicture}}

\begin{document}

\title{Primitive asymptotics in $\phi^4$ vector theory}

\author{Paul-Hermann Balduf\footnote{Mathematical Institute, University of Oxford, Andrew Wiles Building, Woodstock Road, Oxford, OX2 6GG, United Kingdom}, 
Johannes Thürigen%
\footnote{Institute for Analysis and Numerics, University of M\"unster, Orl\'eans-Ring 10, 48149 M\"unster, Germany}
}
\maketitle

\begin{abstract}
A longstanding conjecture in $\phi^4_4$ theory is that primitive graphs dominate the beta function asymptotically at large loop order in the minimal-subtraction scheme. 
Here we investigate this issue by exploiting additional combinatorial structure coming from an extension to vectors with $\O(N)$ symmetry.
For the 0-dimensional case, we calculate the $N$-dependent generating function of primitive graphs and its asymptotics, including arbitrarily many subleading corrections.
We find that the leading  asymptotic growth rate becomes visible only above $\approx 25$ loops, while data at lower order is suggestive of a wrong asymptotics.  Our results also yield the symmetry-factor weighted sum of 3-connected cubic graphs, and the exact asymptotics of Martin invariants. 

For individual Feynman graphs, we give bounds on their degree in $N$ depending on their coradical degree, and  construct the primitive graphs of highest degree explicitly. 
We calculate the 4D primitive beta function numerically up to 17 loops, 
and find its behaviour to be qualitatively similar to the 0D case. The locations of zeros quickly approach their large-loop asymptotics at negative integer $N$, while the growth rate of the beta function differs from the asymptotic prediction even at 17 loops.

\end{abstract}

\newpage

\setcounter{tocdepth}{2}
\tableofcontents

\newpage

\section{Introduction}

\subsection{Background and Motivation}

In the renormalization of perturbative quantum field theory (QFT), 
one can view \emph{primitive} graphs \cite{kreimer_overlapping_1999} as the building blocks of the theory.
At a fixed spacetime dimension, a graph is primitive when its Feynman integral is superficially divergent and has no subdivergence. In the present work, our use of the term primitive always refers to four spacetime dimensions.  One can generate all other divergent Feynman graphs by recursive insertion of
primitive graphs into primitive or  convergent ones, potentially modulo Ward identities \cite{kreimer_anatomy_2006, vansuijlekom_renormalization_2007}.  
In typical renormalizable QFTs, the number of Feynman graphs grows factorially with the loop order,  and  primitive graphs asymptotically represent a non-vanishing fraction of all Feynman graphs \cite{borinsky_graphs_2018}.
In the case of scalar $\phi^4_4$ theory in four dimensions, asymptotically $e^{-\frac{15}{4}}\approx  2.3\%$ of all 1PI graphs are primitive, compare \cref{tab:primitive_asymptotic_coefficients}. At the same time, it is conjectured that primitive graphs dominate the large-loop asymptotics of the beta function in the minimal subtraction scheme in $\phi^4$ theory \cite{mckane_instanton_1978, mckane_perturbation_2019}. 
This conjecture has never been proven, and numerical calculations up to 18 loops are inconclusive \cite{balduf_statistics_2023}\footnote{Very recently, more numerical data has appeared in \cite{borinsky_tropicalized_2025}, and the model of \emph{tropical field theory} has   been analyzed with the result that primitives  likely do not dominate the MS beta function there \cite{balduf_renormalized_2025}.}.

\medskip
One way to improve the control of scalar field theory is to extend the field to an $N$-dimensional vector and demand invariance of the theory under $\O(N)$ symmetry transformations; this setup is called \emph{vector model}.
It allows for another perturbative expansion, using  
$\frac 1 N$ as expansion parameter.
This large-$N$ expansion of $\O(N)$ symmetric scalar theories  has been studied extensively starting from the 1970s \cite{dolan_symmetry_1974,schnitzer_hartree_1974,vasilev_simple_1981,abbott_bound_1976,vasiliev_1_1981,vasiliev_1_1982}, reviews include \cite{zinn-justin_vector_1998,moshe_quantum_2003}. 
The limit is dominated by \enquote{bubble} graphs, which are chains of 1-loop multiedges (\emph{fish} graphs, \cref{fig:fish}). Analogous large-$N$ expansions  have also been fruitful for example in the Gross-Neveu model \cite{boehmer_large_2009,gracey_large_2018}, in   gauge theories \cite{thooft_planar_1974}, or in matrix models  \cite{bessis_quantum_1980,difrancesco_2d_2004}.  Specifically for   higher rank tensors, the large-$N$ limit is dominated by the so-called melonic graphs \cite{gurau_notes_2022,jepsen_rg_2024}. Melonic graphs have an iterative structure, where insertion of melonic graphs does not change the degree \cite{gurau_random_2016} and reducing all melonic subgraphs leads to so-called schemes \cite{gurau_regular_2016}. However, we stress that this recursive structure is distinct from the one of primitives; in particular, schemes are not necessarily primitive.

\medskip 

Being subleading in the large-$N$ limit of $O(N)$ symmetric $\phi^4_4$ theory, primitive graphs have not received  particular attention in the large-$N$ literature yet. The present article is meant to fill that gap. There are a number of reasons to undertake this endeavour: 

Firstly, for the perturbative expansion of $\phi_4^4$ theory in small coupling constant, the asymptotics at large loop order $L \rightarrow \infty$ has been determined from instanton calculations \cite{mckane_perturbation_2019}. It is a nontrivial question whether this large-order asymptotics constitutes a numerically accurate prediction at a given finite loop order $L$.  If this prediction is accurate, we will say that this $L$ lies in the \enquote{asymptotic regime}. For   $\phi^4$ theory, various quantities related to the number  of graphs are in the asymptotic regime already for $L\geq 10$ loops, while the growth rate of the Feynman integrals of primitive graphs even at $L=18$ loops is not   \cite{balduf_statistics_2023}. By studying the $N$ dependence of such quantities, potential further insights can be gained regarding their large-order asymptotics.

Secondly, the large-$N$ limit of the vector model is given by chains of 1-loop \emph{fish} graphs. The 1-loop graph, at the same time, determines the leading term of the log-expansion in perturbative renormalization and certain classes of chained and nested 1-loop  graphs can be solved analytically to all orders \cite{delbourgo_dimensional_1996,broadhurst_exact_2001,delbourgo_dimensional_1997,kreimer_etude_2006}. 
On the Hopf algebra side, it is well understood how the subleading orders of the log-expansion are determined by primitive graphs of higher loop order \cite{kruger_filtrations_2015,kruger_log_2020,courtiel_nexttok_2020,courtiel_terminal_2017}. 
The $N$ dependence of primitive graphs therefore has direct bearing on the $N$ dependence of the leading-log expansion.

Thirdly,  the \emph{Martin invariant} \cite{panzer_feynman_2025}, which is the (analytic continuation of the) $\O(N)$~symmetry factor evaluated at $N=-2$, is strongly correlated with the numerical value of the Feynman integral of primitive graphs  \cite{balduf_predicting_2024}.
It therefore appears that the $\O(N)$~symmetry factors of primitive graphs have a meaning for Feynman integrals \emph{as such}, beyond being a weighting factor in the perturbative expansion.

Fourthly,   the global $\O(N)$ symmetry  
is a combinatorially tractable example of a symmetry group that can serve as an illustrative case for the interplay of renormalization and symmetries in the form of Hopf ideals \cite{kreimer_anatomy_2006,kreimer_remark_2008,kissler_offshell_2021}. Factorization relations of $\O(N)$ symmetry factors in that sense are Ward identities which guarantee renormalizability.

\subsection{Outline and results}\label{sec:outline}
In the present subsection, we give a brief summary of all major findings of the article, with references to the corresponding later sections. The subsequent subsection (\cref{sec:discussion}) is a discussion of the takeaways and remaining open questions.

\medskip 

\Cref{sec:setup} reviews $\O(N)$-symmetric $\phi^4$ theory and the central objects to be examined in the article:
Each Feynman graph has not only a symmetry factor of automorphisms $\frac{1}{\abs{\Aut}}$, but also  an $\O(N)$-symmetry factor $T(G,N)$ (\cref{def:TGN}), which   is a polynomial in $N$. $T(G,N)$ can be construct graphically by decomposing every vertex of the graph, and it factorizes under insertion of subgraphs (\cref{sec:symmetric_theory}). 
The \emph{beta function} $\beta$  can be computed from a sum of all one-particle irreducible (1PI)
vertex-type Feynman integrals, while the \emph{primitive beta function} $\beta^\text{prim}$ is this sum restricted to primitive vertex-type graphs (\cref{betaprim_expansion}), called \emph{periods} $\period(G)\in \mathbb R$ (\cref{sec:periods}).
The $L$-loop coefficients $\beta_L$ of the beta function and its primitive contribution $\beta^\text{prim}_L$ grow factorially in $L$.
An open question (\cref{beta_conjecture}~\cite{mckane_instanton_1978}) is whether asymptotically $\beta^\text{prim}_L\rightarrow \beta_L$ as $L\rightarrow \infty$.
Finally,   the \emph{Martin invariant} $M^{[1]}(G)\in \mathbbm N$ is a simple transformation of $T(G,-2)$ (\cref{sec:Martin_invariant}).

\medskip 

\noindent
In \cref{sec:0dim}, we use a method known as \emph{zero-dimensional QFT} to enumerate the Feynman graphs of $\O(N)$-symmetric $\phi^4$ theory. Our use of \enquote{primitive}, \enquote{subdivergence} etc. in this section   refers to the power counting of the 4-dimensional theory.

\begin{itemize}
	\item The  path integral in 0-dimensional QFT is an ordinary integral. Its power series solution for arbitrary (not necessarily integer) $N$ can be found exactly (\cref{sec:0dim_vacuum}). This series enumerates   graphs with a given number of vertices and external edges, weighted by $T(G,N)$ and $\frac{1}{\abs{\Aut(G)}}$ (\cref{fig:Z_graphs_vaccum,fig:Z_graphs_2legs,fig:Z_graphs_4legs}).
    Algebraic transformations 
    allow to enumerate connected and  1PI graphs (\cref{sec:0dim_external}), and primitive graphs and their sums of Martin invariants (\cref{sec:0dim_counterterms})\footnote{Data files and implementations are available from the first author's website \href{https://www.paulbalduf.com/research}{www.paulbalduf.com}, and some of them as ancillary files to the electronic version of the preprint.\label{fn:dat}}.
	\item For the sum $p_L$ of primitive graphs at   $L\rightarrow \infty$ loops, we prove the  asymptotics  \cref{primitive_asymptotics}:
	\begin{align*} 
		p_L &\sim \frac{ 3^{\frac{N-1}{2} } e^{-\frac{12+3N}{4}} }  {\frac 4 3 \sqrt{2\pi} \Gamma   \big(  \frac{N+4}{2} \big) } \left( \frac 2 3 \right) ^{ L + \frac{N+5}{2} } \Gamma \Big( 
        { L+ \frac{N+5}{2}}\Big)  \left( 36- \frac{9(3 N^2 - 4N -80)}{4 (L + \frac{N+3}{2})} +\ldots \right).
	\end{align*}
	Arbitrarily many subleading terms can be computed, the sequence of subleading terms again grows factorially.  As a corollary, we obtain the all-order asymptotics of the sum of Martin invariants  (\cref{sec:0dim_asymptotics_Martin}).
	\item When the asymptotic expansion is used as an approximation at finite $L$, the numerical reliability varies substantially, depending on the class of graphs. For the sum of all (not necessarily connected) vertex-type graphs, the leading asymptotics differs from the true value by 1\% at order 16, and including the first subleading term reduces the error to $\ll 1\%$ already at order 2(!) (\cref{fig:Z_asymptotics_both}). Conversely, for the primitive graphs ($p_L$), the discrepancy  is larger than 20\% even at loop order 20 when the leading asymptotics is used, and   the first subleading asymptotic correction reaches 1\% only at 9 loops (\cref{fig:prim_asymptotics_both}). 
	\item  Additionally, for fixed $L$ the obtainable accuracy strongly depends on the value of $N$. We examine this dependence in detail for  the example of $p_L(N)$ in \cref{fig:minimum_loop_order}.
	\item For   $p_L(N)$, the ratio of successive loop orders behaves asymptotically as \cref{primitives_ratio},
	\begin{align}\label{primitives_ratio_2} 
		r_L(N) &:= \scalemath{.9}{\frac 1 L \frac{p_{L+1}(N)}{p_L(N)} \sim \frac 2 3 + \frac{N+5}{3}\frac 1 L - \frac{3N^2-4N-80}{24}\frac{1}{L^2}  +\ldots}, \qquad L \rightarrow \infty.
	\end{align}
	The limiting value $\frac 2 3$ is independent of $N$, and the first correction is linear in $N$. If one tries to measure $r_L(N)$ from the data at  $L\leq 18$, this measurement suggests an $N$-independent limiting value, and linear correction term, but these do \emph{not} coincide with the true asymptotics   (\cref{fig:0dim_primitive_ratio}). Conversely, the correct limiting value $\frac 2 3$ can only be measured from data of at least $L\geq 25$ loops.  
	\item A transformation of variables $U=N \hbar$ leads to power series in $U$ instead of $\hbar$, which represent the large-$N$ limit   (\cref{sec:0dim_largeN}). At fixed $L$,   $p_L(N)$  is a polynomial in $N$. The degree of this polynomial grows like $\frac 2 3 L$ (\cref{lem:TGN_bound}), but the average     degree of the monomials       grows like ~$\frac 1 2 \ln(L)$ (\cref{average_order_asymptotics}). This illustrates why the large-$N$ expansion is convergent while the large-$L$ expansion is not: At fixed $L$, the    leading    coefficients in $N$ of the polynomials $p_L(N)$ are vanishingly small compared to the low-order coefficients in $N$. 
\end{itemize}

\medskip 

\noindent

In \cref{sec:leading} we turn to the symmetry factors $T(G,N)$ of individual graphs at a fixed loop order $L$.
We identify those primitive graphs where $T(G,N)$ has largest degree in $N$, and give bounds on the degree for non-primitive graphs.
\begin{itemize}
	\item \Cref{sec:dual} introduces certain \emph{dual} graphs $\dual$. These are the Feynman graphs of the field variable $\sigma$ introduced through the Hubbard-Stratonovich transformation in \cref{sec:0dim_external}. Every vertex in the dual graph corresponds to a factor of $N$. 
	\item In general, the map from original graphs to dual graphs is not invertible.  An exception are the primitive graphs of leading degree in $N$ with $L=1\mod 3$ loops. They are in bijection with  3-connected cubic  graphs (\cref{thm:leading_3n+1}) and the \emph{line graph} construction (\cref{def:line_graph}) defines a unique inverse. This allows to compute the sum of 3-connected cubic graphs, weighted by symmetry factors (\cref{thm:cubic_count}).
	\item When $L=0\mod 3$ (\cref{sec:leading_2}) or $L=2\mod 3$ (\cref{sec:leading_3}), the leading-degree primitive graphs $G$ are constructed algorithmically by a slight variation of the line graph.  We have generated these graphs  for $L\leq 25$   (\cref{tab:leading_graphs_count})\textsuperscript{\ref{fn:dat}}.
	\item For an $L$-loop primitive vertex-type graph, the degree in $N$ is bounded by $\lfloor \frac 23L-\frac 23 \rfloor$ (\cref{lem:TGN_bound}). The $L$-loop \enquote{zigzag} graph, i.e. (1,2)-circulant, has degree $\lfloor \frac 12 L+2 \rfloor$, and therefore does not contribute at leading order in $N$ 
    (\cref{lem:zigzag}).   
	\item In \cref{tab:loworder_leading_graphs_count}, we give the number of primitive graphs of fixed degree in $N$  up to 16 loops.   The average degree grows slower than the degree bound $\lfloor \frac 2 3 L -\frac 2 3\rfloor$ (\cref{fig:mean_max_order}); that is, at large loop order most graphs do not saturate the bound. 
	\item By factorization properties, the degree bounds of primitive graphs  imply degree bounds for all other graphs (\cref{sec:nonprimitive}).
\end{itemize}

\medskip 

\noindent
In \cref{sec:4d}, we examine numerical results for the 4-dimensional theory, where each primitive graph obtains a non-trivial value, given by its Feynman period $\period(G)$ (\cref{def:period}). 
\begin{itemize}
	\item  In \cref{sec:4d_coefficients}, numerical values for the coefficients in $N$ of the primitive beta function $\beta^\text{prim}_L$, and evaluations at integer $N$, are given for $L\leq 16$   (\cref{tab:coefficients,tab:evaluations}).
	\item   At finite $L$, the $N$ dependence of the 4-dimensional $\beta^\text{prim}_L(N)$  is reasonably well described by that of the zero-dimensional  $p_L(N)$ (\cref{sec:0dim_4dim_relation}). Both functions have zeros close to negative even integer $N\leq -4$  (\cref{fig:largest_roots}), and the ratio $\beta^\text{prim}_L(N) / p_L(N)$ is a slowly and monotonically growing function of $N$ for $N>-4$ (\cref{fig:beta_N_ratio}).
	\item  The growth ratio $r_L(N)$ of $\beta^\text{prim}_L(N)$ is defined analogously to \cref{primitives_ratio_2}.  
    For $L\leq 18$, it  appears to converge towards an $N$-independent value,   which however is different from the expected asymptotics  (\cref{fig:4dim_beta_ratio}).     In view of the similar situation in zero dimension (\cref{fig:0dim_primitive_ratio}), we conclude that our numerical data does not confirm nor refute \cref{beta_conjecture}, that primitive graphs dominate at large $L$. An extrapolation of numerical data suggests that one will need at least 25 loops to measure the correct leading asymptotics from numerical data (\cref{sec:4d_asymptotics}).
\end{itemize}

\subsection{Discussion and Outlook}\label{sec:discussion}

We have not been able to clearly confirm or refute \cref{beta_conjecture} stating that primitive graphs dominate the beta function in $\phi^4_4$ theory. Rather, we found that for primitive graphs in the 0-dimensional theory, the numerical values are in reasonable agreement with the leading large-order asymptotics (i.e. below $20\%$ difference at $N=1$) only upwards of $\approx 25$ loops   (\cref{fig:prim_asymptotics_correction,fig:0dim_primitive_ratio}). For the 4-dimensional theory, an extrapolation of the low-order numerical data meets the leading asymptotics at around 25 loops  (\cref{fig:4dim_beta_ratio}), which suggests that again, the leading large-order asymptotics becomes a reliable estimate for the value at finite $L$  only upwards of 25 loops.

On the other hand, the $N$ dependence (for $N$ not too large) of these growth rates $r_L(N)$ is close to the leading-asymptotic pattern   already at $L \approx 10$ loops. In particular, the zeros of the $N$-dependent beta function converge towards their asymptotic values $N\in \left \lbrace -4, -6, \ldots \right \rbrace $ very rapidly already at  $L<10$ (\cref{fig:largest_roots}). 
This shows that the $N$-dependence of a quantity is not a reliable indicator for whether the absolute value of that same quantity is close to the leading large-order asymptotics.

It is furthermore remarkable that the reliability of the asymptotic expansion is very different for different classes of graphs: while the growth rate $r_L(N)$ of primitive graphs is captured with $20\%$ relative error by the leading asymptotics only beyond 25 loops, the sum of all graphs, or all connected graphs, are very close to the leading asymptotics already at much lower loop order (\cref{fig:Z_asymptotics_both}).  
We have seen in \cref{sec:0dim_4dim_relation} that for primitive graphs, the 4-dimensional theory behaves quite similar to the 0-dimensional one. However, it is not clear if that should be expected for non-primitive generating functions, too, because the renormalization of non-primitive graphs introduces an additional degree of freedom, which has no counterpart in the 0-dimensional theory.

By \cref{lem:circuit_polynomial_zeros}, a generating function of graphs vanishes at $N=-2j$ if it contains a sum of all orientations (channels) of a $(2j+2)$-valent subgraph. Asymptotically, the  generating functions of \cref{sec:0dim_asymptotics} for graphs with $2k$ legs have such zeros for $j\geq k$, in particular the primitive generating function $\prim(\hbar_\ren)$ of \cref{primitive_generating_function} has zeros for $j\geq 2$. 
It would be interesting to understand to what extent the opposite direction of \cref{lem:circuit_polynomial_zeros} holds asymptotically, that is, if we know a generating function vanishes asymptotically at $N=-2j$, does this imply that asymptotically, almost all orientations of $(2j+2)$-valent subgraphs are present?  For the sum of \emph{all} graphs, this statement is trivial: A sum of all graphs in particular contains all possible orientations of every subgraph. However,   in the case of primitive graphs, inserting a given 6-valent 1PI subgraph in one specific orientation might give rise to a primitive graph, while inserting the same subgraph in the same location in a different orientation might yield a non-primitive graph. Does the fact that the generating function asymptotically has a zero at $N=-4$ imply that asymptotically, almost all such insertions of 6-valent subgraphs yield primitives, and/or that a non-vanishing fraction of all 1PI graphs is primitive?

Regarding the classification of the $N$ dependence of the symmetry factors $T(G,N)$ of non-primitive graphs (\cref{sec:nonprimitive}), we reproduce the known fact that \enquote{bubble} graphs (fish chains) dominate at $N\rightarrow \infty$. Conversely, the fish-free graphs of highest degree in $N$ at every $L$ are indeed primitive, because inserting any subgraph that is not a multiedge effectively lowers the degree in a large-$N$ expansion. In that respect, our results are analogous to the dominance of \emph{melonic} graphs   \cite[Sect.~4]{gurau_random_2016}  \cite{bonzom_critical_2011} in tensor theories.
The   classification of $N$ dependence in terms of Gurau degree identifies \emph{schemes} as building blocks, analogously to the primitive graphs in our case, but details are different: Primitiveness and the coradical degree only implies bounds on the $N$ dependence (\cref{thm:degree_bound}), but not an exact classification. Conversely, the coradical degree amounts to the exact classification in terms of leading-log expansions. One of the main results of the present work is that these two, seemingly unrelated, expansions in fact have a lot in common. We leave a precise classification of the $N$ dependence of the leading-log expansion for future work.

Regarding the role of degree bounds, it is worth stressing that the exact asymptotics of the average order in $N$, $\left \langle k \right \rangle _L$ from \cref{average_order_asymptotics}, gives a very nice and explicit understanding of the interplay between loop expansion and large-$N$ expansion: For the large-$N$ expansion, one is concerned with degree bounds of classes of graphs, which are linear in the loop order $L$. However, the average order is only logarithmic in $L$, which implies that at large fixed $L$, the overwhelming majority of terms in the $\O(N)$ symmetry factors is of very low order in $N$, much lower than the degree bound. Both the loop expansion and the large-$N$ expansion ultimately include the same infinite set of graphs, and it is the smallness of the large-order coefficients in $N$ in the polynomials $T(G,N)$ that makes it possible for the large-$N$ expansion converge, while the loop expansion does not.

The fact that the $\O(N)$ symmetry factor $T(G,N)$ factorizes for insertion of subgraphs of \emph{arbitrary} valence (\cref{lem:factorization_T}), not just 4-valent ones, is an example of how the symmetry of the theory is reflected in an infinite set of identities. In fact, one could include arbitrary interaction terms of the form $(\vec \phi)^{2k}$ into the Lagrangian without spoiling the factorization. 
Such vertices would render the theory non-renormalizable by power counting (unless the spacetime dimension is changed), but the counterterms of higher valence would still have the same $T(G,N)$ as the vertices. Moreover, the vertices of higher valence would be \enquote{proportional} to products of vertices of lower valence in the sense that the sum of trees produces the same $T(G,N)$ as a single vertex. This behaviour is a very simple, but thereby instructive, example of how a symmetry imposes relation between arbitrarily many terms, analogous to the mechanism that has been suggested for gravity \cite{kreimer_anatomy_2006,kreimer_remark_2008,kissler_offshell_2021}. The case of gravity is of course more complicated since for a \emph{local} symmetry, the identities involve derivatives and thereby the kinematics of the vertices, which is not the case for the present global $\O(N)$ symmetry.

\newpage

\section{Setup}\label{sec:setup}

\subsection{$\O(N)$-symmetric $\phi^4$ theory} \label{sec:symmetric_theory}
We consider a Euclidean scalar $\phi^4$ theory where the scalar field $\vec \phi$ has an internal $\O(N)$ symmetry. For now, we assume $N$ to be a positive integer such that $\vec \phi$ is an $N$-component vector $(\phi_1, \ldots, \phi_N)$, and the Lagrangian density reads
\begin{align}\label{lagrangian_phi4}
	\mathcal L &= -\frac 12 \left( \partial_\mu \phi_1 \partial^\mu \phi_1 + \partial_\mu \phi_2 \partial^\mu \phi_2 + \ldots + \partial_\mu \phi_N \partial^\mu \phi_N \right)- \frac{\lambda}{4!} \left( \phi_1^2 + \phi_2^2 + \ldots + \phi_N^2 \right) ^2.
\end{align}
For each positive integer $N$, \cref{lagrangian_phi4} represents a different universality class of critical phenomena.
More background and a review of the physical significance of $\O(N)$ symmetry can be found in \cite{pelissetto_critical_2002}, perturbation theory of $\O(N)$-symmetric $\phi^4$ theory, with example calculations and further generalizations, for example in \cite{nickel_compilation_1977,kleinert_critical_2001}.

Feynman graphs of $\phi^4$ theory are 4-regular up to vertices which are connected to external edges. 
The structure of the $\O(N)$-symmetric interaction term in \cref{lagrangian_phi4} implies that each 4-valent vertex gives rise to a sum of three terms in Feynman amplitudes, corresponding to the three different possibilities to partition the four edges into pairs. 
Graphically, such a partition can be interpreted as a \emph{decomposition} of the vertex into two lines connecting the four adjacent edges
\footnote{These vertex decompositions have also been called \emph{vertex states} in  \cite{ellis-monaghan_new_1998}, and \emph{transitions} in \cite{panzer_feynman_2025}.
Alternatively, one can view each such decomposition as a strand graph \cite{thurigen_renormalization_2021, thurigen_renormalization_2021a, hock_combinatorial_2024}, that is a graph with the additional structure of pairings (green lines) between edges,
$
\begin{tikzpicture}[scale=.6,baseline={([yshift=-.59ex]current bounding box.center)}] 
    \node[v](v) at (0,0){};
    \foreach \i in {1,2,3,4}{
    \begin{scope}[rotate=\i*90]
      \node[c] (e\i) at (.5,.5){};
      \draw[e] (v) -- (e\i);
		\end{scope}
    }
		
		\draw[->, line width=.3mm] (v)++(1,0) -- ++(.8,0);

    \begin{scope}[xshift = 3cm]
    \node[v](v) at (0,0){};
    \foreach \i in {1,2,3,4}{
    \begin{scope}[rotate=\i*90]
    \node[vb] (e\i) at (.5,.5){};
    \draw[e] (v) -- (e\i);
		\end{scope}
    }
    \draw[eb] (e1) -- (e2);
    \draw[eb] (e3) -- (e4);
    \end{scope}
		
		\node at (4.25,0){$+$};
 
		\begin{scope}[xshift = 5.5cm]
    \node[v](v) at (0,0){};
    \foreach \i in {1,2,3,4}{
    \begin{scope}[rotate=\i*90]
    \node[vb] (e\i) at (.5,.5){};
    \draw[e] (v) -- (e\i);
		\end{scope}
    }
    \draw[eb] (e1) -- (e4);
    \draw[eb] (e3) -- (e2);
    \end{scope}
		
		\node at (6.75,0){$+$};
		
		\begin{scope}[xshift = 8cm]
    \node[v](v) at (0,0){};
    \foreach \i in {1,2,3,4}{
    \begin{scope}[rotate=\i*90]
    \node[vb] (e\i) at (.5,.5){};
    \draw[e] (v) -- (e\i);
		\end{scope}
    }
    \draw[eb] (e1) to[bend angle=30,bend left] (e3);
    \draw[eb] (e4) to[bend angle=30,bend right] (e2);
    \end{scope}
\end{tikzpicture}
$.
The resulting circuits define faces such that a strand graph becomes a 2-complex. The case of vector fields gives rise to a single strand adjacent to each edge such that the 2-complex is only 0-connected, i.e.~faces meet only at vertices. In this sense, the stranding is a ``decomposition'' in this case.
In contrast, matrix and order-$r$ tensor multi-scalar fields have $2$ or $r>2$ strands per edge resulting in 2-complexes which are connected surfaces \cite{thooft_planar_1974} or $r$-dimensional pseudo manifolds \cite{gurau_lost_2010}, respectively. 
} 
as shown in \cref{fig:vertex_decomposition}. 

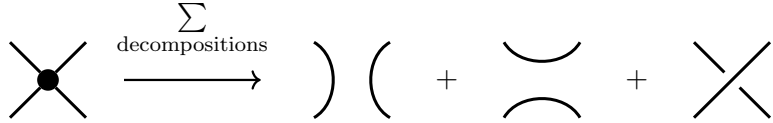
\begin{figure}[htb]
	\centering 
	\begin{tikzpicture}
		\node[vertex](v1) at (0,0){};
		\draw[edge] (v1) -- +(.5,.5);
		\draw[edge] (v1) -- +(.5,-.5);
		\draw[edge] (v1) -- +(-.5,.5);
		\draw[edge] (v1) -- +(-.5,-.5);
		
		\draw[->, line width=.3mm] (v1)++(1,0) -- node[label=above: $\sum\limits_{\text{decompositions}}$]{} ++(1.8,0);
		
		\node (v2) at (4,0){};
		\draw[edge] (v2)+(.5,.5) to[bend angle=60,bend right] ($(v2)+(.5,-.5) $);
		\draw[edge] (v2)+(-.5,.5) to[bend angle=60,bend left] ($(v2)+(-.5,-.5) $);
		
		\node at (5.25,0){$+$};
		\node (v3) at (6.5,0){};
		\draw[edge] (v3)+(.5,.5) to[bend angle=60,bend left] ($(v3)+(-.5,.5) $);
		\draw[edge] (v3)+(.5,-.5) to[bend angle=60,bend right] ($(v3)+(-.5,-.5) $);
		
		\node at (7.75,0){$+$};
		
		\node (v4) at (9,0){};
		\draw[edge] (v4)+(.5,.5) -- ($(v4)+(-.5,-.5) $);
		\draw[edge] (v4)+(-.5,.5) -- ($(v4)+(-.1,.1) $);
		\draw[edge] (v4)+(.5,-.5) -- ($(v4)+(.1,-.1) $);

	\end{tikzpicture}
	\caption{A 4-valent vertex allows for three non-isomorphic \emph{decompositions}, each of which consists of a pair of edges. To compute a Feynman amplitude, all three decompositions need to be summed. }
	\label{fig:vertex_decomposition}
\end{figure}

Then, a decomposition of a graph $G$ is a choice of decomposition of all of its vertices. 
A decomposition of $G$ yields a set of disjoint circuits and paths, where the paths begin and end at external edges of $G$. 

\begin{definition}\label{def:circuit_partition_polynomial} 
	The \emph{circuit partition polynomial} $J$ assigns a factor $N$ to each circuit in the sum of all decompositions of a graph,
	\begin{align}
		J(G,N):=\sum_{\substack{\textnormal{decompositions}\\\textnormal{of }G}} N^{\# \textnormal{circuits}}. 
\end{align}
\end{definition}
If $f(N)$ is a formal power series in $N$, we use the notation $[N^k]f(N)$ to denote extraction of the coefficient of $N^k$. 
For the circuit partition polynomial, $[N^k]J(G,N)$ counts how many different ways there are to obtain exactly $k$ circuits from the graph $G$.

The polynomial $J(G,N)$ yields a symmetry factor for each Feynman graph $G$. Due to the factor $\frac{1}{4!}$ in the interaction term in the Lagrangian density (\cref{lagrangian_phi4}), ordinary $\phi^4$ theory is restored in the case $N=1$ since the sum of all three decompositions of a vertex is scaled by $\frac 1 3$:
\begin{definition}\label{def:TGN}
	Let $J(G,N)$ be the circuit partition polynomial (\cref{def:circuit_partition_polynomial}) of a Feynman graph $G$ with vertices $V_G$. 	The \emph{$\O(N)$-symmetry factor} of $\phi^4$ theory is defined as
	\begin{align*}
		T(G,N) &:=\frac{J(G,N)}{J(G,1)}= \frac{1}{3^{\abs{V_G}}} J(G,N). 
	\end{align*}
\end{definition}
\noindent
For the latter equality, we have used that in $\phi^4$ theory, every vertex is 4-valent. One can set up an analogous sum over decompositions for arbitrary even vertex valence $2p$. In that case,   the factor 3 in the denominator of \cref{def:TGN}  needs to be replaced by $(2p-1)!!$, which is the number of non-isomorphic vertex decompletions.

\begin{figure}[htbp]
	\centering 
	\begin{tikzpicture}[scale=.9]
	
		\coordinate(x0) at (-3.5,.5);
		\node[vertex](v1) at ($(x0) +(-1.2,0)$){};
		\node[vertex](v2) at ($(x0) +(0,0)$){};
		\draw[edge,bend angle=45,bend left] (v1) to (v2);
		\draw[edge,bend angle=45,bend right] (v1) to (v2);
		\draw[edge] (v1) -- +(135:.7);
		\draw[edge] (v1) -- +(225:.7);
		\draw[edge] (v2) -- +(45:.7);
		\draw[edge] (v2) -- +(-45:.7);
		
		\draw[->, line width=.3mm] (v2)++(1,0) -- ++(.8,0);
		
		\node (v2) at (0,.5){};
		\node at ($(v2) + (-1,0)$) {$4\times $};
		\draw[edge] (v2)+(.5,.5) to[bend angle=60,bend right] ($(v2)+(.5,-.5) $);
		\draw[edge] (v2)+(-.5,.5) to[bend angle=60,bend left] ($(v2)+(-.5,-.5) $);

		\node (v3) at (2.5,.5){};
		\node at($(v3) + (-1,0)$){$+2 \times $};
		\draw[edge] (v3)+(.5,.5) to[bend angle=60,bend left] ($(v3)+(-.5,.5) $);
		\draw[edge] (v3)+(.5,-.5) to[bend angle=60,bend right] ($(v3)+(-.5,-.5) $);
		
		\node (v4) at (5,.5){};
		\node at ($(v4) + (-1,0)$){$+2 \times $};
		\draw[edge] (v4)+(.5,.5) -- ($(v4)+(-.5,-.5) $);
		\draw[edge] (v4)+(-.5,.5) -- ($(v4)+(-.1,.1) $);
		\draw[edge] (v4)+(.5,-.5) -- ($(v4)+(.1,-.1) $);
		
		\node (v5) at (8,.5){};
		\node at ($(v5) + (-1.3,0)$){$+1 \times $};
		\draw[edge] (v5)+(.8,.5) to[bend angle=60,bend right] ($(v5)+(.8,-.5) $);
		\draw[edge] (v5)+(-.8,.5) to[bend angle=60,bend left] ($(v5)+(-.8,-.5) $);
		\draw[edge ] (v5) circle(.4);

	\end{tikzpicture}
	\caption{The \emph{fish} graph (also known as double edge, 1-loop multiedge, banana, or bubble) has two vertices. In a sum over all decompositions, both vertices produce three terms according to \cref{fig:vertex_decomposition} which yields nine terms in total, only one of which contains a circuit. 
    Thus, the circuit partition polynomial of this graph is $N+8$. Notice that the circuit partition polynomial depends on whether the four external edges are open or pairwise connected, compare  \cref{ex:fish_chain}.}
	\label{fig:fish}
\end{figure}
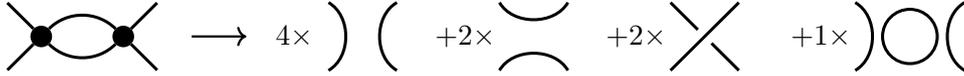

\begin{example}\label{ex:fish_chain}
	A Feynman graph that will be relevant in many computations later on is the \emph{fish}, shown in \cref{fig:fish}. 
	As an example for how \cref{def:TGN,def:circuit_partition_polynomial} allow to compute $J(G,N)$ by repeatedly decomposing vertices, consider chains of fish graphs. These chains are also called \enquote{bubble} graphs. We consider four different ways to terminate the two ends of the chain. In all cases, the number of loops is $L$.
\begin{enumerate}
\item For a chain of $(L-1)$ fish graphs with the two edges on one of the sides joined into a self loop,
\begin{align}\label{eq:bubble-tadpole}
	J\big(\propbubble,N \big)&= \left( {N+2} \right)^L.
\end{align}
This follows by induction as 
$J(\tadpole,N)=N+2$  and extending $\propbubble$ from $L$ to $L+1$ loops yields an additional factor $N+2$ (decompose the last vertex and find that in two cases the decomposition yields the same as for $L$ loops while one case yields a new disconnected loop, thus a factor $N$).
\item Closing also the other two open edges into a self loop simply yields a factor $N$, thus such a
vacuum graph with $L$ loops has
\begin{align}\label{eq:bubble-vacuum}
	J\big(\vacbubble, N \big)&= N \left( {N+2} \right) ^{L-1}.
\end{align}
\item A ring consisting of $(L-1)$ fish graphs is an $L$-loop vacuum graph with
\begin{align}\label{eq:bubble-ring}
	J\Big(\necklace,N \Big) &=   2^{L-2}(N-1)(N+2) + (N+2)^{L-1}.
\end{align}
The argument is again by induction, starting with $J(\melon,N)=3N(N+2)$ (see also \cref{eq:melon polynomial}) and noting that decomposition of any one vertex in the ring yields either the tadpole chain \eqref{eq:bubble-vacuum} or in the other two cases again the ring polynomial \eqref{eq:bubble-ring} of one loop less.
\item A chain of $L$ fish graphs, where the edges on both ends are open, has 
\begin{align}\label{eq:bubble-vertex}
	J\big(\vtxbubble,N\big)&= \frac{(N-1) 2^{L+1} + (N+2)^{L+1}}{N}.
\end{align}
    In this case, induction starts with $J(\fish)=N+8$ (\cref{fig:fish})
    and decomposing the last vertex gives two copies of $J(\vtxbubble,N)$ with $L-1$ loops and one tadpole chain \cref{eq:bubble-tadpole}.
\end{enumerate}

\end{example}

In the present article, we typically let $N$ be arbitrary, but \cref{ex:fish_chain} already makes it clear that one can eliminate specific graphs from the theory by choosing certain values of~$N$. In particular, all self-loops (tadpoles) vanish at $N=-2$ and this value determines the \emph{Martin invariant}, see \cref{sec:Martin_invariant}.

\bigskip 
The circuit partition polynomial, and the related \emph{interlace polynomial} and \emph{Martin polynomial} (to be discussed later in \cref{sec:Martin_invariant}), have long been known in graph theory \cite{martin_enumerations_1977,arratia_interlace_2004,bollobas_evaluations_2002,ellis-monaghan_new_1998,ellis-monaghan_exploring_2004,ellis-monaghan_identities_2004}. In the following, we briefly introduce some further definitions, as well as elementary combinatorial properties of the polynomials $T(G,N)$ and $J(G,N)$.   
Although we are ultimately only interested in $\phi^4$ theory, it is useful to formulate these statements for graphs with vertices of arbitrary even valence $(2p)$ because, as shown below, the insertion of a subgraph with $2p$ external edges in $\phi^4$ theory factorizes in the same way as the insertion of a $2p$-valent vertex in a more general theory. Concretely, a vertex of even valence $2p$ can be decomposed into $(2p-1)!!$ matchings of its legs analogously to \cref{fig:vertex_decomposition}, and the polynomial $J(G,N)$ for such a graph is defined as the sum over all decompositions as in \cref{def:circuit_partition_polynomial}.  Further details, examples, and proofs for the below statements can be found in \cref{sec:combinatorial_properties}.
\begin{lemma}\label{lem:multiedge_J}
	A  $(2p-1)$-loop multiedge graph without external edges has 
	\begin{align*}
		J\big(\pmelon,N \big) &= (2p-1)!!\prod \nolimits_{j=0}^{p-1}(N+2j), \qquad T\big(\pmelon,N \big) = \frac {\prod \nolimits_{j=0}^{p-1}(N+2j)} {(2p-1)!!}.
	\end{align*}
\end{lemma}

\begin{definition}\label{def:uplus}
	Let $g_1$ and $g_2$ be graphs, each of which has exactly $(2p)$ external edges. The \emph{joining} operation $g_1 \uplus g_2$ is defined as the average of all $(2p)!$ matchings between the external edges of $g_1$ and $g_2$,  
	\begin{align*}
		g_1 \uplus g_2 := \frac{1}{(2p)!}\sum_{G \textnormal{ matching of } (g_1,g_2)}G.
	\end{align*}
\end{definition}
A graph is called \emph{$n$-valent} if it has $n$ external edges. In the present article, being concerned with  $\phi^4$ theory, we use the terms \emph{vertex-type graphs} for 4-valent graphs,  \emph{propagator-type graphs} for   2-valent graphs, and   \emph{vacuum graphs} for   0-valent graphs. Vacuum graphs in $\phi^4$ theory are 4-regular, which means that every vertex is 4-valent.

\begin{definition}\label{def:completion_decompletion}
	Let $g$ be a connected $n$-valent graph. The \emph{completion} $G
	= g\uplus \left \lbrace v \right \rbrace $ of $g$ is the graph where the $n$ external edges of $g$ have been joined to a new $n$-valent vertex~$v$. 
	Conversely, if $g$ arises from $G$ upon removing any one vertex  $v\in V_G$, then $g=G\setminus \left \lbrace v \right \rbrace   $ is called a \emph{decompletion} of $G$. 
\end{definition}

Notice that the sum of matchings in \cref{def:uplus} contains exactly $(2p)!$ graphs, which may be isomorphic. If $g_2=v$ is a single vertex, then all  $(2p)!$ graphs are isomorphic and the completion $g \uplus \left \lbrace v \right \rbrace=G $ is a single graph. Viewed as a vacuum graph, a completion has larger loop order than the corresponding decompletion. Later in the article, we will only be interested in completions as a way to organize vertex-type graphs in $\phi^4$ theory. To avoid shifts in the loop number, we define   the loop number of a completion to refer to  that of its decompletions. In particular, in $\phi^4$ theory a \emph{$L$-loop completion}  is a graph with $(L+2)$ 4-valent vertices and no external legs, such that its decompletions are vertex-type graphs of $L$ loops, see \cref{fig:completion}.

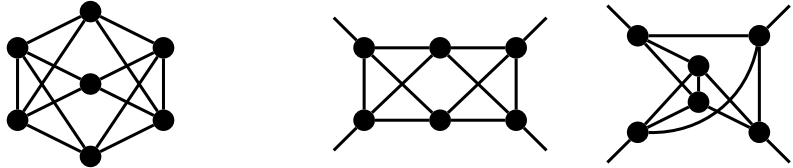
\begin{figure}[htbp]
	\centering 
	\begin{tikzpicture}[scale=.8]

		\node[vertex](v1) at (0,-1.2){};
		\node[vertex](v2) at (0,0){};
		\node[vertex](v3) at (0,1.2){};
		\node[vertex](v4) at (-1.2,.6){};
		\node[vertex](v5) at (-1.2,-.6){};
		\node[vertex](v6) at (1.2,.6){};
		\node[vertex](v7) at (1.2,-.6){};
		\draw[edge] (v1) -- (v4);
		\draw[edge] (v2) -- (v4);
		\draw[edge] (v3) -- (v4);
		\draw[edge] (v1) -- (v5);
		\draw[edge] (v2) -- (v5);
		\draw[edge] (v3) -- (v5);
		\draw[edge] (v4) -- (v5);
		
		\draw[edge] (v1) -- (v6);
		\draw[edge] (v2) -- (v6);
		\draw[edge] (v3) -- (v6);
		\draw[edge] (v1) -- (v7);
		\draw[edge] (v2) -- (v7);
		\draw[edge] (v3) -- (v7);
		\draw[edge] (v6) -- (v7);

		\node[vertex](v1) at (4.5,-.6 ){};
		\node[vertex](v2) at (4.5,.6){};
		\node[vertex](v3) at (5.75,-.6){};
		\node[vertex](v4) at (5.75,.6){};
		\node[vertex](v5) at (7,-.6){};
		\node[vertex](v6) at (7,.6){};
		\draw[edge] (v1) -- (v2);
		\draw[edge] (v1) -- (v3);
		\draw[edge] (v1) -- (v4);
		\draw[edge] (v2) -- (v4);
		
		\draw[edge] (v2) -- (v3);
		\draw[edge] (v3) -- (v5);
		\draw[edge] (v3) -- (v6);
		\draw[edge] (v4) -- (v6);
		\draw[edge] (v5) -- (v6);
		
		\draw[edge] (v4) -- (v5);
		\draw[edge] (v1) --+(-.5,-.5);
		\draw[edge] (v2) --+(-.5,.5);
		\draw[edge] (v5) --+(.5,-.5);
		\draw[edge] (v6) --+(.5,.5);
		
		\node[vertex](v1) at (9,.8 ){};
		\node[vertex](v2) at (9,.-.8){};
		\node[vertex](v3) at (10,.3){};
		\node[vertex](v4) at (10,-.3){};
		\node[vertex](v5) at (11,.8){};
		\node[vertex](v6) at (11,-.8){};
		
		\draw[edge] (v1) -- (v3);
		\draw[edge] (v1) -- (v4);
		\draw[edge] (v2) -- (v3);
		\draw[edge] (v2) -- (v4);
		\draw[edge] (v3) -- (v4);
		\draw[edge] (v1) -- (v5);
		\draw[edge, bend angle=40, bend right] (v2) to (v5);
		\draw[edge] (v6) -- (v5);
		\draw[edge] (v3) -- (v6);
		\draw[edge] (v4) -- (v6);
		\draw[edge] (v1) --+(-.5,.5);
		\draw[edge] (v2) --+(-.5,-.5);
		\draw[edge] (v5) --+(.5,.5);
		\draw[edge] (v6) --+(.5,-.5);

	\end{tikzpicture}
	\caption{A 5-loop completion (left) has seven vertices. Removing one of the vertices yields a 5-loop decompletion. In this case, only two of the seven possible decompletions are non-isomorphic.}
	\label{fig:completion}
\end{figure}

The automorphism symmetry factor $\frac{1}{\abs{\Aut(G)}}$ is compatible with  completion and decompletion (\cref{def:completion_decompletion}). Let $G$ be a   $L$-loop completion, then 
\begin{align}\label{decompletion_symmetry_factor}
	\frac{L+2}{\abs{\Aut(G)}}=	\sum_{\substack{g \text{ decompletion of } G  }} \frac{1}{\abs{\Aut(g)}}.
\end{align}
The sum  extends over non-isomorphic decompletions $g$ of $G$, and the automorphisms $\Aut(g)$ allow permutations of the external legs. To obtain symmetry factors where external edges are fixed, but different channels are summed over (which is the usual convention in QFT), both sides of \cref{decompletion_symmetry_factor} need to be multiplied by $4!$.

\begin{example}\label{ex:symmetry_factors}
	For the graphs in \cref{fig:completion} from left to right, \cref{decompletion_symmetry_factor} reads $\frac{7}{48}=\frac{1}{16}+\frac{1}{12}$. 
	In the traditional QFT convention, the graph in the middle of \cref{fig:completion} has three non-isomorphic channels, and a symmetry factor (with external legs fixed) of $\frac 12$. Indeed,  $3\cdot \frac 12 = \frac{24}{16}= 4! \cdot \frac{1}{16}$, that is, we can reproduce the QFT convention by multiplying with $4!$. The graph at the right has four channels and a leg-fixed symmetry factor $\frac 12$, and $4 \cdot \frac 12 = 2 = 4!\cdot \frac{1}{12}$.
\end{example}

The operation $g_1\uplus g_2$ of joining along external edges (\cref{def:uplus}) can equivalently be interpreted as insertion of $g_1$ into the completion $G_2=g_2\cup \left \lbrace v_2 \right \rbrace  $ in place of a $(2p)$-valent vertex~$v_2$, denoted by $ g_1 \circ_{v_2} G_2:=g_1 \uplus g_2$. This notation can be extended to insertion of 2-valent subgraphs into edges with the convention that an edge is formally identified with a 2-valent vertex of the graph. 
In \cref{sec:combinatorial_properties}, we show that the circuit partition polynomial $J(G,N)$ factorizes under joining of graphs (\cref{lem:factorization_J}).   Recall from \cref{def:TGN} that    we divide by $ (2p-1)!!$ for each $(2p)$-valent vertex to go from $J(G,N)$ to $T(G,N)$.

\begin{lemma}\label{lem:factorization_T}
	\begin{enumerate}
		\item Let $g_1$ be a $(2p)$-valent graph, $G_1=g_1 \uplus \left \lbrace v \right \rbrace $ its completion (\cref{def:completion_decompletion}),  and  $\pmelon$ the multiedge graph with $(2p-1)$ loops (\cref{lem:multiedge_J}). Then
		\begin{align*}
		T(G_1,N)&=  T\big(\pmelon,N\big)\cdot T(g_1,N) =  \frac {\prod \nolimits_{j=0}^{p-1}(N+2j) } {(2p-1)!!}\cdot T(g_1,N).
		\end{align*} 
		\item Let $g_1$ be a  $(2p)$-valent graph, and let $v_2$ be a $(2p)$-valent vertex in a graph $g_2$ (which may or may not have external edges).  Then $g_1 \circ_{v_2} g_2$  is a sum of graphs which have the same valence as $g_2$, and
		\begin{align*}
			T \big(g_1 \circ_{v_2} g_2,N \big) &= T\big(\pmelon,N\big)\cdot T(g_1,N)\cdot T(g_2\setminus\left \lbrace v_2 \right \rbrace   ,N)= T(g_1,N) \cdot T\left( g_2,N \right).
		\end{align*}
	\end{enumerate}
\end{lemma}

The factorization of $\O(N)$-symmetry factors upon insertion of subgraphs (\cref{lem:factorization_T}) can also be interpreted \enquote{in the opposite direction}: The $(2p-1)$-loop multiedge has vertices of valence $(2p)$, and its circuit partition polynomial has zeros at negative even integers $N\in \left \lbrace 0,-2,\ldots 2p-2 \right \rbrace $ (\cref{lem:multiedge_J}). Namely, the presence of such zeros indicates the presence of either a vertex of valence $(2p)$, or of a sum of $(2p)$-valent subgraphs.
\begin{lemma}\label{lem:circuit_polynomial_zeros}
	\begin{enumerate} 
		\item 	Let $G$ be a completion which contains a vertex of valence $2p$. Then $J(G,N)$ and $T(G,N)$ have zeros at $N\in\left \lbrace 0, -2, \ldots, -2p+2 \right \rbrace  $.
		\item Let $g_1, g_2$ be $(2p)$-valent, and let $g=G\setminus \left \lbrace v \right \rbrace  $ be a sum of vertex-type graphs, where the corresponding sum of completions $G=g_1 \uplus g_2$ (\cref{def:uplus}) arises from joining  $g_1$ and  $g_2$ in all possible orientations (channels). Then $J(g,N)$ and $T(g,N)$ have zeros at $N\in \left \lbrace -4, -6, \ldots, -2p+2 \right \rbrace  $.
	\end{enumerate}

\end{lemma}

\subsection{Renormalization, primitive graphs, and periods}\label{sec:periods}

The renormalization of Feynman graphs has the structure of a Hopf algebra \cite{kreimer_hopf_1998,kreimer_overlapping_1999,connes_renormalization_2001}, where the \emph{primitive} elements are the Feynman integrals without subdivergences, and the coproduct of the Hopf algebra realizes the Zimmermann forest formula \cite{zimmermann_convergence_1969} of extracting and renormalizing subdivergences. 
We remind the reader of the following facts which will be used in the rest of the paper:
\begin{itemize}
\item As $\phi^4$ theory is renormalizable in four dimensions, a graph is superficially divergent if and only if it has four or less external edges. 
\item  A graph is primitive if it is \emph{cyclically 6-edge connected}, that is, if it has no  proper subgraphs of valence less than 6, except for trees.
\item Every propagator-type graph except the 2-loop multiedge (the "sunrise" graph) has a vertex-type subgraph upon removing a vertex adjacent to an external edge. 
\item Consequently, except for the 2-loop multiedge, all primitive graphs of $\phi^4$ theory are vertex-type graphs. 
\item The 2-loop multiedge contains three distinct  1-loop multiedges (\emph{fish}, \cref{fig:fish}) as divergent subgraphs. Contracting any one of theses subgraphs produces a 1-loop tadpole  (case $L=1$ of \cref{eq:bubble-tadpole}). In the massless theory, and in a massive theory with kinematic renormalization conditions, tadpole Feynman integrals vanish. In that case, the 2-loop multiedge is primitive.  
\end{itemize}

The Feynman integral of primitive vertex-type graphs depends  logarithmically on the momentum scale. 
The coefficient of this dependence is the \emph{Feynman period} \cite{broadhurst_knots_1995,schnetz_quantum_2010,brown_periods_2010,balduf_statistics_2023}. One possible representation of the Feynman period in a 4-dimensional theory is  by the parametric Feynman integral of  the first Symanzik polynomial  $\psi_g$ of $g$:
\begin{align}\label{def:period}
	\period (g)&:= \Big( \prod_{e\in E_g} \int \nolimits_0^\infty \d a_e \Big)  \; \delta \Big( 1-\sum_{e\in E_G} a_e \Big) \frac{1} {\psi_g^{ 2}}.
\end{align} 
The period is defined for vertex-type graphs $g$, but the periods of all decompletions of the same completion $G$ (\cref{def:completion_decompletion}) have the same value, denoted by $\period(G):=\period(g)$. 

\subsection{The beta function}\label{sec:beta_function}

\subsubsection{Definition and asymptotics}\label{sec:beta_function_asymptotics}
The beta function describes the running of the effective coupling with the energy scale. It is related to critical exponents   \cite{kleinert_critical_2001,vasilev_field_2004,zinn-justin_quantum_2021}.
For $\phi^4$ theory regularized in $D=4-2\epsilon$ dimensions, it is given by the derivative of the  counterterm $z^{(\alpha)}= z^{(4)}/(z^{(2)})^2$ of the invariant charge, where $z^{(4)}$ is the vertex- and $z^{(2)}$ the propagator counterterm,
\begin{align}\label{def:beta_function}
	\beta(\alpha,\epsilon,N) &:= \frac{-\epsilon}{\partial_\alpha \ln \left( \alpha \cdot z^{(\alpha)}(\alpha,\epsilon,N ) \right) } = \sum_{L\geq 1} (-\alpha)^{L+1} \beta_L(\epsilon,N).
\end{align}
The expansion parameter is $\alpha := -\lambda \mu^{-2\epsilon} (4\pi)^{-2}$, where $\lambda$ is the coupling constant  of our Lagrangian (\cref{lagrangian_phi4}), and $\mu$ is an energy reference scale. In the following, we will only be interested in the beta function at $\epsilon=0$. 
For every fixed order $L$, the coefficient $\beta_L(N):=\beta_L(0,N)$ is a polynomial in~$N$. The  counterterm in \cref{def:beta_function} is determined by 1PI  graphs, the beta function receives contributions from primitive (\cref{sec:periods}) and from non-primitive graphs, $\beta_L(N) = \beta^{\text{prim}}_L(N) + \beta^{\text{nonprim}}_L(N)$. The second summand depends on the renormalization scheme, for a recent review see the first author's PhD thesis \cite{balduf_dyson_2024}. 
The full beta function in the minimal-subtraction scheme is known up to $L=7$ loops \cite{kompaniets_minimally_2017,schnetz_numbers_2018,schnetz_phi4_2023}, and starts with 
\begin{align}\label{beta_MS_coefficients}
 \beta(\alpha,N) &= -2\epsilon \alpha + \frac{N+8}{3}\alpha^2 - \frac{3N+14}{3}\alpha^3 \\
&\qquad + \frac{96(5N+22)\zeta_3 + 33N^2 + 922N + 2960}{216}\alpha^4 + \ldots, \nonumber 
\end{align}

Recall that an \emph{asymptotic  expansion} of a function $f(t)$  around $t=\infty$ is a sequence of continuous functions $f_k(L)$ such that
\begin{align}\label{def:asymptotic_expansion}
f(t) &= \sum_{k=0}^K f_k(t) + R_K(t), \quad \text{where} \quad \lim_{t\rightarrow \infty}  \abs{\frac{R_K(t)}{f_k(t)}}=0 \quad \text{for every fixed $K$}.
\end{align} 
A series expansion is convergent for a  fixed finite $t$ if, in the same setting, $\lim_{K \rightarrow \infty} \abs{R_K(t)}=0$. An   asymptotic series is not necessarily convergent, and therefore there is no guarantee that a truncated asymptotic series produces numerically accurate values for any value of its parameter $t$ other than the expansion point. A central question of the present work is to establish whether asymptotic expansions around loop order $L=\infty$ in $\phi^4$ theory produce numerically useful predictions for finite $L$. If that is the case, we call this range of finite loop orders $L$ an 	\emph{asymptotic regime}. The asymptotic regime depends on the quantity in question, on the required numerical accuracy, and on the number of subleading asymptotic terms to be included, i.e. on the cutoff $K$ in \cref{def:asymptotic_expansion}.

The leading asymptotics of $\beta_L$ for $L\rightarrow \infty$ has been obtained from an instanton computation in \cite{lipatov_divergence_1977,brezin_perturbation_1977a,mckane_nonperturbative_1984,komarova_asymptotic_2001}. 
The review \cite{mckane_perturbation_2019} gives a consolidated result, with slightly different conventions, namely in $D=4-\varepsilon$ dimensions for an interaction term $\frac{g}{4}\phi^4$. There, $\beta=-\varepsilon g + \frac{(N+8)}{8\pi^2}g^2 + \ldots = \sum_k g^k b_k$, and the leading asymptotics of $b_k$ as $k\rightarrow \infty$ is $\bar b_k$,  
\begin{align*}
	\bar b_k &= C_\beta (-1)^k ~ k! ~ k^{\frac {7}2} \left( \frac{3}{8 \pi^2} \right) ^k, \qquad C_\beta = 2^{\frac {13}2} \cdot 3^{\frac{2}{2}} e^{\frac{3 \zeta'(2)}{\pi^2} - \frac 7 2 \gamma_E - \frac {15}4}.
\end{align*}
The $N$ dependence of this asymptotics can be found in \cite{mckane_nonperturbative_1984}.
One can write $e^{\frac{\zeta'(2)}{\pi^2}}= (2\pi e^{\gamma_E})^{\frac{1}{6}} / A^2$ with the Glaisher-Kinkelin constant \cite{kinkelin_ueber_1860} $A\approx 1.28242713$. In our convention (\cref{def:beta_function}),   $\bar b_k$ amounts to the leading asymptotics of $\beta_{k-1}$.
Hence, $\beta_L$ grows asymptotically as $\beta_L(N) \sim \tilde \beta_L (N)\left(1 + \mathcal O \left( \frac 1 L \right) \right)   $, where  
\begin{align}\label{beta_asymptotics}
	\tilde \beta_L &=(L+1)! ~(L+1)^{ \frac{N+6}{2}} \frac{36 \cdot 3^{ \frac{N+1}{2}} }{\pi  \Gamma \! \left( 2+\frac{N}{2} \right) A^{2N+4} } e^{-\frac 32 - \frac{N+8}{3}\left( \frac 3 4 + \gamma_\text{E} \right)  }  .
\end{align}

Factorial asymptotic growth can be expressed in many equivalent forms by redefining subleading coefficients. It turns out that \cref{beta_asymptotics} is not the most suitable form for our purpose. Instead, we will use  the notation of \cite{borinsky_generating_2018,borinsky_renormalized_2017} throughout the article:
\begin{align}\label{beta_asymptotic_ansatz}
	\beta_L &\underset{L \rightarrow \infty}\sim   \sum_{r=0}^\infty \Gamma \left( L + c_s -r \right)a^{-L-c_s+r} c_r \\
& = \left( \frac 1 a \right) ^{L+c_s} \cdot \Gamma \big(L+c_s\big) \left( c_0 + \frac{ac_1}{(L+c_s-1)}+\frac{a^2c_2}{(L+c_s-1)(L+c_s-2)}+ \ldots \right) . \nonumber 
\end{align}
Here, the parameters $a,c_s,\left \lbrace c_j \right \rbrace _{j\in \mathbb N_0} $ are independent of $L$, but possibly functions of $N$.
In order to determine the growth parameters of a factorially growing quantity, it is often useful to consider a growth ratio which has a finite limit as $L\rightarrow\infty$. We use
\begin{align}\label{def:rL}
	r_L(\beta)&:= \frac{\beta_{L+1}}{L\cdot \beta_L}\underset{L \rightarrow \infty}\sim \scalemath{.9}{ \frac 1 a + \frac{ c_s}{a}\frac 1 L - \frac{c_1}{c_0 }\frac{1}{L^2}+ \left( (c_s-1) \frac{c_1}{c_0} + a \left( \frac {c_1}  {c_0} \right)^2 - 2a \frac{c_2}{c_0}\right) \frac{1}{L^3}
	+ \ldots}.
\end{align}
By construction, $r_L(\beta)$ is invariant under rescaling of $\beta$ with an overall factor; $r_L(\beta)$ depends on $a$ and $c_s$ individually, but on the subleading corrections only through the ratios $\frac{c_j}{c_0}$. Notice in particular that knowing the \emph{leading} factorial asymptotics amounts to knowing the limit \emph{and} the $\frac 1 L$ correction of the growth ratio. 

For the asymptotic growth of the beta function in \cref{beta_asymptotics}, a simple computation yields
\begin{align*}
	\frac{\tilde \beta_{L+1}}{L\cdot \tilde \beta_L}&= \frac{(L+2)! (L+2)^{\frac{N+6}{2}} }{L(L+1)! (L+1)^{\frac{N+6}{2}} } = \frac{(L+2)}{L} \left( \frac{L+2}{L+1} \right) ^{\frac{N+6}2}=\left( 1+\frac 2 L \right)  \left( 1+ \frac{1}{L+1} \right) ^{\frac{N+6}{2}} .
\end{align*}
Expanding the last factor in orders of $\frac 1 L$ and comparing with \cref{def:rL}, we can read off the growth parameters that correspond to \cref{beta_asymptotics}: The beta function in minimal subtraction is expected to grow according to \cref{beta_asymptotic_ansatz}, with
\begin{align}\label{beta_asymptotics_predicted}
	a &= 1, \qquad 
    c_s = \frac {N+10}2, \qquad 
    c_0 = \frac{ 36 \cdot 3^{ \frac{N+1}{2}} e^{-\frac 3 2 - \frac{N+8}{3}\left( \frac 3 4 + \gamma_E \right) } }{\pi \Gamma(2+\frac N2) A^{2N +4} }.
\end{align}
It is important to realize that \cref{beta_asymptotics} only specifies the \emph{leading} asymptotic growth, and \cref{beta_asymptotics_predicted} represents the same leading asymptotic growth, but they differ in terms $\propto \frac 1 L$. For example, for $N=1$ at $L=13$, \cref{beta_asymptotics_predicted} is larger than \cref{beta_asymptotics} by approximately $65\%$. Without information about subleading coefficients in \cref{beta_asymptotics}, it is impossible to know which of the two conventions is numerically closer to the true value at finite $L$.

\subsubsection{Primitive beta function}\label{sec:primitive_beta_function}

The \emph{primitive} beta function $\beta^\text{prim}$ consists of the contributions of primitive vertex-type graphs~$g$ only, which is given by their Feynman periods $\period(g)$ (\cref{def:period}). Owing to the completion invariance of the period and \cref{lem:factorization_T,decompletion_symmetry_factor},  $\beta^\text{prim}$  can equivalently be written as a sum over completions:
\begin{align}\label{betaprim_expansion}
\beta^\text{prim}(N)&= 2\sum_{g \text{ decomp.}} \frac{4!~T(g,N) }{\abs{\Aut(g)}}\cdot \period (g)
= 2\sum_{G \text{ comp.}}  \frac{ 4!(L_G+2)}{\abs{\Aut(G)}} \frac{3T(G,N)}{N(N+2)} \period (G).
\end{align}

$\beta^\text{prim}_L$ is independent of the renormalization scheme, while the full $\beta_L$ is not. It has been conjectured \cite{mckane_perturbation_2019} that the leading asymptotic growth of $\beta^\text{prim}_L$ coincides, for $L\rightarrow\infty$, with that of $\beta_L$ in minimal subtraction:
\begin{conjecture}\label{beta_conjecture}
	The leading growth of the \emph{primitive} beta function is given by \cref{beta_asymptotics},
	\begin{align*}
		\beta^\textnormal{prim}_L &\sim \bar \beta_L \scalemath{.8}{\left( 1+ \mathcal{O}\left( \frac 1 L \right)   \right) } , \quad \bar \beta_L:= \Gamma \!\left( L + \frac{N+10}{2} \right)   \frac{36 \cdot 3^{ \frac{N+1}{2}} }{\pi  \Gamma \! \left( 2+\frac{N}{2} \right) A^{2N+4} } e^{-\frac 32 - \frac{N+8}{3}\left( \frac 3 4 + \gamma_\text{E} \right)  }   .
	\end{align*} 
\end{conjecture}

By \cref{betaprim_expansion}, the conjecture would imply a certain leading growth of the \emph{average} period. 
Let $p_L(N)=\sum_G \frac{T(G,N)}{\abs{\Aut(G)}}$, where the sum is over $L$-loop primitives $G$, see \cref{primitive_generating_function} below. Equivalently, $\beta^\text{prim}_L \big|_{\period \mapsto 1}=2p_L$.
Dividing $\beta^\text{prim}_L(N)$ by $2p_L(N)$, we obtain the average period per graph:
\begin{align}\label{def:average_period}
	\left \langle \period \right \rangle_{\frac{T}{\Aut},L}(N) &:= \frac{\sum_{G ~\text{compl}} \frac{ T(G,N) \period(G)}{\abs{\Aut(G)} \, }}{ \sum_{G ~\text{compl}} \frac{ T(G,N)} {\abs{\Aut(G)} \, } } = \frac{\beta^\text{prim}_L (N)}{2 \, p_L(N)}.
\end{align}
Notice that $\left \langle \period \right \rangle _{T/\Aut,L}$ is the average of a sample weighted proportional to $\frac{T}{\Aut}$, which is not exactly the same as the average of $\frac{T\cdot \period}{\Aut}$ in a uniform sample, see the appendix of~\cite{balduf_statistics_2023}. 
In \cref{sec:0dim_primitive_asymptotics} below, we compute the $L\rightarrow \infty$ asymptotics of $p_L(N)$ (\cref{primitive_asymptotics}).
Consequently,  \cref{beta_conjecture} amounts to an asymptotics of $\left \langle \period \right \rangle _{T/\Aut}$ according to
\begin{align}\label{mean_period_asymptotics}
\left \langle \period \right \rangle_{\frac{T}{\Aut},L} &\sim \scalemath{.9}{\left( \frac 3 2 \right) ^L L^{\frac 5 2}  \cdot  \delta^N \cdot\frac{ 9 \left( \frac 3 2 \right) ^{\frac{1}{2}} e^{-\frac 12-\frac {8\gamma_E  }3 }} {A^{4}\sqrt{2\pi} }    \left( 1 +  \mathcal O \left( \frac{1}{L} \right)  \right) },\qquad  \delta = \frac{\sqrt{\frac 3 2} e^{\frac 12 - \frac {\gamma_E}3}}{A^2}.
\end{align}
Notice\footnote{Erik Panzer drew our attention to the numerical value of this quantity.} that $\delta \approx 1.013$ such that the $N$ dependence of the leading asymptotics is relatively gentle. $\left \langle \period \right \rangle_{\frac{T}{\Aut},L}$ grows  exponentially, not factorially, with $L$.  To investigate its growth, we consider a ratio similar to \cref{def:rL}, but without the denominator $L$, namely
\begin{align}\label{def:fL}
f_L\big( \left \langle \period \right \rangle _{\frac{T}{\Aut}} \big)  &:= \frac{\left \langle \period \right \rangle_{\frac{T}{\Aut},L+1} (N) }{\left \langle \period \right \rangle_{\frac{T}{\Aut},L} (N) } = \frac 3 2 + \frac{15}{4} \frac 1 L + \mathcal O \left( \frac 1 {L^2} \right) . 
\end{align} 
Thereby, if \cref{beta_conjecture} is true, the first two coefficients of $f_L$ in \cref{def:fL} should be independent of $N$. 
In the following sections, we will first examine $p_L(N)$ and related quantities in detail from the perspective of 0-dimensional $\phi^4$ theory. 
We return to numerical values of $\beta^\text{prim}_L$ in the 4-dimensional theory in \cref{sec:4d_asymptotics}.

\subsection{Martin invariant}\label{sec:Martin_invariant}

The \emph{Martin polynomial} \cite{martin_enumerations_1977} is a simple transformation of the circuit partition polynomial (\cref{def:circuit_partition_polynomial}), $M(G,N) := \frac{J(G,N-2)}{N-2}$. Using \cref{lem:factorization_T}, the Martin polynomial can be related to the $\O(N)$ symmetry factor of the decompletion $g=G\setminus \left \lbrace v \right \rbrace $:
\begin{align}\label{def:Martin_polynomial}
	T \big( g , N \big) &=\frac{1}{3^{\abs{V_G}-1}} \frac{ M(G,N+2)}{N+2}, \quad \text{where $G$ is a primitive completion}.
\end{align}
The linear coefficient of the Martin polynomial is  the (first) \emph{Martin invariant} \cite{panzer_feynman_2025},
\begin{align}\label{def:Martin_invariant}
M^{[1]} &:= \frac 1 6 \frac{\partial}{\partial N}M(G,N) \Big|_{N=0}=\frac{3^L}2 \frac{\partial}{\partial N} \Big( N\cdot T \left( g , N-2 \right)  \Big) \Big|_{N=0} = \frac{3^L}{2}T \left( g,-2 \right) .
\end{align}
	At a fixed loop order, we use \cref{decompletion_symmetry_factor} to relate the sum of all Martin invariants, weighted by symmetry factor, to the sum of all primitive decompletions: 
\begin{align}\label{Martin_sum}
M^{[1]}_L:=	\sum_{\textnormal{completions }G} \frac{(L+2)~ M^{[1]}(G)}{\abs{\Aut(G)}} &= \frac{3^L}2 \sum_{\textnormal{decomp. }g} \frac{T \left( g , -2 \right) }{\abs{\Aut(g)}} .
\end{align}

\begin{example}\label{ex:Martin_invariants}
	At $L=1$ loop, the unique decompletion $g$ (the fish graph from \cref{fig:fish}) has $T(g ,N)=\frac{N+8}9$. 
    Thus, its completion has the Martin polynomial 
	\begin{align*}
	M(G,N)= 3^2 \cdot N\cdot T(G\setminus \left \lbrace v \right \rbrace , N-2)= N(N+6),
	\end{align*}
	 and the Martin invariant is $M^{[1]}=1$. The completion has symmetry factor $\abs{\Aut(G)}=3! \cdot 2^3=48$, and the decompletion (the fish) has $\abs{\Aut(g)}=2^4=16$. The sums in \cref{Martin_sum} at $L=1$ contain only one term each, namely
	\begin{align*}
	\frac{(1+2) \cdot 1}{48}&= \frac{3^1}{2} \Big( \frac{N+8}{9\cdot 16}\Big)\Big|_{N=-2}=\frac{1}{16}.
	\end{align*}
	For $L\in \left \lbrace 1, \ldots, 6 \right \rbrace $, explicit computation yields the sums of Martin invariants:
	\begin{align*}
	M^{[1]}_L\in \left \lbrace \frac{1}{16}, \quad 0, \quad \frac 14, \quad \frac 74, \quad \frac{89}4, \quad \frac{1255}{4} \right \rbrace  .
	\end{align*} 

\end{example}

\section{QFT in zero dimensions}\label{sec:0dim}
In \cref{sec:setup}, we have introduced the $\O(N)$~symmetry factors $T(G,N)$ in a \enquote{bottom up} fashion, from their definition as sums over circuits (\cref{def:circuit_partition_polynomial}) and their behaviour upon inserting and joining graphs. 

A complimentary perspective on the $N$ dependence of the coefficients of the perturbation series, which is \enquote{top down} in the sense that it only considers the sum of all graphs in question, is provided by \emph{QFT in zero dimensions}. This term is motivated from formally taking the limit $D \rightarrow 0$ in the path integral, which eliminates the kinetic term and replaces each Feynman integral by a constant independent of kinematics, so that the Green functions of 0-dimensional QFT become the generating functions of the corresponding graphs \cite{argyres_zerodimensional_2001,bessis_quantum_1980}. Hence, despite the name, 0-dimensional QFT is not a field theory in a physical sense.   For a recent overview with emphasis on enumeration of graphs, we refer to \cite{borinsky_graphs_2018}.

\subsection{Vacuum graphs}\label{sec:0dim_vacuum}
The vacuum path integral of $N$-dependent 0-dimensional $\phi^4$ theory is the $N$-fold ordinary integral (understood through its formal power series expansion in $\hbar$. Analytic properties for $\hbar<0$ as function of $N$ have been studied in \cite{benedetti_smalln_2024})
\begin{align}\label{Z0_definition}
	Z(\hbar)&:=\int_{\R^N} \frac{\d^N \vec \phi}{(2\pi \hbar)^{\frac N2}}\; e^{\frac{1}{\hbar}\left( -\frac{\vec \phi ^2}{2}+ \frac{(\vec \phi^2)^2}{4!} \right) }.
\end{align}
For fixed integer $N$, a power series expansion of this integral can in principle be computed by expanding all dot products $\vec \phi^2=\vec \phi \cdot \vec \phi =\sum_{j=1}^N \phi_j^2$ in terms of their components, and repeatedly using the Gaussian integral formula
\begin{align}\label{gaussian_integral}
	\int_\R \frac{\d x}{\sqrt{2\pi \hbar } }  e^{ - \frac{x ^2}{2\hbar}  } x^{2n}&= (2n-1)!! \, \hbar^n .
\end{align}
However, it is more convenient to make use of the identity 
\begin{align}\label{Hubbard_Stratonovich_transformation}
	\int_\R \frac{\d \sigma}{\sqrt{2\pi a}} \; e^{\frac 1 a \left(-\frac{\sigma ^2}{2} + b \sigma \right)} &= e^{\frac{ b^2}{2a}}.
\end{align}
Choosing $b=\vec \phi^2/\sqrt{12}$ and $a=\hbar$, this transformation eliminates the $(\vec \phi^2)^2$ interaction term from \cref{Z0_definition}, and the remaining integral over $\d^N \vec \phi$ is Gaussian. \Cref{Hubbard_Stratonovich_transformation} is called \emph{Hubbard-Stratonovich transformation} \cite{stratonovich_method_1958,hubbard_calculation_1959,byczuk_generalized_2023}, and it can be interpreted physically as introducing an auxiliary \enquote{mean field} $\sigma=\vec \phi^2$, which turns the Lagrangian of $\phi^4$ theory into that of a $\sigma$-model \cite{novikov_twodimensional_1984} (see also e.g. \cite{gracey_progress_1997,moshe_quantum_2003,gracey_large_2018}). We discuss the perturbative expansion of $\sigma$ in \cref{sec:dual}. For an arbitrary constant $c$, \cref{Hubbard_Stratonovich_transformation}  yields
\begin{align}\label{phi4_gaussian}
	\int_{\R^N} \frac{\d^N \vec \phi}{(2\pi \hbar)^{\frac N 2}} \; e^{\frac{1}{\hbar}\left( -\frac{\vec \phi^2}{2c}+ \frac{(\vec \phi^2)^2}{4!} \right) }	
  &=\int_\R \frac{\d \sigma}{\sqrt{2\pi \hbar}} \frac{1}{\left( \frac 1 c-\frac{\sigma}{\sqrt 3} \right) ^{\frac N 2}} e^{-\frac{\sigma^2}{2\hbar}}.
\end{align}
To expand this as a power series, use \cref{gaussian_integral} and the binomial series
\begin{align}\label{geometric_series}
	\left( 1 -t \right) ^{-\frac N2} 
  &= \sum_{n=0}^\infty \binom{-\frac N2}{n}(-t)^n
  = \sum_{n=0}^\infty \binom{\frac N2 + n-1}{n}t^n
  = \sum_{n=0}^\infty \frac{ \Gamma\left(n+\frac N 2\right)}{\Gamma(n+1) \Gamma \left( \frac N 2 \right) } t ^n.
\end{align}
The vacuum path integral \cref{Z0_definition} amounts to $c=1$ and becomes
\begin{align}\label{Z0_series} 
	Z(\hbar)
	&= \sum_{n=0}^\infty \frac{ \Gamma\left(2n+\frac N 2\right)  } { \Gamma \left( \frac N 2 \right) \Gamma(n+1) 6^n } \hbar^n= \sum_{n=0}^\infty \frac{ \hbar^n } {  \Gamma(n+1) 24^n } \prod_{k=0}^{2n-1} (N+2k).
\end{align}
Note that the $N$ dependence in \cref{Z0_series} is extremely simple. In fact, this finding is the \enquote{analogue} of \cref{lem:factorization_T}: $Z(\hbar)$ enumerates the sum of \emph{all} vacuum graphs of a given loop order. Those graphs can be viewed as arising from gluing subgraphs in all possible ways, where each gluing produces factors of the form $(N+2j)$ for $j\in \mathbb N_0$.

\subsection{External edges, connected and 1PI graphs}\label{sec:0dim_external}

To enumerate graphs with external edges, we need a source term. We will only be interested in $\O(N)$-invariant observables
where two external edges share the same index and the index is being summed over. To this end, it is convenient to introduce a (scalar) source $\eta$ for~$ \vec \phi^2$ (instead of one for each individual component $\phi_k$, $k\in \left \lbrace 1, \ldots, N \right \rbrace $):
\begin{align}\label{Z_eta_definition}
	Z(\hbar, \eta) &:= \int_{-\infty}^\infty \frac{\d^N \vec \phi}{(2\pi \hbar)^{\frac N2}} e^{\frac{1}{\hbar}\left( -\frac{\vec \phi^2}{2}+ \frac{(\vec \phi^2)^2}{4!} +\frac{\eta}{\hbar} \vec \phi^2 \right) }
	= \int_{-\infty}^\infty \frac{\d^N \vec \phi}{(2\pi \hbar)^{\frac N2}} e^{\frac{1}{\hbar}\left( -\frac{\vec \phi^2}{2}\left( 1-2\frac \eta \hbar \right) + \frac{(\vec \phi^2)^2}{4!} \right) }.
\end{align}
We have included an extra factor $\hbar^{-1}$ to be consistent with the ordinary power counting, namely the power of $\hbar$ in the series expansion of $Z(\hbar, \eta)$ counts (vertices minus edges), and a 1-particle source term is one source vertex, the $\eta$ source term must therefore count like two vertices and have an overall factor $\hbar^{-2}$ in the exponent. 
Of course, having a source term for $\vec \phi^2$ is equivalent to marking exactly one of the propgators in a graph. To solve the integral in \cref{Z_eta_definition}, we apply the Hubbard-Stratonovich transformation \cref{Hubbard_Stratonovich_transformation}, use \cref{phi4_gaussian} with $c=\frac{1}{1-2\frac{\eta}{\hbar}}$, and finally use \cref{geometric_series}:
\begin{align}\label{Z_eta_series}
	Z(\hbar, \eta)&= \int_\R \frac{\d \sigma}{\sqrt{2\pi \hbar}} \frac{1}{\left( 1-2\frac \eta \hbar -\frac{\sigma}{\sqrt 3} \right) ^{\frac N 2}} e^{-\frac{\sigma^2}{2\hbar}}= \sum_{n=0}^\infty \sum_{k=0}^\infty \frac{ \Gamma\left(k+2n+\frac N 2\right) 2^k }{ \Gamma\left( \frac N2 \right) \Gamma(n+1)\Gamma(k+1)  6^n} \eta^k \hbar^{n-k}.
\end{align}
By construction, the $\eta$-derivatives of the generating function $Z(\hbar, \eta)$ in \cref{Z_eta_series} enumerate Green functions with pairs of external edges, whose indices are summed. 
This is a common setup in field theory, compare for example the so-called \emph{generalized effective actions} \cite{cornwall_effective_1974,berges_nparticle_2004,carrington_techniques_2011}. 
However, for the following steps it is useful if we can directly compare to the examination of $\phi^4$ theory (without $\O(N)$ symmetry) in \cite{borinsky_renormalized_2017}, therefore we   prefer a source term $j$ that generates a \emph{single} external line. This is possible by replacing $\eta^k$ in \cref{Z_eta_series}, but one needs to carefully adjust combinatorial factors since now a $2k$\textsuperscript{th} derivative with respect to $j$ replaces a $k$\textsuperscript{th} derivative with respect to $\eta$. Furthermore, we divide by the $N$ dependence at $n=0$ to cancel summation over vector indices of external edges that was implied in the $\eta$-formalism.
All in all, the required replacement is 
\begin{align}\label{eta_replacement}
	\eta^k &\mapsto \frac{\Gamma \left( \frac N 2 \right) \Gamma(k+\frac 12)}{\Gamma \left( \frac 12 \right) \Gamma\left( k+\frac N 2 \right) }\frac{\Gamma(k+1)}{\Gamma(2k+1)}j^{2k} .
\end{align}
We remark that this is a purely formal replacement, $j$ is not a \enquote{source term} for any field that exists in this theory. Conceptually, $j$ would be a source term for $\sqrt{  \phi^2}$, which is not the same as a (vector-valued) source $\vec j$ for the field $\vec \phi$.  We redefine the summation index $n$ to expose the dependence on $\hbar$, the resulting transformation of \cref{Z_eta_series} is 
\begin{align}\label{Z_j_series}
	Z(\hbar, j) &=  \sum_{k=0}^\infty\sum_{n=-k}^\infty  \frac{  \left( k+\frac N 2 \right) _{2n+2k} \left( \frac 12 \right) _k }{  6^n 3^k ~  (n+k)! (2k)! } j^{2k} \hbar^{n} \, ,
\end{align}
where we have used Pochhammer symbols $(x)_k = x (x+1) ... (x+k-1)= \Gamma(x+k)/ \Gamma(x)$ to illustrate that the dependence on $N$ is a polynomial for finite $n$ and $k$. 

Now, the coefficient of $Z$ at $j^{2k}$ is a formal power series that enumerates graphs with exactly $2k$ external edges, for example
\begin{align}\label{Z_derivatives}
	\scalemath{.9}{\partial^0_j Z(\hbar, j)\Big|_{j=0}} &\scalemath{.9}{= Z(\hbar, 0)=1 + \frac{N(N+2)}{24}\hbar + \frac{N(N+2)(N+4)(N+6)}{1152} \hbar^2 + \ldots } \\
	\scalemath{.9}{ \partial^2_j Z(\hbar, j)\Big|_{j=0} }& \scalemath{.9}{=\frac 1 \hbar + \frac{ (N+2)(N+4)}{24} + \frac{ (N+2)(N+4)(N+6)(N+8)}{1152}\hbar + \ldots} \nonumber \\
	\scalemath{.9}{\partial^4_j Z(\hbar, j)\Big|_{j=0}} & \scalemath{.9}{= \frac{3}{\hbar^2}+ \frac{ (N+4)(N+6)}{8 \hbar } } \scalemath{.9}{+ \frac{ (N+4)(N+6)(N+8)(N+10)}{384} + \ldots.} \nonumber
\end{align}
Here, the power of $\hbar$ counts (edges minus vertices), where each external edge counts as a 1-valent $\phi$-vertex.

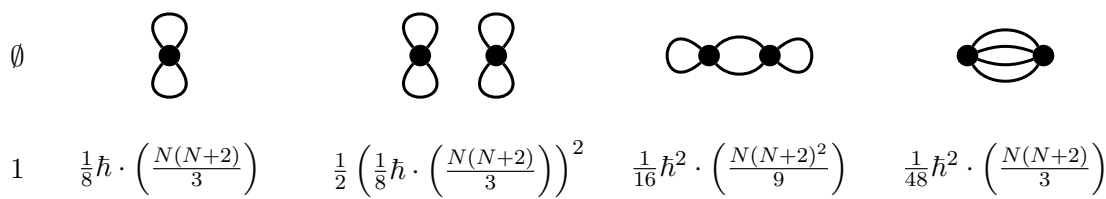
\begin{figure}[htb]
	\begin{tikzpicture}
		\coordinate(x0) at (0,0);
		\node at (x0) {$\emptyset$};
		\node at ($(x0) +(0,-1.5)$){$1$};
		
		\coordinate(x0) at (2,0);
		\node[vertex,fill=black] (v1) at ($(x0) +(0,0)$){};
		\draw[edge ] (v1) ..controls +(-.6,.7) and +(.6,.7) ..(v1);
		\draw[edge ] (v1) ..controls +(-.6,-.7) and +(.6,-.7) .. (v1);
		\node at ($(x0) +(0,-1.5)$){$\frac 1 8 \hbar \cdot \left( \frac{N(N+2)}{3} \right) $};

		\coordinate(x0) at (5.8,0);
		\node[vertex,fill=black] (v1) at ($(x0) +(-.5,0)$){};
		\draw[edge ] (v1) ..controls +(-.6,.7) and +(.6,.7) ..(v1);
		\draw[edge ] (v1) ..controls +(-.6,-.7) and +(.6,-.7) .. (v1);
		\node[vertex,fill=black] (v2) at ($(x0) +(.5,0)$){};
		\draw[edge ] (v2) ..controls +(-.6,.7) and +(.6,.7) ..(v2);
		\draw[edge ] (v2) ..controls +(-.6,-.7) and +(.6,-.7) .. (v2);
		\node at ($(x0) +(0,-1.5)$){$\frac 12 \left( \frac 1 8 \hbar \cdot \left( \frac{N(N+2)}{3} \right) \right) ^2 $};
		
		\coordinate(x0) at (9.5,0);
		\node[vertex,fill=black] (v1) at ($(x0) +(-.4,0)$){};
		\draw[edge ] (v1) ..controls +(-.7,.6) and +(-.7,-.6) ..(v1);
		\node[vertex,fill=black] (v2) at ($(x0) +(.4,0)$){}; 
		\draw[edge, bend angle =50,bend right] (v1) to (v2);
		\draw[edge, bend angle =50,bend left] (v1) to (v2);
		\draw[edge ] (v2) ..controls +(.7,-.6) and +(.7,.6) .. (v2);
		\node at ($(x0) +(0,-1.5)$){$\frac{1}{16}\hbar^2 \cdot \left( \frac{N(N+2)^2}{9} \right) $};
		
		\coordinate(x0) at (13,0);
		\node[vertex,fill=black] (v1) at ($(x0) +(-.5,0)$){};
		\node[vertex,fill=black] (v2) at ($(x0) +(.5,0)$){}; 
		\draw[edge, bend angle =60,bend right] (v1) to (v2);
		\draw[edge, bend angle =20,bend left] (v1) to (v2);
		\draw[edge, bend angle =20,bend right] (v1) to (v2);
		\draw[edge, bend angle =60,bend left] (v1) to (v2);
		\node at ($(x0) +(0,-1.5)$){$\frac{1}{48}\hbar^2 \cdot \left( \frac{N(N+2)}{3} \right) $};

	\end{tikzpicture}
	\caption{Feynman graphs representing the first three summands of $Z(\hbar, 0)$. This partition function enumerates vacuum graphs. The first line shows the graph, where the solid dot is a vertex to be decomposed according to \cref{fig:vertex_decomposition}. The second line is the automorphism symmetry factor and the power of $\hbar$, multiplied by the circuit partition polynomial. }
	\label{fig:Z_graphs_vaccum}
\end{figure}

The derivation of $Z(\hbar, \eta)$ did not make explicit reference to circuit partition polynomials (\cref{def:circuit_partition_polynomial}), but the resulting series is consistent with the Feynman graph expansion. We give the correspondence for the first few terms of the series expansion. The vacuum partition function $Z(\hbar, 0)$ is a sum of vacuum Feynman graphs as shown in \cref{fig:Z_graphs_vaccum}, the polynomials $J(G,N)$ for some of the graphs had been computed in \cref{ex:fish_chain}. 

\begin{example}
	
	The sum of the terms at order $\hbar^2$ in \cref{fig:Z_graphs_vaccum} is
	\begin{align*}
		\scalemath{.9}{\frac 12 \left( \frac 1 8  \left( \frac{N(N+2)}{3} \right) \right) ^2 +\frac{1}{16} \left( \frac{N(N+2)^2}{9} \right) +\frac{1}{48} \left( \frac{N(N+2)}{3} \right) =\frac{N(N+2)(N+4)(N+6)}{1152}},
	\end{align*}
	which coincides with the first line of \cref{Z_derivatives} as claimed. This short calculation already illustrates that the circuit partition polynomials of individual graphs will often appear much more \enquote{random} than the $\O(N)$-dependence of the sum of all graphs, as remarked below \cref{Z0_series}.
\end{example}

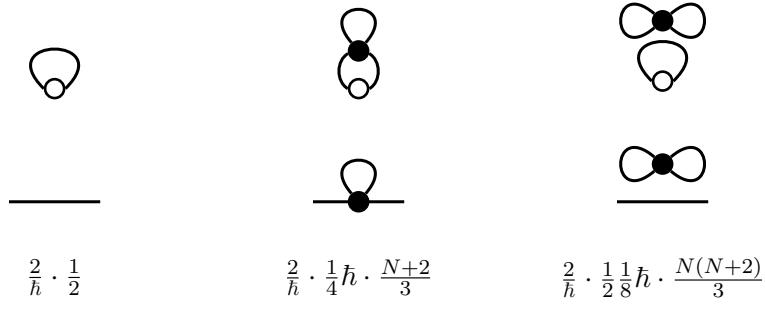
\begin{figure}[htb]
	\centering
	\begin{tikzpicture}
		
		\coordinate(x0) at (2,0);
		\node[vertex,fill=white] (v1) at ($(x0) +(0,-.5)$){};
		\draw[edge] (v1.west) ..controls +(-.7,.7) and +(.7,.7) ..(v1.east);
		
		\node (v1) at ($(x0) +(0,-2)$){};
		\draw[edge ] ($(v1) +(-.6,0)$) -- ($(v1) +(.6,0)$);
		
		\node at ($(x0) +(0,-3)$){$\frac 2 \hbar \cdot \frac 12  $};

		\coordinate(x0) at (6,0);
		\node[vertex,fill=black] (v1) at ($(x0) +(0,0)$){};
		\node[vertex,fill=white] (v2) at ($(x0) +(0,-.5)$){};
		\draw[edge, bend angle =60, bend left ] (v1) to (v2.east);
		\draw[edge, bend angle=60,bend right ] (v1)to (v2.west);
		\draw[edge ] (v1) ..controls +(-.6,.7) and +(.6,.7) .. (v1);

		\node [vertex](v1) at ($(x0) +(0,-2)$){};
		\draw[edge ] ($(v1) +(-.6,0)$) -- ($(v1) +(.6,0)$);
		\draw[edge ] (v1) ..controls +(-.6,.7) and +(.6,.7) .. (v1);
		
		\node at ($(x0) +(0,-3)$){$\frac 2 \hbar \cdot \frac 14 \hbar \cdot \frac{ N+2 }{3} $};
		
		\coordinate(x0) at (10,0);
		\node[vertex,fill=black] (v1) at ($(x0) +(0,.4)$){};
		\draw[edge ] (v1) ..controls +(-.7,.6) and +(-.7,-.6) ..(v1);
		\draw[edge ] (v1) ..controls +(.7,-.6) and +(.7,.6) .. (v1);
		\node[vertex,fill=white] (v2) at ($(x0) +( 0,-.4)$){}; 
		\draw[edge] (v2.west) ..controls +(-.7,.7) and +(.7,.7) ..(v2.east);
		
		\node[vertex,fill=black] (v1) at ($(x0) +(0,-1.5)$){};
		\draw[edge ] (v1) ..controls +(-.7,.6) and +(-.7,-.6) ..(v1);
		\draw[edge ] (v1) ..controls +(.7,-.6) and +(.7,.6) .. (v1);
		\node (v2) at ($(x0) +( 0,-2)$){}; 
		\draw[edge ] ($(v2) +(-.6,0)$) -- ($(v2) +(.6,0)$); 
		
		\node at ($(x0) +(0,-3)$){$ \frac 2 \hbar \cdot \frac{1}{2} \frac{1}{8}\hbar \cdot \frac{N(N+2)}{3} $};

	\end{tikzpicture}
	\caption{Feynman graphs representing the first two summands of $\partial^2_j Z(\hbar, j)$. These are vacuum graphs with one marked edge, indicated with a white 2-valent vertex in the first row. Equivalently, upon cutting the white vertex, one obtains graphs with exactly two external edges, shown in the second row. The replacement \cref{eta_replacement} guarantees that there is no sum over the $\O(N)$ index $j$ of the external $\phi_j$. }
	\label{fig:Z_graphs_2legs}
\end{figure}

The power series $\partial^2_jZ \big|_{j=0}$ enumerates graphs with two external $\phi$-edges, or equivalently vacuum graphs with one marked internal edge. In the latter perspective, the marking amounts to a 2-valent vertex with Feynman rule $\frac 2 \hbar$, and cycles containing 2-valent vertices do not give rise to a factor $N$. Note that this 2-valent vertex influences the automorphism symmetry factor. Of course, the outcome in both perspectives is equivalent. \Cref{fig:Z_graphs_2legs} shows the first terms of $\partial^2_j Z(\hbar, j)\big|_{j=0}$.

\begin{example}\label{ex:Z_second_derivative}
	The sum at order $\hbar^0$ in \cref{fig:Z_graphs_2legs} reproduces the term in the second line of \cref{Z_derivatives},
	\begin{align}
		\frac{2}{\hbar} \cdot \frac 1 4 \hbar \frac{(N+2)}{3} + \frac{2}{\hbar} \frac 1{16}\hbar \cdot \frac{N(N+2)}{3}&= \frac{(N+2)(N+4)}{24}.
	\end{align}
\end{example} 
Analogously, the higher derivatives $\partial^{2k}_j Z(\hbar, j)$ enumerate graphs with $2k$ external edges, or with $k$ 2-valent vertices. To be consistent with usual definitions, these 2-valent vertices may not be exchanged by automorphisms, which can be realized by giving an extra factor $k!$ to the graph in question. \Cref{fig:Z_graphs_4legs} shows the first graphs of $\partial^4_j Z(\hbar, j)$. 

\begin{figure}[htb]
	\centering
	\begin{tikzpicture}
		
		\coordinate(x0) at (0,0);
		\node[vertex,fill=white] (v1) at ($(x0) +(-.4,-.3)$){};
		\node[vertex,fill=white] (v2) at ($(x0) +(.4,-.3)$){};
		\draw[edge] (v1.north) ..controls +(.1,.3) and +(-.1,.3) ..(v2.north);
		\draw[edge] (v1.south) ..controls +(.1,-.3) and +(-.1,-.3) ..(v2.south);

		\node (v1) at ($(x0) +(-.6,-1.8)$){};
		\node (v2) at ($(x0) +(-.6,-2.2)$){};
		\draw[edge ] ($(v1) +(-.4,0)$) -- ($(v1) +(.4,0)$);
		\draw[edge ] ($(v2) +(-.4,0)$) -- ($(v2) +(.4,0)$);
		
		\node (v1) at ($(x0) +(.6,-2)$){};
		\node (v2) at ($(x0) +(.6,-2)$){};
		\draw[edge ] ($(v1) +(-.4,.2)$) -- ($(v1) +(.4,-.2)$);
		\draw[edge ] ($(v2) +(-.4,-.2)$) -- ($(v2) +(.4,.2)$);
		
		\node at ($(x0) +(0,-3)$){$2\frac 4 {\hbar^2} \cdot \frac 14 $};

		\coordinate(x0) at (3,0);
		\node[vertex,fill=white] (v1) at ($(x0) +(-.4,-.3)$){};
		\node[vertex,fill=white] (v2) at ($(x0) +(.4,-.3)$){};
		\draw[edge] (v1.west) ..controls +(-.7,.7) and +(.7,.7) ..(v1.east);
		\draw[edge] (v2.west) ..controls +(-.7,.7) and +(.7,.7) ..(v2.east);

		\node (v1) at ($(x0) +(-.3,-2)$){};
		\node (v2) at ($(x0) +(.3,-2)$){};
		\draw[edge, bend angle=80,bend left ] ($(v1) +(0,.3)$) to ($(v1) +(0,-.3)$);
		\draw[edge, bend angle=80,bend right ] ($(v2) +(0,.3)$) to ($(v2) +(0,-.3)$);
		
		\node at ($(x0) +(0,-3)$){$2\frac 4 {\hbar^2} \cdot \frac 12\left( \frac 12 \right) ^2 $};

		\coordinate(x0) at (8,0);
		\node[vertex,fill=black] (v1) at ($(x0) +(0,0)$){};
		\node[vertex,fill=white] (v2) at ($(x0) +(-.2,-.5)$){};
		\node[vertex,fill=white] (v3) at ($(x0) +(.2,-.5)$){};
		\draw[edge, bend angle =60, bend left ] (v1) to (v3.east);
		\draw[edge, bend angle =10, bend right ] (v2) to (v3);
		\draw[edge, bend angle=60,bend right ] (v1)to (v2.west);
		\draw[edge ] (v1) ..controls +(-.6,.7) and +(.6,.7) .. (v1);

		\node [vertex](v1) at ($(x0) +(0,-1.9)$){};
		\node (v2) at ($(x0) +(0,-2.2)$){};
		\draw[edge ] ($(v1) +(-.6,0)$) -- ($(v1) +(.6,0)$);
		\draw[edge ] (v1) ..controls +(-.6,.7) and +(.6,.7) .. (v1);
		\draw[edge ] ($(v2) +(-.6,0)$) -- ($(v2) +(.6,0)$);
		
		\node at ($(x0) +(0,-3)$){$2\frac 4 {\hbar^2} \cdot \frac 14  \cdot \frac{ N+2 }{3}\hbar $};

		\coordinate(x0) at (12,0);
		
		\node[vertex,fill=black] (v1) at ($(x0) +(-.4,0)$){};
		\node[vertex,fill=white] (v2) at ($(x0) +(-.4,-.5)$){};
		\node[vertex,fill=white] (v3) at ($(x0) +(.4,-.3)$){};
		\draw[edge, bend angle =60, bend left ] (v1) to (v2.east);
		\draw[edge, bend angle=60,bend right ] (v1)to (v2.west);
		\draw[edge ] (v1) ..controls +(-.6,.7) and +(.6,.7) .. (v1);
		\draw[edge] (v3.west) ..controls +(-.7,.7) and +(.7,.7) ..(v3.east);

		\node [vertex](v1) at ($(x0) +(0,-1.9)$){};
		\node (v2) at ($(x0) +(0,-2.2)$){};
		\draw[edge ] ($(v1) +(-.6,0)$) -- ($(v1) +(.6,0)$);
		\draw[edge ] (v1) ..controls +(-.6,.7) and +(.6,.7) .. (v1);
		\draw[edge ] ($(v2) +(-.6,0)$) -- ($(v2) +(.6,0)$);
		
		\node at ($(x0) +(0,-3)$){$2\frac 4 {\hbar^2} \cdot \frac 12  \cdot\frac 14 \frac{ N+2 }{3}\hbar $};

		\coordinate(x0) at (1,-6);
		\node[vertex,fill=black] (v1) at ($(x0) +(0,.4)$){};
		\draw[edge ] (v1) ..controls +(-.7,.6) and +(-.7,-.6) ..(v1);
		\draw[edge ] (v1) ..controls +(.7,-.6) and +(.7,.6) .. (v1);
		\node[vertex,fill=white] (v2) at ($(x0) +( -.4,-.5)$){}; 
		\node[vertex,fill=white] (v3) at ($(x0) +(.4,-.5)$){}; 
		\draw[edge] (v2.west) ..controls +(-.7,.7) and +(.7,.7) ..(v2.east);
		\draw[edge] (v3.west) ..controls +(-.7,.7) and +(.7,.7) ..(v3.east);
		
		\node[vertex,fill=black] (v1) at ($(x0) +(0,-1.5)$){};
		\draw[edge ] (v1) ..controls +(-.7,.6) and +(-.7,-.6) ..(v1);
		\draw[edge ] (v1) ..controls +(.7,-.6) and +(.7,.6) .. (v1);
		\node (v2) at ($(x0) +( 0,-2)$){}; 
		\node (v3) at ($(x0) +( 0,-2.3)$){}; 
		\draw[edge ] ($(v2) +(-.6,0)$) -- ($(v2) +(.6,0)$);
		\draw[edge ] ($(v3) +(-.6,0)$) -- ($(v3) +(.6,0)$); 
		
		\node at ($(x0) +(0,-3)$){$ 2\frac 4 {\hbar^2} \cdot \frac{1}{2} \left( \frac 12  \right) ^2\cdot \frac{1}{8} \cdot \frac{ N(N+2)}{3} \hbar $};

		\coordinate(x0) at (6,-6);
		\node[vertex,fill=black] (v1) at ($(x0) +(0,.4)$){};
		\draw[edge ] (v1) ..controls +(-.7,.6) and +(-.7,-.6) ..(v1);
		\draw[edge ] (v1) ..controls +(.7,-.6) and +(.7,.6) .. (v1);
		\node[vertex,fill=white] (v2) at ($(x0) +(-.4,-.3)$){};
		\node[vertex,fill=white] (v3) at ($(x0) +(.4,-.3)$){};
		\draw[edge] (v2.north) ..controls +(.1,.3) and +(-.1,.3) ..(v3.north);
		\draw[edge] (v2.south) ..controls +(.1,-.3) and +(-.1,-.3) ..(v3.south);
		
		\node[vertex,fill=black] (v1) at ($(x0) +(0,-1.5)$){};
		\draw[edge ] (v1) ..controls +(-.7,.6) and +(-.7,-.6) ..(v1);
		\draw[edge ] (v1) ..controls +(.7,-.6) and +(.7,.6) .. (v1);
		\node (v2) at ($(x0) +( 0,-2)$){}; 
		\node (v3) at ($(x0) +( 0,-2.3)$){}; 
		\draw[edge ] ($(v2) +(-.6,0)$) -- ($(v2) +(.6,0)$);
		\draw[edge ] ($(v3) +(-.6,0)$) -- ($(v3) +(.6,0)$); 
		
		\node at ($(x0) +(0,-3)$){$ 2\frac 4 {\hbar^2} \cdot \frac{1}{4}  \cdot \frac{1}{8} \cdot \frac{N(N+2)}{3} \hbar $};

		\coordinate(x0) at (11,-6);
		\node[vertex,fill=black] (v1) at ($(x0) +(0,0)$){};
		\node[vertex,fill=white] (v2) at ($(x0) +(-.7,0)$){};
		\node[vertex,fill=white] (v3) at ($(x0) +(.7,0)$){};
		\draw[edge] (v1) ..controls +(.1,.3) and +(-.1,.3) ..(v3.north);
		\draw[edge] (v1) ..controls +(.1,-.3) and +(-.1,-.3) ..(v3.south);
		\draw[edge] (v1) ..controls +(-.1,.3) and +(.1,.3) ..(v2.north);
		\draw[edge] (v1) ..controls +(-.1,-.3) and +(.1,-.3) ..(v2.south);
		
		\node[vertex,fill=black] (v1) at ($(x0) +(0,-2)$){}; 
		\draw[edge ] ($(v1) +(-.6,.3)$) -- ($(v1) +(.6,-.3)$);
		\draw[edge ] ($(v1) +(-.6,-.3)$) -- ($(v1) +(.6,.3)$); 
		
		\node at ($(x0) +(0,-3)$){$ 2\frac 4 {\hbar^2} \cdot \frac{1}{8}  \hbar $};

	\end{tikzpicture}
	\caption{Feynman graphs representing the first two summands of $\partial^4_j Z(\hbar, j)$. Notice that the three decompositions of the vertex (\cref{fig:vertex_decomposition}) are generated from two topoplogically distinct vacuum graphs; together they form the first term of the third line of \cref{Z_derivatives}. A similar effect occurs at higher loop order, but we do not draw all permutations of external edges explicitly. Only the very last one of the shown graph is connected and contributes as leading summand \enquote{1} to the series expansion of $W$ (\cref{W_derivatives}).} 
	\label{fig:Z_graphs_4legs}
\end{figure}
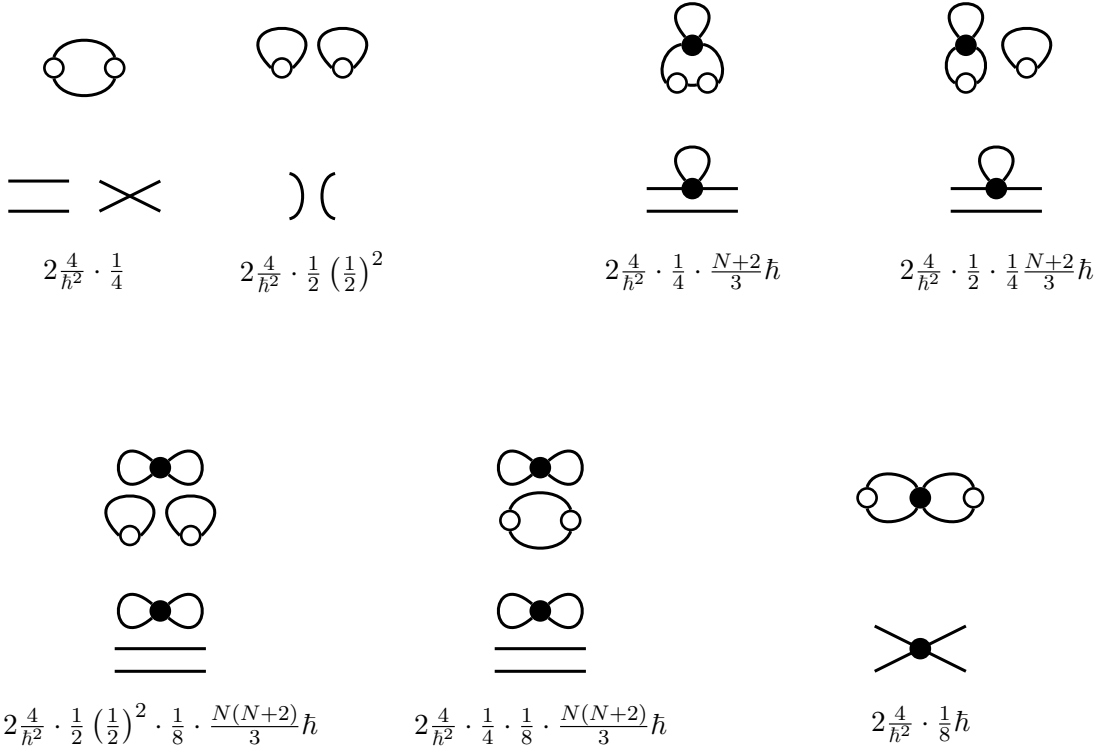

While the partition function $Z$ enumerates \emph{all} Feynman graphs, the free energy $W:= \hbar \ln (Z)$ enumerates the \emph{connected} Feynman graphs. Notice that connectedness here refers to the graphs with external edges, not the vacuum graphs with marked vertices, and that the extra factor $\hbar$ ensures that the order in $\hbar$ coincides with the loop order of the graphs. The series start with
\begin{align}\label{W_derivatives}
	\scalemath{.9}{\partial^0_j W \Big|_{j=0}}&	\scalemath{.9}{=  \frac{N(N+2)}{24}\hbar^2 + \frac{N(N+2)(N+3)}{144}\hbar^3 + \frac{N(N+2)( 60+34N+5N^2)}{2592}\hbar^4 + \ldots } \nonumber \\
	\scalemath{.9}{\partial^2_j W \Big|_{j=0}}&	\scalemath{.9}{= 1 + \frac{ N+2 }{6}\hbar + \frac{ (N+2)(N+3)}{18}\hbar^2 + \frac{(N+2)( 60+34N+5N^2)}{216} \hbar^3 + \ldots} \\
	\scalemath{.9}{\partial^4_j W \Big|_{j=0}}&	\scalemath{.9}{=1+ \frac{16+5N}{6}\hbar + \frac{92+50N+7N^2}{12}\hbar^2 + \frac{2784+2000N+493N^2+42N^3 }{108}\hbar^3 + \ldots .} \nonumber 
\end{align}

\begin{example}
	The first terms of $\partial^0_j W$ in \cref{W_derivatives} correspond to the three \emph{connected} graphs shown in \cref{fig:Z_graphs_vaccum}, where one confirms that
	\begin{align*}
		\frac{1}{16}\hbar^3 \cdot \left( \frac{N(N+2)^2}{9} \right) +\frac{1}{48}\hbar^3 \cdot \left( \frac{N(N+2)}{3} \right) &=\frac{N(N+2)(N+3)}{144} \hbar^3.
	\end{align*}
	Similarly, the connected graphs from \cref{fig:Z_graphs_2legs,fig:Z_graphs_4legs} yield the expected terms of \cref{W_derivatives}. Note that of the vertex-type graphs (\cref{fig:Z_graphs_4legs}), only the very last one is connected, and it corresponds to the leading term $1 \cdot \hbar^0$ of $\partial^4_j W$ in \cref{W_derivatives}.
\end{example}

\medskip
In order to enumerate 1PI graphs, one would conventionally introduce the \enquote{classical} field variable $\varphi:= \partial_j W(\hbar, j)$, invert that series to obtain $j(\hbar, \varphi)$, and define the generating function $	G(\hbar, \varphi):= W \big(\hbar, j(\hbar, \varphi)\big)- j(\hbar, \varphi)\varphi$. 
However, it is computationally faster to directly relate the derivatives of $G$ to those of $W$, namely 
\begin{align}\label{G_W_derivatives}
	\partial^0_\varphi G \big|_{\varphi=0} &= \partial^0_j W \big|_{j=0},  \hspace{2cm} \partial^2_\varphi G \big|_{\varphi=0}= \frac{-1} {\partial^2_j W \big|_{j=0}}, \\
	\scalemath{.9}{\partial^4_\varphi G \big|_{0} =\frac{ \partial^4_j W \big|_{0} } {\left( \partial^2_j W \big|_{0} \right) ^4} }~&\scalemath{.9}{= 1 + \frac{N+8}{6}\hbar + \frac{140 + 46N + 3N^2}{36}\hbar^2 + \frac{(N+4)(384+97N+5N^2)}{108}\hbar^3 +\ldots .} \nonumber 
\end{align}

\subsection{Counterterms, primitive graphs, Martin invariants} \label{sec:0dim_counterterms}

In the 0-dimensional theory, there are no divergent Feynman integrals. Nevertheless, one can use this setting to enumerate the graphs which would be divergent, or primitive, in four dimensions because these graphs can be identified purely combinatorially as stated in \cref{sec:periods}: A graph is declared to be divergent if and only if it has four or less external edges.
Using the 1PI Green functions from \cref{G_W_derivatives}, we first compute the invariant charge, 
\begin{align*}
	\scalemath{.9}{	Q(\hbar)} &\scalemath{.9}{:= \frac{ \partial^4_\varphi G \big|_{\varphi=0} }{\left( \partial^2_\varphi G \big|_{\varphi=0}\right)^2 } = 1 + \frac{N+4}{2}\hbar + \frac{5(N+4)(N+5)}{18}\hbar^2 + \frac{7(N+4)(148+54N+5N^2)}{216}\hbar^3+\ldots. }
\end{align*}
The invariant charge defines the running coupling. 
In our setup, the expansion parameter~$\hbar$ takes the role of the coupling, hence we obtain a \enquote{renormalized} expansion parameter by $\hbar_\ren(\hbar) := \hbar \cdot Q(\hbar)$. Inverting this power series, we find 
\begin{align}\label{hbar_ren}
	\scalemath{.9}{	\hbar(\hbar_\ren)} &= \scalemath{.9}{\hbar_\ren - \frac{N+4}{2}\hbar^2_\ren + \frac{(N+4)(4N+11)}{18}\hbar^3_\ren + \frac{(N+4)(49+27N+5N^2)}{54}\hbar^4_\ren + \ldots. }
\end{align}
Finally, the vertex counterterm is the inverse (under multiplication) of the 1PI four-point function (\cref{G_W_derivatives}), where we replace $\hbar$ by \cref{hbar_ren},
\begin{align} 
	\scalemath{.9}{	z^{(4)}(\hbar_\ren) } &\scalemath{.9}{:=\frac{1}{\partial^4_\varphi G \big|_{\varphi=0}} = 1 - \frac{N+8}{6}\hbar_\ren + \frac{20+6N+N^2}{36}\hbar_\ren^2 - \frac{ 224 +64N + 8N^2 + N^3}{216}\hbar^3_\ren + \ldots. }
\end{align}
The vertex counterterm enumerates the graphs with four external edges which, in four dimensions, are either primitive or contain an overlapping divergence. 
As established in \cite[Sec.~6]{borinsky_graphs_2018}, for $\phi^4$ theory the only overlapping divergences are the $L$-loop chains of multiedge graphs. 
These chains have an automorphism symmetry factor $\frac{1}{2^L}$ 
and we have computed the $\O(N)$-symmetry factor in \cref{ex:fish_chain}. 
Therefore the generating function of primitive graphs is
\begin{align}\label{primitive_generating_function}
	&\text{prim}(\hbar_\ren)=\sum_L p_L \hbar^L_\ren :=\scalemath{.95}{ 1-z^{(4)}(\hbar_\ren)+\sum_{L=2}^\infty (-1)^L \frac{(N-1) 2^{L+1} + (N+2)^{L+1}}{6^{L}N}  \hbar^L_\ren } \\
	&= \frac{N+8}{6}\hbar_\ren + \frac{22+5N}{27}\hbar^3_\ren + \frac{186+55N+2N^2}{81}\hbar^4_\ren + \frac{5440+1946N+147N^2}{486}\hbar^5_\ren + \ldots .\nonumber 
\end{align}

The coefficients $p_L=p_L(N)$ are polynomials in $N$, and $p_2$  vanishes since there is no primitive two-loop graph in $\phi^4$ theory. Setting $N=1$ in \cref{primitive_generating_function}, we reproduce the sum of automorphism symmetry factors given in \cite{borinsky_renormalized_2017}, which had been called $N_L^{(\text{Aut})}$ in \cite{balduf_statistics_2023}:
\begin{align}\label{NLAut}
	\operatorname{prim}(\hbar_\ren)\big|_{N=1}=\frac 3 2 \hbar_\ren + \hbar^3_\ren + 3 \hbar^4_\ren + \frac{31}{2}\hbar^5_\ren + \frac{529}{6}\hbar^6_\ren + \ldots .
\end{align}
We give further coefficients ~$3^{L+1}p_L$ in \cref{tab:0dim_primitive_coefficients} in the appendix. 
These numbers are in general not integers because the automorphism symmetry factors give rise to a denominator. 
Conversely, in \cref{tab:leading_decompositions_count}, we give the sum of $\O(N)$-symmetry factors \emph{without} the weighting by automorphism symmetry factors.
In both cases, we observe that the coefficients of low order in $N$ are typically larger than the coefficient of the higher orders in~$N$. 
To quantify this, we define the \emph{average order}
\begin{align}\label{def:average_order}
	\left \langle k \right \rangle _L := \frac{1}{p_L\big|_{N=1}}\sum_{k=0}^\infty k \left[ N^k \right] p_L(N)
    = \frac{ \frac{\partial}{\partial \,N} p_L(N) }{p_L(N) }\Big|_{N=1}
    = \frac{\partial \log p_L(N)}{\partial \log N}\Big|_{N=1}.
\end{align}
This quantity is positive by construction, but it grows with loop order only slowly, staying below $\frac 12$ for all $L\leq 9$. This indicates that $p_L(N)$ is, for small $N$, mostly determined by the coefficients of the low-order monomials $N^0, N^1$, and not by the highest-order monomial (which, of course, dominates as $N \rightarrow \infty$). 

\medskip 

The generating function of primitives (\cref{primitive_generating_function}) equals, up to a factor $4!$, the sum on the right hand side of \cref{Martin_sum}. Consequently, we can read off the sum of the Martin invariants at $L$ loops, 
\begin{align*}
	M_L^{[1]}:=& \sum_{\substack{\text{completion }G\\ \text{with }L \text{ loop dec.}}} \frac{(L+2)}{\abs{\Aut(G)}} M^{[1]}(G) = \frac{3^L}{2\cdot 4!} \  [\hbar_\ren^L]\operatorname{prim}(\hbar_\ren) \Big|_{N=-2} = \frac{3^L}{48} p_L(-2).
\end{align*}
The factor $3^L$ can be absorbed into the argument of $\operatorname{prim}(\hbar_\ren)$, namely 
\begin{align}\label{Martin_generating_function}
	\sum_{L=0}^\infty M_L^{[1]} \hbar_\ren^L = \frac{1}{48} \operatorname{prim}(3\hbar_\ren) \big|_{N=-2}.
\end{align}
We reproduce the explicit computation from \cref{ex:Martin_invariants}, and find for $L\in \left \lbrace 1, \ldots, 10 \right \rbrace $: 
\begin{align*}
	M_L^{[1]}\in \left \lbrace \frac{1}{16}, \ 0, \quad \frac{1}{4}, \quad \frac 7 4 , \quad\frac{89}{4}, \quad \frac{1255}{4}, \quad \frac{20405}{4}, \quad \frac{372139}{4}, \quad \frac{7510709}{4}, \quad \frac{166012747}{4} \right \rbrace . 
\end{align*}

\subsection{Asymptotics}\label{sec:0dim_asymptotics}
So far, we have obtained generating functions for certain sums of graphs, which deliver their exact counts at low loop order. Now, we turn to the asymptotic expansion around infinite order. One way to obtain the asymptotics of the 0-dimensional path integral $Z(\hbar,j)$ is to start from closed formula \cref{Z_j_series} for the $n$\textsuperscript{th} coefficient, 
\begin{align}\label{Z_hj_coefficient}
	\left[ \hbar^n  j^{2k}\right] Z(\hbar, j) &= \frac{\Gamma \!\left( 2n+3k+ \frac N 2 \right) }{\Gamma(n+k+1)} \frac{  \Gamma \!\left( k+\frac 12 \right)  }{  6^n 3^k ~ \Gamma \!\left( k+\frac N 2 \right) \Gamma \! \left( \frac 12 \right)  \Gamma(2k+1) }.
\end{align}
Here, $k$ is a fixed number, while the asymptotic expansion concerns $n\rightarrow \infty$. Using the multiplication identity for the gamma function, $\Gamma(2n+A) = \Gamma \left( n+\frac A 2 \right) \Gamma \left( n+ \frac A 2 + \frac 1 2 \right) \frac{4^n 2^A } { 2 \sqrt \pi }$, we can directly apply the known asymptotic expansion \cite{olver_asymptotic_1995}
\begin{align*}
	\frac{\Gamma(n+x)\Gamma(n+y)}{\Gamma(n+z)}&\underset{n \rightarrow \infty}\sim \sum_{r=0}^\infty \Gamma \left(n+ x + y-z -r \right) (-1)^r \frac{(z-x)_r (z-y)_r}{r!}. 
\end{align*}
In our case, $x= \frac 3 2 k + \frac N 4$, ~ $y= \frac 3 2 k + \frac{N+2}{4}$, and $z= k+1$.
The Pochhammer symbols $(A)_r:=  \Gamma(A+r)/\Gamma(A)$ in the formula can be simplified, 
\begin{align*}
	\left( 1-\frac 12 k - \frac N 4 \right) _r 	\left( \frac 12 - \frac 12 k - \frac N 4 \right) _r &= \frac{1}{16^r} \prod_{t=1}^{2r} \left( N+2k-2t \right) = \frac{ \Gamma \!\left( 1-k - \frac N 2 + 2r \right) }{4^r ~\Gamma \!\left( 1-k - \frac N 2 \right) }.
\end{align*}
Combining all parts, the right hand side of \cref{Z_hj_coefficient} has the expansion
\begin{align}\label{Z_hj_coefficient2}
	\hspace*{-.5cm}\underset{n\rightarrow \infty}\sim  \sum_{r=0}^\infty  \Gamma \! \left(n+ 2k + \frac {N-1} 2 -r \right) \frac{ (-1)^r~ 2^n 8^k 2^{ \frac N 2}\Gamma \!\left( 1-k - \frac N 2 + 2r \right) \Gamma \!\left( k+\frac 12 \right) } { 2 \pi r!~ 3^n 3^k 4^r \Gamma \!\left( 1-k - \frac N 2 \right)\Gamma \!\left( k+\frac N 2 \right) \Gamma(2k+1)}. 
\end{align}
We transform this expression to the notation of \cite{borinsky_generating_2018,borinsky_renormalized_2017} that we already used in \cref{beta_asymptotic_ansatz}:
\begin{align}\label{asymptotics_series}
	\left[ \hbar^n  j^{2k}\right] Z(\hbar, j) 
	&\underset{n\rightarrow \infty} \sim \sum_{r=0}^\infty c_r a^{-n-c_s+r} ~ \Gamma\! \left( n+c_s-r \right),
\end{align}
\begin{align}\label{Z_hj_asymptotics_coefficients}
	\text{where} &\qquad a=\frac 3 2, \qquad c_s = 2k+\frac{N-1}{2}, \\
	c_r &= \frac{(-1)^{k+r}~  3^k ~3^{ \frac{N-1}{2}} ~ \sin \left( \frac N 2 \pi  \right)  \Gamma \!\left( 1-k - \frac N 2 + 2r \right)  } { 6^r 2^k \sqrt 2 ~\pi^{\frac 3 2}~ r!~\Gamma( k+1)}  =c_0 \cdot \frac{(-1)^r}{6^r} \left( 1-k-\frac N 2 \right) _{2r}.\nonumber 
\end{align}
\begin{example}
	The coefficient $k=0$ of the path integral, hence the power series $Z(\hbar, 0)= \sum \hbar^n z_n$, has the asymptotics
	\begin{align*}
		\scalemath{.85}{z_n\sim\frac{3^{\frac{N-1}{2}} \left( \frac 3 2 \right) ^{-n-\frac{N-1}{2}} \Gamma \! \left( n+ \frac{N-1}{2} \right) }{ \sqrt {2\pi} \Gamma \!\left( \frac N 2 \right) } \left( 1 - \frac{(N-2)(N-4)} {24 (2n +N-3) } + \frac{(N-2)(N-4)(N-6)(N-8)}{1152 (2n+N-3)(2n+N-5)}+\ldots \right)} . 
	\end{align*}
\end{example}

A plot of the coefficients for $k=2$, that is, the 4-point function $[\hbar^n j^4]Z(\hbar, j)$, is shown in \cref{fig:Z_asymptotics}.   We see that already the leading asymptotics $c_0 A^{-n-c_s} \Gamma(n+c_s)$ matches the data points quite well, the deviation is $\approx 1\%$ at $n=15$. If one divides by this leading asymptotics, the remaining difference is captured well by the subleading terms in \cref{asymptotics_series}. Note that \cref{asymptotics_series} is an \emph{asymptotic} expansion: Including higher terms in the sum improves the description for $n\rightarrow \infty$, but the sum does not converge for any finite $n$, as is clearly visible for the small values of $n$ in \cref{fig:Z_asymptotics_correction}.

\begin{figure}[htbp]
	\begin{subfigure}{ .49 \linewidth}
		\centering
		{\small Coefficients $[\hbar^n j^{2k}]Z$ for $N=1$}\\
		\includegraphics[width=\linewidth]{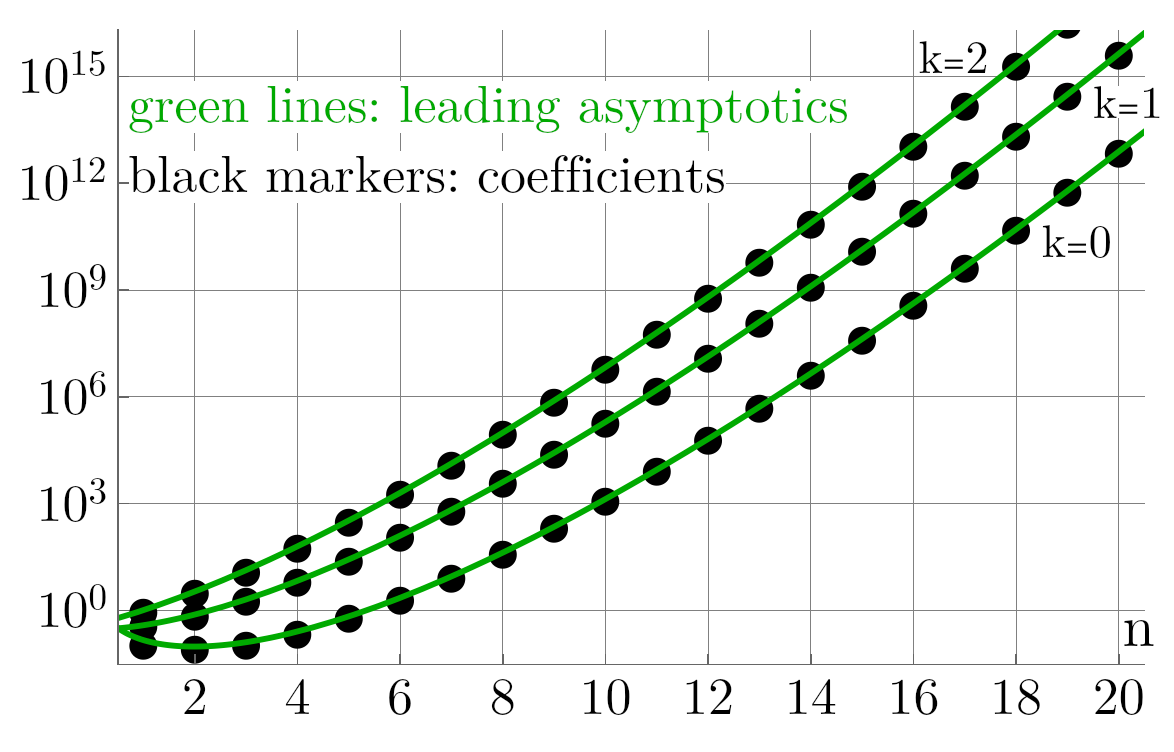}
		\subcaption{}
		\label{fig:Z_asymptotics}
	\end{subfigure}
	\begin{subfigure}{ .49 \linewidth}
		\centering
		{\small Subleading correctios to $[\hbar^n j^4]Z$ for $N=1$}\\
		\includegraphics[width=\linewidth]{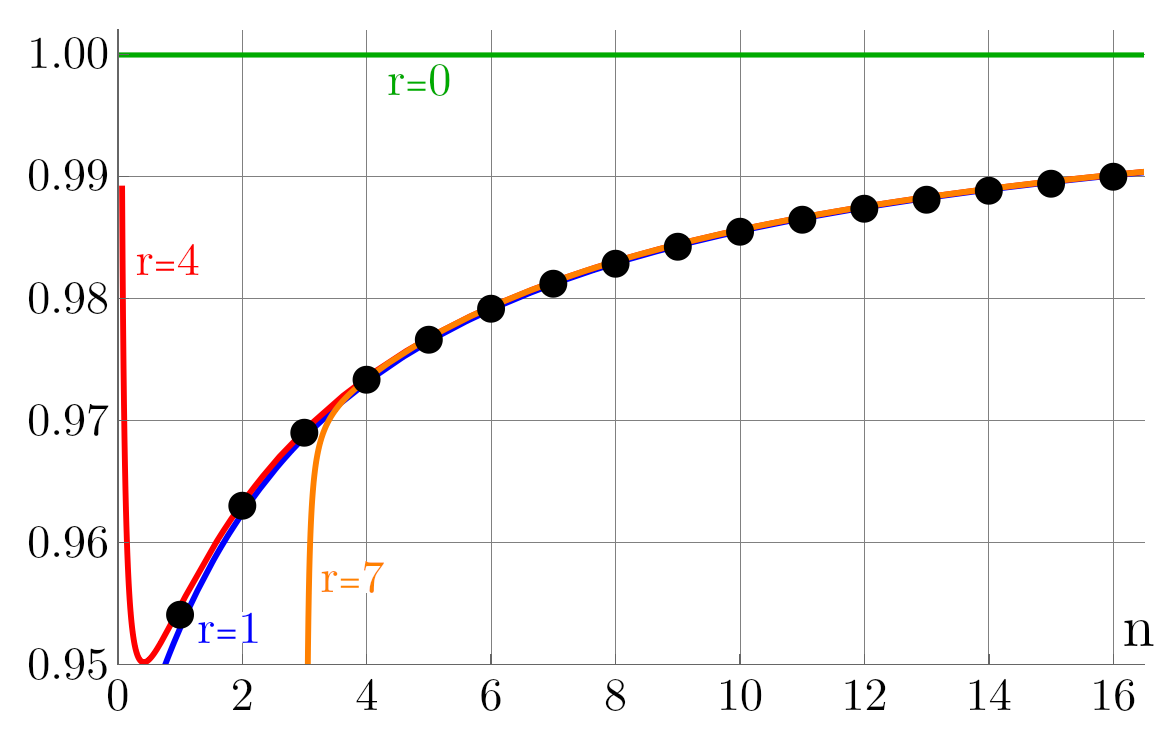}
		\subcaption{}
		\label{fig:Z_asymptotics_correction}
	\end{subfigure}
	\caption{\textbf{(a)} First 20 coefficients of $[\hbar^n j^{2k}]Z(\hbar, j)$, for $N=1$ and $k\in \left \lbrace 0,1,2 \right \rbrace $ , logarithmic plot. The green line shows the leading asymptotics according to \cref{asymptotics_series}, only including the term $r=0$  \textbf{(b)} The data points for $k=2$, scaled to the leading asymptotics. Colored lines show the subleading corrections up to the indicated $r$. Note that this asymptotic expansion matches the data remarkably well even at very low orders $n$. }
	\label{fig:Z_asymptotics_both}
\end{figure}

The mapping from $Z(\hbar, j)$, viewed as a formal power series in $\hbar$, to its asymptotics (\cref{asymptotics_series}) is the alien derivative operator $\mathcal A^a_\hbar$ of \cite{borinsky_generating_2018,borinsky_renormalized_2017},
\begin{align}\label{alien_derivative_series}
	\mathcal A^{\frac 3 2}_\hbar Z  (\hbar,j) &:= \hbar^{-c_s} \sum_{r=0}^\infty\sum_{k=0}^\infty c_r(k,N) \hbar^r j^{2k}.
\end{align}
This operator has well-defined product- and chain rules, and it commutes with the ordinary derivative $\partial_j$.
To compute the asymptotics of the generating function of connected graphs, $W(\hbar, j)= \hbar \ln Z(\hbar, j)$, we use the chain rule and note that the power series of $\ln(x)$ is not divergent, therefore $\mathcal A^{\frac 3 2}_x \ln(x)=0$ and
\begin{align*}
	\mathcal A^{\frac 32}_\hbar W(\hbar,j) &= \frac{\hbar}{Z(\hbar, j)} \mathcal A^{\frac 32}_\hbar Z(\hbar, j).
\end{align*}
The first four terms of this power series are given in \cref{tab:primitive_asymptotic_coefficients} in \cref{sec:tables}. Note that by \cref{alien_derivative_series}, the exponent of $\hbar$ in the prefactor encodes the constant $c_s$ (which thereby does not always take the value of \cref{Z_hj_asymptotics_coefficients}, whereas $a=\frac 3 2$ in all cases).
Setting $N=1$ reproduces the coefficients given in \cite{borinsky_renormalized_2017} for $\phi^4$ theory without $\O(N)$ symmetry.

The asymptotics of the generating function $G(\hbar, \varphi)$ of 1PI graphs follows from \cref{G_W_derivatives},
\begin{align*}
	\mathcal A^{\frac 3 2}_\hbar \partial^0 _\varphi  G\big|_{\varphi=0} &= \mathcal A^{\frac 3 2}_\hbar\partial^0_j W\big|_{j=0}, \qquad 
	\mathcal A^{\frac 3 2}_\hbar\partial^2 _\varphi G\big|_{\varphi=0} = \frac{1}{\left( \partial^2_j W \big|_{j=0} \right) ^2 }\mathcal A^{\frac 3 2}_\hbar \partial^2_j W\big|_{j=0} , \\
	\mathcal A^{\frac 3 2}_\hbar \partial^4_\varphi G \big|_{\varphi=0}&= \frac{1}{\left( \partial^2_j W \big|_{j=0} \right) ^4} \mathcal A^{\frac 3 2}_\hbar \partial^4_j W \big|_{j=0} - \frac{4 ~\partial^4_j W\big|_{j=0}}{\left( \partial^2_jW\big|_{j=0} \right) ^5} \mathcal A^{\frac 3 2}_\hbar \partial^2_j W \big|_{j=0}.
\end{align*}
Again, the first coefficients of these power series are given explicitly in \cref{tab:primitive_asymptotic_coefficients}. Similarly, we use the product rule to compute the asymptotics of the invariant charge:
\begin{align*}
	\mathcal A^{\frac 32}_\hbar Q &= \mathcal A^{\frac 32}_\hbar \left( \frac{\partial^4_j W \big|_{j=0}}{ \left( \partial^2_j W\big|_{j=0} \right) ^2} \right) = \frac{\mathcal A^{\frac 32}_\hbar \partial^4_j W \big|_{j=0}}{ \left( \partial^2_j W\big|_{j=0} \right) ^2} -\frac{2~\partial^4_j W \big|_{j=0}}{ \left( \partial^2_j W\big|_{j=0} \right) ^3} \mathcal A^{\frac 32}_\hbar\partial^2_j W\big|_{j=0}.
\end{align*}
This is the asymptotics of the renormalized coupling $\mathcal A^{\frac 3 2}_\hbar \hbar_\ren = \hbar \mathcal A^{\frac 3 2}_\hbar Q$. To obtain the asymptotics of the inverse function, let $x$ be the variable of the inverse series $\hbar_\ren^{-1}$, then
\begin{align*}
	\mathcal A^{\frac 3 2}_{\hbar_\ren} \hbar(\hbar_\ren) = \mathcal A^{\frac 32}_x \hbar_\ren^{-1}(x) &= - e^{\frac 32 \left( \frac 1 x - \frac 1 {\hbar_\ren^{-1}(x)} \right) } \frac{\mathcal A^{\frac 3 2}_\hbar \hbar_\ren(\hbar)}{\partial_\hbar \hbar_\ren(\hbar)} \big|_{\hbar= \hbar_\ren^{-1}(x)}.
\end{align*}
We provide their coefficients in \cref{tab:primitive_asymptotic_coefficients}. Finally, we need the asymptotics of the vertex counterterm $z^{(4)}$, expressed as a power series of $\hbar_\ren$. 
With 
$\mathcal A^{\frac 3 2}_\hbar z^{(4)}(\hbar) = - \big( \partial^4_\varphi G\big|_{ 0} \big) ^{-2} \mathcal A^{\frac 32}_\hbar \partial^4_\varphi G\big|_{ 0}$, and, since both $G(\hbar)$ and $\hbar(\hbar_\ren)$ have factorially divergent series, 
\begin{align}\label{z4_asymptotics}
	\mathcal A^{\frac 32}_{\hbar_\ren} z^{(4)}\big(\hbar(\hbar_\ren) \big) &= \partial_\hbar z^{(4)}\big|_{\hbar=\hbar(\hbar_\ren)} \cdot \mathcal A^{\frac 32}_{\hbar_\ren} \hbar(\hbar_\ren) + e^{\frac 3 2 \left( \frac{1}{\hbar_\ren}-\frac{1}{\hbar(\hbar_\ren)} \right) }\mathcal A^{\frac 32}_\hbar z^{(4)}\big|_{\hbar=\hbar(\hbar_\ren)}.
\end{align}

\subsection{Asymptotics of primitive graphs} \label{sec:0dim_primitive_asymptotics}

The asymptotics of the primitive graphs is unchanged by the addition of multiedge chains, and coincides with \cref{z4_asymptotics} as given in \cref{tab:primitive_asymptotic_coefficients}. Concretely, the asymptotic growth of $p_L$  as $L\rightarrow \infty$ in the series $\operatorname{prim}(\hbar_\ren)=:\sum p_L \hbar_\ren^L$ (\cref{primitive_generating_function}) is 
\begin{align}\label{primitive_asymptotics}
	p_L &\sim \frac{ 3^{\frac{N-1}{2} }}  {\frac 4 3 \sqrt{2\pi} \Gamma \! \left( \frac{N+4}{2} \right) } e^{-\frac{12+3N}{4}} \left( \frac 2 3 \right) ^{ L + \frac{N+5}{2} } \Gamma \!\left( L+ \frac{N+5}{2}\right) \\
	&\quad \cdot \left( 36- \frac{9(3 N^2 - 4N -80)}{4 (L + \frac{N+3}{2})} +\frac{9 ( 9 N^4 -40N^3-488N^2-416N+1664)}{128 (L+\frac{N+1}{2})(L+\frac{N+3}{2})}+\ldots \right). \nonumber
\end{align}

\begin{figure}[htbp]
	\begin{subfigure}{ .49 \linewidth}
		\centering
		{\small Coefficients $p_L$ for $N=1$}\\
		\includegraphics[width=\linewidth]{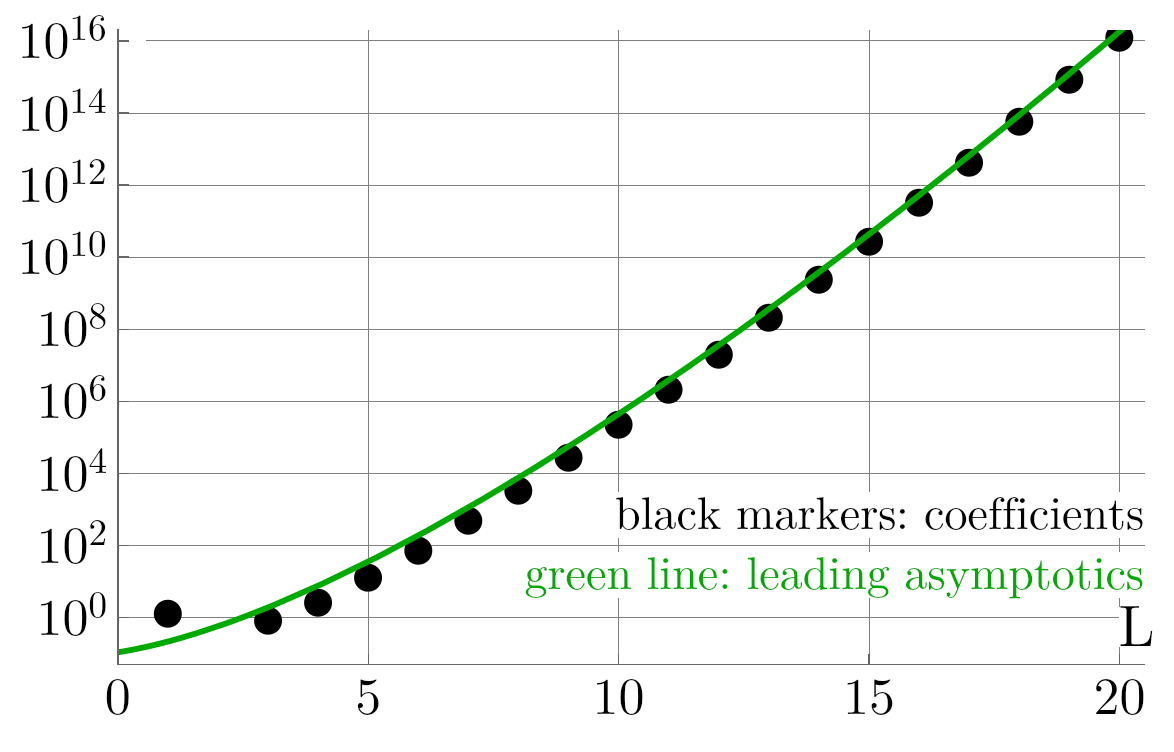}
		\subcaption{}
		\label{fig:prim_asymptotics}
	\end{subfigure}
	\begin{subfigure}{ .49 \linewidth}
		\centering
		{\small Subleading correction to $p_L$ for $N=1$}\\
		\includegraphics[width=\linewidth]{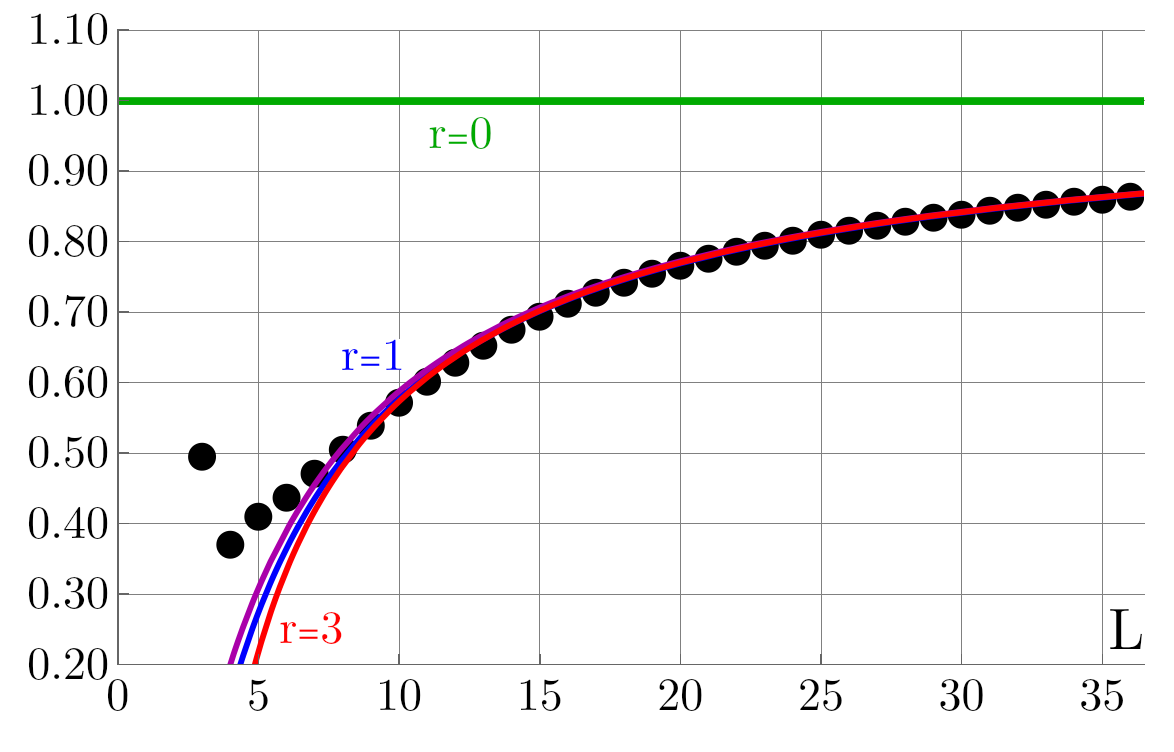}
		\subcaption{}
		\label{fig:prim_asymptotics_correction}
	\end{subfigure}
	\caption{\textbf{(a)} First 20 coefficients $p_L$ of $\text{prim}(\hbar_\ren)$, for $N=1$, logarithmic plot. The green line shows the leading asymptotics according to \cref{primitive_asymptotics}, only including the first term in the parenthesis.  \textbf{(b)} The same data points, but scaled to the leading asymptotics. They converge to unity as $L\rightarrow \infty$. Colored lines show the first three subleading corrections. Compare to \cref{fig:Z_asymptotics_correction} and note the different plot scale. }
	\label{fig:prim_asymptotics_both}
\end{figure}

For low loop orders, the exact terms $p_L$ from \cref{primitive_generating_function} are shown together with the leading term of \cref{primitive_asymptotics} in \cref{fig:prim_asymptotics}. Comparing this to the analogous plot for $\partial^4_j Z(\hbar, j)$ (\cref{fig:Z_asymptotics}), we note that the asymptotic expansion of primitive graphs is much less accurate in capturing the terms at low order, compared to the asymptotic expansion of $\partial_j^4 Z$. For example, at $L=16$ loops, the leading asymptotics for $p_{16}$ is roughly $30\%$ off the true value, while for $[\hbar^{16}]\partial^4_j Z$, the difference is only about 1\%. A similar picture holds for the subleading corrections to the asymptotics: From \cref{fig:Z_asymptotics_correction}, we see that the first subleading order, $r=1$, starts to capture the data only for $L >8$ loops. A closer investigation of the remaining residue indicates that the second subleading order, $r=2$, reliably improves the accuracy only at $L>25$. Both these thresholds depend on the value of $N$. In \cref{fig:minimum_loop_order}, we show the minimum loop order required to reach a certain accuracy with the leading, first subleading, 5\textsuperscript{th} subleading, or 9\textsuperscript{th} subleading order. The qualitative takeaway is that below roughly 25 loops, the asymptotic expansion is unreliable and accurate results are mostly coincidence of cancellations for particular values of $N$.

\begin{figure}[htb]
	\begin{subfigure}{ .49 \linewidth}
		\centering
		{\small $p_{20}$ relative to asymptotics}\\
		\includegraphics[width=\linewidth]{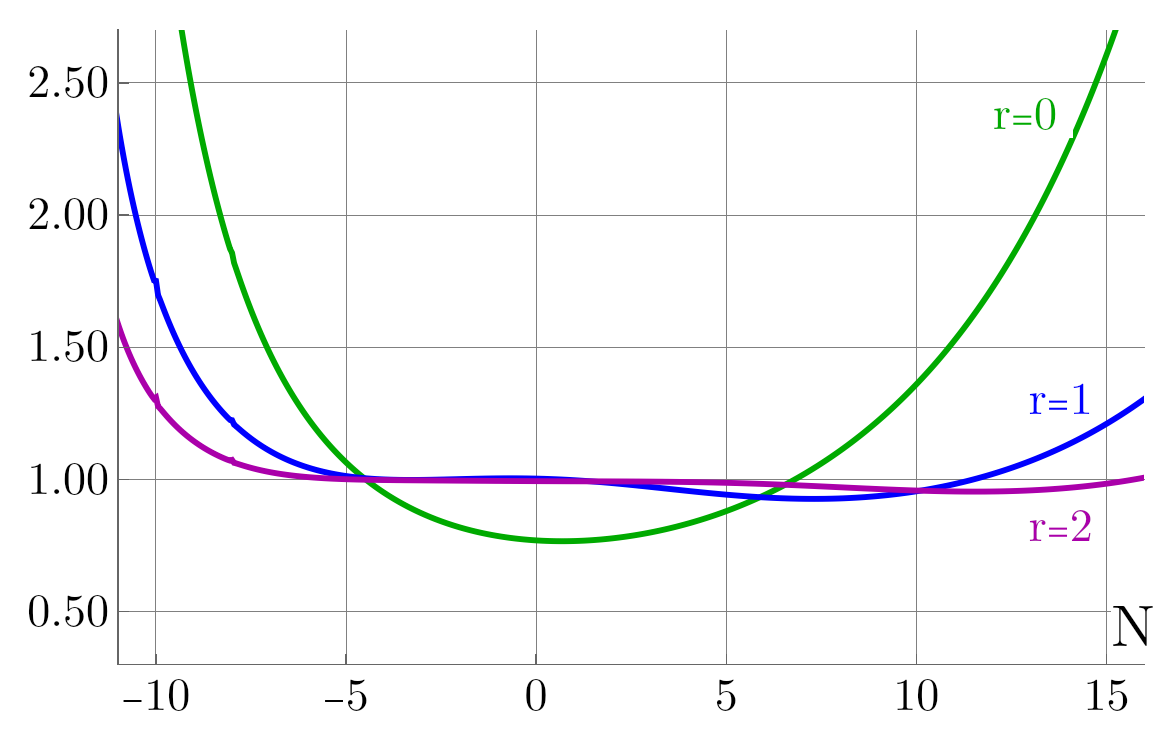}
		\subcaption{}
		\label{fig:prim_relative_N}
	\end{subfigure}
	\begin{subfigure}{ .49 \linewidth}
		\centering
		{\small Minimum loop order for asymptotics}\\
		\includegraphics[width=\linewidth]{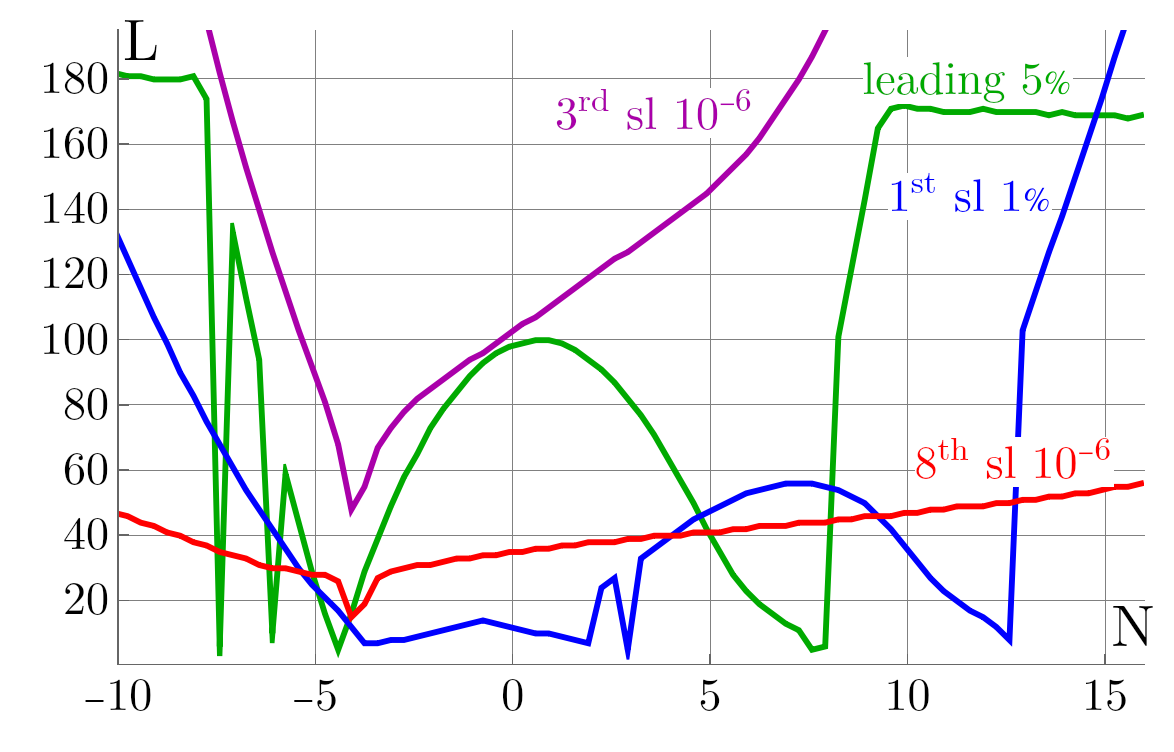}
		\subcaption{}
		\label{fig:minimum_loop_order}
	\end{subfigure}
	\caption{\textbf{(a)} $p_{20}(N)$ from \cref{primitive_generating_function} divided by the asymptotics \cref{primitive_asymptotics}, where both are functions of $N$. The leading asymptotics ($r=0$) captures the $N$ dependence up to a factor $\sim 2$ in this regime, higher corrections are close to unity for not too large $\abs{N}$. The accuracy might appear poor at first, but $p_L(N)$ is  strongly dependent on $N$, we have $p_{20}(15)/p_{20}(1)>10^3$. \textbf{(b)} Minimum loop order where the asymptotic expansion coincides with the true count of primitive graphs to a given accuracy. The leading asymptotics might need as much as $L\geq 180$ to reach 5\% accuracy. From $L\geq 30$ loops, one can expect that including the first subleading terms improves accuracy, as evidenced from the lines where the 3\textsuperscript{rd} respective 8\textsuperscript{th} corrections reach $10^{-6}$ accuracy.
	}
\end{figure}

To investigate the $N$ dependence, we divide $p_{20}(N)$ by the asymptotic series \cref{primitive_asymptotics} for $L=20$ up to a certain maximum order $r\in \left \lbrace 0,1,2 \right \rbrace $, see \cref{fig:prim_relative_N}. We see that the leading asymptotics, $r=0$, essentially fails at reproducing the actual $N$ dependence, while already $r=1$ gives a reasonable result for $-5 \leq N \leq 10$, and $r=2$ is even more accurate, but still differs visibly for $N \approx 10$. Despite the agreement, note that the limit $N\rightarrow \infty$ does not commute with the asymptotic expansion: Every term in the perturbation series is a polynomial, and therefore diverges as $N \rightarrow \infty$ (at fixed order $L$). On the other hand, the asymptotic expansion \cref{primitive_asymptotics} vanishes as $N\rightarrow \infty$ due to the factor $\frac{1}{\Gamma(\frac{N+4}{2})}$. 

The asymptotics \cref{primitive_asymptotics} has zeros at all negative integer $N$ smaller than $-2$, due to the factor $\Gamma \big(\frac{N+4}{2}\big)$ in the denominator. In view of \cref{lem:circuit_polynomial_zeros}, a zero at $N=-2j$ represent a summation over all channels of a $(2j+2)$-valent subgraph. A priori, not all such channels would give rise to a primitive graph, but in the limit $L\rightarrow \infty$, indeed almost all channels are present. This is a curious alternative perspective on the known fact that asymptotically, almost all simple graphs are primitive, that is, the leading asymptotic number of primitive graphs coincides with the leading asymptotic number of all simple graphs. Moreover, despite the rather slow convergence of the absolute value of $p_L$ to its asymptotics (\cref{fig:prim_asymptotics_correction}), the location of these zeroes converges rapidly towards negative even integers, see \cref{fig:largest_roots}. This effect was called \emph{superfast convergence} in \cite{pobylitsa_superfast_2008}. An alternative perspective is that $p_L(N)$ is not factorially divergent when $N$ is  exactly a negative even integer $\leq 4$ \cite{pobylitsa_anharmonic_2008}.

Later, when we compare to numerical data of 4-dimensional $\phi^4$ theory, we will consider the growth ratio $r_L$ defined in \cref{def:rL}. By \cref{primitive_asymptotics}, this ratio is an asymptotic power series in $\frac 1 L$, starting with
\begin{align}\label{primitives_ratio}
	r_L(p) &:= \frac1 L \frac{p_{L+1}}{p_L} \sim \frac 2 3 + \frac{N+5}{3}\frac 1 L - \frac{3N^2-4N-80}{24}\frac{1}{L^2} + \frac{5N^3+8N^2+40N+352}{48}\frac{1}{L^3}+\ldots
\end{align}
Consequently, if we plot the exact data (\cref{primitive_generating_function}) as a function of $\frac 1 L$, we expect to see a finite limit $\frac 2 3$ as $\frac 1 L \rightarrow 0$, and a linear convergence towards this limit with a slope $\frac{N+5}{3}$. This is indeed what we observe in \cref{fig:0dim_primitive_ratio}, but it is striking that this behaviour only sets in for $L \geq 25$. Conversely, for smaller $L$, the data points almost lie on a linear function, too, but with a \emph{wrong} slope and limit which does not reflect the true behaviour as $L\rightarrow \infty$. Not only that, 
but these wrong would-be asymptotics even mimics the $N$ dependence of \cref{primitives_ratio}: For various values of $N$, one obtains a slope that linearly depends on $N$, and the $y$-axis intersect is almost independent of $N$, close to $0.64$. This value is  $\approx 5\%$ smaller than the true limit $\frac 2 3$.
To illustrate this effect, we did linear fits to the data points $r_L$ for $10\leq L\leq 18$, (marked red in \cref{fig:0dim_primitive_ratio}), and extrapolated their slopes to $L=\infty$ (red lines). For comparison, the true asymptotics (\cref{primitives_ratio}) is drawn as green lines.

\begin{figure}[htbp]
	\begin{subfigure}{ .49 \linewidth}
		\centering
		{\small Growth rate $r_L$ of $p_L$,~ 0 dimensions}\\
		\includegraphics[width=\linewidth]{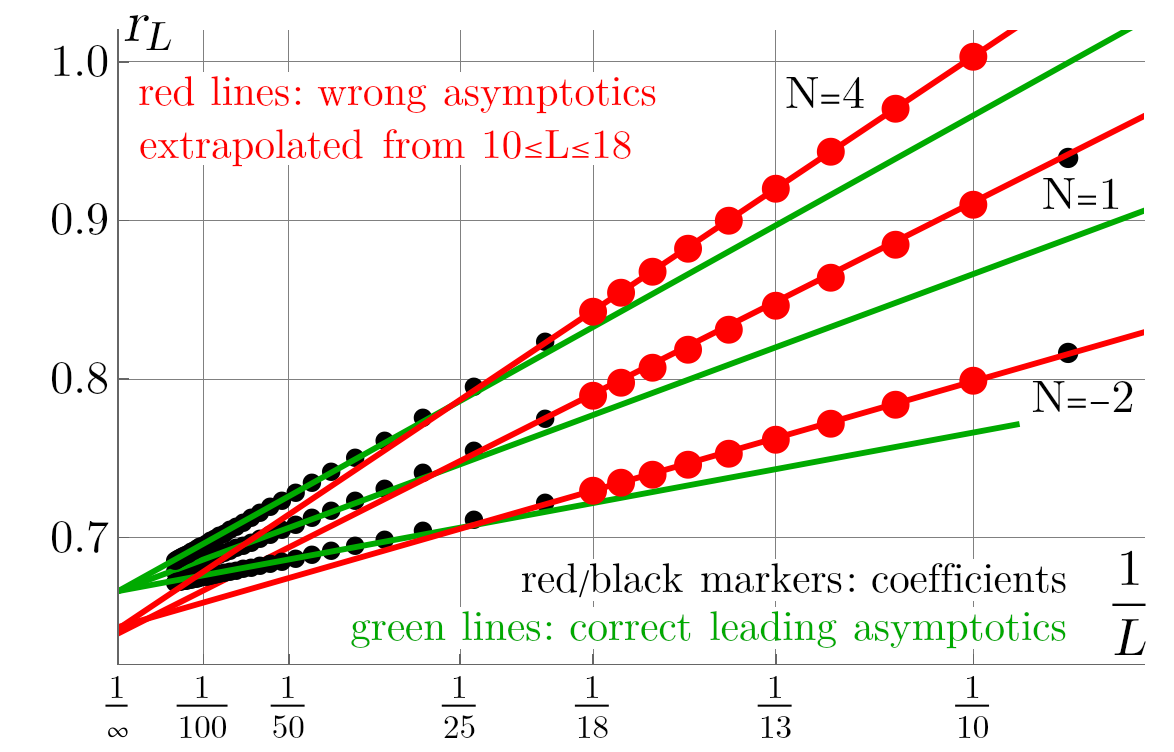}
		\subcaption{}
		\label{fig:0dim_primitive_ratio}
	\end{subfigure}
	\begin{subfigure}{ .49 \linewidth}
		\centering
		{\small Mean order $\left \langle k \right \rangle_L $ of $T(g,N)$}\\
		\includegraphics[width=\linewidth]{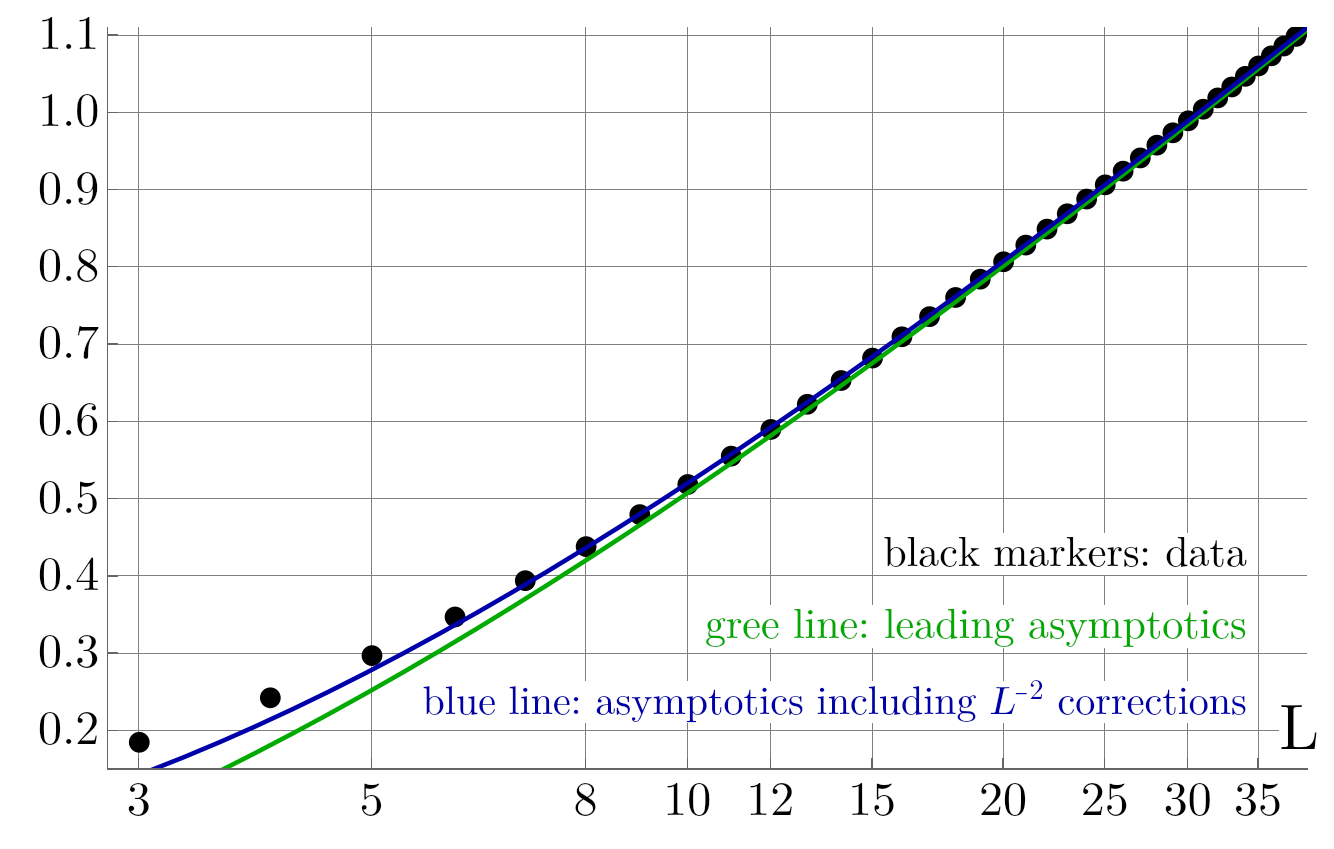}
		\subcaption{}
		\label{fig:mean_order}
	\end{subfigure}
	\caption{	\textbf{(a)} 	Growth ratio $r_L$ (\cref{primitives_ratio}) for $L\leq 150$, for three selected values of $N$. Green lines show the exact asymptotics including $\frac 1 L$ terms. Red lines show the wrong asymptotics obtained from a linear fit of $10 \leq L \leq 18$ loops (data points marked in red). Note that the red points give rise to the expected $N$ dependence (i.e. convergence to a $N$-independent limit at $L\rightarrow \infty$, with a slope that grows linearly in $N$, but the numerical value of these constants is not the true asymptotic one. The data starts to match the true asymptotics at $L\approx 25$. \textbf{(b)} Mean order $\left \langle k \right \rangle _L$ of the polynomial $p_L(N)$ (\cref{def:average_order}, together with its asymptotics (\cref{average_order_asymptotics}). The asymptotics grows logarithmically, which is a straight line in this log linear plot.}
\end{figure}

\subsection{Asymptotics of Martin invariants} \label{sec:0dim_asymptotics_Martin}
By \cref{Martin_generating_function}, the sum of all Martin invariants at $L$ loops, $M^{[1]}_L$, is a simple transformation of the generating function of primitive graphs. We can infer its asymptotics from \cref{primitive_asymptotics}, 
\begin{align*}
	M^{[1]}_L &\sim \scalemath{.9}{\frac{2^Le^{-\frac 32}}{24 \sqrt{\pi}} \Gamma \!\left( L+\frac 32 \right) \left( 1-\frac{15}{4(L+\frac 12)} + \frac{63}{32 (L+\frac 12)(L-\frac 12)} - \frac{1263}{128(L+\frac 12)(L-\frac 12) (L-\frac 32)}+\ldots \right) .}
\end{align*}
The growth rate is three times the one of primitives (\cref{primitives_ratio}) at $N=-2$,
\begin{align}\label{Martin_growth_rate}
	r_L\big( M^{[1]} \big)   &:= \frac{M^{[1]}_{L+1}}{L\cdot M^{[1]}_L}\sim 2 + 3 \frac 1 L + \frac{15}{2}\frac{1}{L^2}+ \frac{33}{2}\frac{1}{L^3}+ \frac{753}{8}\frac{1}{L^4} + \frac{3243}{4}\frac{1}{L^5}+ \ldots 
\end{align}
In \cref{fig:0dim_primitive_ratio}, we see that the sum of Martin invariants shows the same effect as the sum of primitive graphs, namely the growth rate $r_L$ assumes the true asymptotic behaviour only above 25 loops.
Dividing $M^{[1]}_L$ by $N^{(\text{Aut})}_L$, we obtain the average Martin invariant of an $L$-loop primitive, where the graphs are weighted by their automorphism symmetry factor. Its asymptotics at $L\rightarrow \infty$ is
\begin{align}\label{Martin_average}
	\frac{M^{[1]}_L}{N^{(\text{Aut})}_L } &= 3^L L^{-\frac 32} \frac{ e^{\frac 94} \sqrt \pi }{ 128 \sqrt 2} \left( 1-\frac{21}{16} \frac 1 L + \frac{673}{512}\frac{1}{L^2} + \frac{82809}{8192} \frac{1}{L^3} + \ldots \right) .
\end{align}

\subsection{The large-$N$ limit}\label{sec:0dim_largeN}

To examine the large-$N$ limit, introduce a variable $U=\hbar N$, so that each occurence of $\hbar $ in the generating functions in \cref{sec:0dim_external} is replaced by $\frac{U}{N}$. The coefficients of $Z(U, j)$ (\cref{Z_j_series}) still contain positive powers of $N$, but those are entirely due to disconnected graphs. In $W$, and all subsequent series, the coefficients are polynomials in $\frac 1 N$. In a large-$N$ limit (where $U$ is kept fixed), all subleading terms vanish and one obtains a power series in $U$. In particular, the coefficients of $\operatorname{prim}(\hbar)$ (\cref{primitive_generating_function}) have no leading contribution beyond order $U^1$.  This says that to a fixed order $\frac{1}{N^k}$, only finitely many primitive graphs will contribute, and all primitives beyond a certain loop order vanish to that order in $\frac{1}{N}$. 

\begin{example}
	By \cref{W_derivatives}, the series for connected vertex-type graphs starts with
	\begin{align*}
		\partial^4_j W\Big|_{j=0}\Big(U, \frac 1 N\Big) &= 1 + \left( \frac 5 6+\frac 8 3 \frac 1 N \right) U + \left( \frac{7}{12} + \frac{25}{6}\frac 1 N + \frac{23}{3}\frac{1}{N^2} \right) U^2 + \ldots. 
	\end{align*}
	The large-$N$ limit of this series is $1+\frac 5 6 U + \frac{7}{12}U^2 + \ldots$. Conversely, the generating series for primitive graphs starts with 
	\begin{align*}
		\operatorname{prim} \Big(U, \frac 1 N\Big)  &= \scalemath{.9}{\left( \frac 1 6 + \frac 4 3 \frac 1 N \right) U + \left( 0 +  \frac{5}{27} \frac{1}{N^2} + \frac{22}{27} \frac{1}{N^3} \right) U^3 +\left( 0+\frac{2}{81}\frac{1}{N^2}+\ldots \right) U^4 + \ldots .}
	\end{align*}
	In this case, the large-$N$ limit of higher-order terms vanishes, and $\prim(U,0)=\frac 1 6 U $. 
\end{example}

The replacement $\hbar \mapsto \frac U N$ leads to an extra factor $N^n$ in \cref{Z_hj_coefficient2}. The asymptotic growth coefficient $a=\frac 3 2$ of \cref{Z_hj_asymptotics_coefficients} then turns into $a=\frac 3 2 N$, which diverges in the limit $N\rightarrow \infty$. The asymptotic growth of series coefficients is $\propto a^{-n}\Gamma(n+c_s)$, hence, unbounded growth of $a$ suggests that the large-$N$ series does not diverge factorially\footnote{The growth coefficient $a$ refers to doing the the large-$L$ expansion first, before the large-$N$ expansion. In order to \emph{prove} convergence of the large-$N$ expansion at finite coupling, one would instead have to do the large-$N$ expansion first. }.
Convergence of the large-$N$ expansion means that with growing $L$, the coefficient of highest order in $N$ in the  polynomial $p_L(N)$ grows much slower than the sum of all terms (which is the evaluation of $p_L(N)$ at fixed $N$). This causes the sum to diverge factorially with growing $L$, while every finite order in the large-$N$ expansion is convergent.  

We can make this behaviour more explicit by considering the average order $\left \langle k \right \rangle _L$ of the   polynomials $p_L(N)$, introduced in \cref{def:average_order}. Computing the $N$-derivative of the asymptotics of $p_L(N)$ from \cref{primitive_asymptotics}, we find the asymptotics of $\left \langle k \right \rangle _L$: 
\begin{align}\label{average_order_asymptotics}
\left \langle k \right \rangle _L = \frac{\partial_N \, p_L}{p_L}\Big|_{N=1} &\sim \frac 12 \ln(L) + \frac 12 \gamma_E - \frac{25}{12} + \frac 3 2 \ln(2) + \frac {11} 8 \frac 1 L -\frac{145}{96} \frac{1}{L^2}+ \mathcal O \left( \frac{1}{L^3} \right) .
\end{align}
Here, $\gamma_E$ is the Euler-Mascheroni constant. This asymptotics together with data points for $L\leq 40$ is shown in \cref{fig:mean_order}. 

The average order $\left \langle k \right \rangle _L$ grows only logarithmically with the loop order, whereas the \emph{degree} of $T(G,N)$ grows linearly (see \cref{lem:TGN_bound} below). This means that, as $L$ increases, the polynomials are more and more dominated by terms of low order in $N$, that is, the coefficient of largest order in $N$ are vanishingly small compared to the coefficients at low order. 
It explains why it is possible that each fixed order in the large-$N$ expansion converges, while the full perturbative series (which would be the sum of all orders in the large-$N$ expansion) is factorially divergent. More details of the resummed perturbation series for the case of vacuum graphs  can be found in \cite{benedetti_smalln_2024}.

\section{Classification of leading $N$ dependence}\label{sec:leading}

\subsection{Dual graphs}\label{sec:dual}

In studying the 0-dimensional path integral in \cref{sec:0dim}, we have used the Hubbard-Stratonovich transformation \cite{stratonovich_method_1958,hubbard_calculation_1959,byczuk_generalized_2023} as a formal identity between certain integrals that share the same power series expansion (\cref{Hubbard_Stratonovich_transformation}). 
In the present section, we interpret the same transformation as a mapping between the $\O(N)$-symmetric $\phi^4$ theory, and a \enquote{dual} theory of a field $\sigma$ with interaction $\frac N 2 \hbar \ln \big( 1-\frac{\sigma}{\sqrt 3} \big) $ whose vacuum path integral (in the 0-dimensional case) is \cref{phi4_gaussian}. 
This is a common tool also used e.g. in tensor models \cite{bonzom_colored_2017,lionni_multicritical_2019, lionni_colored_2018}, loop vertex expansion \cite{rivasseau_loop_2018}, or for a duality between 3-dimensional $\phi^4$ theory and string theory \cite{klebanov_ads_2002}.

Both the original field $\phi$ and the dual field $\sigma$ admit a perturbative expansion in terms of Feynman graphs. These expansions produce the same power series when Feynman rules are applied, but the relevant graphs and their Feynman rules are different in both cases. In Feynman graphs of the original $\phi^4$ theory (\cref{fig:Z_graphs_vaccum}), all vertices are 4-valent and contribute $\hbar$, and the $N$ dependence arises from decompositions of these vertices (\cref{fig:vertex_decomposition}). 
The dual field $\sigma$ has interaction vertices of all valences $n\geq 1$ for which the Feynman rules give $\frac{(n-1)!}{2\cdot 3^{\frac n2}}N$. These vertex Feynman rules do not contain $\hbar$. Instead, each edge of $\sigma$ contributes a factor $\hbar$. The vertices of $\sigma$ correspond to the circuits in the decompositions of the graphs of $\phi$, and the edges of $\sigma$ correspond to the vertices of $\phi$. 
The mapping between graphs $G$ of $\phi^4$ theory and dual graphs $\dual$ of the $\sigma$ model is therefore one to many, in the sense that it does not map each graph $G$, but each decomposition of a graph $G$, to one dual graph $G^\star$.

\begin{definition}\label{def:dual}
	Let $G$ be a completion (\cref{def:completion_decompletion}). For a fixed decomposition of $G$ into circuits, the Hubbard-Stratonovich \emph{dual graph} $\dual$ is constructed as follows: Each of the circuits becomes a vertex of $\dual$ and for each vertex in $G$ adjacent to two (possibly identical) circuits there is an edge of $\dual$ connecting the corresponding vertices.
\end{definition}
Visually, each circuit in the decomposition of $G$ is being contracted to a vertex, preserving adjacency, see \cref{fig:dual}. In particular, 
\begin{align}\label{dual_edges_vertices}
	\abs{E_{\dual}}= \abs{V_G}.
\end{align}
\Cref{fig:dual} shows an example of non-isomorphic duals $G^\star$ that are obtained from distinct decompositions of the same $\phi^4$-graph $G$. 
The first few connected vacuum graphs of $\sigma$ are shown in \cref{fig:Z_graphs_vaccum_sigma}, they reproduce the terms of the connected graphs of $\phi$ in \cref{fig:Z_graphs_vaccum}.

Since the concept of \emph{planar} duality is an important ingredient to the large-$N$ expansion of many quantum field theories, we want to stress at this point that the duality of \cref{def:dual} is \emph{not} planar duality, and neither the graphs $G$ nor their dual graphs $G^\star$ are typically planar.

\begin{figure}[htb]
	\centering 
	\begin{tikzpicture}[scale=.75]
		
		\coordinate(x0) at (-4,3);
		\node[vertex](v1) at ($(x0) +(0,-1.2)$){};
		\node[vertex](v2) at ($(x0) +(0,0)$){};
		\node[vertex](v3) at ($(x0) +(0,1.2)$){};
		\node[vertex](v4) at ($(x0) +(-1.2,.6)$){};
		\node[vertex](v5) at ($(x0) +(-1.2,-.6)$){};
		\node[vertex](v6) at ($(x0) +(1.2,.6)$){};
		\node[vertex](v7) at ($(x0) +(1.2,-.6)$){};
		\draw[edge] (v1) -- (v4);
		\draw[edge] (v2) -- (v4);
		\draw[edge] (v3) -- (v4);
		\draw[edge] (v1) -- (v5);
		\draw[edge] (v2) -- (v5);
		\draw[edge] (v3) -- (v5);
		\draw[edge] (v4) -- (v5);
		\draw[edge] (v1) -- (v6);
		\draw[edge] (v2) -- (v6);
		\draw[edge] (v3) -- (v6);
		\draw[edge] (v1) -- (v7);
		\draw[edge] (v2) -- (v7);
		\draw[edge] (v3) -- (v7);
		\draw[edge] (v6) -- (v7);
		
		\draw[line width=.4mm, ->] ($(x0)+(1.7,0) $) --+(.7,0);
		
		\coordinate(x0) at (0,3);
		\coordinate (v1) at ($(x0) +(-1.2,.6)$){};
		\coordinate (v2) at ($(x0) +(-1.2,-.6)$){};
		\coordinate (v3) at ($(x0) +(0,1.2)$){};
		\coordinate (v4) at ($(x0) +(0,0)$){};
		\coordinate (v5) at ($(x0) +(0,-1.2)$){};
		\coordinate (v6) at ($(x0) +(1.2,.6)$){};
		\coordinate (v7) at ($(x0) +(1.2,-.6)$){};
		\node[vertex,lightgray] at (v1){};
		\node[vertex,lightgray] at (v2){};
		\node[vertex,lightgray] at (v3){};
		\node[vertex,lightgray] at (v4){};
		\node[vertex,lightgray] at (v5){};
		\node[vertex,lightgray] at (v6){};
		\node[vertex,lightgray] at (v7){};
		\draw[edge, rounded corners=8pt ] (v1) -- (v2)--(v4)-- cycle;
		\draw[edge, rounded corners=8pt ] (v4) -- (v6)--(v7) --cycle;
		\draw[edge, rounded corners=8pt ] (v1) -- (v3)--(v6)--(v5) --cycle;
		\draw[edge, rounded corners=8pt ] (v2) -- (v3)--(v7)--(v5) --cycle;
		
		\draw[line width=.4mm, ->] ($(x0)+(1.7,0) $) --+(.7,0);
		
		\coordinate(x0) at (4,3);
		\coordinate (v1) at ($(x0) +(-1.2,.6)$){};
		\coordinate (v2) at ($(x0) +(-1.2,-.6)$){};
		\coordinate (v3) at ($(x0) +(0,1.2)$){};
		\coordinate (v4) at ($(x0) +(0,0)$){};
		\coordinate (v5) at ($(x0) +(0,-1.2)$){};
		\coordinate (v6) at ($(x0) +(1.2,.6)$){};
		\coordinate (v7) at ($(x0) +(1.2,-.6)$){};
		\draw[thick, rounded corners=15pt, fill=red, fill opacity=.2] (v1) -- (v2)--(v4)-- cycle;
		\draw[thick, rounded corners=15pt, fill=green, fill opacity=.2] (v4) -- (v6)--(v7) --cycle;
		\draw[thick, rounded corners=15pt, fill=blue, fill opacity=.2] (v1) -- (v3)--(v6)--(v5) --cycle;
		\draw[thick, rounded corners=15pt, fill=yellow, fill opacity=.2] (v2) -- (v3)--(v7)--(v5) --cycle;
		
		\draw[line width=.4mm, ->] ($(x0)+(1.7,0) $) --+(.7,0);
		
		\coordinate(x0) at (8,3);
		\node[vertex,fill=red] (v1) at ($(x0) +(- .8,0)$){};
		\node[vertex,fill=blue] (v2) at ($(x0) +(0,.7)$){};
		\node[vertex,fill=yellow] (v3) at ($(x0) +(0,-.7)$){};
		\node[vertex,fill=green] (v4) at ($(x0) +(.8,0)$){};
		\draw[edge, bend angle =20,bend right] (v2) to (v3);
		\draw[edge, bend angle =20,bend left] (v2) to (v3);
		\draw[edge] (v1) to (v2);
		\draw[edge] (v1) to (v3);
		\draw[edge] (v1) to (v4);
		\draw[edge ] (v2) to (v4);
		\draw[edge ] (v3) to (v4);

		\coordinate(x0) at (-4,0);
		\node[vertex](v1) at ($(x0) +(0,-1.2)$){};
		\node[vertex](v2) at ($(x0) +(0,0)$){};
		\node[vertex](v3) at ($(x0) +(0,1.2)$){};
		\node[vertex](v4) at ($(x0) +(-1.2,.6)$){};
		\node[vertex](v5) at ($(x0) +(-1.2,-.6)$){};
		\node[vertex](v6) at ($(x0) +(1.2,.6)$){};
		\node[vertex](v7) at ($(x0) +(1.2,-.6)$){};
		\draw[edge] (v1) -- (v4);
		\draw[edge] (v2) -- (v4);
		\draw[edge] (v3) -- (v4);
		\draw[edge] (v1) -- (v5);
		\draw[edge] (v2) -- (v5);
		\draw[edge] (v3) -- (v5);
		\draw[edge] (v4) -- (v5);
		
		\draw[edge] (v1) -- (v6);
		\draw[edge] (v2) -- (v6);
		\draw[edge] (v3) -- (v6);
		\draw[edge] (v1) -- (v7);
		\draw[edge] (v2) -- (v7);
		\draw[edge] (v3) -- (v7);
		\draw[edge] (v6) -- (v7);
		
		\draw[line width=.4mm, ->] ($(x0)+(1.7,0) $) --+(.7,0);
		
		\coordinate(x0) at (0,0);
		\coordinate (v1) at ($(x0) +(-1.2,.6)$){};
		\coordinate (v2) at ($(x0) +(-1.2,-.6)$){};
		\coordinate (v3) at ($(x0) +(0,1.2)$){};
		\coordinate (v4) at ($(x0) +(0,0)$){};
		\coordinate (v5) at ($(x0) +(0,-1.2)$){};
		\coordinate (v6) at ($(x0) +(1.2,.6)$){};
		\coordinate (v7) at ($(x0) +(1.2,-.6)$){};
		\node[vertex,lightgray] at (v1){};
		\node[vertex,lightgray] at (v2){};
		\node[vertex,lightgray] at (v3){};
		\node[vertex,lightgray] at (v4){};
		\node[vertex,lightgray] at (v5){};
		\node[vertex,lightgray] at (v6){};
		\node[vertex,lightgray] at (v7){};
		\draw[edge, rounded corners=8pt] (v1) -- (v2)--(v4)-- (v6) -- (v3) -- cycle;
		\draw[edge, rounded corners=8pt] (v1) -- (v4)--(v7)-- (v5) --cycle;
		\draw[edge, rounded corners=8pt] (v2) -- (v5)--(v6)--(v7) -- (v3)--cycle;
		
		\draw[line width=.4mm, ->] ($(x0)+(1.7,0) $) --+(.7,0);
		
		\coordinate(x0) at (4,0);
		\coordinate (v1) at ($(x0) +(-1.2,.6)$){};
		\coordinate (v2) at ($(x0) +(-1.2,-.6)$){};
		\coordinate (v3) at ($(x0) +(0,1.2)$){};
		\coordinate (v4) at ($(x0) +(0,0)$){};
		\coordinate (v5) at ($(x0) +(0,-1.2)$){};
		\coordinate (v6) at ($(x0) +(1.2,.6)$){};
		\coordinate (v7) at ($(x0) +(1.2,-.6)$){};
		\draw[thick, rounded corners=15pt, fill=red, fill opacity=.2] (v1) -- (v2)--(v4)-- (v6) -- (v3) -- cycle;
		\draw[thick, rounded corners=15pt, fill=green, fill opacity=.2] (v1) -- (v4)--(v7)-- (v5) --cycle;
		\draw[thick, rounded corners=15pt, fill=blue, fill opacity=.2] (v2) -- (v5)--(v6)--(v7) -- (v3)--cycle;
		
		\draw[line width=.4mm, ->] ($(x0)+(1.7,0) $) --+(.7,0);
		
		\coordinate(x0) at (7.5,0);
		\node[vertex,fill=red] (v1) at ($(x0) +(- .2,.8)$){};
		\node[vertex,fill=blue] (v2) at ($(x0) +(.6,.3)$){};
		\node[vertex,fill=green] (v3) at ($(x0) +(.5,-.6)$){};
		\draw[edge, bend angle =50,bend left] (v1) to (v2);
		\draw[edge, bend angle =30,bend right] (v1) to (v2);
		\draw[edge, bend angle =10,bend left] (v1) to (v2);
		\draw[edge, bend angle =20,bend right] (v1) to (v3);
		\draw[edge, bend angle =50,bend right] (v1) to (v3);
		\draw[edge, bend angle =30,bend right] (v2) to (v3);
		\draw[edge, bend angle =30,bend left] (v2) to (v3);

		\coordinate(x0) at (-4,-3);
		\node[vertex](v1) at ($(x0) +(0,-1.2)$){};
		\node[vertex](v2) at ($(x0) +(0,0)$){};
		\node[vertex](v3) at ($(x0) +(0,1.2)$){};
		\node[vertex](v4) at ($(x0) +(-1.2,.6)$){};
		\node[vertex](v5) at ($(x0) +(-1.2,-.6)$){};
		\node[vertex](v6) at ($(x0) +(1.2,.6)$){};
		\node[vertex](v7) at ($(x0) +(1.2,-.6)$){};
		\draw[edge] (v1) -- (v4);
		\draw[edge] (v2) -- (v4);
		\draw[edge] (v3) -- (v4);
		\draw[edge] (v1) -- (v5);
		\draw[edge] (v2) -- (v5);
		\draw[edge] (v3) -- (v5);
		\draw[edge] (v4) -- (v5);
		
		\draw[edge] (v1) -- (v6);
		\draw[edge] (v2) -- (v6);
		\draw[edge] (v3) -- (v6);
		\draw[edge] (v1) -- (v7);
		\draw[edge] (v2) -- (v7);
		\draw[edge] (v3) -- (v7);
		\draw[edge] (v6) -- (v7);
		
		\draw[line width=.4mm, ->] ($(x0)+(1.7,0) $) --+(.7,0);
		
		\coordinate(x0) at (0,-3);
		\coordinate (v1) at ($(x0) +(-1.2,.6)$){};
		\coordinate (v2) at ($(x0) +(-1.2,-.6)$){};
		\coordinate (v3) at ($(x0) +(0,1.2)$){};
		\coordinate (v4) at ($(x0) +(0,0)$){};
		\coordinate (v5) at ($(x0) +(0,-1.2)$){};
		\coordinate (v6) at ($(x0) +(1.2,.6)$){};
		\coordinate (v7) at ($(x0) +(1.2,-.6)$){};
		\node[vertex,lightgray] at (v1){};
		\node[vertex,lightgray] at (v2){};
		\node[vertex,lightgray] at (v3){};
		\node[vertex,lightgray] at (v4){};
		\node[vertex,lightgray] at (v5){};
		\node[vertex,lightgray] at (v6){};
		\node[vertex,lightgray] at (v7){};
		\draw[edge, rounded corners=8pt] (v1) -- (v5)--(v7)-- (v3) -- (v6) --(v7) -- (v4) --(v2) --(v3) -- cycle;
		\draw[edge, rounded corners=8pt] (v1) -- (v2)--(v5)-- (v6) --(v4) --cycle;

		\draw[line width=.4mm, ->] ($(x0)+(1.7,0) $) --+(.7,0);
		
		\coordinate(x0) at (4,-3);
		\coordinate (v1) at ($(x0) +(-1.2,.6)$){};
		\coordinate (v2) at ($(x0) +(-1.2,-.6)$){};
		\coordinate (v3) at ($(x0) +(0,1.2)$){};
		\coordinate (v4) at ($(x0) +(0,0)$){};
		\coordinate (v5) at ($(x0) +(0,-1.2)$){};
		\coordinate (v6) at ($(x0) +(1.2,.6)$){};
		\coordinate (v7) at ($(x0) +(1.2,-.6)$){};
		\draw[thick, rounded corners=15pt, fill=red, fill opacity=.2](v1) -- (v5)--(v7)-- (v3) -- (v6) --(v7) -- (v4) --(v2) --(v3) -- cycle;
		\draw[thick, rounded corners=15pt, fill=green, fill opacity=.2] (v1) -- (v2)--(v5)-- (v6) --(v4) --cycle;

		\draw[line width=.4mm, ->] ($(x0)+(1.7,0) $) --+(.7,0);
		
		\coordinate(x0) at (7.5,-3);
		\node[vertex,fill=red] (v1) at ($(x0) +(.5,-.3)$){};
		\node[vertex,fill=green] (v2) at ($(x0) +(-.4,.3)$){};
		\draw[edge, bend angle =70,bend left] (v1) to (v2);
		\draw[edge, bend angle =30,bend left] (v1) to (v2);
		\draw[edge ] (v1) to (v2);
		\draw[edge, bend angle =30,bend right] (v1) to (v2);
		\draw[edge, bend angle =70,bend right] (v1) to (v2);
		\draw[edge ] (v1) ..controls +(.6,1.2) and +(1.5,0) ..(v1);
		\draw[edge ] (v1) ..controls +(1.3,-.4) and +(.2,-1.5) .. (v1);

		\coordinate(x0) at (-4,-6);
		\node[vertex](v1) at ($(x0) +(0,-1.2)$){};
		\node[vertex](v2) at ($(x0) +(0,0)$){};
		\node[vertex](v3) at ($(x0) +(0,1.2)$){};
		\node[vertex](v4) at ($(x0) +(-1.2,.6)$){};
		\node[vertex](v5) at ($(x0) +(-1.2,-.6)$){};
		\node[vertex](v6) at ($(x0) +(1.2,.6)$){};
		\node[vertex](v7) at ($(x0) +(1.2,-.6)$){};
		\draw[edge] (v1) -- (v4);
		\draw[edge] (v2) -- (v4);
		\draw[edge] (v3) -- (v4);
		\draw[edge] (v1) -- (v5);
		\draw[edge] (v2) -- (v5);
		\draw[edge] (v3) -- (v5);
		\draw[edge] (v4) -- (v5);
		
		\draw[edge] (v1) -- (v6);
		\draw[edge] (v2) -- (v6);
		\draw[edge] (v3) -- (v6);
		\draw[edge] (v1) -- (v7);
		\draw[edge] (v2) -- (v7);
		\draw[edge] (v3) -- (v7);
		\draw[edge] (v6) -- (v7);
		
		\draw[line width=.4mm, ->] ($(x0)+(1.7,0) $) --+(.7,0);
		
		\coordinate(x0) at (0,-6);
		\coordinate (v1) at ($(x0) +(-1.2,.6)$){};
		\coordinate (v2) at ($(x0) +(-1.2,-.6)$){};
		\coordinate (v3) at ($(x0) +(0,1.2)$){};
		\coordinate (v4) at ($(x0) +(0,0)$){};
		\coordinate (v5) at ($(x0) +(0,-1.2)$){};
		\coordinate (v6) at ($(x0) +(1.2,.6)$){};
		\coordinate (v7) at ($(x0) +(1.2,-.6)$){};
		\node[vertex,lightgray] at (v1){};
		\node[vertex,lightgray] at (v2){};
		\node[vertex,lightgray] at (v3){};
		\node[vertex,lightgray] at (v4){};
		\node[vertex,lightgray] at (v5){};
		\node[vertex,lightgray] at (v6){};
		\node[vertex,lightgray] at (v7){};
		\draw[edge, rounded corners=8pt] (v1) -- (v4)--(v2)-- (v3) -- (v7)  -- (v6) --(v3) --(v1) --(v5)--(v7) --(v4)--(v6)--(v5)--(v2)-- cycle;

		\draw[line width=.4mm, ->] ($(x0)+(1.7,0) $) --+(.7,0);
		
		\coordinate(x0) at (4,-6);
		\coordinate (v1) at ($(x0) +(-1.2,.6)$){};
		\coordinate (v2) at ($(x0) +(-1.2,-.6)$){};
		\coordinate (v3) at ($(x0) +(0,1.2)$){};
		\coordinate (v4) at ($(x0) +(0,0)$){};
		\coordinate (v5) at ($(x0) +(0,-1.2)$){};
		\coordinate (v6) at ($(x0) +(1.2,.6)$){};
		\coordinate (v7) at ($(x0) +(1.2,-.6)$){};
		\draw[thick, rounded corners=15pt, fill=red, fill opacity=.2] (v1) -- (v4)--(v2)-- (v3) -- (v7)  -- (v6) --(v3) --(v1) --(v5)--(v7) --(v4)--(v6)--(v5)--(v2)-- cycle;

		\draw[line width=.4mm, ->] ($(x0)+(1.7,0) $) --+(.7,0);
		
		\coordinate(x0) at (7.5,-6);
		\node[vertex,fill=red] (v1) at ($(x0) +(.2,0)$){};
		\draw[edge ] (v1) ..controls +(0:1.2) and +(50:1.2) ..(v1);
		\draw[edge ] (v1) ..controls +(51:1.2) and +(101:1.2) .. (v1);
		\draw[edge ] (v1) ..controls +(102:1.2) and +(152:1.2) ..(v1);
		\draw[edge ] (v1) ..controls +(154:1.2) and +(204:1.2) ..(v1);
		\draw[edge ] (v1) ..controls +(206:1.2) and +(256:1.2) ..(v1);
		\draw[edge ] (v1) ..controls +(257:1.2) and +(307:1.2) ..(v1);
		\draw[edge ] (v1) ..controls +(308:1.2) and +(358:1.2) ..(v1);

	\end{tikzpicture}
	\caption{Construction of four distinct duals $\dual$ for four possible choices of circuits of the completion $G$ shown in \cref{fig:completion}. For the same graph $G$, the duals $\dual$ can have different numbers of vertices, but in all cases the number of edges in $\dual$ equals the number of vertices in $G$ (\cref{dual_edges_vertices}). In particular, an Eulerian circuit in the original graph amounts to $\dual$ being a rose graph on a single vertex (bottom row).}
	\label{fig:dual}
\end{figure}
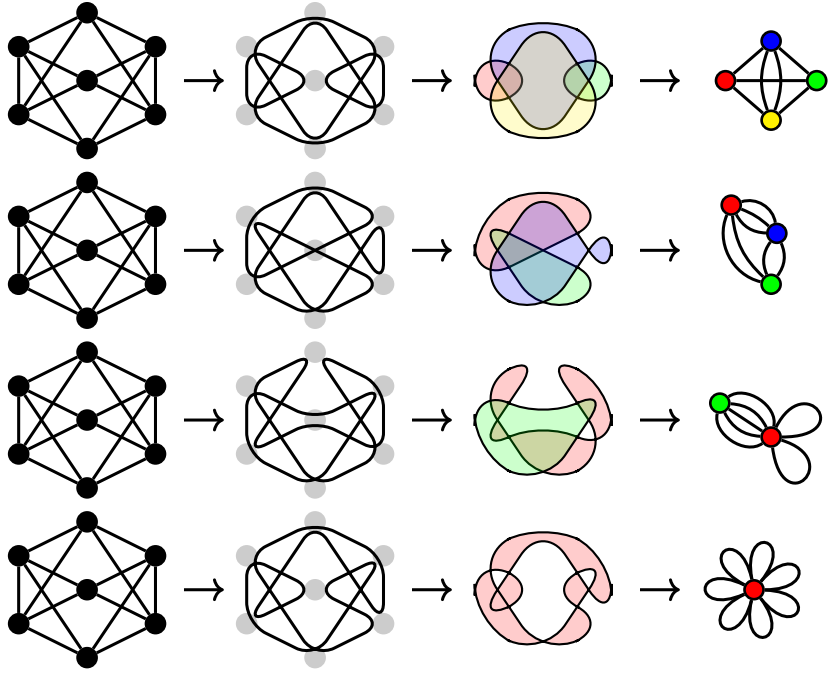

\begin{figure}[htb]
	\begin{tikzpicture}
		\coordinate(x0) at (0,0);
		\node at (x0) {$\emptyset$};
		\node at ($(x0) +(0,-1.5)$){$1$};

		\coordinate(x0) at (2,0);
		\node[vertex,fill=black] (v1) at ($(x0) +(0, .25)$){};
		\node[vertex,fill=black] (v2) at ($(x0) +(0, -.25)$){};
		\draw[edge ] (v1) --(v2); 
		\node at ($(x0) +(0,-1.5)$){$\frac 1 2 \hbar \left( \frac{N}{\sqrt{12}} \right)^2 $};

		\coordinate(x0) at (4,0);
		\node[vertex,fill=black] (v1) at ($(x0) +(0, -.2)$){}; 
		\draw[edge ] (v1) ..controls +(-.6,.7) and +(.6,.7) ..(v1); 
		\node at ($(x0) +(0,-1.5)$){$\frac 1 2 \hbar \frac{N}{6} $};
		
		\coordinate(x0) at (6.5,0);
		\node[vertex,fill=black] (v1) at ($(x0) +(0, .5)$){};
		\node[vertex,fill=black] (v2) at ($(x0) +(0, 0)$){};
		\node[vertex,fill=black] (v3) at ($(x0) +(0, -.5)$){};
		\draw[edge ] (v1) --(v2)--(v3); 
		\node at ($(x0) +(0,-1.5)$){$\frac 1 2 \hbar^2 \big( \frac{N}{\sqrt{12}} \big)^2 \frac N 6 $};

		\coordinate(x0) at (8.75,0);
		\node[vertex,fill=black] (v1) at ($(x0) +(0, 0)$){}; 
		\node[vertex,fill=black] (v2) at ($(x0) +(0, -.4)$){}; 
		\draw[edge ] (v1) ..controls +(-.6,.7) and +(.6,.7) ..(v1); 
		\draw[edge] (v1)--(v2);
		\node at ($(x0) +(0,-1.5)$){$\frac 1 2 \hbar^2 \frac{N}{\sqrt{27}} \frac{N}{\sqrt{12}} $};
		
		\coordinate(x0) at (11,0);
		\node[vertex,fill=black] (v1) at ($(x0) +(0, .4)$){}; 
		\node[vertex,fill=black] (v2) at ($(x0) +(0, -.4)$){}; 
		\draw[edge ] (v1) ..controls +(-.3,-.1) and +(-.3,.1) ..(v2); 
		\draw[edge ] (v1) ..controls +(.3,-.1) and +(.3,.1) ..(v2); 
		\node at ($(x0) +(0,-1.5)$){$\frac 1 4 \hbar^2 \big(\frac N 6\big)^2 $};

		\coordinate(x0) at (13,0);
		\node[vertex,fill=black] (v1) at ($(x0) +(0,0)$){};
		\draw[edge ] (v1) ..controls +(-.6,.7) and +(.6,.7) ..(v1);
		\draw[edge ] (v1) ..controls +(-.6,-.7) and +(.6,-.7) .. (v1);
		\node at ($(x0) +(0,-1.5)$){$\frac 1 8 \hbar^2 \frac N 3 $};

	\end{tikzpicture}
	\caption{Feynman graphs of the dual field $\sigma$, representing the first three terms of $\partial^0_jW|_0$ (\cref{W_derivatives}). The number of edges determines the order of $\hbar$, and the number of vertices is the order in $N$. Hence, this expansion produces the full $N$ dependence, whereas for the corresponding Feynman graphs of $\phi$ (the connected graphs in \cref{fig:Z_graphs_vaccum}), one needs to insert the $\O(N)$ symmetry factor $T(G,N)$ (\cref{def:TGN}) in addition. }
	\label{fig:Z_graphs_vaccum_sigma}
\end{figure}
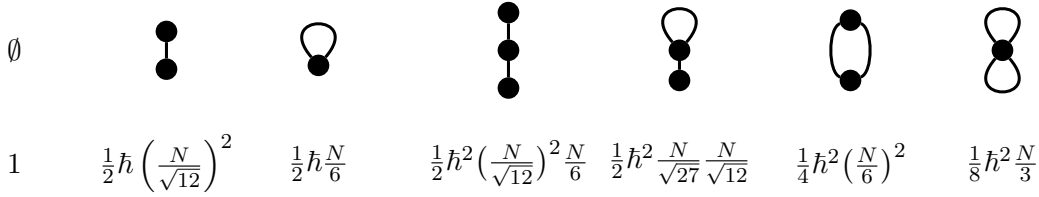

Our goal is to use the dual graphs $G^\star$ to understand the leading $N$ dependence of primitive graphs of $\phi^4$ theory. To that end, we need to understand how to reconstruct $G$ from a given $G^\star$, and under which conditions the resulting $G$ is  primitive (\cref{sec:periods}). We first establish two elementary lemmata. 

\begin{lemma}\label{lem:no_2valent}
	If $G$ is a primitive completion with $L>1$ loops, then the dual $\dual$ is connected and does not contain 1-valent or 2-valent vertices.
\end{lemma}
\begin{proof}
	By \cref{def:dual}, the Hubbard-Stratonovich transformation maps connected graphs to connected graphs. The valence of a vertex in $G^\star$ is larger than or equal to the number of vertices in the corresponding circuit in the decomposition of $G$. If a vertex of $G^\star$ had valence two, there would be a circuit in $G$ with at most two vertices. Such circuit is either a tadpole or a double edge, both of which are not allowed since $G$ is primitive. 
\end{proof}

\begin{lemma}\label{lem:TGN_bound}
	Let $G$ be a primitive completion (\cref{def:completion_decompletion}), where $L>1$ is the loop order of the decompletion $g=G\setminus \left \lbrace v \right \rbrace $. Then the degree of the $\O(N)$ symmetry factor (\cref{def:TGN}) is bounded in the following way:
	\begin{align*} 
		&\textnormal{Degree in $N$ of the polynomial } T(G,N) \leq N_\textnormal{max}:=  \frac{\abs{E_G}}{3} = \frac{2}{3}L + \frac 4 3 \\
		&\textnormal{Degree in $N$ of the polynomial } T(g,N) \leq N_\textnormal{max}-2 = \frac 2 3 L - \frac 2 3.
	\end{align*}
\end{lemma}
\begin{proof}
	In the absence of multiedges, the smallest number of edges in a circuit is three. A completion with $L$ loops has $L+2$ vertices and $2(L+2)$ edges. Distributing these edges into circuits of length three gives the claimed formula. 
	The bound for the decompletion $g=G\setminus\left \lbrace v \right \rbrace $ follows from \cref{lem:factorization_T}. 
	
	Alternatively, consider the dual $\dual$ (\cref{def:dual}). The number of vertices in $\dual$ equals the order in $N$ of the corresponding term in $T(G,N)$.  By \cref{lem:no_2valent}, each vertex in $\dual$ has valence at least 3, therefore $\abs{V_{\dual}}\leq \frac 23 \abs{E_{\dual}}$. A 4-regular completion with $L$ loops has $L+2$ vertices, hence by \cref{dual_edges_vertices} $\abs{E_{\dual}}=L+2$.
\end{proof}

\begin{example}
	The graph in \cref{fig:dual} has $L=5$ and hence $\lfloor N_\textnormal{max}\rfloor = \lfloor \frac{14}{3} \rfloor = 4$. The top row in the picture shows one choice of circuits that exhausts this bound. 
	
	 \Cref{lem:TGN_bound} does not hold for the decompletion at $L=1$, which is the fish graph (\cref{fig:fish}). The completion of the fish is a ring of three fish, which, according to \cref{ex:fish_chain}, has degree $3$ in $N$. This is larger than $N_\text{max}=\frac 2 3 \cdot 1 + \frac 4 3 = 2$.
\end{example}

\bigskip

For general $\dual$, the mapping $G \mapsto \dual$ of \cref{def:dual} is not invertible. However, it is possible to construct an inverse mapping for certain restricted classes of dual graphs;these happen to be the ones of leading order in $N$ and we will show that in this case one can construct $G$ from the \emph{line graph} of $\dual$.

\begin{definition}\label{def:line_graph}
	Let $\dual$ be a graph. The \emph{line graph} $L(\dual)$ is a graph where every edge of $\dual$ is a vertex of $L(\dual)$, and two vertices of $L(\dual)$ are connected by an edge if and only if the corresponding edges in $\dual$ are adjacent to the same vertex in $\dual$.
\end{definition}

Note that for general $G^\star$, the line graph $L(G^\star)$ is not 4-regular.

\begin{example}
	Every 4-regular graph $G$ with $n$ vertices has at least one Eulerian circuit, that is, a circuit that traverses every edge exactly once. The dual $\dual$ of such a decomposition is a rose graph with $n$ petals, shown in the last row of \cref{fig:dual}. Consequently, every 4-regular graph $G$ has a rose graph $G^\star$ among its duals. From this $\dual$, it is not possible to infer any information about $G$ except for the number of vertices. 
	
	The line graph $L(\dual)$ of the rose on $n$ petals is the complete graph $K_n$, which is not 4-regular in general. 
\end{example}

\subsection{The duals of leading-order primitive graphs}\label{sec:duals_leading}

\subsubsection[Primitive graphs with 3n+1 loops]{Primitive graphs with $3n+1$ loops} \label{sec:leading_1}

We first consider those loop orders where the ratio $N_\text{max}$ in \cref{lem:TGN_bound} is an integer, namely $L \in \left \lbrace 4, 7, 10, 13, \ldots \right \rbrace $, or generally $L = 3n+1$ for $n\in \mathbb N$. To exhaust the bound $N_\text{max}$, a dual graph  $\dual$ (\cref{def:dual}) must have the maximum number of vertices for a given number of edges, 
\begin{align}\label{leading_EV_bound}
	\abs{V_{\dual}} &\overset != \frac{2 L_G + 4}{3} =\frac 2 3  \abs{E_{\dual}} .
\end{align}
In the absence of 2-valent vertices (\cref{lem:no_2valent}),  \cref{leading_EV_bound} implies that all vertices of $\dual$ are 3-valent. Hence, only such completions $G$ can contribute at leading order in $N$ which have a 3-regular dual graph $\dual$. The interesting part is the converse: Under which conditions can we obtain a primitive completion $G$ from a 3-regular $\dual$? We first show that the leading decompositions are unique:

\begin{lemma}\label{lem:leading_no_multiple}
	Let $L=3n+1$ for $n\in \mathbb N$. 
    Let $G$ be a  primitive completion that has a decomposition of leading order $N_\textnormal{max}$ (\cref{lem:TGN_bound}). 
    At $L=4$, or $n=1$,  $G$ is the six-vertex zigzag (the octahedron graph, \cref{fig:leading_no_multiple_2}), and it has two isomorphic decompositions. If $n>1$, there is exactly one decomposition at leading order.
\end{lemma}
\begin{proof}
	The decompositions at leading order are such that every circuit has length three. Hence, for there to be more than one such decomposition, it must be possible to decompose $G$ into triangles in more than one way. 
	
	Consider a vertex $v\in G$. Assume that all three decompositions (\cref{fig:vertex_decomposition}) of that vertex give rise to a term of leading order in $T(G,N)$. This means that in all three cases, the pairs of edges must belong to triangles. This is only possible if all four neighboring vertices to $v$ are also adjacent to each other, see \cref{fig:leading_no_multiple_1}. Hence, $G$ is $K_5$. But $K_5$ has five vertices and thus $L=3$ in contradiction to the assumption $L=3n+1$. Consequently, it is impossible that all three decompositions of any vertex contribute at leading order.
	
	Now assume that there are two compositions that contribute at leading order. Then, the four neighboring vertices of $v$ each have to be adjacent to two (but not all three) others. In order to be a decomposition at leading order for the whole graph, these vertices each have to be adjacent another triangle in each of the two configurations, see \cref{fig:leading_no_multiple_2}. The only way to realize this is to have all four vertices adjacent to one vertex $v'$. Hence, $G$ is the octahedron.
\end{proof}

\begin{figure}[htbp]
	\centering
	\begin{tikzpicture}

		\coordinate(x0) at (-4,0);
		
		\coordinate(v0) at ($(x0) $){};
		\coordinate(v1) at ($(x0) +(60:1)$){};
		\coordinate(v2) at ($(x0) +(120:1)$){};
		\coordinate(v3) at ($(x0) +(210:1)$){};
		\coordinate(v4) at ($(x0) +(330:1)$){}; 
		\draw[thick,red] (v1) -- (v2);
		\draw[thick,red] (v3) -- (v4);
		\node[vertex,draw=lightgray,fill=lightgray, label=below:$v$] at (v0){};
		\node[vertex,draw=lightgray,fill=lightgray] at (v1){};
		\node[vertex,draw=lightgray,fill=lightgray] at (v2){};
		\node[vertex,draw=lightgray,fill=lightgray] at (v3){};
		\node[vertex,draw=lightgray,fill=lightgray] at (v4){};
		
		\draw[edge, rounded corners=5pt ] (v1) -- (v0)--(v2);
		\draw[edge, rounded corners=7pt ] (v3) -- (v0)--(v4);

		\coordinate(x0) at (-1,0);
		
		\coordinate(v0) at ($(x0) $){};
		\coordinate(v1) at ($(x0) +(60:1)$){};
		\coordinate(v2) at ($(x0) +(120:1)$){};
		\coordinate(v3) at ($(x0) +(210:1)$){};
		\coordinate(v4) at ($(x0) +(330:1)$){};
		
		\draw[thick,red] (v1) -- (v4);
		\draw[thick,red] (v2) -- (v3); 
		\node[vertex,draw=lightgray,fill=lightgray, label=below:$v$] at (v0){};
		\node[vertex,draw=lightgray,fill=lightgray] at (v1){};
		\node[vertex,draw=lightgray,fill=lightgray] at (v2){};
		\node[vertex,draw=lightgray,fill=lightgray] at (v3){};
		\node[vertex,draw=lightgray,fill=lightgray] at (v4){};
		
		\draw[edge, rounded corners=5pt ] (v1) -- (v0)--(v4);
		\draw[edge, rounded corners=7pt ] (v2) -- (v0)--(v3);

		\coordinate(x0) at (2,0);
		
		\coordinate(v0) at ($(x0) $){};
		\coordinate(v1) at ($(x0) +(60:1)$){};
		\coordinate(v2) at ($(x0) +(120:1)$){};
		\coordinate(v3) at ($(x0) +(210:1)$){};
		\coordinate(v4) at ($(x0) +(330:1)$){}; 
		\draw[thick,red] (v1) -- (v3);
		\draw[thick,red] (v2) -- (v4);
		\node[vertex,draw=lightgray,fill=lightgray, label=below:$v$] at (v0){};
		\node[vertex,draw=lightgray,fill=lightgray] at (v1){};
		\node[vertex,draw=lightgray,fill=lightgray] at (v2){};
		\node[vertex,draw=lightgray,fill=lightgray] at (v3){};
		\node[vertex,draw=lightgray,fill=lightgray] at (v4){};
		
		\draw[edge, rounded corners=5pt ] (v1) -- (v0)--(v3);
		\draw[edge, rounded corners=7pt ] (v2) -- (v0)--(v4);

		\draw[edge, ->] (4,0)--(5,0);
		
		\coordinate(x0) at (6.5,0);
		
		\coordinate(v0) at ($(x0) $){};
		\coordinate(v1) at ($(x0) +(60:1)$){};
		\coordinate(v2) at ($(x0) +(120:1)$){};
		\coordinate(v3) at ($(x0) +(210:1)$){};
		\coordinate(v4) at ($(x0) +(330:1)$){};
		
		\draw[thick,red] (v1) -- (v4);
		\draw[thick,red] (v2) -- (v3); 
		\node[vertex, label=below:$v$] at (v0){};
		\node[vertex ] at (v1){};
		\node[vertex ] at (v2){};
		\node[vertex ] at (v3){};
		\node[vertex ] at (v4){};
		
		\draw[edge] (v1) -- (v0)--(v4)--(v1);
		\draw[edge] (v2) -- (v0)--(v3) -- (v2);
		\draw[edge ] (v1) -- (v2)--(v3)--(v4)--(v1);
		\draw[edge] (v1)--(v3);
		\draw[edge](v2)--(v4);

	\end{tikzpicture}
	\caption{Proof of \cref{lem:leading_no_multiple}: If we assume that each of the three decompositions of a vertex $v$ leads to  two triangles, the graph must contain all of the edges indicated in red. This implies that  the graph is necessarily $K_5$, which has $L=3$.}
	\label{fig:leading_no_multiple_1}
\end{figure}
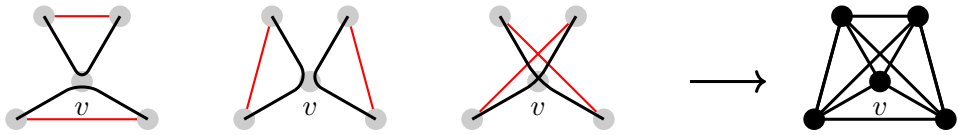

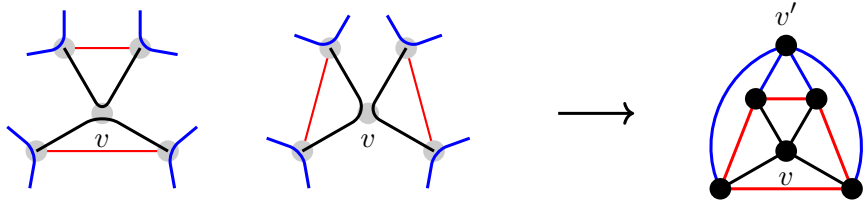
\begin{figure}[htbp]
	\centering
	\begin{tikzpicture}

		\coordinate(x0) at (-4,0);
		
		\coordinate(v0) at ($(x0) $){};
		\coordinate(v1) at ($(x0) +(60:1)$){};
		\coordinate(v2) at ($(x0) +(120:1)$){};
		\coordinate(v3) at ($(x0) +(210:1)$){};
		\coordinate(v4) at ($(x0) +(330:1)$){}; 
		\draw[thick,red] (v1) -- (v2);
		\draw[thick,red] (v3) -- (v4);
		\node[vertex,draw=lightgray,fill=lightgray, label=below:$v$] at (v0){};
		\node[vertex,draw=lightgray,fill=lightgray] at (v1){};
		\node[vertex,draw=lightgray,fill=lightgray] at (v2){};
		\node[vertex,draw=lightgray,fill=lightgray] at (v3){};
		\node[vertex,draw=lightgray,fill=lightgray] at (v4){};
		
		\draw[edge, rounded corners=5pt ] (v1) -- (v0)--(v2);
		\draw[edge, rounded corners=7pt ] (v3) -- (v0)--(v4);

		\draw[edge,blue, rounded corners=5pt] ($(v1)+(90:.5)$) --(v1) -- ($(v1) +(-10:.5) $);
		\draw[edge,blue, rounded corners=5pt] ($(v2)+(90:.5)$) --(v2) -- ($(v2) +(190:.5) $); 
		\draw[edge,blue, rounded corners=5pt] ($(v3)+(140:.5)$) --(v3) -- ($(v3) +(260:.5) $); 
		\draw[edge,blue, rounded corners=5pt] ($(v4)+( 40:.5)$) --(v4) -- ($(v4) +(280:.5) $); 
		
		\coordinate(x0) at (-.5,0);
		
		\coordinate(v0) at ($(x0) $){};
		\coordinate(v1) at ($(x0) +(60:1)$){};
		\coordinate(v2) at ($(x0) +(120:1)$){};
		\coordinate(v3) at ($(x0) +(210:1)$){};
		\coordinate(v4) at ($(x0) +(330:1)$){};
		
		\draw[thick,red] (v1) -- (v4);
		\draw[thick,red] (v2) -- (v3); 
		\node[vertex,draw=lightgray,fill=lightgray, label=below:$v$] at (v0){};
		\node[vertex,draw=lightgray,fill=lightgray] at (v1){};
		\node[vertex,draw=lightgray,fill=lightgray] at (v2){};
		\node[vertex,draw=lightgray,fill=lightgray] at (v3){};
		\node[vertex,draw=lightgray,fill=lightgray] at (v4){};
		
		\draw[edge, rounded corners=5pt ] (v1) -- (v0)--(v4);
		\draw[edge, rounded corners=7pt ] (v2) -- (v0)--(v3);
		\draw[edge,blue, rounded corners=5pt] ($(v1)+(120:.5)$) --(v1) -- ($(v1) +(20:.5) $);
		\draw[edge,blue, rounded corners=5pt] ($(v2)+(60:.5)$) --(v2) -- ($(v2) +(160:.5) $); 
		\draw[edge,blue, rounded corners=5pt] ($(v3)+(160:.5)$) --(v3) -- ($(v3) +(280:.5) $); 
		\draw[edge,blue, rounded corners=5pt] ($(v4)+(20:.5)$) --(v4) -- ($(v4) +(260:.5) $);

		\draw[edge, ->] (2,0)--(3,0);
		
		\coordinate(x0) at (5 ,-.5);
		
		\coordinate(v0) at ($(x0) $){};
		\coordinate(v1) at ($(x0) +(60:.8)$){};
		\coordinate(v2) at ($(x0) +(120:.8)$){};
		\coordinate(v3) at ($(x0) +(210:1)$){};
		\coordinate(v4) at ($(x0) +(330:1)$){};
		\coordinate(v5) at ($(x0) +(90:1.4)$){};
		
		\draw[edge,red] (v1) -- (v2) -- (v3) -- (v4)--(v1);
		\draw[edge] (v1) -- (v0) -- (v2) ;
		\draw[edge] (v3) -- (v0)--(v4);
		\draw[edge,blue ] (v1) -- (v5);
		\draw[edge,blue] (v2) -- (v5);
		\draw[edge,blue,bend angle=50,bend left] (v3) to (v5);
		\draw[edge,blue,bend angle=50,bend right] (v4) to (v5);
		
		\node[vertex, label=below:$v$] at (v0){};
		\node[vertex] at (v1){};
		\node[vertex] at (v2){};
		\node[vertex] at (v3){};
		\node[vertex] at (v4){};
		\node[vertex,label=above:$v'$] at (v5){};

	\end{tikzpicture}
	\caption{Proof of \cref{lem:leading_no_multiple}: If we assume that two decompositions of a vertex $v$ lead to  two triangles each, the graph must contain the red edges. When the four red edges are present, they, together with $v$, form a 4-valent subgraph. The only way to obtain a primitive (=internally six-edge connected) graph is to connect the four blue edges to a single vertex $v'$. The resulting graph  is the octahedron (=six-vertex zigzag).}
	\label{fig:leading_no_multiple_2}
\end{figure}

\begin{theorem}\label{thm:leading_3n+1}
	Let  $L = 3n+1$ for $n\in \mathbb N$. The primitive completed graphs $G$ which contribute at leading order in $N$ are the line graphs (\cref{def:line_graph}) of the 3-edge-connected 3-regular  graphs $\dual$ on $2n+2$ vertices.
\end{theorem}
\begin{proof}
	By construction and \cref{lem:TGN_bound}, the dual graphs $\dual$ of the leading contributions in $T(G,N)$ are 3-regular. We need to show two things: 1. When $\dual$ is furthermore three-connected, then its line graph is primitive. And 2. this construction produces all primitive $G$ that contribute at leading order in $N$. 
	
	1. If $\dual$ is 3-edge-connected, then every edge in $\dual$ is adjacent to two \emph{distinct} vertices. These, in turn, are adjacent to a total of four distinct other edges. Hence, the line graph $L(\dual)$ is 4-regular, and every vertex of $L(\dual)$ is adjacent to another four distinct vertices. In particular, there are no multiedges and no tadpoles. Furthermore, since every vertex of $\dual$ is 3-valent, it corresponds to a triangle of edges in $L(G^\star)$. Since every edge in $\dual$ is adjacent to four distinct edges, every vertex of $L(\dual)$ is adjacent to two triangles which have no edges in common (there may be more than two triangles in total if they share edges), as shown below:
	\begin{center}
		\begin{tikzpicture}
			\node at (-2,0){$\dual \supseteq $};
			\node[vertex] (v1) at (-0.4,0) {};
			\node[vertex] (v2) at (0.4,0) {};
			\node[vertex] (v3) at (-0.8,.4) {};
			\node[vertex] (v4) at (-0.8,-.4) {};
			\node[vertex] (v5) at (0.8,.4) {};
			\node[vertex] (v6) at (0.8,-.4) {};
			\draw[edge] (v1)--(v2);
			\draw[edge] (v3)--(v1) --(v4);
			\draw[edge] (v5)--(v2) --(v6);
			
			\node at (1.6,0){$\longleftrightarrow$};
			
			\node[vertex] (v1) at (3,0){};
			\node[vertex] (v2) at (2.5,-.3){};
			\node[vertex] (v3) at (2.5,.3){};
			\node[vertex] (v4) at (3.5,-.3){};
			\node[vertex] (v5) at (3.5,.3){};
			
			\draw[edge] (v1) -- (v2) -- (v3) -- (v1);
			\draw[edge] (v1) -- (v4) -- (v5) -- (v1);
			\draw[edge] (v2) -- +(-.2,-.1);
			\draw[edge] (v3) -- +(-.2,.1);
			\draw[edge] (v2) -- +(.1,-.2);
			\draw[edge] (v3) -- +(.1,.2);
			\draw[edge] (v4) -- +(.2,-.1);
			\draw[edge] (v5) -- +(.2,.1);
			\draw[edge] (v4) -- +(-.1,-.2);
			\draw[edge] (v5) -- +(-.1,.2);
			
			\node at (5,0){$\subseteq L(\dual)$};
			
		\end{tikzpicture}
	\end{center}
	
	When $L(\dual)$ is 4-regular, and every vertex is adjacent to at least two edge-disjoint triangles,  then every edge of $L(\dual)$ is part of at least one triangle. 
    Now assume that a subgraph $\gamma \subset L(\dual)$ is connected to its complement through four edges. Since each of the four edges must be part of some triangle, they must be arranged such that two of them share one triangle and the other two share another triangle. Hence, $\gamma$ is connected to its complement in $L(\dual)$ through two vertices as shown below:
    	\begin{center}
    	\begin{tikzpicture}
    		 
    		\node[circle,minimum size=1.1cm,fill=gray] (v1) at (0.0,0) {$\gamma$};
    		\node[vertex] (v3) at (-0.4,.4) {};
    		\node[vertex] (v4) at (-0.4,-.4) {};
    		\node[vertex] (v5) at (0.4,.4) {};
    		\node[vertex] (v6) at (0.4,-.4) {};
    		
    		\node[vertex] (v7) at (-1,0) {};
    		\node[vertex] (v8) at (1,0) {};
    		
    		\draw[edge] (v3)--(v7) --(v4);
    		\draw[edge] (v5)--(v8) --(v6);
    		\draw[edge](v3)--(v4);
    		\draw[edge](v5)--(v6);
    		\draw[edge] (v7) -- +(-.2,.2);
    		\draw[edge] (v7) -- +(-.2,-.2);
    		\draw[edge] (v8) -- +(.2,.2);
    		\draw[edge] (v8) -- +(.2,-.2);
    		
    		\draw[edge] (v3) -- +(.2,-.1);
    		\draw[edge] (v4) -- +(.2,.1);
    		\draw[edge] (v3) -- +(.25,.05);
    		\draw[edge] (v4) -- +(.25,-.05);
    		\draw[edge] (v5) -- +(-.2,-.1);
    		\draw[edge] (v6) -- +(-.2,.1);
    		\draw[edge] (v5) -- +(-.25,.05);
    		\draw[edge] (v6) -- +(-.25,-.05);

    		\node at (2.5,0){$\subseteq L(\dual)$};
    		
    	\end{tikzpicture}
    \end{center}
    
     Removing these two vertices disconnects the two parts.  
    Therefore, removing the two corresponding edges in $\dual$ disconnects $\dual$. Consequently,   $\dual$ is two-edge-reducible.  Analogously, a subgraph of valence two in $L(\dual)$ leads to $\dual$ being one-edge-reducible. We conclude that, if $\dual$ is 3-edge-connected, no such subgraphs can be present in~$L(\dual)$. Therefore $L(\dual)$ is primitive.

	2. By \cref{lem:leading_no_multiple}, apart from the special case, each  primitive graph $G$ has at most one decomposition at leading order. Conversely, knowing the dual $\dual$, we can reconstruct all triangles of the original graph $G$ and how they are connected. But for a triangle, knowing its adjacent vertices is enough to fix all edges in the triangle, consequently, all edges of $G$ are fixed. Hence, at leading order in $N$, knowing $\dual$ is enough to uniquely fix $G$, and therefore no two distinct $G$ lead to the same $\dual$.
	
	Note that the bijectivity rests on the  vertices of $\dual$ corresponding to \emph{triangles}: For a cycle of larger size, one would have multiple ways to arrange the edges while still being incident to the same set of vertices, compare \cref{fig:4vertex_decompositions} (a).
\end{proof}

For $L=4$ and $L=7$, the duals and their corresponding primitive decompositions are shown in \cref{fig:leading_dual_graphs}.
We have written a small program, based on \texttt{nauty} \cite{mckay_practical_2014}, to generate the   dual graphs and construct the leading primitives, the counts  are shown  in \cref{tab:leading_graphs_count} below.

\begin{figure}[htb]
	\centering
	\begin{tikzpicture}[scale=.7]
		\coordinate(x0) at (-4,3.2);
		\node[vertex,green,draw=black](v1) at ($(x0) $){};
		\node[vertex,yellow,draw=black](v2) at ($(x0) +(90:1)$){};
		\node[vertex,red,draw=black](v3) at ($(x0) +(210:1)$){};
		\node[vertex,blue,draw=black](v4) at ($(x0) +(330:1)$){};
		\draw[edge] (v1) -- (v4);
		\draw[edge] (v2) -- (v4);
		\draw[edge] (v3) -- (v4);
		\draw[edge] (v1) -- (v2);
		\draw[edge] (v2) -- (v3);
		\draw[edge] (v1) -- (v3);
		
		\coordinate(x0) at (1,3.2);
		\node[vertex,orange,draw=black](v1) at ($(x0) +(0:1) $){};
		\node[vertex,white,draw=black](v2) at ($(x0) +(60:.7)$){};
		\node[vertex,green,draw=black](v3) at ($(x0) +(120:.7)$){};
		\node[vertex,blue,draw=black](v4) at ($(x0) +(180:1)$){};
		\node[vertex,yellow,draw=black](v5) at ($(x0) +(240:.7)$){};
		\node[vertex,red,draw=black](v6) at ($(x0) +(300:.7)$){};
		\draw[edge] (v1) -- (v2);
		\draw[edge] (v2) -- (v3);
		\draw[edge] (v3) -- (v4);
		\draw[edge] (v4) -- (v5);
		\draw[edge] (v5) -- (v6);
		\draw[edge] (v6) -- (v1);
		\draw[edge] (v3) -- (v6);
		\draw[edge] (v2) -- (v5);
		\draw[edge] (v1) ..controls ($(v1) +(90:1.5) $) and ($(v4) +(90:1.5) $).. (v4);
		
		\coordinate(x0) at (5,3.2);
		\node[vertex,white,draw=black](v1) at ($(x0) +(0:1) $){};
		\node[vertex,orange,draw=black](v2) at ($(x0) +(60:.7)$){};
		\node[vertex, green,draw=black](v3) at ($(x0) +(120:.7)$){};
		\node[vertex, blue, draw=black](v4) at ($(x0) +(180:1)$){};
		\node[vertex, red, draw=black](v5) at ($(x0) +(240:.7)$){};
		\node[vertex, yellow, draw=black](v6) at ($(x0) +(300:.7)$){};
		\draw[edge] (v1) -- (v2);
		\draw[edge] (v2) -- (v3);
		\draw[edge] (v3) -- (v4);
		\draw[edge] (v4) -- (v5);
		\draw[edge] (v5) -- (v6);
		\draw[edge] (v6) -- (v1);
		\draw[edge] (v3) -- (v5);
		\draw[edge] (v2) -- (v6);
		\draw[edge] (v1) ..controls ($(v1) +(90:1.5) $) and ($(v4) +(90:1.5) $).. (v4);

		\coordinate(x0) at (-4,0);
		
		\coordinate(v1) at ($(x0) +(30:1.5)$){};
		\coordinate(v2) at ($(x0) +(150:1.5)$){};
		\coordinate(v3) at ($(x0) +(270:1.5)$){};
		\coordinate(v4) at ($(x0) +(90:.3)$){};
		\coordinate(v5) at ($(x0) +(210:.3)$){};
		\coordinate(v6) at ($(x0) +(330:.3)$){};
		\node[vertex,lightgray] at (v1){};
		\node[vertex,lightgray] at (v2){};
		\node[vertex,lightgray] at (v3){};
		\node[vertex,lightgray] at (v4){};
		\node[vertex,lightgray] at (v5){};
		\node[vertex,lightgray] at (v6){};
		
		\draw[edge, rounded corners=2pt, fill=green, fill opacity=.2] (v4) -- (v5)--(v6) -- cycle;
		\draw[edge, rounded corners=4pt, fill=yellow, fill opacity=.1 ] (v1) -- (v2)--(v4) -- cycle;
		\draw[edge, rounded corners=4pt, fill=blue, fill opacity=.1 ] (v1) -- (v3)--(v6) -- cycle;
		\draw[edge, rounded corners=4pt, fill=red, fill opacity=.3] (v2) -- (v3)--(v5) -- cycle;

		\coordinate(x0) at (1,0);
		\coordinate(v1) at ($(x0) +(90:1.5)$){};
		\coordinate(v2) at ($(x0) +(150:1.5)$){};
		\coordinate(v3) at ($(x0) +(210:1.5)$){};
		\coordinate(v4) at ($(x0) +(270:1.5)$){};
		\coordinate(v5) at ($(x0) +(330:1.5)$){};
		\coordinate(v6) at ($(x0) +(30:1.5)$){};
		\coordinate(v7) at ($(x0) +( 90:.5)$){};
		\coordinate(v8) at ($(x0) +(210:.5)$){};
		\coordinate(v9) at ($(x0) +(330:.5)$){};
		\node[vertex,lightgray] at (v1){};
		\node[vertex,lightgray] at (v2){};
		\node[vertex,lightgray] at (v3){};
		\node[vertex,lightgray] at (v4){};
		\node[vertex,lightgray] at (v5){};
		\node[vertex,lightgray] at (v6){};
		\node[vertex,lightgray] at (v7){};
		\node[vertex,lightgray] at (v8){};
		\node[vertex,lightgray] at (v9){};
		
		\draw[edge, rounded corners=4pt, fill=green, fill opacity=.1 ] (v1) -- (v8) -- (v2) -- cycle;
		\draw[edge, rounded corners=4pt, fill=red, fill opacity=.3] (v8) -- (v4) -- (v5) -- cycle;
		\draw[edge, rounded corners=4pt, fill=blue, fill opacity=.1] (v3) -- (v2) -- (v7) -- cycle;
		\draw[edge, rounded corners=4pt, fill=orange, fill opacity=.2] (v5) -- (v6) -- (v7) -- cycle;
		\draw[edge, rounded corners=4pt, fill=white, fill opacity=.15] (v1) -- (v6) -- (v9) -- cycle;
		\draw[edge, rounded corners=4pt, fill=yellow, fill opacity=.15] (v3) -- (v4) -- (v9) -- cycle;

		\coordinate(x0) at (5,0);

		\coordinate(v1) at ($(x0) +(90:1.5)$){};
		\coordinate(v2) at ($(x0) +(150:1.5)$){};
		\coordinate(v3) at ($(x0) +(210:1.5)$){};
		\coordinate(v4) at ($(x0) +(270:1.5)$){};
		\coordinate(v5) at ($(x0) +(330:1.5)$){};
		\coordinate(v6) at ($(x0) +(30:1.5)$){};
		\coordinate(v7) at ($(x0) +( 180:.7)$){};
		\coordinate(v8) at ($(x0) +(0:0)$){};
		\coordinate(v9) at ($(x0) +(0:.7)$){};
		\node[vertex,lightgray] at (v1){};
		\node[vertex,lightgray] at (v2){};
		\node[vertex,lightgray] at (v3){};
		\node[vertex,lightgray] at (v4){};
		\node[vertex,lightgray] at (v5){};
		\node[vertex,lightgray] at (v6){};
		\node[vertex,lightgray] at (v7){};
		\node[vertex,lightgray] at (v8){};
		\node[vertex,lightgray] at (v9){};
		
		\draw[edge, rounded corners=5pt, fill=green, fill opacity=.1 ] (v1) -- (v7) -- (v2) -- cycle;
		\draw[edge, rounded corners=5pt, fill=red, fill opacity=.3] (v3) -- (v4) -- (v7) -- cycle;
		\draw[edge, rounded corners=8pt, fill=blue, fill opacity=.1] (v2) -- (v8) -- (v3) -- cycle;
		\draw[edge, rounded corners=5pt, fill=orange, fill opacity=.2] (v1) -- (v9) -- (v6) -- cycle;
		\draw[edge, rounded corners=8pt, fill=white, fill opacity=.15] (v6) -- (v5) -- (v8) -- cycle;
		\draw[edge, rounded corners=5pt, fill=yellow, fill opacity=.15] (v4) -- (v5) -- (v9) -- cycle;

	\end{tikzpicture}
	\caption{The 3-regular three-connected graphs $\dual$ on four and six vertices. Below them are their line graphs, which are the leading-order decompositions of completions $G$ at four and seven loops.}
	\label{fig:leading_dual_graphs}
\end{figure}
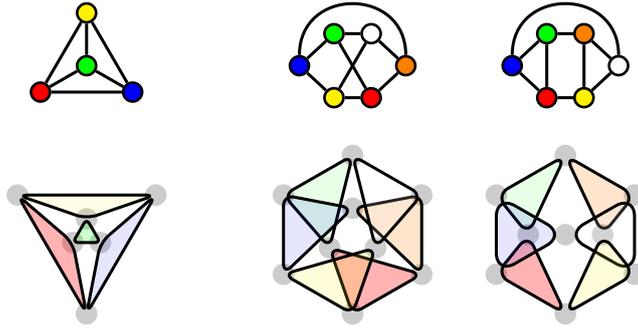

\subsubsection{Counts of  3-connected cubic graphs}

By \cref{thm:leading_3n+1},  the results regarding the $N$-dependent sums of primitive graphs obtained in \cref{sec:0dim} imply statements about sums of 3-connected cubic graphs. 
\begin{proposition}\label{thm:cubic_count}
	Let $2<n\in \mathbb{N}$, and let $C_{2n}$ be the set of 3-connected cubic graphs on   $2n$ vertices, and let $S_n$ be the sum of their symmetry factors. Then
	\begin{align*}
	S_{n}:=\sum_{G^\star\in C_{2n}} \frac{1}{\abs{\Aut(G^\star)}} &=	\frac{3^{3n-2}}{4! \; n} [N^{2n-1}]p_{3n-2}(N),
	\end{align*}
	where the polynomial $p_L(N)$ is the $L-$loop sum of primitives computed in \cref{primitive_generating_function}.
\end{proposition}
\begin{proof}
	By \cref{thm:leading_3n+1}, the 3-connected cubic graphs $G^\star$ on $2n$ vertices are dual to the leading-order decompositions  of completions $G$   whose decompletions $g$ have $L=3(n-1)+1=3n-2$ loops. According to \cref{lem:TGN_bound}, the leading coefficient of $T(g,N)$ has order $N^{\frac 2 3 L - \frac 2 3}=N^{2n-2}$, and according to \cref{lem:leading_no_multiple}, its coefficient is unity (except $n=1$). According to \cref{def:TGN}, the polynomial $T(g,N)$ includes a normalization factor $3^{-\abs{V_g}}$, where $\abs{V_g}=L+1=3n-1$ is the number of vertices of the decompletion. Combining all these statements, it follows that for a decompletion $g$ at $L=3n-2$,  the coefficient
	\begin{align*}
		3^{3n -1} \cdot 	[N^{2n-2}] T(g,N) &=\begin{cases}
			\frac{1}{\abs{\Aut(g)}} & \text{if $G$ is dual to a 3-connected cubic graph,}\\
			0 & \text{else}.
		\end{cases}
	\end{align*}
	In that expression, $\abs{\Aut(g)}^{-1}$ is the symmetry factor of the decompletion, in QFT convention (i.e. with 4 external legs fixed, but summed over all channels). The same convention is used in 
	our generating function $\prim(\hbar_\ren)$ of \cref{primitive_generating_function}. For the purpose of counting cubic graphs, we require the symmetry factor of the completion $G$.  
	According to \cref{decompletion_symmetry_factor}, it is obtained by dividing by $4!$ and by $(L+2)$, the latter being $3n$ in our case.

	What remains to be shown is that $\abs{\Aut(G)}=\abs{\Aut(G^\star)}$, where $G$ is the line graph of $G^\star$. For $n>1$, both $G$ and $G^\star$ are free of multi edges or self loops, and any potential automorphism thus arises from a permutation of vertices.  Such automorphisms of the line graph are known to be induced by automorphisms of $G^\star$  \cite{erdos_automorphisms_1980,lovasz_combinatorial_2007}.
\end{proof}

\begin{example}{}\label{ex:cubic_graphs_count}
	The number of 3-connected cubic graphs on $2n $ vertices is currently known for $2n\leq 28$ \cite[A204198]{oeis}. Starting from $n=3$, the sequence is
	\begin{align*} 
		\left \lbrace  2,~4,~14,~57,~341,~2828,~30468,~396150, ~5909292,~98101019,~1782392646,~\ldots \right \rbrace . 
	\end{align*}
	These numbers correspond to loop orders $L\leq 40$ of primitive decompletions, they are consistent with \cref{tab:leading_graphs_count}.
	The sum of symmetry factors   according to \cref{thm:cubic_count} is
		\begin{align*}
		\scalemath{.95}{S_{n\geq 3}=\left \lbrace  \frac 7 {72},~\frac 5{12}, ~ \frac {611}{240},~ \frac{159}{8},~\frac{5985}{32},~\frac{262185}{128},~\frac{44082575}{1728},~\frac{682544387}{1920}, ~\frac{2891383547}{528},~\ldots \right \rbrace.}
	\end{align*}
	The first two entries, corresponding to $L=7$ and $L=10$,  can also be read off from  \cref{tab:0dim_primitive_coefficients}. We have explicitly computed the symmetry factors for the graphs listed in \cref{tab:leading_graphs_count} at $L=3n-2\leq 25$, that is $n\leq 9$, and found them to agree with the values quoted above. 
	
	$S_n$ is a lower bound to $\abs{C_{2n}}$, the count of 3-connected cubic graphs, because $\abs{\Aut(G^\star)}\geq 1$. As $n$ increases, the number of symmetry factors approaches the true count: 
	\begin{align*}
		\textnormal{at $n=3$}:~\frac{\frac 7 {72}}{2} \approx 0.049, \quad \qquad \textnormal{at $n=14$}:~ \frac{\frac{180310293533225}{5376}}{35085504243}\approx 0.956.
	\end{align*}	
\end{example}

\subsubsection[Primitive graphs with 3n loops]{Primitive graphs with $3n $ loops}\label{sec:leading_2}

Completions with $3n$ loops have $3n+2$ vertices, and therefore $6n+4$ edges. By \cref{lem:TGN_bound}, the leading graphs of $3n$ loops have order $2n+1$ in $N$, hence their duals (\cref{def:dual}) have $2n+1$ vertices and $3n+2$ edges. The leading graphs are those whose cycles are all triangles except for one square. Consequently, the dual graphs  have $2n$ 3-valent vertices and a single 4-valent vertex. 

In \cref{sec:leading_1}, the analogous conditions implied that the dual graphs are \emph{simple}, that is, free of multiedges. This is not true in the case of $3n$ loops. 
By \cref{lem:no_2valent}, $\dual$~still has no 2-valent vertices and, by the same proof as \cref{thm:leading_3n+1}, there are no 2-valent subgraphs. Hence, the dual $\dual$ must be 3-edge-connected.
But this leaves the possibility that it can contain a double edge adjacent to the single 4-valent vertex, as shown in \cref{fig:4vertex_decompositions} (b). 

Furthermore, the line graph alone is not suitable to recover $G$ from $\dual$.
When $\dual$ contains a 4-valent vertex, the line graph $L(\dual)$ (\cref{def:line_graph}) contains five-valent vertices. 
Still, we can use the line graph construction for all edges that are not adjacent to the single 4-valent vertex $v$ and replace $v$ by a circuit of length four. 
There are exactly three different possibilities of an undirected circuit on four vertices, shown in \cref{fig:4vertex_decompositions} (a). 
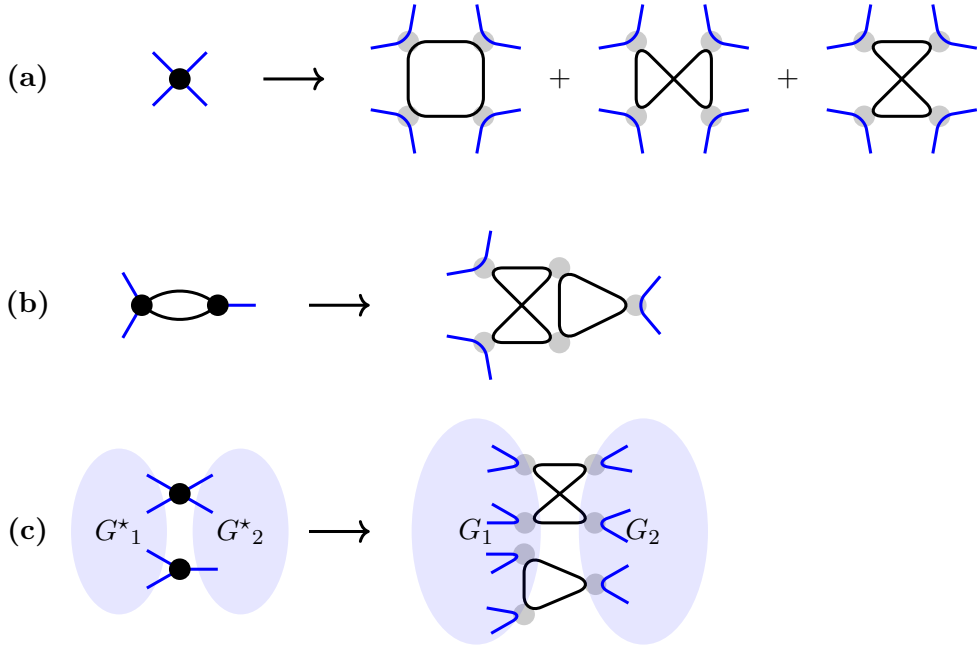
\begin{figure}[htb]
	\centering
	
	\begin{tikzpicture}
		
		\coordinate(x0) at (-3.5,0);
		\node at ($(x0) +(-2,0)$){\textbf{(a)}};
		
		\coordinate(v1) at ($(x0) $){};
		
		\draw[edge,blue](v1)-- +(45:.5);
		\draw[edge,blue](v1)-- +(135:.5);
		\draw[edge,blue](v1)-- +(225:.5);
		\draw[edge,blue](v1)-- +(315:.5);
		
		\node[vertex] at (v1){};
		
		\draw [edge,->] (-2.4,0) -- (- 1.6,0);
		
		\coordinate(x0) at (0,0);
		
		\coordinate(v0) at ($(x0) $){};
		\coordinate(v1) at ($(x0) +(45:.7)$){};
		\coordinate(v2) at ($(x0) +(135:.7)$){};
		\coordinate(v3) at ($(x0) +(225:.7)$){};
		\coordinate(v4) at ($(x0) +(315:.7)$){}; 
		
		\node[vertex, lightgray] at (v1){};
		\node[vertex, lightgray] at (v2){};
		\node[vertex, lightgray] at (v3){};
		\node[vertex, lightgray] at (v4){};
		
		\draw[edge, rounded corners=7pt ] (v1) -- (v2)--(v3)--(v4)--cycle;
		
		\draw[edge,blue, rounded corners=5pt] ($(v1)+(100:.5)$) --(v1) -- ($(v1) +(-10:.5) $);
		\draw[edge,blue, rounded corners=5pt] ($(v2)+(80:.5)$) --(v2) -- ($(v2) +(190:.5) $); 
		\draw[edge,blue, rounded corners=5pt] ($(v3)+(170:.5)$) --(v3) -- ($(v3) +(280:.5) $); 
		\draw[edge,blue, rounded corners=5pt] ($(v4)+( 10:.5)$) --(v4) -- ($(v4) +(260:.5) $); 
		
		\node at ($(x0)+ (1.5,0)$){+};
		
		\coordinate(x0) at (3,0);

		\coordinate(v0) at ($(x0) $){};
		\coordinate(v1) at ($(x0) +(45:.7)$){};
		\coordinate(v2) at ($(x0) +(135:.7)$){};
		\coordinate(v3) at ($(x0) +(225:.7)$){};
		\coordinate(v4) at ($(x0) +(315:.7)$){}; 
		
		\node[vertex,draw=lightgray,fill=lightgray] at (v1){};
		\node[vertex,draw=lightgray,fill=lightgray] at (v2){};
		\node[vertex,draw=lightgray,fill=lightgray] at (v3){};
		\node[vertex,draw=lightgray,fill=lightgray] at (v4){};
		
		\draw[edge, rounded corners=7pt ] (v1) -- (v3)--(v2)--(v4)--cycle;
		
		\draw[edge,blue, rounded corners=5pt] ($(v1)+(100:.5)$) --(v1) -- ($(v1) +(-10:.5) $);
		\draw[edge,blue, rounded corners=5pt] ($(v2)+(80:.5)$) --(v2) -- ($(v2) +(190:.5) $); 
		\draw[edge,blue, rounded corners=5pt] ($(v3)+(170:.5)$) --(v3) -- ($(v3) +(280:.5) $); 
		\draw[edge,blue, rounded corners=5pt] ($(v4)+( 10:.5)$) --(v4) -- ($(v4) +(260:.5) $); 
		
		\node at ($(x0)+ (1.5,0)$){+};
		
		\coordinate(x0) at (6,0);

		\coordinate(v0) at ($(x0) $){};
		\coordinate(v1) at ($(x0) +(45:.7)$){};
		\coordinate(v2) at ($(x0) +(135:.7)$){};
		\coordinate(v3) at ($(x0) +(225:.7)$){};
		\coordinate(v4) at ($(x0) +(315:.7)$){}; 
		
		\node[vertex,draw=lightgray,fill=lightgray] at (v1){};
		\node[vertex,draw=lightgray,fill=lightgray] at (v2){};
		\node[vertex,draw=lightgray,fill=lightgray] at (v3){};
		\node[vertex,draw=lightgray,fill=lightgray] at (v4){};
		
		\draw[edge, rounded corners=7pt ] (v1) -- (v2)--(v4)--(v3)--cycle;
		
		\draw[edge,blue, rounded corners=5pt] ($(v1)+(100:.5)$) --(v1) -- ($(v1) +(-10:.5) $);
		\draw[edge,blue, rounded corners=5pt] ($(v2)+(80:.5)$) --(v2) -- ($(v2) +(190:.5) $); 
		\draw[edge,blue, rounded corners=5pt] ($(v3)+(170:.5)$) --(v3) -- ($(v3) +(280:.5) $); 
		\draw[edge,blue, rounded corners=5pt] ($(v4)+( 10:.5)$) --(v4) -- ($(v4) +(260:.5) $);

		\coordinate(x0) at (-4,-3);
		\node at ($(x0) +(-1.5,0)$){\textbf{(b)}};
		\coordinate(v1) at ($(x0) $){};
		\coordinate(v2) at ($(x0) +(1,0)$){};
		\draw[edge,blue](v2)-- +(0:.5);
		\draw[edge,blue](v1)-- +(120:.5);
		\draw[edge,blue](v1)-- +(240:.5);
		\draw[edge,bend angle=40,bend left] (v1) to (v2);
		\draw[edge,bend angle=40,bend right] (v1) to (v2);
		\node[vertex] at (v1){};
		\node[vertex] at (v2){};
		
		\draw [edge,->] (-1.8,-3) -- (- 1,-3);

		\coordinate(x0) at (1,-3);
		
		\coordinate(v0) at ($(x0) $){};
		\coordinate(v1) at ($(x0) +(45:.7)$){};
		\coordinate(v2) at ($(x0) +(135:.7)$){};
		\coordinate(v3) at ($(x0) +(225:.7)$){};
		\coordinate(v4) at ($(x0) +(315:.7)$){}; 
		\coordinate(v5) at ($(x0) +(0:1.5)$){};
		
		\node[vertex, lightgray] at (v1){};
		\node[vertex, lightgray] at (v2){};
		\node[vertex, lightgray] at (v3){};
		\node[vertex, lightgray] at (v4){};
		\node[vertex, lightgray] at (v5){};
		
		\draw[edge, rounded corners=7pt ] (v1) -- (v2)--(v4)--(v3)--cycle;
		
		\draw[edge, rounded corners=7pt ] (v1) -- (v4)--(v5)-- cycle;
		\draw[edge,blue, rounded corners=5pt] ($(v2)+(80:.5)$) --(v2) -- ($(v2) +(190:.5) $); 
		\draw[edge,blue, rounded corners=5pt] ($(v3)+(170:.5)$) --(v3) -- ($(v3) +(280:.5) $); 
		\draw[edge,blue, rounded corners=5pt] ($(v5)+(50:.5)$) --(v5) -- ($(v5) +(-50:.5) $);

		\coordinate(x0) at (-4,-6);
		\node at ($(x0) +(-1.5,0)$){\textbf{(c)}};
		\coordinate(v1) at ($(x0)+(.5,.5) $){};
		\coordinate(v2) at ($(x0)+(.5,-.5) $){};
		
		\draw[edge,blue](v1)-- +(30:.5);
		\draw[edge,blue](v1)-- +(-30:.5);
		\draw[edge,blue](v1)-- +(150:.5);
		\draw[edge,blue](v1)-- +(210:.5);
		
		\draw[edge,blue](v2)-- +(0:.5);
		\draw[edge,blue](v2)-- +(150:.5);
		\draw[edge,blue](v2)-- +(210:.5);
		\node[vertex] at (v1){};
		\node[vertex] at (v2){};
		
		\node[ellipse ,minimum height=2.2cm , minimum width=1cm, fill=blue,fill opacity=.1, text opacity=1,text=black] at ($(x0)+(-.3,0) $) { $\dual_1$ };
		\node[ellipse ,minimum height=2.2cm , minimum width=1cm, fill=blue,fill opacity=.1, text opacity=1,text=black] at ($(x0)+(1.3,0) $) { $\dual_2$ };

		\draw [edge,->] (-1.8,-6) -- (- 1,-6);
		
		\coordinate(x0) at (1.5,-5.5);
		
		\coordinate(v0) at ($(x0) $){};
		\coordinate(v1) at ($(x0) +(40:.6)$){};
		\coordinate(v2) at ($(x0) +(140:.6)$){};
		\coordinate(v3) at ($(x0) +(220:.6)$){};
		\coordinate(v4) at ($(x0) +(320:.6)$){}; 
		\coordinate(v5) at ($(x0) +(-.47,-.8)$){};
		\coordinate(v6) at ($(x0) +(-.47,-1.6)$){};
		\coordinate(v7) at ($(x0) +(.47,-1.2)$){}; 
		
		\node[vertex, lightgray] at (v1){};
		\node[vertex, lightgray] at (v2){};
		\node[vertex, lightgray] at (v3){};
		\node[vertex, lightgray] at (v4){};
		\node[vertex, lightgray] at (v5){};
		\node[vertex, lightgray] at (v6){};
		\node[vertex, lightgray] at (v7){};
		
		\draw[edge, rounded corners=7pt ] (v1) -- (v2)--(v4)--(v3)--cycle;
		\draw[edge, rounded corners=7pt ] (v7) -- (v5)--(v6)-- cycle;
		
		\draw[edge,blue, rounded corners=5pt] ($(v2)+(150:.5)$) --(v2) -- ($(v2) +(190:.5) $); 
		\draw[edge,blue, rounded corners=5pt] ($(v3)+(150:.5)$) --(v3) -- ($(v3) +(180:.5) $); 
		\draw[edge,blue, rounded corners=5pt] ($(v5)+(180:.5)$) --(v5) -- ($(v5) +(210:.5) $); 
		\draw[edge,blue, rounded corners=5pt] ($(v6)+(170:.5)$) --(v6) -- ($(v6) +(210:.5) $); 
		
		\draw[edge,blue, rounded corners=5pt] ($(v1)+(30:.5)$) --(v1) -- ($(v1) +(-10:.5) $); 
		\draw[edge,blue, rounded corners=5pt] ($(v4)+(20:.5)$) --(v4) -- ($(v4) +(-30:.5) $); 
		\draw[edge,blue, rounded corners=5pt] ($(v7)+(30:.5)$) --(v7) -- ($(v7) +(-30:.5) $); 
		
		\node[ellipse ,minimum height=3cm , minimum width=1.7cm, fill=blue,fill opacity=.1, text opacity=1,text=black] at ($(x0)+(-1.1,-.5) $) { $ G_1$ };
		\node[ellipse ,minimum height=3cm , minimum width=1.7cm, fill=blue,fill opacity=.1, text opacity=1,text=black] at ($(x0)+(1.1,-.5) $) { $ G_2$ };

	\end{tikzpicture}

	\caption{\textbf{(a)} A 4-valent vertex in $\dual$ gives rise to one of three possible circuits. In a three-vertex-connected simple $\dual$, all three give rise to (potentially distinct) primitives $G$. \textbf{(b)} If~$\dual$ contains a double edge adjacent to a 4-valent vertex, there is only one possibility to decompose this vertex which makes $G$ primitive. \textbf{(c)} Similarly, if the 4-valent vertex of $\dual$ is a cut vertex of a 2-vertex cut, then only one of the three decompositions gives rise to a primitive~$G$, namely the one that produces four edges between the cut components. }
	\label{fig:4vertex_decompositions}
	
\end{figure}

\begin{theorem}\label{thm:leading_3n}
	The leading decompositions at $L=3n$ loops, $n>1$, arise from 3-edge-connected dual graphs $\dual$, consisting of a single 4-valent vertex $v$, and $2n$ 3-valent vertices,  where either one of the following is true:
	\begin{enumerate}
		\item  $\dual$ is simple (and therefore three-vertex connected), then all three decompositions of $v$ shown in \cref{fig:4vertex_decompositions} (a) contribute.
		\item $\dual$ is three-vertex connected and contains one double edge, adjacent to $v$. Only that decomposition of $v$ contributes which does not lead to a double edge in $G$, see (b) in \cref{fig:4vertex_decompositions}.
		\item $\dual$ is simple and has one 2-vertex cut, where  one of the cut vertices is the 4-valent vertex $v$. Then only that one decomposition of $v$ contributes that joins the two cut components of $G$ with four and not two edges, see \cref{fig:4vertex_decompositions} (c).
	\end{enumerate}
\end{theorem} 

Since  \cref{thm:leading_3n} potentially involves a sum over distinct decompositions of a given 4-valent vertex,  the number $n_0$ of leading-order graphs generated is in general larger than the number of non-isomorphic duals $\dual$. 
Algorithmically, the graphs in \cref{thm:leading_3n} (b) can be generated from 3-regular, 3-edge-connected graphs on $2n$ vertices upon replacing one of the 3-valent vertices by a 4-valent vertex with adjacent double edge. The three distinct choices of where to put the double edge amount to distinct decompositions of the same primitive $G$. We have generated all suitable $\dual$ and from them all primitive $G$ which contribute at leading order in $N$ up to $L=24$ loops, see \cref{tab:leading_graphs_count}.

\begin{example}
	The case $n=1$ is excluded from  \cref{thm:leading_3n} because the dual graph with one 4-valent and two 3-valent vertices contains two multiedges. This $\dual$ is the dual of $K_5$, and the three ways of arranging the four-cycle (\cref{fig:4vertex_decompositions} (a)) amount to the three isomorphic vertex decompositions of $K_5$ discussed in the proof of \cref{lem:leading_no_multiple}, see \cref{fig:dual_K5}. As all five vertices of $K_5$ are equivalent, one obtains in total  $3\cdot 5=15$ decompositions contributing at leading order in $N$. An explicit computation yields
	\begin{align*}
		T \left( K_5,N \right) &= 132N+96N^2+15N^3, \qquad T \left( K_5\setminus \left \lbrace v \right \rbrace ,N \right) = 66+15N.
	\end{align*}
	Notice that a similar sum over isomorphic vertices is not present when $\dual$ is  3-regular because there, all circuits have length three and swapping them results in the same decomposition. 
\end{example}

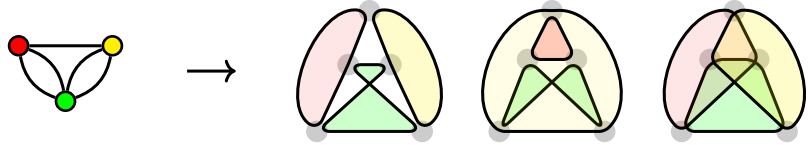
\begin{figure}[htb]
	\centering
	\begin{tikzpicture}[scale=.8]
		
		\coordinate(x0) at (-5,0);
		\node[vertex,green,draw=black](v1) at ($(x0) $){};
		\node[vertex,yellow,draw=black](v2) at ($(x0) +(50:1.2)$){};
		\node[vertex,red,draw=black](v3) at ($(x0) +(130:1.2)$){};
		
		\draw[edge] (v2) -- (v3);
		\draw[edge, bend angle=30,bend left] (v1) to (v2);
		\draw[edge, bend angle=30,bend right] (v1) to (v2);
		\draw[edge, bend angle=30,bend left] (v1) to (v3);
		\draw[edge, bend angle=30,bend right] (v1) to (v3);

		\draw[edge,->](-3,0.5)--(-2.2,0.5);

		\coordinate(x0) at (0,0);

		\coordinate(v1) at ($(x0) +(60:.7)$){};
		\coordinate(v2) at ($(x0) +(120:.7)$){};
		\coordinate(v3) at ($(x0) +(210:1)$){};
		\coordinate(v4) at ($(x0) +(330:1)$){};
		\coordinate(v5) at ($(x0) +(90:1.5)$){};
		
		\node[vertex,lightgray] at (v5){};
		\node[vertex,lightgray] at (v1){};
		\node[vertex,lightgray] at (v2){};
		\node[vertex,lightgray] at (v3){};
		\node[vertex,lightgray] at (v4){};
		
		\draw[edge,fill=green,fill opacity=.2, rounded corners=5pt] (v1) -- (v2) -- (v4) -- (v3)--cycle;
		
		\draw[edge, bend angle=70,bend left,fill=red,fill opacity=.1, rounded corners=5pt] (v3) to (v5)--(v2)--cycle;
		\draw[edge, bend angle=70,bend right,fill=yellow,fill opacity=.2, rounded corners=5pt] (v4) to (v5)--(v1)--cycle;
		
		\coordinate(x0) at (3,0);

		\coordinate(v1) at ($(x0) +(60:.8)$){};
		\coordinate(v2) at ($(x0) +(120:.8)$){};
		\coordinate(v3) at ($(x0) +(210:1)$){};
		\coordinate(v4) at ($(x0) +(330:1)$){};
		\coordinate(v5) at ($(x0) +(90:1.5)$){};
		
		\node[vertex,lightgray] at (v5){};
		\node[vertex,lightgray] at (v1){};
		\node[vertex,lightgray] at (v2){};
		\node[vertex,lightgray] at (v3){};
		\node[vertex,lightgray] at (v4){};
		
		\draw[edge,fill=green,fill opacity=.2, rounded corners=5pt] (v1) -- (v3) -- (v2) -- (v4)--cycle;
		
		\draw[edge, bend angle=70,bend left,fill=yellow,fill opacity=.1, rounded corners=2pt] (v3) to (v5) to (v4)--cycle;
		\draw[edge , fill=red,fill opacity=.2, rounded corners=5pt] (v2) to (v5)--(v1)--cycle;
		
		\coordinate(x0) at (6,0);

		\coordinate(v1) at ($(x0) +(60:.8)$){};
		\coordinate(v2) at ($(x0) +(120:.8)$){};
		\coordinate(v3) at ($(x0) +(210:1)$){};
		\coordinate(v4) at ($(x0) +(330:1)$){};
		\coordinate(v5) at ($(x0) +(90:1.5)$){};
		
		\node[vertex,lightgray] at (v5){};
		\node[vertex,lightgray] at (v1){};
		\node[vertex,lightgray] at (v2){};
		\node[vertex,lightgray] at (v3){};
		\node[vertex,lightgray] at (v4){};
		
		\draw[edge,fill=green,fill opacity=.2, rounded corners=5pt] (v1) -- (v2) -- (v3) -- (v4)-- cycle;
		
		\draw[edge, bend angle=70,bend left,fill=red,fill opacity=.1, rounded corners=3pt] (v3) to (v5)--(v1)--cycle;
		\draw[edge ,bend angle=70,bend right,fill=yellow,fill opacity=.2, rounded corners=3pt] (v4) to (v5)--(v2)--cycle;

	\end{tikzpicture}
	\caption{The unique dual $\dual$ with one 4-valent and two 3-valent vertices. The three ways of arranging its four-cycle give rise to three isomorphic three-loop graphs $G=K_5$. Equivalently, this choice can be understood as the three ways of splitting the upper vertex in $G$ according to \cref{fig:vertex_decomposition}, compare \cref{fig:leading_no_multiple_1}. }
	\label{fig:dual_K5}
\end{figure}

\subsubsection[Primitive graphs with 3n+2 loops]{Primitive graphs with $3n+2$ loops} \label{sec:leading_3}

In this case, the completed graphs $G$ have $3n+4$ vertices and $6n+8$ edges. The leading order in $N$ is $2n+2$. Consequently, the dual graphs $\dual$ have $2n+2$ vertices and $3n+4$ edges. 
These dual graphs either have two 4-valent vertices and the other $2n$ vertices 3-valent, or one vertex is 5-valent and the other $2n+1$ are 3-valent. 
As before, $\dual$ must be 3-edge-connected, but it can have a 2-vertex cut. 
These conditions imply that  $\dual$ must not contain double edges consisting of only 3-valent vertices, but there may be one or two double edges adjacent to a vertex of valence larger than three. It is straightforward to derive the particular constraints on decompositions in each of these cases, analogously to \cref{thm:leading_3n}, and to implement them in a computer program. However, it is little insightful to enumerate them here, so we merely state the (obvious) necessary condition on $\dual$. 

\begin{lemma}\label{thm:leading_3n+2}
	Let $L=3n+2$ with $n\in \mathbbm N$. The primitive graphs $G$ contributing at leading order in $N$ arise from dual graphs $\dual$ which have $2n+2$ vertices, $3n+4$ edges, no vertices of degree lower than three, and which are 3-edge-connected. 
\end{lemma}

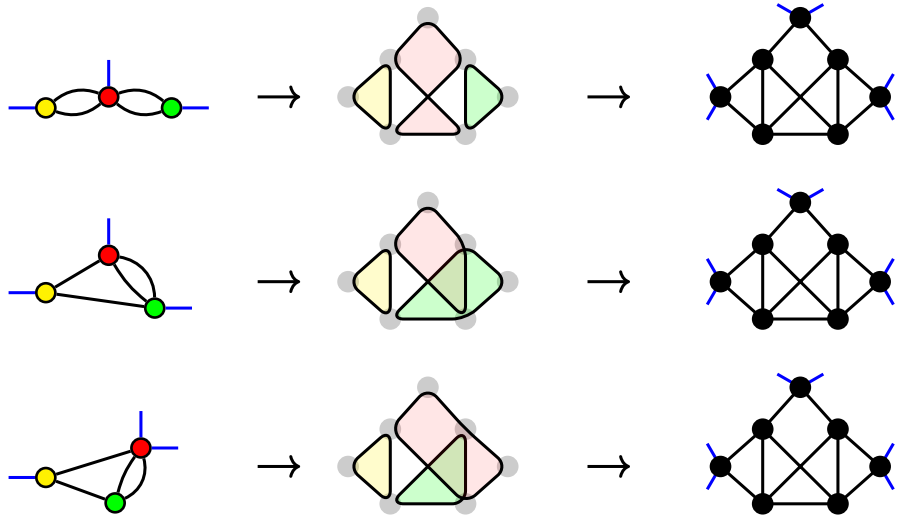
\begin{figure}[htb]
	\centering
	\begin{tikzpicture}[scale=.7]
		
		\coordinate(x0) at (-6,0);
		\node[vertex,red,draw=black](v1) at ($(x0) $){};
		\node[vertex,yellow,draw=black](v2) at ($(x0) +(190:1.2)$){};
		\node[vertex,green,draw=black](v3) at ($(x0) +(-10:1.2)$){};
		
		\draw[edge, bend angle=30,bend left] (v1) to (v2);
		\draw[edge, bend angle=30,bend right] (v1) to (v2);
		\draw[edge, bend angle=30,bend left] (v1) to (v3);
		\draw[edge, bend angle=30,bend right] (v1) to (v3);
		\draw[edge,blue] (v2) -- +(180:.7);
		\draw[edge,blue] (v3) -- +(0:.7);
		\draw[edge,blue] (v1) -- +(90:.7);
		
		\draw[edge,->](-3.2,0 )--(-2.4,0 );

		\coordinate(x0) at (0,0);
		
		\coordinate(v1) at ($(x0) +(45:1)$){};
		\coordinate(v2) at ($(x0) +(135:1)$){};
		\coordinate(v3) at ($(x0) +(225:1)$){};
		\coordinate(v4) at ($(x0) +(315:1)$){};
		\coordinate(v5) at ($(x0) +(90:1.5)$){};
		\coordinate(v6) at ($(x0) +(0:1.5)$){};
		\coordinate(v7) at ($(x0) +(180:1.5)$){};
		
		\node[vertex,lightgray] at (v1){};
		\node[vertex,lightgray] at (v2){};
		\node[vertex,lightgray] at (v3){};
		\node[vertex,lightgray] at (v4){};
		\node[vertex,lightgray] at (v5){};
		\node[vertex,lightgray] at (v6){};
		\node[vertex,lightgray] at (v7){};
		
		\draw[edge,fill=green,fill opacity=.2, rounded corners=5pt] (v1) -- (v6) -- (v4) --cycle;
		
		\draw[edge, fill=red,fill opacity=.1, rounded corners=5pt] (v2) to (v5)--(v1)-- (v3) --(v4)-- cycle;
		\draw[edge, fill=yellow,fill opacity=.2, rounded corners=5pt] (v2) to (v7)--(v3)--cycle;

		\draw[edge,->](3,0 )--(3.8,0 );
		
		\coordinate(x0) at (7,0);
		
		\coordinate(v1) at ($(x0) +(45:1)$){};
		\coordinate(v2) at ($(x0) +(135:1)$){};
		\coordinate(v3) at ($(x0) +(225:1)$){};
		\coordinate(v4) at ($(x0) +(315:1)$){};
		\coordinate(v5) at ($(x0) +(90:1.5)$){};
		\coordinate(v6) at ($(x0) +(0:1.5)$){};
		\coordinate(v7) at ($(x0) +(180:1.5)$){};
		
		\draw[edge] (v1) --(v5)--(v2)--(v4)--(v3)--(v7)--(v2)--(v3)--(v1)--(v6)--(v4)--(v1);
		\draw[edge,blue] ($(v5)+(30:.5)$) --(v5) -- ($(v5) +(150:.5) $); 
		\draw[edge,blue] ($(v6)+(60:.5)$) --(v6) -- ($(v6) +(-60:.5) $); 
		\draw[edge,blue] ($(v7)+(120:.5)$) --(v7) -- ($(v7) +(240:.5) $); 
		
		\node[vertex] at (v1){};
		\node[vertex] at (v2){};
		\node[vertex] at (v3){};
		\node[vertex] at (v4){};
		\node[vertex] at (v5){};
		\node[vertex] at (v6){};
		\node[vertex] at (v7){};

		\coordinate(x0) at (-6,-3.5);
		\node[vertex,red,draw=black](v1) at ($(x0)+(90:.5) $){};
		\node[vertex,yellow,draw=black](v2) at ($(x0) +(190:1.2)$){};
		\node[vertex,green,draw=black](v3) at ($(x0) +(-30:1.0)$){};
		
		\draw[edge] (v1) to (v2);
		\draw[edge] (v2) -- (v3);
		\draw[edge, bend angle=40,bend left] (v1) to (v3);
		\draw[edge, bend angle=10,bend right] (v1) to (v3);
		\draw[edge,blue] (v2) -- +(180:.7);
		\draw[edge,blue] (v3) -- +(0:.7);
		\draw[edge,blue] (v1) -- +(90:.7);
		
		\draw[edge,->](-3.2,-3.5)--(-2.4,-3.5);

		\coordinate(x0) at (0,-3.5);
		
		\coordinate(v1) at ($(x0) +(45:1)$){};
		\coordinate(v2) at ($(x0) +(135:1)$){};
		\coordinate(v3) at ($(x0) +(225:1)$){};
		\coordinate(v4) at ($(x0) +(315:1)$){};
		\coordinate(v5) at ($(x0) +(90:1.5)$){};
		\coordinate(v6) at ($(x0) +(0:1.5)$){};
		\coordinate(v7) at ($(x0) +(180:1.5)$){};
		
		\node[vertex,lightgray] at (v1){};
		\node[vertex,lightgray] at (v2){};
		\node[vertex,lightgray] at (v3){};
		\node[vertex,lightgray] at (v4){};
		\node[vertex,lightgray] at (v5){};
		\node[vertex,lightgray] at (v6){};
		\node[vertex,lightgray] at (v7){};
		
		\draw[edge,fill=green,fill opacity=.2, rounded corners=5pt] (v1) -- (v3) -- (v4)--(v6) --cycle;
		
		\draw[edge, fill=red,fill opacity=.1, rounded corners=5pt] (v1) -- (v5)--(v2)-- (v4) -- cycle;
		\draw[edge, fill=yellow,fill opacity=.2, rounded corners=5pt] (v2) to (v7)--(v3)--cycle;

		\draw[edge,->](3,-3.5 )--(3.8,-3.5 );
		
		\coordinate(x0) at (7,-3.5);
		
		\coordinate(v1) at ($(x0) +(45:1)$){};
		\coordinate(v2) at ($(x0) +(135:1)$){};
		\coordinate(v3) at ($(x0) +(225:1)$){};
		\coordinate(v4) at ($(x0) +(315:1)$){};
		\coordinate(v5) at ($(x0) +(90:1.5)$){};
		\coordinate(v6) at ($(x0) +(0:1.5)$){};
		\coordinate(v7) at ($(x0) +(180:1.5)$){};
		
		\draw[edge] (v1) --(v5)--(v2)--(v4)--(v3)--(v7)--(v2)--(v3)--(v1)--(v6)--(v4)--(v1);
		\draw[edge,blue] ($(v5)+(30:.5)$) --(v5) -- ($(v5) +(150:.5) $); 
		\draw[edge,blue] ($(v6)+(60:.5)$) --(v6) -- ($(v6) +(-60:.5) $); 
		\draw[edge,blue] ($(v7)+(120:.5)$) --(v7) -- ($(v7) +(240:.5) $); 
		
		\node[vertex] at (v1){};
		\node[vertex] at (v2){};
		\node[vertex] at (v3){};
		\node[vertex] at (v4){};
		\node[vertex] at (v5){};
		\node[vertex] at (v6){};
		\node[vertex] at (v7){};

		\coordinate(x0) at (-6,-7);
		\node[vertex,red,draw=black](v1) at ($(x0)+(30:.7) $){};
		\node[vertex,yellow,draw=black](v2) at ($(x0) +(190:1.2)$){};
		\node[vertex,green,draw=black](v3) at ($(x0) +(280:.7)$){};
		
		\draw[edge] (v1) to (v2);
		\draw[edge] (v2) -- (v3);
		\draw[edge, bend angle=40,bend left] (v1) to (v3);
		\draw[edge, bend angle=10,bend right] (v1) to (v3);
		\draw[edge,blue] (v2) -- +(180:.7);
		\draw[edge,blue] (v1) -- +(0:.7);
		\draw[edge,blue] (v1) -- +(90:.7);
		
		\draw[edge,->](-3.2,-7)--(-2.4,-7);

		\coordinate(x0) at (0,-7);
		
		\coordinate(v1) at ($(x0) +(45:1)$){};
		\coordinate(v2) at ($(x0) +(135:1)$){};
		\coordinate(v3) at ($(x0) +(225:1)$){};
		\coordinate(v4) at ($(x0) +(315:1)$){};
		\coordinate(v5) at ($(x0) +(90:1.5)$){};
		\coordinate(v6) at ($(x0) +(0:1.5)$){};
		\coordinate(v7) at ($(x0) +(180:1.5)$){};
		
		\node[vertex,lightgray] at (v1){};
		\node[vertex,lightgray] at (v2){};
		\node[vertex,lightgray] at (v3){};
		\node[vertex,lightgray] at (v4){};
		\node[vertex,lightgray] at (v5){};
		\node[vertex,lightgray] at (v6){};
		\node[vertex,lightgray] at (v7){};
		
		\draw[edge,fill=green,fill opacity=.2, rounded corners=5pt] (v1) -- (v3) -- (v4)--cycle;
		
		\draw[edge, fill=red,fill opacity=.1, rounded corners=5pt] (v1) -- (v5)--(v2)-- (v4) -- (v6)--cycle;
		\draw[edge, fill=yellow,fill opacity=.2, rounded corners=5pt] (v2) to (v7)--(v3)--cycle;

		\draw[edge,->](3,-7)--(3.8,-7 );
		
		\coordinate(x0) at (7,-7);
		
		\coordinate(v1) at ($(x0) +(45:1)$){};
		\coordinate(v2) at ($(x0) +(135:1)$){};
		\coordinate(v3) at ($(x0) +(225:1)$){};
		\coordinate(v4) at ($(x0) +(315:1)$){};
		\coordinate(v5) at ($(x0) +(90:1.5)$){};
		\coordinate(v6) at ($(x0) +(0:1.5)$){};
		\coordinate(v7) at ($(x0) +(180:1.5)$){};
		
		\draw[edge] (v1) --(v5)--(v2)--(v4)--(v3)--(v7)--(v2)--(v3)--(v1)--(v6)--(v4)--(v1);
		\draw[edge,blue] ($(v5)+(30:.5)$) --(v5) -- ($(v5) +(150:.5) $); 
		\draw[edge,blue] ($(v6)+(60:.5)$) --(v6) -- ($(v6) +(-60:.5) $); 
		\draw[edge,blue] ($(v7)+(120:.5)$) --(v7) -- ($(v7) +(240:.5) $); 
		
		\node[vertex] at (v1){};
		\node[vertex] at (v2){};
		\node[vertex] at (v3){};
		\node[vertex] at (v4){};
		\node[vertex] at (v5){};
		\node[vertex] at (v6){};
		\node[vertex] at (v7){};

	\end{tikzpicture}
	\caption{
		The right column shows a subgraph of some completion $G$. This subgraph has five non-isomorphic vertex decompositions into three circuits each, three of them are shown in the central column. Each of them gives rise to a distinct, non-isomorphic dual (left column). 
		On the other hand, the four- and five-valent vertices of the duals can not contribute all of their possible circuit orientations since some of them would produce non-primitive graphs.  }
	\label{fig:dual_non_isomorphic}
\end{figure}

For $L=3n+1$ (\cref{thm:leading_3n+1}), there is a one-to-one correspondence between dual graphs $\dual$ and the leading-order decompositions of primitive completions $G$, given by the line graph (\cref{def:line_graph}). 
Already for $L=3n$ (\cref{thm:leading_3n}), a single dual can give rise to more than one completion, compare \cref{fig:4vertex_decompositions} (a). In the present case $L=3n+2$, the correspondence is even weaker:
\begin{enumerate}
	\item A given dual $\dual$ can produce multiple non-isomorphic leading-order graphs $G$.
	\item A leading order graph $G$ can have more than one leading-order decomposition, that is, the leading-order coefficient of $T(G,N)$ is not necessarily unity.
	\item Not all leading-order decompositions of a fixed $G$ result in isomorphic duals $\dual$.
\end{enumerate}
\Cref{fig:dual_non_isomorphic} shows an example of the latter two effects. 

Again, we have generated all relevant $\dual$ and their corresponding primitive graphs $G$ at leading order up to $L=23$, see \cref{tab:loworder_leading_graphs_count}.
While the necessary constraints on the allowed mappings from dual vertices to circuits are tedious to spell out explicitly, a computer program that checks them is still much faster than the explicit enumeration of \emph{all} primitives and filtering for leading-order decompositions. For example, generating the leading duals at $L=17$ loops took less than two minutes, while generating and filtering all 16-loop graphs for \cref{tab:loworder_leading_graphs_count} took more than one week. 

\Cref{tab:leading_graphs_count} shows the count of graphs for all three cases $L\mod 3 \in \left \lbrace 0,1,2 \right \rbrace$.   An explicit list of these graphs is available from the first author's website\footnote{\href{https://paulbalduf.com/research}{paulbalduf.com/research}}.
We remark that the present construction of completions from their duals can be extended to subleading primitive completions by including dual graphs that have less vertices at a fixed number of edges, and observing an increasing number of restrictions on the allowed mappings from vertices to circuits.

\begin{table}[htb]
	\centering
	\begin{tblr}{ 
			vlines, 
			hline{1}={1-3}{1pt, solid,rightpos=0},
			hline{3}= {1-3}{1pt,solid,rightpos=0},
			hline{ Z}={1-3}{1pt, solid,rightpos=0},
			hline{1}={5-7}{1pt, solid,rightpos=0,leftpos=0},
			hline{3}= {5-7}{1pt,solid,rightpos=0,leftpos=0}, 
			hline{ Z}={5-7}{1pt, solid,rightpos=0,leftpos=0},
			hline{1}={9-11}{1pt, solid,rightpos=0},
			hline{3}= {9-11}{1pt,solid,rightpos=0},
			hline{ Z}={9-11}{1pt, solid,rightpos=0},
			cell{1}{1}={r=1,c=3}{c},
			cell{1}{5}={r=1,c=3}{c},
			cell{1}{9}={r=1,c=3}{c},
			rowsep=1pt,
			cells={font=\fontsize{10pt}{12pt}\selectfont },
			columns={halign=r},
			row{1}  = { rowsep=2pt, halign=c, valign=m,font = \fontsize{10pt}{12pt}\selectfont } 
		}
		$3n$ & & & & $3n+1$ & & & & $3n+2$ & & \\
		$L$ & $\lfloor N_\text{max} \rfloor$ & $n_0$ & &$L$ & $\lfloor N_\text{max} \rfloor$ & $n_0$ & &$L$ & $\lfloor N_\text{max} \rfloor$ & $n_0$  \\
		3 & 3 &  1& & 4 & 4 & 1 & &5 & 4 & 2\\
		6 & 5 & 3 & & 7 & 6 & 2 & &8 & 6 & 22 \\
		9 & 7 & 10 & & 10 & 8 & 4 & &11 & 8 & 205 \\
		12 & 9 & 57 & & 13 & 10 & 14 && 14 & 10 & 2,278 \\
		15 & 11 & 425 & & 16 & 12 & 57 &&17 & 12 & 29,840\\
		18 & 13 & 4,151 & & 19 & 14 & 341 &&20 & 14 & 440,910\\
		21 & 15 & 49,136 & & 22 & 16 & 2,828 & &23 & 16 & 7,181,335 \\
		24 & 17 & 673,350 & &25 & 18 & 30,468
	\end{tblr}
	\caption{Number $n_0$ of primitive completed graphs contributing at leading order in $N$. All these graphs have been explicitly generated from their dual; for $L=3n+1$, the count equals the number of 3-vertex connected graphs \cref{ex:cubic_graphs_count} as expected.
	}
	\label{tab:leading_graphs_count}
\end{table}

\subsubsection{Families of leading primitive graphs}\label{sec:specific_families}

The zigzag-graphs, or $(1,2)$-circulants, are famously the only irreducible infinite family of graphs in $\phi^4$ theory whose period is known analytically \cite{brown_singlevalued_2015}.

\begin{lemma}\label{lem:zigzag}
	Let $Z_L$ be the zigzag graph whose decompletion has $L$ loops. In the polynomial $T(Z_L,N)$, the coefficient of highest-order in $N$  has order $\lfloor \frac L 2 +2\rfloor$, and is $(2L+4)$ when $L$ is odd, and 2 when $L$ is even. 
\end{lemma}
\begin{proof}
	We need to identify a vertex decomposition that gives rise to the maximum number of circuits. Concretely, as many circuits as possible need to be triangles. For a zigzag, all vertices can be drawn on a circle such that they are only ever adjacent to their nearest and next-nearest neighbors. In that sense, any potential short cycle must be \enquote{local}. 
	
	Starting from any triangle, we see that the next triangle can only touch the current one at a vertex, but they can not share an edge. Hence, the ring of $(L+2)$ vertices can have at most $\frac{L+2}{2}$ triangles. Every other vertex has two edges which are not in any triangle, they form another \enquote{special} circuit, giving a total of $\frac{L}{2}+2$ circuits.

	When $L$ is even, there are exactly two distinct decompositions of this type. When $L$ is odd, there remains a \enquote{defect} of one edge which does not fit a triangle. Either the special circuit is extended to cover this edge, or one triangle is sacrificed to make a 4-circuit to cover the edge, leaving the special circuit intact. The defect can be located at any of the $(L+2)$ vertices, always giving rise to two distinct possibilities to resolve it, for a total of $2L+4$ decompositions that yield the same maximum number of circuits.

	One might suspect that another decomposition of the same, or potentially larger, number of circuits is possible based on squares, but this is not the case. Like the triangles, the squares require one long \enquote{special} circuit, and there is effectively only one square for every three vertices, so they result in a number of circuits $\sim \frac L 3$. 
\end{proof}

\begin{figure}[htb]
	\centering
	\begin{tikzpicture}[scale=.6]
		
		\coordinate(x0) at (-6,0);
		\draw[edge] (x0) circle(1.2);
		\draw[edge] (x0) circle(0.5);
		
		\node[vertex,red,draw=black](v1) at ($(x0) + (60:.5) $){};
		\node[vertex,yellow,draw=black](v2) at ($(x0) +(180:.5)$){};
		\node[vertex,green,draw=black](v3) at ($(x0) +(300:.5)$){};
		\node[vertex,red,draw=black](v4) at ($(x0) + (60:1.2) $){};
		\node[vertex,yellow,draw=black](v5) at ($(x0) +(180:1.2)$){};
		\node[vertex,green,draw=black](v6) at ($(x0) +(300:1.2)$){};
		
		\draw[edge] (v1) -- (v4);
		\draw[edge] (v2) --(v5); 
		\draw[edge] (v3) -- (v6);

		\draw[edge,->](-3.2,0 )--(-2.4,0 );

		\coordinate(x0) at (0,0);
		
		\coordinate(v1) at ($(x0) +(0:.5) $){};
		\coordinate(v2) at ($(x0) +(120:.5)$){};
		\coordinate(v3) at ($(x0) +(240:.5)$){};
		\coordinate(v4) at ($(x0) +(0:1.6) $){};
		\coordinate(v5) at ($(x0) +(120:1.6)$){};
		\coordinate(v6) at ($(x0) +(240:1.6)$){};
		\coordinate(v7) at ($(x0) +(300:1)$){};
		\coordinate(v8) at ($(x0) +(60:1)$){};
		\coordinate(v9) at ($(x0) +(180:1)$){};
		
		\node[vertex,lightgray] at (v1){};
		\node[vertex,lightgray] at (v2){};
		\node[vertex,lightgray] at (v3){};
		\node[vertex,lightgray] at (v4){};
		\node[vertex,lightgray] at (v5){};
		\node[vertex,lightgray] at (v6){};
		\node[vertex,lightgray] at (v7){};
		\node[vertex,lightgray] at (v8){};
		\node[vertex,lightgray] at (v9){};
		
		\draw[edge,fill=green,fill opacity=.2, rounded corners=5pt] (v1) -- (v3) -- (v7) --cycle;
		\draw[edge, fill=red,fill opacity=.1, rounded corners=5pt] (v2) -- (v1)--(v8) -- cycle;
		\draw[edge, fill=yellow,fill opacity=.2, rounded corners=5pt] (v2) to (v9)--(v3)--cycle;
		\draw[edge,fill=green,fill opacity=.2, rounded corners=5pt, bend angle=40,bend right] (v6) to (v4) --(v7) --cycle;
		\draw[edge,fill=red,fill opacity=.1, rounded corners=5pt, bend angle=40,bend right] (v4) to (v5) --(v8) --cycle;
		\draw[edge,fill=yellow,fill opacity=.2, rounded corners=5pt, bend angle=40,bend right] (v5) to (v6) --(v9) --cycle;

		\draw[edge,->](3,0 )--(3.8,0 );
		
		\coordinate(x0) at (7,0);
		
		\coordinate(v1) at ($(x0) +(0:.5) $){};
		\coordinate(v2) at ($(x0) +(120:.5)$){};
		\coordinate(v3) at ($(x0) +(240:.5)$){};
		\coordinate(v4) at ($(x0) +(0:1.6) $){};
		\coordinate(v5) at ($(x0) +(120:1.6)$){};
		\coordinate(v6) at ($(x0) +(240:1.6)$){};
		\coordinate(v7) at ($(x0) +(300:.9)$){};
		\coordinate(v8) at ($(x0) +(60:.9)$){};
		\coordinate(v9) at ($(x0) +(180:.9)$){};
		
		\draw[edge] (v1) --(v8)--(v2)--(v9)--(v7)--(v1);
		\draw[edge] (v4) -- (v8) -- (v5) -- (v9) -- (v6) -- (v7) --(v4);
		\draw[edge] (x0) circle(1.6);
		\draw[edge] (x0) circle(0.5);
		
		\node[smallvertex] at (v1){};
		\node[smallvertex] at (v2){};
		\node[smallvertex] at (v3){};
		\node[smallvertex] at (v4){};
		\node[smallvertex] at (v5){};
		\node[smallvertex] at (v6){};
		\node[smallvertex] at (v7){};
		\node[smallvertex] at (v8){};
		\node[smallvertex] at (v9){};

		\coordinate(x0) at (-6,-4);
		\draw[edge] (x0) circle(1.2);
		\draw[edge] (x0) circle(0.5);
		
		\node[vertex,red,draw=black](v1) at ($(x0) + (45:.5) $){};
		\node[vertex,yellow,draw=black](v2) at ($(x0) +(135:.5)$){};
		\node[vertex,green,draw=black](v3) at ($(x0) +(225:.5)$){};
		\node[vertex,blue,draw=black](v4) at ($(x0) + (315:.5)$){};
		\node[vertex,red,draw=black](va) at ($(x0) +(45:1.2)$){};
		\node[vertex,yellow,draw=black](vb) at ($(x0) +(135:1.2)$){};
		\node[vertex,green,draw=black](vc) at ($(x0) +(225:1.2)$){};
		\node[vertex,blue,draw=black](vd) at ($(x0) +(315:1.2)$){};
		
		\draw[edge] (v1) -- (va);
		\draw[edge] (v2) -- (vb); 
		\draw[edge] (v3) -- (vc);
		\draw[edge] (v4) -- (vd);

		\draw[edge,->](-3.2,-4)--+(.8,0);

		\coordinate(x0) at (0,-4);
		
		\coordinate(v1) at ($(x0) +(0:.5) $){};
		\coordinate(v2) at ($(x0) +(90:.5)$){};
		\coordinate(v3) at ($(x0) +(180:.5)$){};
		\coordinate(v4) at ($(x0) +(270:.5) $){};
		\coordinate(v5) at ($(x0) +(0:1.6)$){};
		\coordinate(v6) at ($(x0) +(90:1.6)$){};
		\coordinate(v7) at ($(x0) +(180:1.6)$){};
		\coordinate(v8) at ($(x0) +(270:1.6)$){};
		\coordinate(va) at ($(x0) +(45:1)$){};
		\coordinate(vb) at ($(x0) +(135:1)$){};
		\coordinate(vc) at ($(x0) +(225:1)$){};
		\coordinate(vd) at ($(x0) +(315:1)$){};
		
		\node[vertex,lightgray] at (v1){};
		\node[vertex,lightgray] at (v2){};
		\node[vertex,lightgray] at (v3){};
		\node[vertex,lightgray] at (v4){};
		\node[vertex,lightgray] at (v5){};
		\node[vertex,lightgray] at (v6){};
		\node[vertex,lightgray] at (v7){};
		\node[vertex,lightgray] at (v8){};
		\node[vertex,lightgray] at (va){};
		\node[vertex,lightgray] at (vb){};
		\node[vertex,lightgray] at (vc){};
		\node[vertex,lightgray] at (vd){};

		\draw[edge,fill=red,fill opacity=.1, rounded corners=5pt] (v1) -- (v2) -- (va) --cycle;
		\draw[edge, fill=yellow,fill opacity=.2, rounded corners=5pt] (v2) -- (v3)--(vb) -- cycle;
		\draw[edge, fill=green,fill opacity=.2, rounded corners=5pt] (v3) to (v4)--(vc)--cycle;
		\draw[edge, fill=blue,fill opacity=.2, rounded corners=5pt] (v4) to (v1)--(vd)--cycle;
		\draw[edge,fill=red,fill opacity=.1, rounded corners=5pt, bend angle=40,bend right] (v5) to (v6) --(va) --cycle;
		\draw[edge,fill=yellow,fill opacity=.2, rounded corners=5pt, bend angle=40,bend right] (v6) to (v7) --(vb) --cycle;
		\draw[edge,fill=green,fill opacity=.2, rounded corners=5pt, bend angle=40,bend right] (v7) to (v8) --(vc) --cycle;
		\draw[edge,fill=blue,fill opacity=.2, rounded corners=5pt, bend angle=40,bend right] (v8) to (v5) --(vd) --cycle;
		
		\draw[edge,->](3,-4)--+(.8,0);
		
		\coordinate(x0) at (7,-4);
		
		\coordinate(v1) at ($(x0) +(0:.5) $){};
		\coordinate(v2) at ($(x0) +(90:.5)$){};
		\coordinate(v3) at ($(x0) +(180:.5)$){};
		\coordinate(v4) at ($(x0) +(270:.5) $){};
		\coordinate(v5) at ($(x0) +(0:1.6)$){};
		\coordinate(v6) at ($(x0) +(90:1.6)$){};
		\coordinate(v7) at ($(x0) +(180:1.6)$){};
		\coordinate(v8) at ($(x0) +(270:1.6)$){};
		\coordinate(va) at ($(x0) +(45:.9)$){};
		\coordinate(vb) at ($(x0) +(135:.9)$){};
		\coordinate(vc) at ($(x0) +(225:.9)$){};
		\coordinate(vd) at ($(x0) +(315:.9)$){};
		
		\draw[edge] (v1) --(va)--(v2)--(vb)--(v3)--(vc)--(v4)--(vd)--cycle;
		\draw[edge] (v5) -- (va) -- (v6) -- (vb) -- (v7) -- (vc) --(v8) --(vd) --cycle;
		\draw[edge] (x0) circle(1.6);
		\draw[edge] (x0) circle(0.5);
		
		\node[smallvertex] at (v1){};
		\node[smallvertex] at (v2){};
		\node[smallvertex] at (v3){};
		\node[smallvertex] at (v4){};
		\node[smallvertex] at (v5){};
		\node[smallvertex] at (v6){};
		\node[smallvertex] at (v7){};
		\node[smallvertex] at (v8){};
		\node[smallvertex] at (va){};
		\node[smallvertex] at (vb){};
		\node[smallvertex] at (vc){};
		\node[smallvertex] at (vd){};

		\coordinate(x0) at (-6,-8);
		\draw[edge] (x0) circle(1.2);
		\draw[edge] (x0) circle(0.5);
		
		\node[vertex,red,draw=black](v1) at ($(x0) + (36:.5) $){};
		\node[vertex,yellow,draw=black](v2) at ($(x0) +(108:.5)$){};
		\node[vertex,green,draw=black](v3) at ($(x0) +(180:.5)$){};
		\node[vertex,blue,draw=black](v4) at ($(x0) + (252:.5)$){};
		\node[vertex,purple,draw=black](v5) at ($(x0) + (324:.5)$){};
		\node[vertex,red,draw=black](va) at ($(x0) +(36:1.2)$){};
		\node[vertex,yellow,draw=black](vb) at ($(x0) +(108:1.2)$){};
		\node[vertex,green,draw=black](vc) at ($(x0) +(180:1.2)$){};
		\node[vertex,blue,draw=black](vd) at ($(x0) +(252:1.2)$){};
		\node[vertex,purple,draw=black](ve) at ($(x0) +(324:1.2)$){};
		
		\draw[edge] (v1) -- (va);
		\draw[edge] (v2) -- (vb); 
		\draw[edge] (v3) -- (vc);
		\draw[edge] (v4) -- (vd);
		\draw[edge] (v5) -- (ve);

		\draw[edge,->](-3.2,-8)--+(.8,0);

		\coordinate(x0) at (0,-8);
		
		\coordinate(v1) at ($(x0) +(0:.5) $){};
		\coordinate(v2) at ($(x0) +(72:.5)$){};
		\coordinate(v3) at ($(x0) +(144:.5)$){};
		\coordinate(v4) at ($(x0) +(216:.5) $){};
		\coordinate(v5) at ($(x0) +(288:.5) $){};
		\coordinate(v6) at ($(x0) +(0:1.6)$){};
		\coordinate(v7) at ($(x0) +(72:1.6)$){};
		\coordinate(v8) at ($(x0) +(144:1.6)$){};
		\coordinate(v9) at ($(x0) +(216:1.6)$){};
		\coordinate(v10) at ($(x0) +(288:1.6)$){};
		\coordinate(va) at ($(x0) +(36:1)$){};
		\coordinate(vb) at ($(x0) +(108:1)$){};
		\coordinate(vc) at ($(x0) +(180:1)$){};
		\coordinate(vd) at ($(x0) +(252:1)$){};
		\coordinate(ve) at ($(x0) +(324:1)$){};
		
		\node[vertex,lightgray] at (v1){};
		\node[vertex,lightgray] at (v2){};
		\node[vertex,lightgray] at (v3){};
		\node[vertex,lightgray] at (v4){};
		\node[vertex,lightgray] at (v5){};
		\node[vertex,lightgray] at (v6){};
		\node[vertex,lightgray] at (v7){};
		\node[vertex,lightgray] at (v8){};
		\node[vertex,lightgray] at (v9){};
		\node[vertex,lightgray] at (v10){};
		\node[vertex,lightgray] at (va){};
		\node[vertex,lightgray] at (vb){};
		\node[vertex,lightgray] at (vc){};
		\node[vertex,lightgray] at (vd){};
		\node[vertex,lightgray] at (ve){};

		\draw[edge,fill=red,fill opacity=.1, rounded corners=5pt] (v1) -- (v2) -- (va) --cycle;
		\draw[edge, fill=yellow,fill opacity=.2, rounded corners=5pt] (v2) -- (v3)--(vb) -- cycle;
		\draw[edge, fill=green,fill opacity=.2, rounded corners=5pt] (v3) to (v4)--(vc)--cycle;
		\draw[edge, fill=blue,fill opacity=.2, rounded corners=5pt] (v4) to (v5)--(vd)--cycle;
		\draw[edge, fill=purple,fill opacity=.2, rounded corners=5pt] (v5) to (v1)--(ve)--cycle;
		\draw[edge,fill=red,fill opacity=.1, rounded corners=5pt, bend angle=40,bend right] (v6) to (v7) --(va) --cycle;
		\draw[edge,fill=yellow,fill opacity=.2, rounded corners=5pt, bend angle=40,bend right] (v7) to (v8) --(vb) --cycle;
		\draw[edge,fill=green,fill opacity=.2, rounded corners=5pt, bend angle=40,bend right] (v8) to (v9) --(vc) --cycle;
		\draw[edge,fill=blue,fill opacity=.2, rounded corners=5pt, bend angle=40,bend right] (v9) to (v10) --(vd) --cycle;
		\draw[edge,fill=purple,fill opacity=.2, rounded corners=5pt, bend angle=40,bend right] (v10) to (v6) --(ve) --cycle;
		
		\draw[edge,->](3,-8)--+(.8,0);
		
		\coordinate(x0) at (7,-8);
		
		\coordinate(v1) at ($(x0) +(0:.6) $){};
		\coordinate(v2) at ($(x0) +(72:.6)$){};
		\coordinate(v3) at ($(x0) +(144:.6)$){};
		\coordinate(v4) at ($(x0) +(216:.6) $){};
		\coordinate(v5) at ($(x0) +(288:.6) $){};
		\coordinate(v6) at ($(x0) +(0:1.6)$){};
		\coordinate(v7) at ($(x0) +(72:1.6)$){};
		\coordinate(v8) at ($(x0) +(144:1.6)$){};
		\coordinate(v9) at ($(x0) +(216:1.6)$){};
		\coordinate(v10) at ($(x0) +(288:1.6)$){};
		\coordinate(va) at ($(x0) +(36:1)$){};
		\coordinate(vb) at ($(x0) +(108:1)$){};
		\coordinate(vc) at ($(x0) +(180:1)$){};
		\coordinate(vd) at ($(x0) +(252:1)$){};
		\coordinate(ve) at ($(x0) +(324:1)$){};
		
		\draw[edge] (v1) --(va)--(v2)--(vb)--(v3)--(vc)--(v4)--(vd)--(v5) --(ve) --cycle;
		\draw[edge] (v6) -- (va) -- (v7) -- (vb) -- (v8) -- (vc) --(v9) --(vd) --(v10) --(ve) --cycle;
		\draw[edge] (x0) circle(1.6);
		\draw[edge] (x0) circle(0.6);
		
		\node[smallvertex] at (v1){};
		\node[smallvertex] at (v2){};
		\node[smallvertex] at (v3){};
		\node[smallvertex] at (v4){};
		\node[smallvertex] at (v5){};
		\node[smallvertex] at (v6){};
		\node[smallvertex] at (v7){};
		\node[smallvertex] at (v8){};
		\node[smallvertex] at (v9){};
		\node[smallvertex] at (v10){};
		\node[smallvertex] at (va){};
		\node[smallvertex] at (vb){};
		\node[smallvertex] at (vc){};
		\node[smallvertex] at (vd){};
		\node[smallvertex] at (ve){};
		
	\end{tikzpicture}
	\caption{First instances of an infinite family of duals (first column), whose corresponding completed graphs (last column) exhaust the bound \cref{lem:TGN_bound} at $L=1\mod 3$ loops. The dual on six vertices is the same as in \cref{fig:leading_dual_graphs}, but drawn differently.}
	\label{fig:triangle_family}
\end{figure}
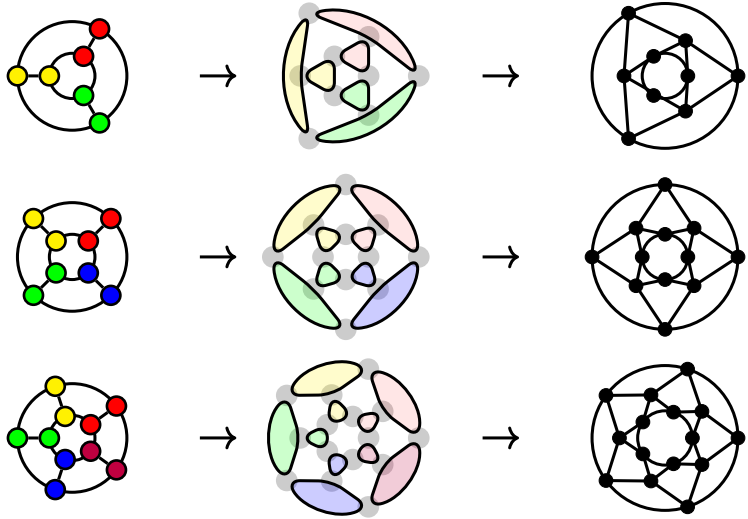

By \cref{lem:zigzag}, the degree of the zigzag coincides with the leading degree $\lfloor \frac 2 3 L + \frac 4 3\rfloor$ (\cref{lem:TGN_bound}) only for $L\in\left \lbrace 4,5,6, 8 \right \rbrace $. For other loop orders, the zigzag is not among the leading graphs at large $N$.

One can find infinite families of completions that contribute to leading order at all loops if one starts from the dual graphs (\cref{def:dual}). For the case of $L=3n+1$ loops, all we have to do is find a family of 3-connected 3-regular graphs. By \cref{thm:leading_3n+1}, the line graph of such graph on $2n$ vertices will be a primitive completion in $\phi^4$ theory at $3n+1$ loops. One such infinite family is shown in \cref{fig:triangle_family}

\subsection{Enumeration of graphs at low loop order} \label{sec:enumeration_low_loop_order}

\begin{table}[htb]
	\centering
	\begin{tblr}{ vlines, 
			vline{3}={1}{-}{solid},
			vline{3}={2}{-}{solid},
			vline{9}={1}{-}{solid},
			vline{9}={2}{-}{solid},
			hline{1}={solid},
			hline{2,Z}={solid},
			rowsep=0pt,
			cells={font=\fontsize{10pt}{12pt}\selectfont },
			columns={halign=r},
			row{1}  = { rowsep=2pt, halign=c, valign=m,font = \fontsize{10pt}{12pt}\selectfont } 
		}
		$L$ & $\lfloor N_\text{max} \rfloor$ & $n_0$ & $n_1$ & $n_2$ & $n_3$ & $n_4$ & $n_5$ &$\left \langle n_\text{max} \right \rangle_L $\\
		3 & 3 &  1& 0 & 0 & 0 &0&0 &3 \\
		4 & \textbf{4} &  1 & 0 & 0 & 0 & 0&0 &4\\
		5 & 4 &  2 & 0 & 0 & 0 &0 &0 &4 \\
		6 & 5 &  3 & 2 & 0 & 0 & 0 &0 &4.60\\
		7 & \textbf{6} &  2 & 12 & 0 & 0 & 0 &0 &5.14 \\
		8 & 6 & 22 & 27 & 0 & 0 & 0 &0 &5.45\\
		9 & 7 & 10 & 200 & 17 & 0 & 0 & 0 & 5.97 \\
		10 & \textbf{8} &  4 & 373 & 973 & 4 & 0 &0 & 6.28\\
		11 & 8 & 205 & 7,127 & 2,390 & 0 & 0 &0 &6.78\\
		12 & 9 & 57 & 11,427 & 68,354 & 1,467 &0 &0 &7.12\\
		13 & \textbf{10} & 14 & 6,850 & 364,479 & 384,170 & 130 &0 & 7.50\\
		14 & 10 & 2,278 & 499,780 & 6,129,768 & 1,003,848 & 3 &0 & 7.93\\
		15 & 11 & 425 & 324,100 & 21,776,598 & 59,847,247 &749,814 &0& 8.26 \\
		16 & \textbf{12} & 57 & 125,057 & 27,008,312 & 569,122,747 & 356,187,442 & 94,550 &  8.65
	\end{tblr}
	\caption{Number $n_k$ of (primitive) completions whose degree in $N$  is the $k$-th subleading order (\cref{def:nk}). Bold numbers indicate loop orders $L$ where $N_\text{max}$ (\cref{lem:TGN_bound}) is integer. With increasing loop order, most graphs have a degree $n_{k\geq 1}$ which is lower than $N_\text{max}$, i.e. they do not contribute to the leading order $n_0$. The last column shows the average degree (\cref{def:average_degree}). }
	\label{tab:loworder_leading_graphs_count}
\end{table}

\begin{figure}[htb]
	\begin{subfigure}{ .49 \linewidth}
		\centering
		{\small Orders of $T(g,N)$, log log plot}\\
		\includegraphics[width=\linewidth]{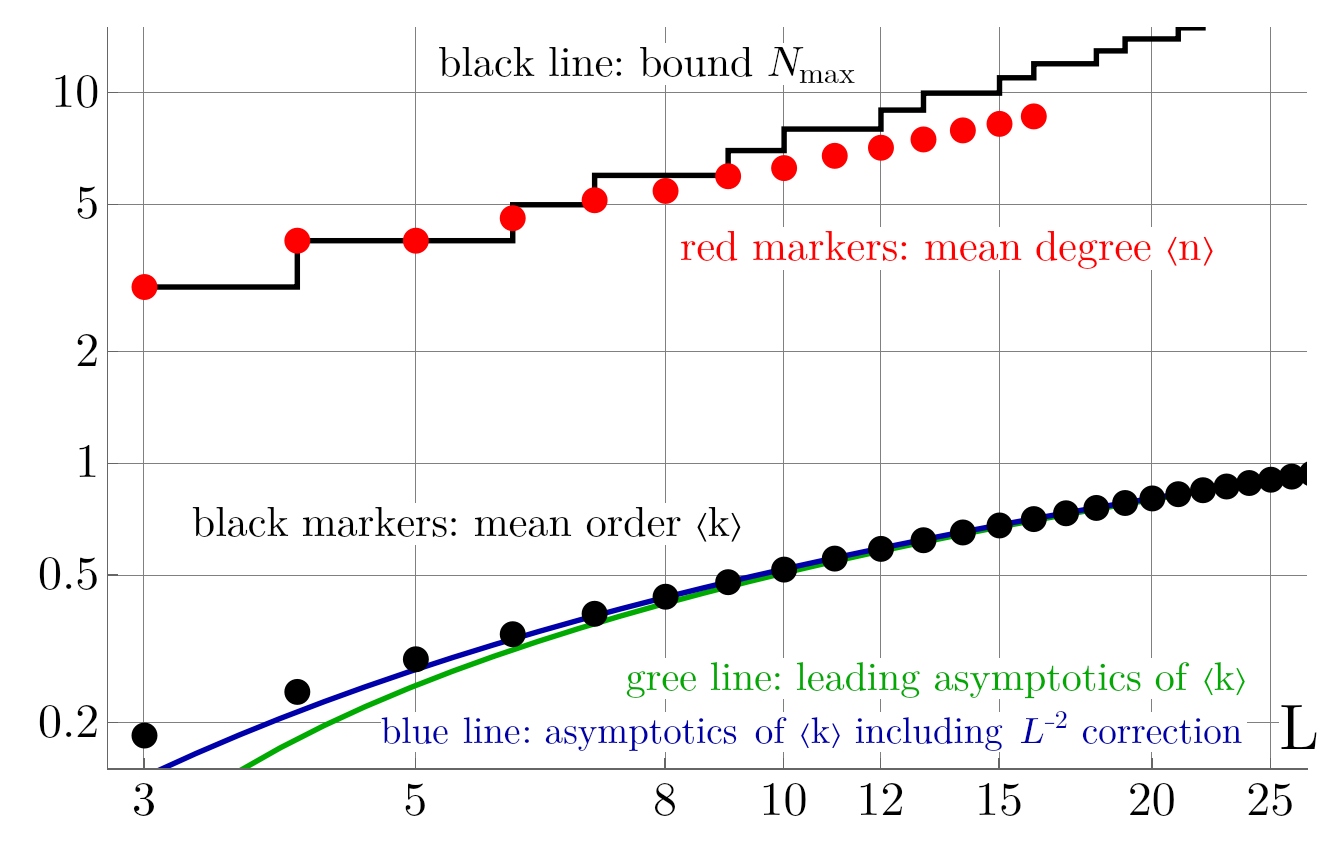}
		\subcaption{}
		\label{fig:mean_max_order}
	\end{subfigure}
		\begin{subfigure}{ .49 \linewidth}
		\centering
		{\small Largest roots of beta function}\\
		\includegraphics[width=\linewidth]{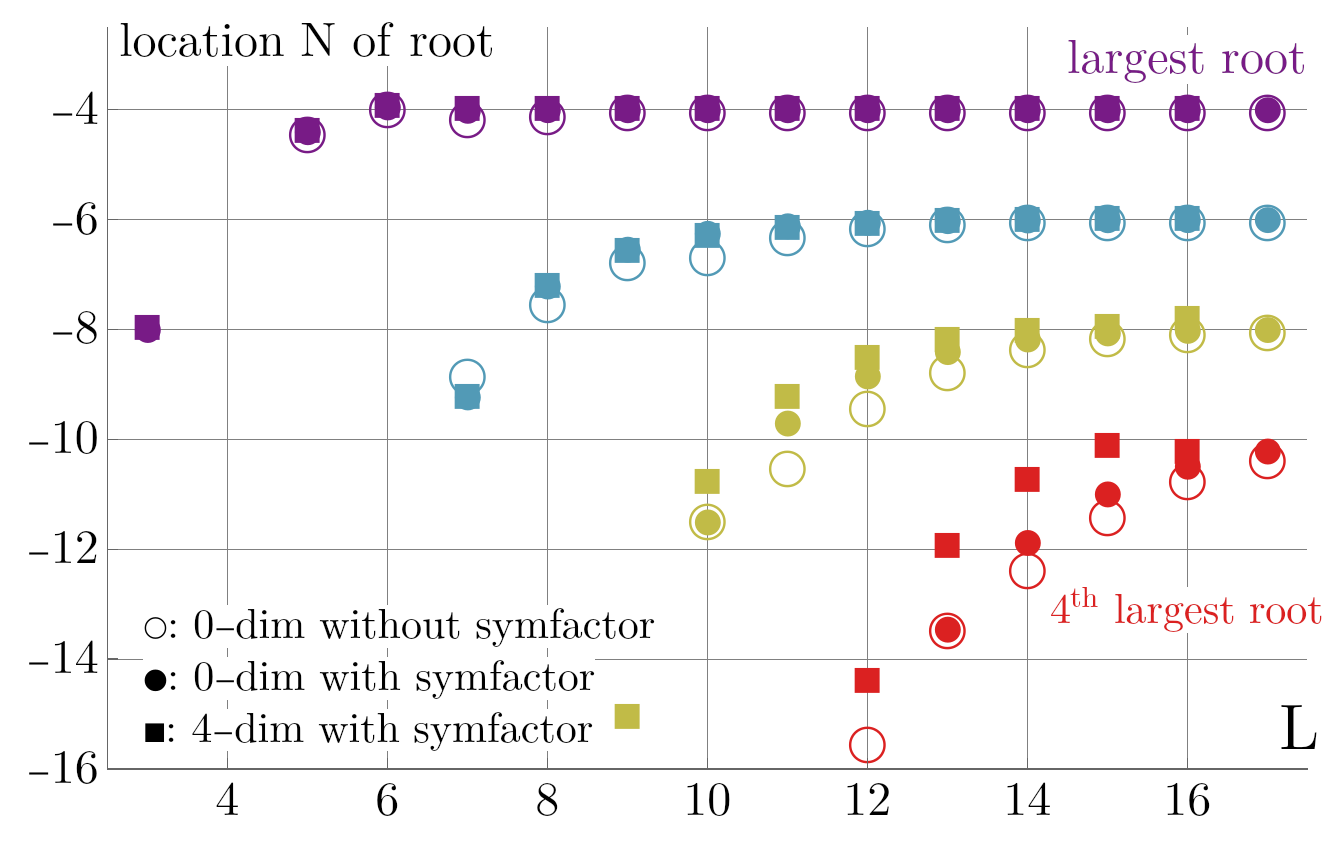}
		\subcaption{}
		\label{fig:largest_roots}
	\end{subfigure}
	
	\caption{\textbf{(a)} The mean degree $\left \langle n \right \rangle _L$ (red markers) appears to lie on a straight line in this log log plot, which suggests that it grows polynomially in $L$. The bound $N_\text{max}$ (black line) grows proportional to $L^1$ (\cref{lem:TGN_bound}). For comparison, the mean order $\left \langle k \right \rangle _L$ (\cref{def:average_order}, black markers) grows only logarithmically in $L$ (\cref{average_order_asymptotics}, blue and green lines). Compare \cref{fig:mean_order}. \textbf{(b)} Largest roots of the primitive beta function in 4-dimensions (squares, \cref{tab:coefficients}), of $p_L(N)$ (filled circles, \cref{tab:0dim_primitive_coefficients}), and of the sum of all primitive $T(G,N)$ without symmetry factor (empty circles, \cref{tab:leading_decompositions_count}). In all cases, the roots converge towards even negative integers as $L$ grows, as expected from \cref{primitive_asymptotics,beta_asymptotics}, and this convergence is rapid already at low loop orders. It indicates the presence of sums over all channels of subgraphs (\cref{lem:circuit_polynomial_zeros}). }
\end{figure}

To verify and complement the combinatorial construction of leading primitive graphs, we have enumerated all primitive graphs of $\phi^4$ theory up to including 16 loops and determined their $\O(N)$ symmetry factors $T(G,N)$. The degree in $N$ is bounded by $N_\text{max}$ (\cref{lem:TGN_bound}), we introduce
\begin{align}\label{def:nk}
	n_k:= \text{Number of graphs where $T(G,N)$ has degree $(N_\text{max}-k)$.}
\end{align}
These numbers are shown in \cref{tab:loworder_leading_graphs_count}. Specifically, $n_0$ is the number of leading graphs in~$N$, and our findings confirm the count of graphs constructed from dual graphs (\cref{tab:leading_graphs_count}). Furthermore, we observe in \cref{tab:loworder_leading_graphs_count} that, as $L$ grows, the degree of most graphs is lower than $N_\text{max}$, or equivalently, the number $n_0$ of leading graphs is increasingly smaller than the numbers of subleading graphs. This is expected from the asymptotic analysis in \cref{sec:0dim_largeN}: The large-$N$ limit is not factorially divergent, therefore, for large enough $L$, the number $n_k$ for any fixed $k$ grows less than factorially. 

In \cref{def:average_order}, we introduced the average \emph{order} in $N$ of the polynomials $T(G,N)$ and found that it grows only logarithmically in $L$ (\cref{average_order_asymptotics}). Analogously, we define the 
average \emph{degree} as 
\begin{align}\label{def:average_degree}
	\left \langle n \right \rangle_L &:= \text{Average degree in $N$ of $T(G,N)$ over all $L$-loop completions $G$}
\end{align}
Numbers are shown in \cref{tab:loworder_leading_graphs_count}. In the log-log plot in \cref{fig:mean_max_order}, the bound $N_\text{max}\sim \frac 2 3 L$ approaches a straight line with slope 1. Similarly, it seems that $\left \langle n \right \rangle _L$ approaches a straight line, but with smaller slope, which suggests that that $\left \langle n \right \rangle _L$ grows polynomially in $L$ with exponent smaller than 1.

\subsection{Non-primitive graphs}\label{sec:nonprimitive}

Knowing a bound for the order in $N$ of primitive graphs, we can use \cref{lem:factorization_T} to obtain a bound for graphs with subdivergences. Recall \cref{lem:factorization_T}: The $\O(N)$-symmetry factor is multiplicative under insertion of propagator-type subgraphs. When $G$ is a completion where $T(G,N)$ is of degree $n$, then removing one edge $e$ produces a propagator-type graph $G\setminus \left \lbrace e \right \rbrace $ where $T(G \setminus \left \lbrace e \right \rbrace ,N)=\frac{T(G,N)}{N}$ is of degree $n-1$, and removing one vertex $v$ produces a vertex-type graph $g:=G\setminus \left \lbrace v \right \rbrace $ where $T(g,N)=\frac{3T(G,N)}{N(N+2)}$ is of degree $n-2$.  

\begin{lemma}\label{lem:subgraph_bound}
	Let $g$ be a vertex-type graph that arises from inserting a primitive  decompletion $g_2$ with $L_2>1$ loops into a primitive  decompletion $g_1$ with $L_1>1$ loops. The resulting graph $g$ is a decompletion with $L_1+L_2$ loops and 
	\begin{align*}
		\textnormal{Degree in $N$ of the polynomial $T(g,N)$} \leq \frac 23 \left( L_1+L_2 \right) - \frac 43.
	\end{align*}
\end{lemma}
\begin{proof}
	The stated loop order of the joined graph trivially follows from counting the edges and vertices, and noting that $g$ is connected.

	Let $G_1=g_1\cup \left \lbrace v \right \rbrace  $ be the completion of $g_1$, and let $G:=G_1\circ_v g_2$ be the sum of graphs that arises from inserting $g_2$ in place of a vertex $v\in G_1$ in all possible ways, where $v\neq v_1$. According to \cref{lem:factorization_T},  $T(G,N)=T(G_1,N)\cdot T(g_2,N)$. By \cref{lem:TGN_bound}, the degree in $N$ of $T(G_1,N)$ is $\leq \frac 2 3 L_1 + \frac 4 3$. Similarly, for $g_2$ the degree is bounded by $\frac 2 3 L_2 -\frac 2 3$.  Hence, the degree of $T(G,N)$ is bounded by 
	\begin{align*}
		\frac 2 3 L_1 + \frac 4 3 + \frac 2 3 L_2 -\frac 2 3 = \frac 23 \left( L_1+L_2 \right) + \frac 2 3.
	\end{align*}
	This bound refers to $G$, which is a sum of graphs, but the coefficients of $T(G,N)$ are positive, therefore there can be no cancellations and the bound holds for every graph individually. The graph $g$ in the statement of the lemma is a decompletion of one of the graphs in $G$ (where the vertex $v$ is the decompletion vertex). The degree of $T(g,N)=\frac{3T(G,N)}{N(N+2)}$ is two less than the degree of $T(G,N)$, which yields the desired bound.
\end{proof}

A Feynman graph has \emph{coradical degree} $n+1$ if it can be obtained from a primitive graph by inserting primitive graphs $n$ times. It is irrelevant whether the inserted graphs are distinct, and whether they are inserted into disjoint locations or into each other. Equivalently, in the Hopf algebra of Feynman graphs \cite{kreimer_overlapping_1999}, the coradical degree is the number of times the coproduct needs to be applied such that all resulting terms contain at least one factor $\mathbbm{1}$. A primitive graph has coradical degree one, a graph with one divergent subgraph has coradical degree two, and so on. Dyson-Schwinger equations, when formulated as identities for generating function of Feynman graphs, ensure that every graph can be built by a (not necessarily unique) sequence of insertions of primitive graphs \cite{kreimer_anatomy_2006}. 

From \cref{lem:subgraph_bound}, one can derive bounds on the degree in $N$ of non-primitive graphs, and the interesting result is that these bounds can be expressed in terms of the coradical degree. The renormalized Feynman amplitude of a superficially divergent graph of coradical degree $n$ depends on the momentum scale $s$ like $\ln(s)^n$. Consequently, bounding the degree of $T(G,N)$ in terms of coradical degree $n$ implies a bound for the $N$ dependence of terms in the leading-log expansion. The precise value of these bounds depends on the class of graphs under consideration (connected/1PI/\ldots), and whether one excludes certain small graphs (tadploes/multiedges/\ldots). As an illustration, we derive just one particular bound of this type.

We call a graph \emph{fish free} if neither the underlying primitive graph, nor any of the  subgraphs inserted at any stage, are the fish graph (i.e. the 1-loop vertex correction given by a double edge, \cref{fig:fish}). Being fish free is a nontrivial assumption, this situation may arise e.g. if a modification of the Feynman rules excludes the presence of double edges, or if double edges have been analytically resummed already. Conversely, the vanishing of tadpole graphs (i.e. graphs with only one external edge, such that they do not depend on momenta) is a choice of renormalization condition that can be made at will, and which is automatic in a massless theory in dimensional regularization.

\begin{theorem}\label{thm:degree_bound}
	Assume that tadpole graphs vanish. 
	Let $g$ be a 1PI fish free $L$-loop vertex-type graph of coradical degree $n$. Then
	\begin{align*}
	\textnormal{Degree in $N$ of the polynomial $T(g,N)$} &\leq \frac 2 3 L - \frac 1 3 n.
	\end{align*}

\end{theorem}
\begin{proof}
	For insertions of vertex-type subgraphs, the result follows recursively from \cref{lem:subgraph_bound}.

	Propagator-type graphs can be obtained from vertex-type graphs by either of two operations: Either by joining two of the external edges to each other, which raises the loop order by one and the degree in $N$ by one, or by joining three external edges to a new vertex, which increases the loop order by two and the degree in $N$ by one. In both cases, the coradical degree of the resulting propagator-type graph is one larger than the one of the vertex-type subgraph. The first operation gives zero contribution if tadpoles vanish because when, during the renormalization process, the vertex-type subgraph is contracted, the resulting cograph is a tadpole and hence the product vanishes. We are left with the second operation. If $L$ is the loop order of the resulting propagator-type graph, the degree of the vertex-type subgraph is bounded by $\frac 2 3 (L-2) - \frac 2 3 (n-1)= \frac 2 3 L-\frac 2 3 n -\frac 2 3$, and the degree of the resulting propagator-type graph is bounded by
	\begin{align*}
		\frac 2 3 L - \frac 2 3 n - \frac 2 3 +1 = \frac 2 3 L - \frac 2 3 n + \frac 1 3.
	\end{align*}
	We see that every such iteration raises the degree by $\frac 1 3$, compared to the bound $\frac 2 3 L - \frac 2 3 n$. Hence, a valid bound is given by $\frac 2 3 L - \frac 1 3 n$. 
	
	When only vertex-type subgraphs are being inserted, the graph has bound $\frac 2 3 L - \frac 2 3 n$, which in particular is stronger than $\frac 2 3 L - \frac 1 3 n$. When a propagator subgraph is inserted, the $T(G,N)$ are multiplicative, too (\cref{lem:factorization_T}), consequently their degree is additive. Likewise, the coradical degree and the loop number are additive, and therefore, the resulting graph again is bounded by $\frac 2 3 L - \frac 1 3 n$.
\end{proof}

\section{$\phi^4$ theory in four dimensions} \label{sec:4d}

Everything so far has been about identifying and enumerating graphs, and understanding their combinatorial properties. 
In this last section, we take into account the numerical values of primitive graphs in $\phi^4$ theory in four spacetime dimensions, which are the Feynman periods (\cref{sec:periods}), in order to examine to what extent the qualitative combinatorial findings stay true when Feynman graphs are weighted by other factors beyond their $O(N)$- and automorphism- symmetry factors.

\subsection{Coefficients and evaluations of $\beta^\text{prim}_L$ in four dimensions } \label{sec:4d_coefficients}

For $L \leq 13$ loops, the periods $\period(G)$ of all primitive graphs had been computed in \cite{balduf_statistics_2023}\footnote{This dataset can be downloaded from \href{https://doi.org/10.5683/SP3/NLEDGH}{DOI 10.5683/SP3/NLEDGH}.}.
Inserting these values into \cref{betaprim_expansion}, 
\begin{align}\label{beta_expansion2}
	\beta^\text{prim}_L(N)&= 2\sum_{ \substack{G \text{ comp.}\\ L \text{ loops}} }  \frac{ 4!(L+2)}{\abs{\Aut(G)}} \frac{3T(G,N)}{N(N+2)} \period (G) ,
\end{align}
we obtain the coefficients in $N$ of $\beta^\text{prim}_L$. For $14 \leq L \leq 18$ loops, only non-complete uniform random samples of periods are available in \cite{balduf_statistics_2023} due to the large number of graphs. As discussed in \cref{sec:enumeration_low_loop_order}, not every graph contributes to every order in $N$, this implies that the statistical uncertainties of the large-$N$ coefficients is large. 
To improve the statistical situation, we have numerically integrated all leading-degree graphs of $L\in \left \lbrace 14,15,16 \right \rbrace $ loops using the methods explained \cite{balduf_statistics_2023}, based on the \texttt{tropical-feynman-quadrature} program \cite{borinsky_tropical_2023,borinsky_tropical_2023a}.
Consequently, for $L>13$, the subleading coefficients suffer from sampling uncertainty, but the leading ones do not. Numerical values are shown in \cref{tab:coefficients} in \cref{sec:tables}. 

\medskip 

The beta function at small positive integer values of $N$ is associated with the critical behaviour of various physical models. In \cref{tab:evaluations} in \cref{sec:tables}, we report the evaluations of $\beta^\text{prim}_L$ at integer $N$. 
Once more, the data for $L \leq 13$ was taken from \cite{balduf_statistics_2023}, and these values are also reported in small print in \cref{tab:evaluations} for larger $L$. However, for $14 \leq L \leq 17$, we additionally used the weighted sampling algorithm developed in \cite{balduf_predicting_2024} to obtain more accurate values, reported in larger print in \cref{tab:evaluations}. We did not compute new data at $L=18$ due to limited computing resources. Note that the sampling algorithm does the weighting for a specific given value of $N$, hence it needs to run multiple times to compute $\beta^\text{prim}_L$ at multiple $N$, but this is still substantially faster than inferring all the evaluations of $\beta^\text{prim}_L$ from a single uniform sample. 

As a remark, it is not trivial to relate the uncertainties in \cref{tab:coefficients} to those in \cref{tab:evaluations}, because both arise from the (numerical and statistical) uncertainties of the periods, while $T(G,N)$ are known exactly. Each column in each of the tables constitutes a different linear combination of periods.

\subsection{$N$ dependence of the beta function in zero and four dimensions} \label{sec:0dim_4dim_relation}

The mean of the 4-dimensional period is much larger than unity, and grows with the loop order, numbers had been given in \cite{balduf_statistics_2023}.  It is therefore clear that the function $2p_L(N)$ from 0-dimensional QFT is not a good approximation of the 4-dimensional  $\beta^\text{prim}_L(N)$ in terms of absolute value; however, it might still be a good approximation for the $N$ dependence. This will be investigated in the present section.

A first characteristic feature of the $N$ dependence is the presence of roots (zeros) at negative $N$. In the limit $L\rightarrow \infty$, both in zero dimensions (\cref{primitive_asymptotics}) and in four dimensions (\cref{beta_asymptotics}), these roots are asymptotically located at negative even integers $N\leq -4$.
For finite $L$, the largest roots are plotted in \cref{fig:largest_roots}. Indeed, we see that both in the 0-dimensional and in the 4-dimensional case, the roots quickly converge to negative even integers as $L$ grows, as had been observed already in  \cite{pobylitsa_superfast_2008,kompaniets_minimally_2017}. In particular, the largest roots are close to their asymptotic location already at low loop orders, long before the absolute value or the growth rate of the beta function is close to its asymptotics (compare \cref{fig:prim_asymptotics_correction,fig:0dim_primitive_ratio}).

Next, we compare the $N$ dependence of $p_L(N)$ and $\beta^\text{prim}_L(N)$ for positive $N$. To this end, recall that in \cref{def:average_period}, we introduced the $N$-dependent average $L$-loop period, weighted by symmetry factors, $\left \langle \period \right \rangle _{T/\Aut, L}(N) = \beta^\text{prim}_L(N) / 2 p_L(N)$. 
If the ratio of $p_L(N)$ and $\beta^\text{prim}_L(N)$ would not depend on $N$,
then $\left \langle \period \right \rangle _{T/\Aut, L}(N) $ would be independent of~$N$. We hence examine whether  $\left \langle \period \right \rangle  _{T/\Aut,L}(N) \overset ?\approx \left \langle \period \right \rangle _{T/\Aut,L}(1)$.  This is equivalent to asking whether
\begin{align}\label{beta_function_factorization}
	\beta^\text{prim}_L (N)~ &\overset{?}\approx ~  \left \langle \period \right \rangle _{\frac T \Aut,L}(1) \cdot 2p_L(N) =  \beta^\text{prim}_L(1) \cdot \frac{p_L(N)}{p_L(1)}.
\end{align} 
We introduce the $N$ dependent ratio between the left- and right-hand side of \cref{beta_function_factorization},
\begin{align}\label{def:beta_N_ratio}
	b_L (N)&:= \frac{\beta^\text{prim}_L(N)}{\beta^\text{prim}_L(1)}\frac{p_L(1)}{p_L(N)}= \frac{\left \langle \period \right \rangle _{\frac T \Aut, L}(N)}{\left \langle \period \right \rangle _{\frac T \Aut,L}(1)} .
\end{align} 

\begin{example}
	At $L=5$, the two primitives have almost the same $\O(N)$ symmetry factors, $T(g_1,N)= 528+186N+15N^2$ vs $T(g_2, N)=526 + 189N+14N^2$, and their periods are very close, $\period(G_1) \approx 52.01$, $\period(G_2)\approx 55.58$. This implies that the ratio $b_5(N)$ is close to unity, and the beta function at 5 loops almost factorizes (\cref{beta_function_factorization}). 
	
	Conversely, for $L=8$, 	Inserting the numerical values of \cref{tab:coefficients} for $\beta^\text{prim}_L$ and the rational numbers of \cref{tab:0dim_primitive_coefficients} for $p_L$, one finds a function that notably differs from unity,
	\begin{align*}
		b_6(N) &\approx \frac{ 13603 + 7150.5N + 1221.0N^2 + 75.809N^3 + 1.2273N^4}  {13704+7097.8N + 1179.6N^2 + 69.410N^3 + 1.0000N^4 } .
	\end{align*}
\end{example}

For finite $L$, $b_L(N)$ is a rational function which has poles at negative $N$ because the zeros of $p_L(N)$ and $\beta^\text{prim}_L(N)$ do not exactly coincide. However,  $b_L(N)$ is smooth for $N\geq -2$, and  it is unity at $N=1$ by construction. 
Plots of $b_L(N)$ for different $L$ are shown in \cref{fig:beta_N_ratio}. We observe that $b_L(N)$ grows monotonically with $N$, which indicates that there is a correlation between the $O(N)$ symmetry factor $T(G,N)$ of  a graph and the value of its period  $\period(G)$.

For finite $L$, the limit  $N \rightarrow \infty$ of $b_L(N)$ is finite, and it measures  the average of  periods which are leading at large $N$ (\cref{sec:duals_leading}), divided by the average of all periods. The limit is shown in \cref{fig:bL_limit}, it grows with $L$ with what could be a power law. 
A value of $b_L(\infty)$ larger than unity means that the average of the leading-$N$-periods is systematically larger than the average of \emph{all} periods. If $b_L(\infty)$ grows only polynomially with $L$, this means that the large-$N$ expansion of \emph{primitive} graphs is convergent not only in 0-dimensional QFT (\cref{sec:0dim_largeN}) but also in 4 dimensions.

\begin{figure}[htbp]
	\begin{subfigure}{ .49 \linewidth}
		\centering
		{\small Ratio $b_L(N)$ between 4-dim. and 0-dim. $\beta_L(N)$} \\
		\includegraphics[width=\linewidth]{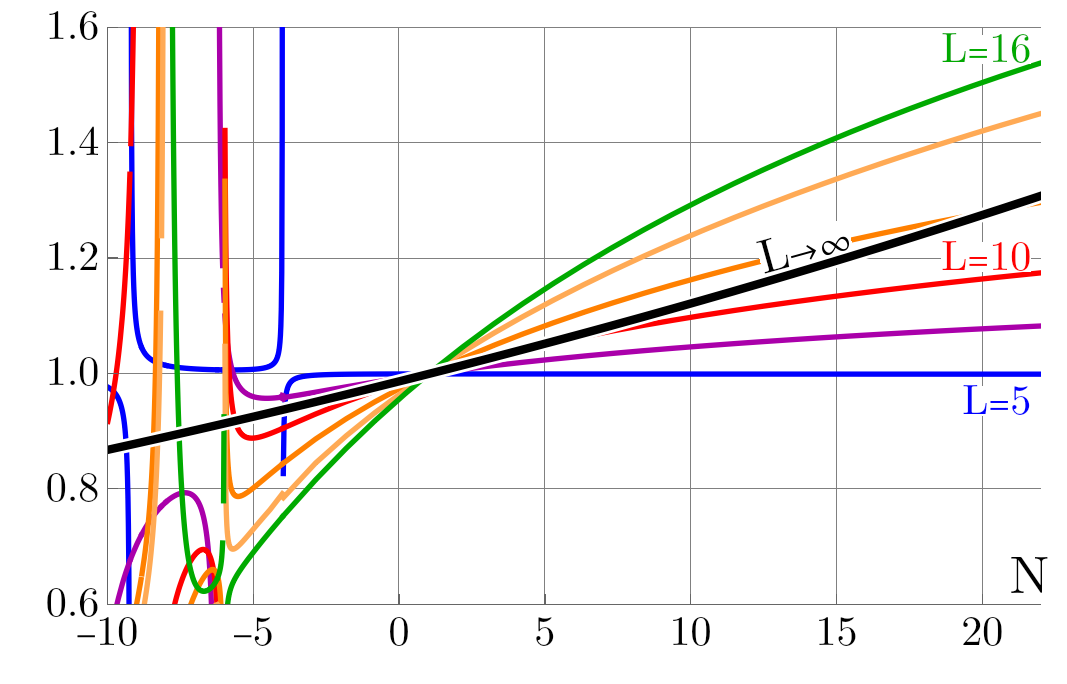}
		\subcaption{}
		\label{fig:beta_N_ratio}
	\end{subfigure}
	\begin{subfigure}{ .49 \linewidth}
		\centering
		{\small Limit $N \rightarrow \infty$ of ratio $b_L(N)$}\\
		\includegraphics[width=\linewidth]{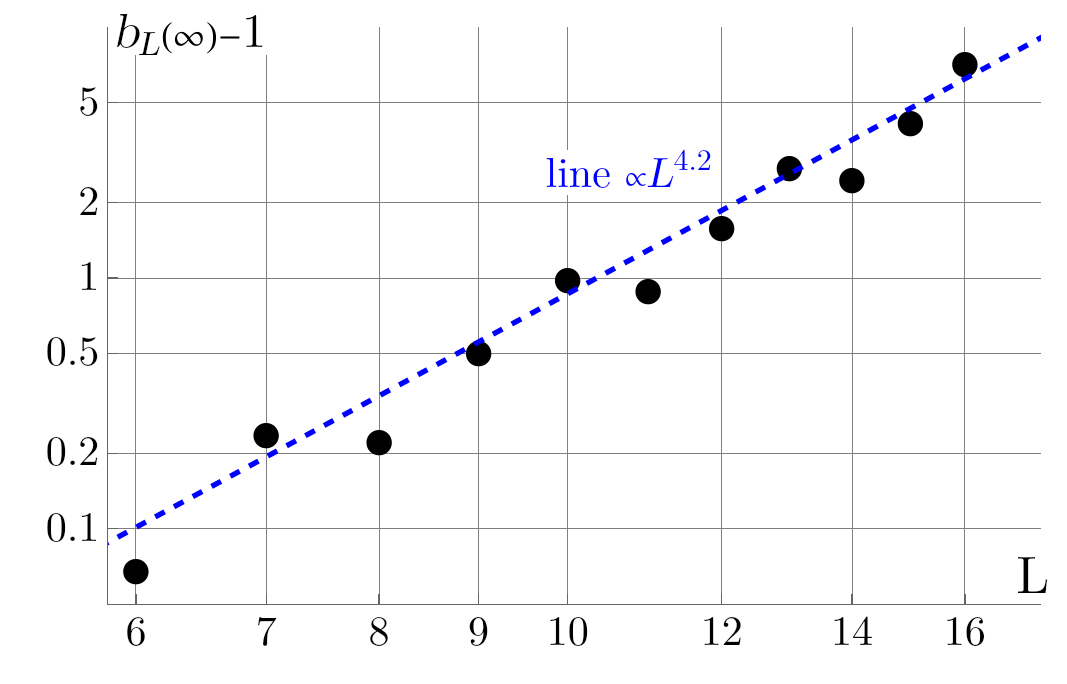}
		\subcaption{}
		\label{fig:bL_limit}
	\end{subfigure}
	
	\caption{\textbf{(a)} Ratio $b_L(N)$ (\cref{def:beta_N_ratio}) between the 4-dimensional and the 0-dimensional primitive beta function, for $L\in \left \lbrace 5,8,10,12,14,16 \right \rbrace $. For finite $L$, the ratio converges to a constant at $N\rightarrow \infty$. The limiting curve for $L\rightarrow \infty$ is shown in black, it is unbounded as $N\rightarrow \infty$. \textbf{(b)} Numerical value of the limit $b_L(\infty)$. Heuristically, the values seem to grow like $L^{4.2}$. }
\end{figure}

Asymptotically for $L\rightarrow \infty$, we know from \cref{mean_period_asymptotics} that $b_L(N) \rightarrow \delta^{N-1}$ with $\delta \approx 1.013$.  This function is shown as a black curve in \cref{fig:4dim_beta_ratio}, it does not match the curves for finite $L$ particularly well, even for $L=16$.  In this sense, $L=16$ is not in the \enquote{asymptotic regime}.  However, note that both $p_L(N)$ and $\beta^\text{prim}_L(N)$ are strongly dependent on $N$ and their values span multiple orders of magnitude  for $N\in [-2, 20]$. The ratio $b_L(N)$ is slowly varying in this domain and  takes values in $[0.8, 1.5]$, which shows that qualitatively, the $N$ dependence of the 0-dimensional $p_L(N)$ is indeed a rather good model for the $N$ dependence of the 4-dimensional $\beta^\text{prim}_L(N)$.

\subsection{Asymptotic growth rate of the beta function at large loop order}\label{sec:4d_asymptotics}

As stated in \cref{sec:primitive_beta_function}, there is a conjecture that the leading asymptotic growth of the primitive beta function should coincide with the one of the full beta function in MS, and therefore it should grow according to \cref{beta_conjecture}.
In \cite{balduf_statistics_2023} we observed a potential discrepancy between \cref{beta_conjecture} and the numerically computed \emph{primitive} beta function. In that work, limited numerical accuracy for $L>13$ loops prevented a definite conclusion. 

The asymptotic growth ratio $r_L$ (\cref{def:rL}) of the primitive beta function behaves as $r_L = a + \frac{c_s}{a} \frac 1 L + \mathcal O (L^{-2})$, where the conjectured asymptotics implies $ a=1,  c_s=5+\frac 12 N$ (\cref{beta_asymptotics_predicted}). We have numerical data for $\beta^\text{prim}_L$ for $L\leq 18$, which corresponds to $r_L$ for $L\leq 17$. We can determine the growth parameters $a$ and $c_s$ by plotting our numerical values of $r_L$ against $\frac 1 L$ and extracting slope and $y$-intersection of a linear fit. The result is shown in \cref{fig:4dim_beta_ratio}. We observe that the linear fits (red lines) result in $a \approx 0.8$, which depends only little on the value of $N$ used for the fit. The slope of the empirical fit lines is approximately described by $c_s \approx 8 + 0.8 N$. These empirical values of $a$ and $c_s$ for $L \leq 18$ loops are incompatible with the conjectured asymptotics (\cref{beta_asymptotics_predicted}), even when numerical uncertainties are taken into account. This is clearly visible from the fact that the red data points in \cref{fig:4dim_beta_ratio} do not lie on the green lines.

\begin{figure}[htb]
	\begin{subfigure}{ .49 \linewidth}
		\centering
		{\small Growth rate $r_L$ of $\beta^\text{prim}$, 4 dimensions}\\
		\includegraphics[width=\linewidth]{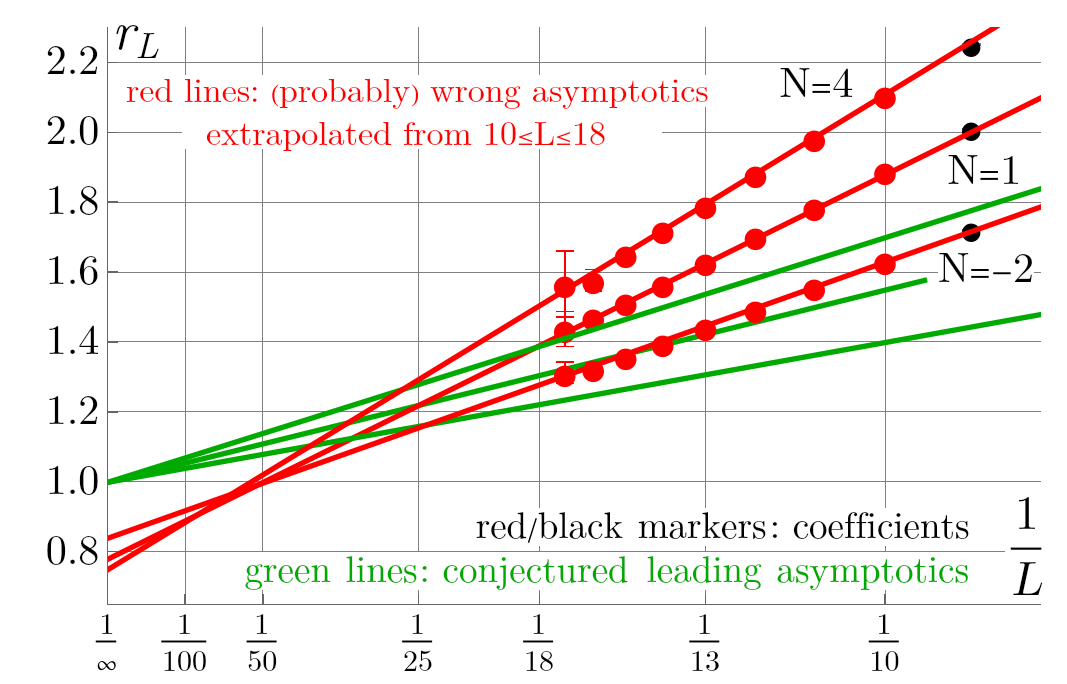}
		\subcaption{}
		\label{fig:4dim_beta_ratio}
	\end{subfigure}
	\begin{subfigure}{ .49 \linewidth}
		\centering
		{\small Growth rate $f_L$ of $\left \langle \period \right \rangle _{T/\Aut}$, 4 dimensions}\\
		\includegraphics[width= \linewidth]{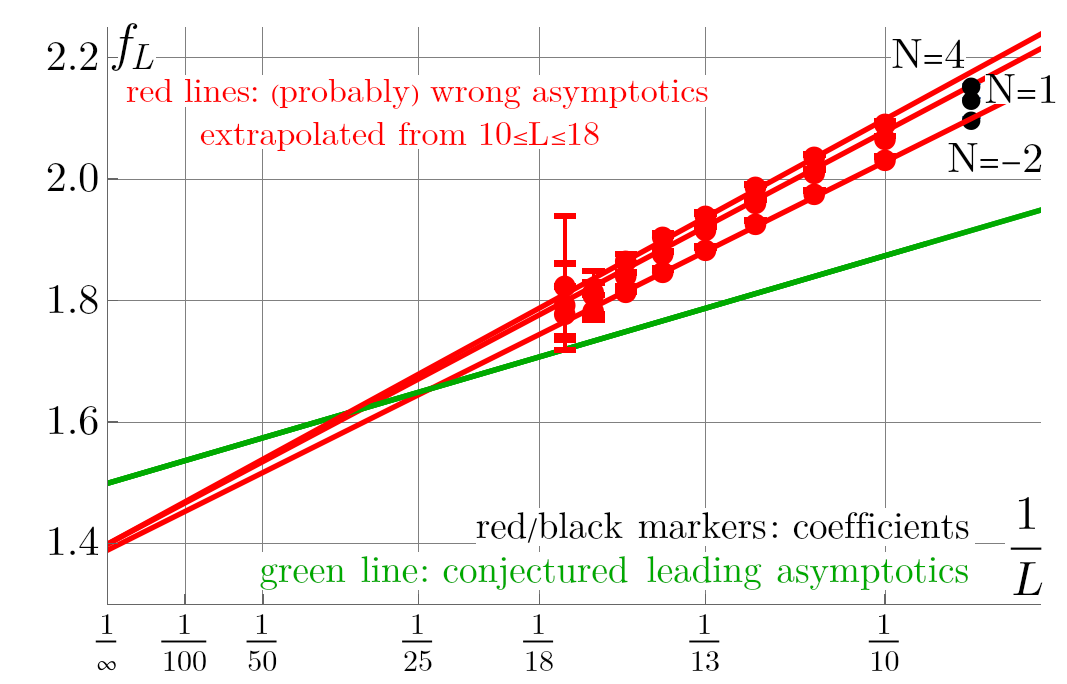}
		\subcaption{}
		\label{fig:4dim_mean_ratio}
	\end{subfigure}
	\caption{\textbf{(a)} Ratio $r_L$ according to \cref{def:rL} for the primitive beta function $\beta^\text{prim}$ in four dimensions. Green lines indicate the leading asymptotics of the full beta function in minimal subtraction, which is conjecturally equal to that of $\beta^\text{prim}$. Red lines indicate the numerically measured growth rate from $10 \leq L \leq 18$ loops. If the pattern continues, the data can be expected to reach the true asymptotics at $L\geq 25$ loops. Compare the strikingly similar situation in zero dimensions, \cref{fig:0dim_primitive_ratio}. \textbf{(b)} Growth ratio $f_L$ (\cref{def:fL}) of the average period for the 4-dimensional theory. This quantity is independent of the sum of symmetry factors $p_L(N)$. \Cref{beta_conjecture}   implies that the data should lie on the green line for large $L$,  for all values of $N$.}
\label{fig:growth_ratio}
\end{figure}

However, the pattern in \cref{fig:4dim_beta_ratio} is strikingly similar to the analogous plot in 0-dimensional QFT, \cref{fig:0dim_primitive_ratio}. In the 0-dimensional case, we have verified that the growth rate $r_L$  changes its slope at $L\geq 25$ loops, and starts to converge towards the correct limiting value (green lines), while it appears to converge towards a too low limit when only $10 \leq L \leq 18$ is extrapolated (red lines). Remarkably, even in the 4-dimensional case (\cref{fig:4dim_beta_ratio}), the extrapolated red lines intersect the conjectured asymptotics at $L\approx 25$. In view of these observations, we conclude that the growth ratio that can be determined from the numerical data for the primitive beta function up to 18 loops is probably not the true asymptotic growth ratio at $L\rightarrow \infty$. We conjecture that this ratio will change when higher order terms are included, and a numerically reliable value will emerge only upwards of $L\approx 25$ loops. 

We have now seen that both the 0-dimensional $p_L(N)$, and the 4-dimensional $\beta_L^\text{prim}(N)$, grow at a ratio $r_L(N)$ that is incompatible with the leading-order asymptotics when $L<25$. This effect would be trivial if all periods had the same value (independent of $L$), since then $\beta_L^\text{prim}(N)$ would be a constant multiple of $p_L(N)$, and necessary show the same phenomena regarding large-order asymptotics and $N$-dependence. However, we know that the periods on average grow exponentially with $L$ (asymptotically at $L\rightarrow \infty$), and that their distribution within a fixed $L$ contains notable outliers \cite{balduf_statistics_2023}. Now there are two scenarios: It could be that the growth of the average period is already very close to its leading asymptotics, and that the distribution within a loop order is independent of the $O(N)$ symmetry factor $T(G,N)$. In that case, one would conclude that the phenomena observed around $L\approx 25$ in \cref{fig:0dim_primitive_ratio,fig:4dim_beta_ratio} are caused by the symmetry factors $T(G,N)$ as encoded in $p_L(N)$. Or, it could be that there is a correlation between value of the period and properties of $T(G,N)$, or that the value of the period itself changes its growth rate around $L\approx 25$. 

To investigate this, we compute the average period per graph, weighted by the symmetry factor $\frac{T(G,N)}{\abs{\Aut(G)}}$, as defined in \cref{def:average_period}. 
The growth ratio $f_L$ of this average (\cref{def:fL}) is expected to behave, asymptotically, as $f_L \sim \frac 3 2 + \frac{15}{4}\frac 1 L$. In a Domb-Sykes plot where the $x$-axis is $\frac 1 L$, we hence expect a linear function whose slope is $\frac {15}4$, independent of $N$. This plot is shown in   \cref{fig:4dim_mean_ratio}, we see that the numerical values of $f_L$ scale almost linearly with $\frac 1 L$, and their slope and $y$-intersect is almost independent of $N$. Nevertheless, the fit lines from $10\leq L \leq 18$ suggest an asymptotic growth ratio of $f_L\rightarrow 1.4$, which is incompatible with the expectation $f_L\rightarrow \frac 3 2$. Moreover, these fit lines intersect the expected leading asymptotic around $L\approx 25$, suggesting once more that the average of the period at fixed loop order is reliably described by its leading-order asymptotics only upwards of $L\approx 25$. We hence conclude that the similarity between \cref{fig:0dim_primitive_ratio} and \cref{fig:4dim_beta_ratio} is not \emph{only} caused by the fact that both quantities contain the same set of $O(N)$ symmetry factors $T(G,N)$, but instead, the mean of the periods itself shows a similar effect of growing at non-asymptotic rates for $L\leq 25$. 

One might wonder if the similarity of these two effects really is just a statistical artefact: We know from  \cite{balduf_predicting_2024} that the numerical value of the period of an individual graph is strongly correlated with its $\O(N)$ symmetry factor $T(G,N)$. It would thus be conceivable that, since the periods in $\left \langle \period \right \rangle_{T/\Aut} $ are weighted proportionally to their $\frac{T(G,N)}{\abs{\Aut(G)}}$, this correlation could be responsible for the observed similarity between the growth of $\left \langle \period \right \rangle_{T/\Aut} $ and $p_L$. This is not the case. Firstly, at $N=1$ the factor $T(G,N)$ is unity for all graphs and hence does not influence the relative weighting, while the pattern in \cref{fig:4dim_mean_ratio} persists at $N=1$. Secondly, the correlation changes sign for different values of $N$\footnote{Empirically, the correlation is roughly ~ $\ln \period(G) \approx \frac{-c_0 }{(N-1)}+\frac{c_1}{(N-1)} \ln T(G,N)$, for $c_0 \sim 150$, $c_1 \sim 10$.}, but the observed growth rate $f_L$ of $\left \langle \period \right \rangle_{T/\Aut} $ in \cref{fig:4dim_mean_ratio} is almost independent of $N$. We thus conclude that the correlation between $T(G,N)$ and $\period(G)$ can not be the reason for the fact that both the sum of $T(G,N)$, and the sum of $\period(G)$, display a similar pattern of reaching their leading large-order growth rates only upwards of $25$ loops.

\section*{Acknowledgements}

We thank Karen Yeats and Erik Panzer for comments and discussions. 

PHB did parts of this work while affiliated with the University of Waterloo and Perimeter Institute. Research at Perimeter
Institute is supported in part by the Government of Canada through the Department of Innovation, Science and Economic Development and by the Province of Ontario through the Ministry of Colleges and Universities.
Work at the University of Oxford was funded through Royal Society grant URF/R1/201473.

JT has been funded by the German research foundation (DFG, grant number 418838388) and in the last stage by the European Union (ERC, GE4SPDE, 101045082), as well as by Germany’s Excellence Strategy EXC 2044–390685587, Mathematics M\"unster: Dynamics–Geometry–Structure.
 
\appendix

\section{Combinatorial properties of the circuit partition polynomial} \label{sec:combinatorial_properties}

In \cref{sec:symmetric_theory}, we stated several facts about the circuit partition polynomial $J(G,N)$ (\cref{def:circuit_partition_polynomial}), and the $\O(N)$-symmetry factor $T(G,N)$ (\cref{def:TGN}). The present section contains derivations and further examples. 

One of the crucial properties of the polynomials $T(G,N)$ and $J(G,N)$ is that they factorize  under insertion of subgraphs (\cref{lem:factorization_T}). To see why this is the case, we first give a proof for a special case of this statement, namely for the completion (\cref{def:completion_decompletion}) of a vertex-type (i.e. 4-valent) graph in $\phi^4$ theory. 

\begin{figure}[htb]
	\centering
	\begin{tikzpicture}
		\coordinate(x0) at (0,0);
		
		\coordinate(v0) at ($(x0) + (-.7,0) $){};
		\coordinate(v1) at ($(x0) +(.6,.9)$){};
		\coordinate(v2) at ($(x0) +(.5, .3)$){};
		\coordinate(v3) at ($(x0) +(.5,-.3)$){};
		\coordinate(v4) at ($(x0) +(.6, -.9)$){}; 
		
		\node[vertex, lightgray] at ($(v0)+(.2,0)$){};
	  \node at ($(v0)+ (-.2,0)$){$v$};
		
		\draw[edge,blue, rounded corners=5pt] ($(v1) $) ..controls +(1, .5) and +(1,-.2) .. ($(v2) $); 
		\draw[edge,blue, rounded corners=5pt] ($(v3) $) ..controls +(1, .3) and +(1.5,-.8) .. ($(v4) $); 
		\node[ellipse, draw=blue,edge, blue, minimum height=.7cm , minimum width= .5cm] at ($(x0)+(1.6,.2) $) { };
		\node[ellipse ,minimum height=3cm , minimum width=1.7cm, fill=blue,fill opacity=.1, text opacity=1,text=black] at ($(x0)+(1.1,0) $) { $ g$ };
		
		\draw[edge,black,bend angle =5,bend right, rounded corners=12pt] (v1) to (v0) to (v2);
		\draw[edge,black,bend angle =5,bend right, rounded corners=12pt] (v3) to (v0) to (v4);
		
		\node[anchor=center] at ($(x0)+ (.5,-2)$){$N^2 \cdot J(g,N)$};

		\coordinate(x0) at (5,0);
		
		\coordinate(v0) at ($(x0) + (-.7,0) $){};
		\coordinate(v1) at ($(x0) +(.6,.9)$){};
		\coordinate(v2) at ($(x0) +(.5, .3)$){};
		\coordinate(v3) at ($(x0) +(.5,-.3)$){};
		\coordinate(v4) at ($(x0) +(.6, -.9)$){}; 
		
		\node[vertex, lightgray] at ($(v0)+(.2,0)$){};
		\node at ($(v0)+ (-.2,0)$){$v$};
		
		\draw[edge,blue, rounded corners=5pt] ($(v1) $) ..controls +(1, .5) and +(1,-.2) .. ($(v2) $); 
		\draw[edge,blue, rounded corners=5pt] ($(v3) $) ..controls +(1, .3) and +(1.5,-.8) .. ($(v4) $); 
		\node[ellipse, draw=blue,edge, blue, minimum height=.7cm , minimum width= .5cm] at ($(x0)+(1.6,.2) $) { };
		\node[ellipse ,minimum height=3cm , minimum width=1.7cm, fill=blue,fill opacity=.1, text opacity=1,text=black] at ($(x0)+(1.1,0) $) { $ g$ };
		
		\draw[edge,black,bend angle =5,bend right, rounded corners=12pt] (v1) to (v0) to (v3);
		\draw[edge,black,bend angle =5,bend right, rounded corners=12pt] (v2) to (v0) to (v4);
		
		\node[anchor=center] at ($(x0)+ (.5,-2)$){$N \cdot J(g,N)$};
		
		\coordinate(x0) at (10,0);
		
		\coordinate(v0) at ($(x0) + (-.7,0) $){};
		\coordinate(v1) at ($(x0) +(.6,.9)$){};
		\coordinate(v2) at ($(x0) +(.5, .3)$){};
		\coordinate(v3) at ($(x0) +(.5,-.3)$){};
		\coordinate(v4) at ($(x0) +(.6, -.9)$){}; 
		
		\node[vertex, lightgray] at ($(v0)+(.2,0)$){};
	\node at ($(v0)+ (-.2,0)$){$v$};
		
		\draw[edge,blue, rounded corners=5pt] ($(v1) $) ..controls +(1, .5) and +(1,-.2) .. ($(v2) $); 
		\draw[edge,blue, rounded corners=5pt] ($(v3) $) ..controls +(1, .3) and +(1.5,-.8) .. ($(v4) $); 
		\node[ellipse, draw=blue,edge, blue, minimum height=.7cm , minimum width= .5cm] at ($(x0)+(1.6,.2) $) { };
		\node[ellipse ,minimum height=3cm , minimum width=1.7cm, fill=blue,fill opacity=.1, text opacity=1,text=black] at ($(x0)+(1.1,0) $) { $ g$ };
		
		\draw[edge,black,bend angle =5,bend right, rounded corners=8pt] (v1) to (v0) to (v4);
		\draw[edge,black,bend angle =5,bend right, rounded corners=13pt] (v2) to (v0) to (v3);
		
		\node[anchor=center] at ($(x0)+ (.5,-2)$){$N \cdot J(g,N)$};

	\end{tikzpicture}
	\caption{Proof of \cref{lem:vertex_decomposition_proportional}. We fix an arbitrary decomposition in the vertex-type graph $g$. This decomposition amounts to a choice to connect the four external edges into pairs, indicated by blue arcs, and zero or more circuits in $g$. The completion vertex $v$ allows for three distinct decompositions, drawn as black lines. One of them (first drawing) produces two circuits and hence a factor of $N^2$, the remaining two decompositions produce only one circuit and hence a factor $N$. Hence, the circuit partition polynomial of the completion is $J(G,N)=(N^2+2N)\cdot J(g,N)$. }
	\label{fig:proof:decomposition}
\end{figure}
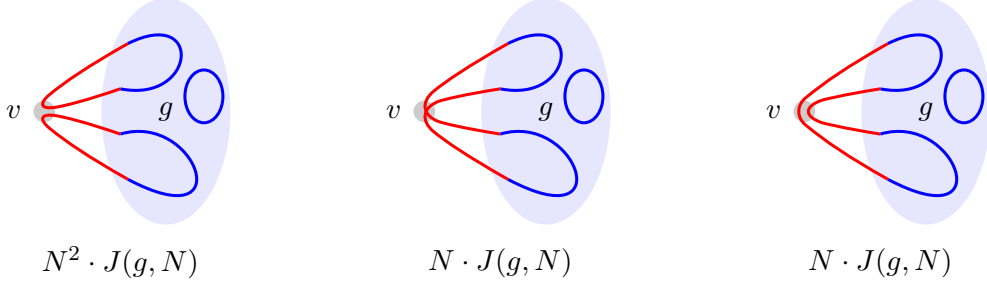

\begin{lemma}[\cite{kleinert_critical_2001} Section~6]\label{lem:vertex_decomposition_proportional}
	Let $g$ be a 4-valent graph and $G    =g\uplus \left \lbrace v \right \rbrace $ its completion. Then the circuit partition polynomials satisfy 
\begin{align*}
	J(G,N)= N(N+2) \cdot J(g,N).
\end{align*}
\end{lemma}
\begin{proof}
	Every possible decomposition of $g$ amounts to a choice to connect the four external edges into two pairs, plus a choice to join all remaining edges into circuits. Now consider the completion $G:=g\uplus \left \lbrace v \right \rbrace $. $G$ has no external edges, therefore, every possible vertex decomposition results in circuits only (and no \enquote{open ends}). The sum over all possible decompositions, assigning the value $N$ to every circuit, is by definition the circuit partition polynomial $J(G,N)$ (\cref{def:circuit_partition_polynomial}).
	
	On the other hand, every decomposition of $G$ amounts to choosing one of the three possible decompositions of the vertex $v$ together with a choice of decomposition of $g$. Since every decomposition of $g$ is, especially, a choice of joining the four external edges of $g$ into pairs, of the three decompositions of $v$ exactly one produces two circuits, and the other ones produce one circuit each, regardless of which particular decomposition of $g$ was chosen. This is shown in \cref{fig:proof:decomposition}. These three summands hence contribute $(N^2+N+N)=N(N+2)$ to the circuit partition polynomial $J(G,N)$. Since this is true for \emph{every} decomposition of $g$, we conclude that $J(G,N) = N(N+2) \cdot J(g,N)$. 
\end{proof}

A special case of \cref{lem:vertex_decomposition_proportional} occurs if the 4-valent graph $g$ is a single vertex. Then, $J(g,N)=3$ 
and we obtain the 3-loop multiedge graph $\melon$ as the completion of $g$ with
\begin{align}\label{eq:melon polynomial}
 J(\melon,N)&= 3 N(N+2),
\end{align}
which is consistent with the $L=3$ case of \eqref{eq:bubble-ring} in \cref{ex:fish_chain}. If we transform $J$ into $T$ (\cref{def:TGN}), the factor 3 in \cref{eq:melon polynomial} is cancelled, and 
\cref{lem:vertex_decomposition_proportional} can   be written as 
\begin{align}
T(G,N) &= T \left( \melon, N \right)
T(g,N), \qquad  T(\melon,N)=  \frac{ N(N+2)}{3}.
\end{align}
The other way round, starting from a vacuum graph $G$, we obtain the $\O(N)$ symmetry factor of any decompletion $G\setminus \left \lbrace v \right \rbrace $ as 
\begin{align}\label{decompletion_TGN}
	T \left( G \setminus \left \lbrace v \right \rbrace , N \right) &= 
	\frac{T(G,N)}{T(\melon,N)} 
\end{align}
This
in particular implies that all decompletions $g$ of a vacuum graph $G$ share the same $\O(N)$~symmetry factor $T(g,N)$, even if they are non-isomorphic graphs. 

\begin{example}\label{ex:TGN_example_completion}
The vacuum graph $G$ shown in \cref{fig:completion} has
\begin{align*}
	T(G,N) &= \frac{1}{3^7}(1056 N + 900 N^2 + 216 N^3 + 15 N^4)
 = \frac{1}{3^6}N(N+2)(N+8)(5N+22).
\end{align*}
Each of the (non-isomorphic) decompletions shown in \cref{fig:completion} has, by explicit calculation,
\begin{align*}
T \left( G\setminus \left \lbrace v \right \rbrace , N \right) &= \frac{1}{729}\left( 176+62N+5N^2 \right) 
= \frac{1}{3^6}(N+8)(5N+22).
\end{align*}
This is the expected result from \cref{decompletion_TGN}.
\end{example}

To generalize \cref{lem:vertex_decomposition_proportional} to the full statement of \cref{lem:factorization_T}, we need decompositions of $(2p)$-valent vertices, where $p$ is an integer. 
A 2-valent vertex ($p=1$) has only a single decomposition, and for $p=2$ we obtain the three choices shown in \cref{fig:vertex_decomposition}. In general, the decomposition of a $(2p)$-valent vertex is the sum over all $(2p-1)!! = (2p-1)(2p-3)(2p-5)\cdots 1$  choices to match the $(2p)$ adjacent edges into pairs.

\begin{lemma}\label{lem:completion_general}
	Let $p$ be a positive integer and let $g$ be a $(2p)$-valent graph   with completion $G=g \uplus \left \lbrace v \right \rbrace $ (\cref{def:completion_decompletion}). Then
	\begin{align*}
	J(G,N) &= \Big( \prod \nolimits_{j=0}^{p-1}(N+2j) \Big) \cdot J(g,N).
	\end{align*} 
\end{lemma}
\begin{proof}
	We use induction. For $p=1$, the completion vertex $v$ is 2-valent and hence has only a single decomposition. This decomposition joins the two external edges of $g$ and gives rise to one new circuit, so $J(G,N)=N\cdot J(g,N)$. For $p=2$, the statement is \cref{lem:vertex_decomposition_proportional}. Assume that the statement is true for $g$ with $2p-2$ external edges. 
	
	Fix any one decomposition of $g$. It implies a  matching of the $2p$ external edges of $g$ into pairs. Without loss of generality, assume that edge 1 is matched to edge 2 within $g$ as shown in \cref{fig:proof_completion_general}. The decompositions of the completion vertex $v$ are pairs, too. Consider the edge that is matched to edge 1. If it is edge 2, a new circuit arises, giving rise to a factor $N$. If it is any one of the $2p-2$ other external edges, no circuit is created and we obtain a factor of $2p-2$ by summing over these cases, see \cref{fig:proof_completion_general}. In all cases, the remaining graph has $2p-2$ unmatched edges, and its completion can be obtained from the induction hypothesis. We thus have $J(G,N)  = (N+2( p-1)) \cdot \big( \prod_{j=0}^{p-2}(N+2j) \big) \cdot J(g,N)$ and the claim follows.
\end{proof}

\begin{figure}[htb]
	\centering
	\begin{tikzpicture}[scale=.8]
		\coordinate(x0) at (0,0);
		
		\coordinate(v0) at ($(x0) + (-.7,0) $){};
		\coordinate(v1) at ($(x0) +(.6,1.2)$){};
		\coordinate(v2) at ($(x0) +(.5, .8)$){};
		\coordinate(v3) at ($(x0) +(.3, .2)$){};
		\coordinate(v4) at ($(x0) +(.3, -.2)$){}; 
		\coordinate(v5) at ($(x0) +(.5,-.8)$){};
		\coordinate(v6) at ($(x0) +(.6, -1.2)$){};
 
		\node[vertex, lightgray] at ($(v0)+(.2,0)$){};
		\node at ($(v0)+ (-.2,0)$){$v$};
		
		\draw[edge,blue, rounded corners=5pt] ($(v1) $) ..controls +(1, .5) and +(1,-.2) .. ($(v2) $); 
		\draw[edge,blue, rounded corners=5pt] ($(v3) $) ..controls +(.3, .3) and +(1.5,-.8) .. ($(v4) $); 
		\draw[edge,blue, rounded corners=5pt] ($(v5) $) ..controls +(1, .3) and +(1 ,-.6) .. ($(v6) $); 
		\node[ellipse, draw=blue,edge, blue, minimum height=.7cm , minimum width= .5cm] at ($(x0)+(1.6,.2) $) { };
		\node[ellipse ,minimum height=3cm , minimum width=1.7cm, fill=blue,fill opacity=.1, text opacity=1,text=black] at ($(x0)+(1.1,0) $) { $ g$ };
		
		\draw[edge,black ,bend angle =5,bend right, rounded corners=10pt] (v1) to (v0) to (v2);

		\node[anchor=center] at ($(x0)+ (1,-2.2)$){$N  $};

		\coordinate(x0) at (5,0);
		
			\coordinate(v0) at ($(x0) + (-.7,0) $){};
		\coordinate(v1) at ($(x0) +(.6,1.2)$){};
		\coordinate(v2) at ($(x0) +(.5, .8)$){};
		\coordinate(v3) at ($(x0) +(.3, .2)$){};
		\coordinate(v4) at ($(x0) +(.3, -.2)$){}; 
		\coordinate(v5) at ($(x0) +(.5,-.8)$){};
		\coordinate(v6) at ($(x0) +(.6, -1.2)$){};
		
		\node[vertex, lightgray] at ($(v0)+(.2,0)$){};
		\node at ($(v0)+ (-.2,0)$){$v$};
		
		\draw[edge,blue, rounded corners=5pt] ($(v1) $) ..controls +(1, .5) and +(1,-.2) .. ($(v2) $); 
		\draw[edge,blue, rounded corners=5pt] ($(v3) $) ..controls +(.3, .3) and +(1.5,-.8) .. ($(v4) $); 
		\draw[edge,blue, rounded corners=5pt] ($(v5) $) ..controls +(1, .3) and +(1 ,-.6) .. ($(v6) $); 
		\node[ellipse, draw=blue,edge, blue, minimum height=.7cm , minimum width= .5cm] at ($(x0)+(1.6,.2) $) { };
		\node[ellipse ,minimum height=3cm , minimum width=1.7cm, fill=blue,fill opacity=.1, text opacity=1,text=black] at ($(x0)+(1.1,0) $) { $ g$ };
		
		\draw[edge,black ,bend angle =5,bend right, rounded corners=6pt] (v1) to (v0) to (v3);
		
		\node[anchor=center] at ($(x0)+ (1,-2.2)$){$1 $};
		
		\node at (8.5,0){$\ldots$};
		
		\coordinate(x0) at (11,0);
		
			\coordinate(v0) at ($(x0) + (-.7,0) $){};
		\coordinate(v1) at ($(x0) +(.6,1.2)$){};
		\coordinate(v2) at ($(x0) +(.5, .8)$){};
		\coordinate(v3) at ($(x0) +(.3, .2)$){};
		\coordinate(v4) at ($(x0) +(.3, -.2)$){}; 
		\coordinate(v5) at ($(x0) +(.5,-.8)$){};
		\coordinate(v6) at ($(x0) +(.6, -1.2)$){};
		
		\node[vertex, lightgray] at ($(v0)+(.2,0)$){};
		\node at ($(v0)+ (-.2,0)$){$v$};
		
		\draw[edge,blue, rounded corners=5pt] ($(v1) $) ..controls +(1, .5) and +(1,-.2) .. ($(v2) $); 
		\draw[edge,blue, rounded corners=5pt] ($(v3) $) ..controls +(.3, .3) and +(1.5,-.8) .. ($(v4) $); 
		\draw[edge,blue, rounded corners=5pt] ($(v5) $) ..controls +(1, .3) and +(1 ,-.6) .. ($(v6) $); 
		\node[ellipse, draw=blue,edge, blue, minimum height=.7cm , minimum width= .5cm] at ($(x0)+(1.6,.2) $) { };
		\node[ellipse ,minimum height=3cm , minimum width=1.7cm, fill=blue,fill opacity=.1, text opacity=1,text=black] at ($(x0)+(1.1,0) $) { $ g$ };
		
		\draw[edge,black,bend angle =5,bend right, rounded corners=3pt] (v1) to (v0) to (v6);
		
		\node[anchor=center] at ($(x0)+ (1,-2.2)$){$1$};
		
		\draw[tikzbrace] (-1,-2.6) -- node[below=2mm] {total factor $N+1+\ldots +1= N+(2p-2)$, remaining $g$ has $2(p-1)$ external edges.} (13,-2.6);

	\end{tikzpicture}
	\caption{Induction step of \cref{lem:completion_general}. Fix an arbitrary decomposition in the $(2p)$-valent graph $g$. This implies a choice to connect the $(2p)$ external edges into pairs, indicated by blue arcs. In the decompositions of $v$, there are $(2p-1)$ ways to connect the first external edge, shown as black args. In the first case, a factor $N$ arises (left), but not in the $(2p-2)$ remaining cases.  }
	\label{fig:proof_completion_general}
\end{figure}

When the graph $g$ in \cref{lem:completion_general} is itself a $2p$-valent vertex, the completion $G=\pmelon$ is a $(2p-1)$-loop multiedge. A $(2p)$-valent vertex has $(2p-1)!!$ decompositions, consequently the circuit partition polynomial of $J(G,N)$ must have $\left( (2p-1)!! \right) ^2$ terms. The product in \cref{lem:completion_general} has only $(2p-1)!!$ factors, therefore every factor occurs $(2p-1)!!$ times. We obtain \cref{lem:multiedge_J}, $	J(\pmelon,N) = (2p-1)!!\prod_{j=0}^{p-1}(N+2j)$.

\begin{example}
Joining two fish graphs (each of which has four external edges) produces 24 summands. Eight of them are isomorphic to a ring of four double edge graphs, and 16 are isomorphic to another five-loop graph, as shown in \cref{fig:joining_fish}. Upon identifying  isomorphic graphs and dividing by $4!=24$, we obtain   $\fish\uplus\fish = \frac 1 3 G_1 + \frac 2 3 G_2$.
\end{example}

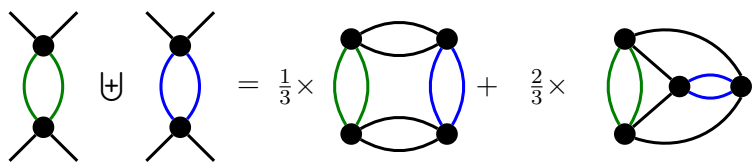
\begin{figure}[htb]
	\centering 
	\begin{tikzpicture}[scale=.9]
	
		\coordinate(x0) at (-3.5,.5);
		\node[vertex](v1) at ($(x0) +(0,-.6)$){};
		\node[vertex](v2) at ($(x0) +(0,.6)$){};
		\draw[edge,darkgreen,bend angle=40,bend left] (v1) to (v2);
		\draw[edge,darkgreen,bend angle=40,bend right] (v1) to (v2);
		\draw[edge] (v2) -- +(135:.7);
		\draw[edge] (v1) -- +(225:.7);
		\draw[edge] (v2) -- +(45:.7);
		\draw[edge] (v1) -- +(-45:.7);
		
		\node at ($(x0)+(1,0)$) {$\biguplus$};
		
		\coordinate(x0) at (-1.5,.5);
		\node[vertex](v1) at ($(x0) +(0,-.6)$){};
		\node[vertex](v2) at ($(x0) +(0,.6)$){};
		\draw[edge,blue,bend angle=40,bend left] (v1) to (v2);
		\draw[edge,blue,bend angle=40,bend right] (v1) to (v2);
		\draw[edge] (v2) -- +(135:.7);
		\draw[edge] (v1) -- +(225:.7);
		\draw[edge] (v2) -- +(45:.7);
		\draw[edge] (v1) -- +(-45:.7);
		
		\node at ($(x0)+(1,0)$) {$=$};

		\coordinate(x0) at (1,.5);
		\node at ($(x0)+(-.8,0)$) {$\frac 13\times$};
		\node[vertex](v1) at ($(x0) +(0,-.7)$){};
		\node[vertex](v2) at ($(x0) +(0,.7)$){};
		\node[vertex](v3) at ($(x0) +(1.4,-.7)$){};
		\node[vertex](v4) at ($(x0) +(1.4,.7)$){};
		\draw[edge,darkgreen,bend angle=30,bend left] (v1) to (v2);
		\draw[edge,darkgreen,bend angle=30,bend right] (v1) to (v2);
		\draw[edge,blue,bend angle=30,bend left] (v3) to (v4);
		\draw[edge,blue,bend angle=30,bend right] (v3) to (v4);
		\draw[edge,bend angle=30,bend left] (v1) to (v3);
		\draw[edge,bend angle=30,bend right] (v1) to (v3);
		\draw[edge,bend angle=30,bend left] (v2) to (v4);
		\draw[edge,bend angle=30,bend right] (v2) to (v4);

		\coordinate(x0) at (5,.5);
		\node at ($(x0)+(-1.5,0)$) {$+\quad \frac 23\times$};
		\node[vertex](v1) at ($(x0) +(0,-.7)$){};
		\node[vertex](v2) at ($(x0) +(0,.7)$){};
		\node[vertex](v3) at ($(x0) +(.8,0)$){};
		\node[vertex](v4) at ($(x0) +(1.7,0)$){};
		\draw[edge,darkgreen,bend angle=30,bend left] (v1) to (v2);
		\draw[edge,darkgreen,bend angle=30,bend right] (v1) to (v2);
		\draw[edge,blue,bend angle=30,bend left] (v3) to (v4);
		\draw[edge,blue,bend angle=30,bend right] (v3) to (v4);
		\draw[edge,bend angle=40,bend right] (v1) to (v4.south);
		\draw[edge,bend angle=40,bend left] (v2) to (v4.north);
		\draw[edge] (v1) to (v3);
		\draw[edge] (v2) to (v3);

	\end{tikzpicture}
	\caption{Joining two fish graphs in all possible ways (\cref{def:uplus}) gives rise to eight copies of $G_1$ and 16 copies of $G_2$. Dividing by $4!$ yields the shown multiplicities. }
	\label{fig:joining_fish}
\end{figure}

\Cref{lem:completion_general} holds not only for completions $g\uplus \left \lbrace v \right \rbrace $, but also for joining two arbitrary $(2p)$-valent graphs $g_1 \uplus g_2$ according to \cref{def:uplus}.  This becomes obvious once one realizes that it is equivalent to sum over all decompositions of the completion vertex $v$ (as done in the proof of \cref{lem:completion_general}), or to fix one decomposition of $v$ and instead sum over all permutations of the external edges of $g$. Stated more formally: 
\begin{lemma}\label{lem:permutations_external_structure}
	The sum over all channels, and over all vertex decompositions, of a $(2p)$-valent graph $g$  has the same external structure (i.e. the same pattern of external edges being joined into pairs) as the sum over all decompositions of a $(2p)$-valent vertex. The constant of proportionality is the circuit partition polynomial $J(g,N)$. 
\end{lemma}
\begin{proof}
	According to \cref{def:circuit_partition_polynomial}, the circuit partition polynomial is a sum over all decompositions of the vertices of $g$. 
	Let $\zeta$ be one of these decompositions. $\xi$ implies a partition  of the external edges of $g$ into pairs, as illustrated in \cref{fig:fish}.
	If we join $g$ to another $(2p)$-valent graph $g_2$, the operation $g \uplus g_2$ involves a sum over all permutations of the external edges of $g$. This sum results in an induced sum for each of the decompositions~$\zeta$. Instead of first summing over all decompositions, and then summing each term over the permutations, we may   exchange the order of summation. Regardless of the details of the individual decomposition $\zeta$, the sum over its external edges is always the same: It is a symmetric sum over all ways of pairing the external edges. But the latter is exactly what it means to decompose a single vertex, see \cref{fig:vertex_decomposition}, so we have
	\begin{align*}
		\sum_{\substack{\text{permutations of}\\ \text{external edges of }g}} \sum_{\substack{\text{decompositions }\zeta \\ \text{of permuted }g}} N^{\# \text{circuits}} &= \sum_{\substack{\text{decompositions }\zeta\\ \text{of }g}} \sum_{\substack{\text{permutations of}\\\text{external edges of }\zeta}} N^{\# \text{circuits}} \\
		&= \Big( \sum_{\substack{\text{decompositions }\zeta\\ \text{of }g}} N^{\# \text{circuits}} \Big)  \cdot \sum_{\substack{\text{decompositions }\\\text{of a $(2p)$-valent vertex}}}
	\end{align*}
	The factor in parentheses is the circuit partition polynomial of $g$. 
\end{proof}

\begin{figure}[htb]
	\centering 
	\begin{tikzpicture}[scale=.8]

		\coordinate(x0) at (-4,-2);
		\node[smallvertex](v1) at ($(x0) +(-.8,0)$){};
		\node[smallvertex](v2) at ($(x0) +(0,0)$){};
		\draw[edge,bend angle=40,bend left] (v1) to (v2);
		\draw[edge,bend angle=40,bend right] (v1) to (v2);
		\draw[edge] (v1) -- +(135:.5);
		\draw[edge] (v1) -- +(225:.5);
		\draw[edge] (v2) -- +(45:.5);
		\draw[edge] (v2) -- +(-45:.5);
		
		\coordinate(x0) at (-2.5,-2);
		\node at ($(x0) + (-.6,0)$){$+ $};
		\node[smallvertex](v1) at ($(x0) +(0,.4)$){};
		\node[smallvertex](v2) at ($(x0) +(0,-.4)$){};
		\draw[edge,bend angle=40,bend left] (v1) to (v2);
		\draw[edge,bend angle=40,bend right] (v1) to (v2);
		\draw[edge] (v1) -- +(135:.5);
		\draw[edge] (v2) -- +(225:.5);
		\draw[edge] (v1) -- +(45:.5);
		\draw[edge] (v2) -- +(-45:.5);
		
		\coordinate(x0) at (-1,-2);
		\node at ($(x0) + (-.6,0)$){$+ $};
		\node[smallvertex](v1) at ($(x0) +(0,.4)$){};
		\node[smallvertex](v2) at ($(x0) +(0,-.4)$){};
		\draw[edge,bend angle=10,bend left] (v1) to (v2);
		\draw[edge,bend angle=50,bend right] (v1) to (v2);
		\draw[edge, bend angle=20, bend left] (v1) to +(-70:1.2);
		\draw[edge] (v2) -- +(225:.5);
		\draw[edge] (v1) -- +(135:.5);
		\draw[edge, bend angle=20, bend right] (v2) to +(70:1.2);
		
		\draw[->, line width=.3mm] (x0)++(1,0) -- ++(.8,0);
		
		\node (v2) at (3.5,-2){};
		\node at ($(v2) + (-1.3,0)$) {$8\times \left(\rule{0em}{7mm}\right.$};
		\draw[edge,blue] (v2)+(.3,.3) to[bend angle=60,bend right] ($(v2)+(.3,-.3) $);
		\draw[edge,blue] (v2)+(-.3,.3) to[bend angle=60,bend left] ($(v2)+(-.3,-.3) $);
	
		\node (v3) at (4.8,-2){};
		\node at($(v3) + (-.6,0)$){$+$};
		\draw[edge,blue] (v3)+(.3,.3) to[bend angle=60,bend left] ($(v3)+(-.3,.3) $);
		\draw[edge,blue] (v3)+(.3,-.3) to[bend angle=60,bend right] ($(v3)+(-.3,-.3) $);
		
		\node (v4) at (6.1,-2){};
		\node at ($(v4) + (-.6,0)$){$+$};
		\draw[edge,blue] (v4)+(.3,.3) -- ($(v4)+(-.3,-.3) $);
		\draw[edge,blue] (v4)+(-.3,.3) -- ($(v4)+(-.1,.1) $);
		\draw[edge,blue] (v4)+(.3,-.3) -- ($(v4)+(.1,-.1) $);
		\node at ($(v4) + (1,0)$) {$\left.\rule{0em}{7mm}\right) $};

		\node at (1.3,-4) {$+1 \times~ $};
		\draw[edge] (2.5,-4) circle(.4);
		
		\node (v2) at (4.5,-4){};
		\node at ($(v2) + (-.9 ,0)$) {$ \left(\rule{0em}{7mm}\right.$};
		\draw[edge,blue] (v2)+(.3,.3) to[bend angle=60,bend right] ($(v2)+(.3,-.3) $);
		\draw[edge,blue] (v2)+(-.3,.3) to[bend angle=60,bend left] ($(v2)+(-.3,-.3) $);
		
		\node (v3) at (5.8,-4){};
		\node at($(v3) + (-.6,0)$){$+$};
		\draw[edge,blue] (v3)+(.3,.3) to[bend angle=60,bend left] ($(v3)+(-.3,.3) $);
		\draw[edge,blue] (v3)+(.3,-.3) to[bend angle=60,bend right] ($(v3)+(-.3,-.3) $);
		
		\node (v4) at (7.1,-4){};
		\node at ($(v4) + (-.6,0)$){$+$};
		\draw[edge,blue] (v4)+(.3,.3) -- ($(v4)+(-.3,-.3) $);
		\draw[edge,blue] (v4)+(-.3,.3) -- ($(v4)+(-.1,.1) $);
		\draw[edge,blue] (v4)+(.3,-.3) -- ($(v4)+(.1,-.1) $);
		\node at ($(v4) + (1,0)$) {$\left.\rule{0em}{7mm}\right) $};

	\end{tikzpicture}
	\caption{
	By \cref{lem:permutations_external_structure}, the sum over all vertex decompositions, and additionally over all channels, is proportional to the decomposition of a single vertex (highlighted in blue). The proportionality factor is  $J(g,N)$, which is $(N+8)$ for the fish graph, compare \cref{fig:fish}.}
	\label{fig:fish_example2}
\end{figure}

\Cref{fig:fish_example2} illustrates \cref{lem:permutations_external_structure}. From the viewpoint of renormalization theory, \cref{lem:permutations_external_structure} asserts that quantum corrections to a $(2p)$-valent vertex, at every loop order, can be interpreted as a redefinition of that vertex (and \emph{not} as introduction of a new $(2p)$-valent vertex with different tensor structure). In that sense, \cref{lem:permutations_external_structure} is the analogue of what would be a Ward- or Slavnov-Taylor identity in a gauge theory: It guarantees that counterterms do not break the symmetry of the theory. For a (local) gauge symmetry, such identities involve kinematic dependences, but this is not the case for the present (global) $\O(N)$ symmetry. 

Combining \cref{lem:completion_general} with \cref{lem:permutations_external_structure}, we see that circuit partition polynomials factorize not only upon adjoining a completion vertex $v$, but whenever two graphs are joined in all possible ways according to \cref{def:uplus}.

\begin{lemma}\label{lem:factorization_J}
	Let $p$ be a positive integer. Denote by $J(G,N)$  the circuit partition polynomial (\cref{def:circuit_partition_polynomial}), and by $\uplus$ the joining operation (\cref{def:uplus}).
	\begin{enumerate}
		\item Let $g_1$ be a $(2p)$-valent graph with completion $G_1=g_1 \uplus \left \lbrace v \right \rbrace $ (\cref{def:completion_decompletion}). Then
		\begin{align*}
			J(G,N) &= \Big( \prod \nolimits_{j=0}^{p-1}(N+2j) \Big) \cdot J(g_1,N).
		\end{align*} 
		\item Let $g_2$ be another $(2p)$-valent graph (possibly isomorphic to $g_1$). Then
		\begin{align*}
			J(g_1 \uplus g_2,N)&= \Big( \prod \nolimits_{j=0}^{p-1}(N+2j) \Big)\cdot J(g_1,N)\cdot J(g_2,N).
		\end{align*} 
	\end{enumerate}	
\end{lemma}

The union $g_1 \uplus g_2$ of graphs (\cref{def:uplus}) can equivalently be interpreted as first completing $G_2=g_2 \uplus \left \lbrace v \right \rbrace $, and then inserting $g_1$ in place of the completion vertex $v$. This allows for one more level of generalization of \cref{lem:factorization_J}, namely, one can insert into a vertex of a non-completed graph, and obtain another non-completed graph (whereas the operation $\uplus$ always produces completions). The appropriate factorization formula for that case can be obtained by first completing $g_2\cup \left \lbrace v \right \rbrace $, then inserting into $v_2\in g_2$, where $v_2\neq v$, and finally removing $v$ from the resulting completion. If all $J(G,N)$ are transformed into $T(G,N)$, this finally proves  \cref{lem:factorization_T} from \cref{sec:symmetric_theory}. 
 \Cref{lem:circuit_polynomial_zeros} is then an immediate corollary, upon recalling \cref{lem:permutations_external_structure}.

\begin{example}
Consider the two possibilities of merging two fish graphs. A fish graph has three distinct channels, shown in \cref{fig:fish_example2}. Merging the two graphs produces the two non-isomorphic graphs shown in \cref{fig:joining_fish}. One of the three possibilities results in a five-loop \enquote{ring} graph $G_1$ of four fish, whose $\O(N)$-symmetry factor has been computed in \cref{ex:fish_chain},
\begin{align*}
	T(G_1,N)&= \frac{N(N+2)(N^2+6N+20)}{81}.
\end{align*}
The remaining two channels produce isomorphic graphs. Decomposing any of their vertices results in a three-loop multiedge with one tadpole and two copies of four-loop ring graphs. Together, the $\O(N)$-symmetry factor of these graphs is
\begin{align*}
	T(G_2,N) &= \frac 13 \frac{N+2}{3}N\frac{N+2}{3} + \frac 2 3 \frac{N(N+2)(N+8)}{27}= \frac{N(N+2)(5N+22)}{81}.
\end{align*}
By adding these two terms, we confirm \cref{lem:factorization_T},
Indeed, we find
\begin{align*}
	\frac 1 3 T(G_1,N) + \frac 2 3 T(G_2,N) = \frac{N(N+2)(N+8)^2}{243} 
 &= \frac{N(N+2)}{3} \frac{N+8}{9} \frac{N+8}{9} \nonumber \\
&=T(\melon,N)\, T(g,N)\, T(g,N).
\end{align*}
\end{example}

\begin{landscape}

	\section{Tables} \label{sec:tables}

	\begin{table}[htb]
		\centering
		\begin{tblr}{ vlines, 
				vline{2}={1}{-}{solid},
				vline{2}={2}{-}{solid},
				hline{1}={solid},
				hline{2,Z}={solid},
				rowsep=0pt,
				cells={font=\fontsize{8pt}{10pt}\selectfont, mode=math },
				columns={halign=r,colsep=2pt},
				row{1}  = { rowsep=2pt, halign=c, valign=m,font = \fontsize{10pt}{12pt}\selectfont } 
			}
			L & N^0 & N^1 & N^2 & N^3 & N^4 & N^5 & N^6 & N^7 & N^8 & N^9 & N^{10}\\
			1 & 2.66667(1) & 0.333335(2) & & & & \\
			3 & 11.75344(3) & 2.671236(3) & & & & & &   \\
			4 & 95.2440(3) & 28.16354(8) & 1.024129(3) & & & & & \\
			5 &  1.226284(7)^3 & 438.766(3) & 33.1116(2) & & & & & \\
			6 &  16.49034(6)^3  & 6.87280(3)^3 & 751.562(3) & 16.06524(5)   & & & &  \\
			7 & 240.5389(6)^3  &  113.6758(3)^3 & 16.03469(4)^3 & 672.775(2)  & 2.59286(2)  & & &  \\
			8 & 3.739433(7)^6 & 1.965620(4)^6  & 335.6488(7)^3 &  20.83929(5)^3 & 337.3748(8)  & & &  \\
			9 &  61.4643(2)^6  & 35.41073(7)^6 & 7.03789(2)^6 & 572.453(2)^3  &  16.79838(5)^3 & 88.9654(4)  &  &  \\
			10 & 1.061828(1)^9  & 662.7130(6)  & 149.2471(2)^6  & 14.88098(2)^6  & 634.029(1)^3  & 8.68937(2)^3 & 9.51925(9) &  \\
			11 & 19.197780(8)^9 & 12.857492(6)^9 & 3.217478(2)^9 & 377.0775(2)^6 & 20.96349(1)^6 & 481.1632(3)^3 & 2.822003(4)^3 & \\
			12 & 362.13909(8)^9 & 258.20912(6)^9 & 70.73210(2)^9 & 9.465852(2)^9 & 645.2312(2)^6 & 20.914424(6)^6 & 254.8850(2)^3 & 523.581(2) \\
			13 & 7.111467(2)^{12} & 5.361932(2)^{12} & 1.5889419(4)^{12} & 237.73172(6)^9 & 19.076888(5)^9 & 797.4419(3)^6 & 15.244803(7)^6 & 93.2975(2)^3 & 42.1726(2) \\
			14 & 145.3(3)^{12} & 115.2(3)^{12} & 36.57(9)^{12} & 6.02(2)^{12} & 553(3)^9 & 28.2(2)^6 & 743(9)^6 & 8.3(3)^6  & 22.5672(1)^3 \\
			15 & 3.068(8)^{15} & 2.543(8)^{15} & 858(3)^{12} & 153.3(7)^{12} & 15.8(1)^{12} & 943(9)^9 & 31.5(5)^9 &  0.53(2)^9 & 3.5(5)^6 &  3.246120(8)^3 \\
			16 & 67.0(6)^{15} & 57.7(6)^{15} &  20.5(3)^{15} &  3.92(6)^{15} & 441(8)^{12} & 29.7(7)^{12} & 1.16(4)^{12} & 24(2)^9 & 0.23(4)^9 & >9(1)^3 & 209.836(3) 
		\end{tblr}
		\caption{Coefficients of $N$ of the primitive beta function in four dimensions. The notation $12.34(5)^6$ is a shorthand for $(12.34 \pm 0.05)\cdot 10^6$, where the uncertainty is one standard deviation. The leading coefficients for $L\geq 11$ result from new numerical integration, all other values from the data set of \cite{balduf_predicting_2024}. For $L=16$, the random sample of  \cite{balduf_predicting_2024} does not include any graphs that contribute at order $N^9$ or $N^10$. The value $9(1)^3$ at $N^9$ is the contribution at $N^9$ from those graphs that contribute also at $N^{10}$, it is a lower bound for the coefficient at $N^9$ since it leaves out those graphs that contribute at $N^9$, but not at $N^{10}$. For all $L\leq 15$, the bounds obtained by the leading-order graphs are weaker than the numerical values obtained from random sampling. \newline We confirm the observation of \cref{sec:enumeration_low_loop_order}, namely that, despite the presence of terms of high order in $N$, the coefficients of $N^0$ or $N^1$ are numerically larger by far.}
		\label{tab:coefficients}
	\end{table}

\end{landscape}

\begin{landscape}

	\begin{table}[htb]
		\centering
		\begin{tblr}{ vlines, 
				vline{2}={1}{-}{solid},
				vline{2}={2}{-}{solid},
				hline{1}={solid},
				hline{2,Z}={solid},
				rowsep=0pt,
				cells={font=\fontsize{8pt}{10pt}\selectfont,valign=m, mode=math},
				cell{14-19}{2-11}={ font=\tiny\selectfont},
				columns={halign=r,colsep=2pt},
				cell{14-16}{4-11}={mode=text},
				cell{17}{7}={mode=text},
				row{1}  = { rowsep=4pt, halign=c, valign=m,font = \fontsize{10pt}{12pt}\selectfont } 
			}
			L & N=-6 & N=-4 & N=-2 & N=-1 & N=0 & N=1 & N=2 & N=3& N=4 & N=5 \\
			1 & 0.66666(1) & 1.33334(1) & 2.00001(2) & 2.33335(2) & 2.66668(2) & 3.00002(2) & 3.33335(3) & 3.66668(3) &4.00003(4) & 4.33336(2) \\
			3 & -4.273978(9)  & 1.068494(3) &6.41097(2) & 9.08220(2) & 11.75344(3) & 14.42467(4) & 17.09591(4) & 19.76715(5) & 22.43838(5) & 25.10962(6) \\
			4 & -36.8686(1) & -1.024129(3) & 43.0134(2) & 68.1046(2) & 95.2440(3) & 124.4316(4)& 155.6676(5) & 188.9517(5) & 224.2842(6) & 261.6649(7)\\
			5 & -214.291(2) & 1.00792(9) & 481.199(3) & 820.630(5) & 1.126284(7)^3 & 1.69816(1)^3&2.23626(2)^3 & 2.84059(2)^3 & 3.51113(2)^3 & 4247.90(3)\\
			6 & -1.160331(4)^3 &  -4.0506(2) & 5.62246(2)^3 & 10.35304(4)^3 & 16.49035(6)^3 & 24.13078(8)^3 & 33.3707(2)^3 & 44.3066(2)^3 & 57.0347(2)^3 &71.6516(3)^3 \\
			7 & -6.22640(2)^3 & -3.235(2) & 71.9853(2)^3 & 142.2276(4)^3 & 240.5389(6)^3 & 370.9248(9)^3 & 537.453(2)^3 & 744.253(2)^3 & 995.519(3)^3 & 1.295503(4)^6 \\
			8 & -34.9763(2)^3 & -11.99(1) & 989.470(2)^3 & 2.088956(3)^6 & 3.739425(6)^6 & 6.06186(1)^6 & 9.18536(2)^6 & 13.2471(3)^6 & 18.39233(3)^6 & 24.77452(5)^6\\
			9 & -206.792(3)^3 & -8.6(2) & 14.48076(3)^6 &32.53576(6)^6 & 61.4644(2)^6 & 104.5024(2)^6 & 165.2889(4)^6 & 247.8765(5)^6 & 356.7427(7)^6 & 496.799(1)^6\\
			10 & -1.26654(2)^6 & -6.8(5) & 224.2097(2)^6 & 534.1068(5)^6 & 1.061829(1)^9 & 1.889313(2)^9 & 3.113716(3)^9 & 4.848456(5)^9 & 7.22427(7)^9 & 10.39026(1)^9\\
			11 & -7.89133(7)^6 & 31(2) & 3.656288(2)^9 & 9.201175(4)^9 & 19.197780(8)^9 & 35.67128(2)^9 & 61.1029(3)^9 & 98.72568(4)^9 & 152.11128(7)^9 & 225.7068(1)^9\\
			12 & -49.5188(5)^6 & 84(9) & 62.59314(2)^9 & 165.82079(4)^9 & 362.13909(8)^9 & 701.2126(2)^9 & 1.2482219(3)^{12} & 2.0864662(5)^{12} & 3.3201518(7)^{12} & 5.077369(2)^{12} \\
			13 & -312.076(7)^6 & 110 (111) & 1.1221927(3)^{12} & 3.1190391(8)^{12}  & 7.111465(2)^{12} & 14.319958(4)^{12} & 26.424680(7)^{12} & 45.66681(2)^{12} & 74.96131(2)^{12} & 118.02171(3)^{12} \\
			14& 0(2)^9 & -25(16)^6 & {$21.04(3)^{12}$ \\ \fontsize{8pt}{10pt}\selectfont $\bm{21.03(1)^{12}}$} & {{\tiny $61.2(1)^{12}$} \\ \fontsize{8pt}{10pt}\selectfont$\bm{61.12(1)^{12}}$} & {$145.3(3)^{12}$ \\ \fontsize{8pt}{10pt}\selectfont $\bm{145.11(7)^{12}}$} & {{\tiny$303.6(7)^{12}$}\\ \fontsize{8pt}{10pt}\selectfont$\bm{303.30(4)^{12}}$} &  {{\tiny $580(2)^{12}$}  \\ \fontsize{8pt}{10pt}\selectfont$\bm{579.3(1)^{12}}$ } & {{ $1.0337(3)^{15}$} \\ \fontsize{8pt}{10pt}\selectfont $\bm{1.0336(3)^{15}}$} & { $1.750(5)^{15} $ \\ \fontsize{8pt}{10pt}\selectfont $\bm{1.746(1)^{15}}$} & { $2.834(9)^{15} $ \\ \fontsize{8pt}{10pt}\selectfont $\bm{2.831(2)^{12}}$} \\
			15 & -38(27)^9 & 100(240)^6 &  { $410.7(8)^{12}$ \\ \fontsize{8pt}{10pt}\selectfont $\bm{411.47(6)^{12}}$} & {$1.244(3)^{15}$ \\ \fontsize{8pt}{10pt}\selectfont $\bm{1.2470(2)^{15}}$} & {$3.08(2)^{15}$ \\ \fontsize{8pt}{10pt}\selectfont $\bm{3.081(5)^{15}}$} & {$6.64(2)^{15}$ \\ \fontsize{8pt}{10pt}\selectfont $\bm{6.655(1)^{15} }$} & {$13.10(5)^{15}$ \\ \fontsize{8pt}{10pt}\selectfont $\bm{13.131(9)^{15}}$} & {$24.09(9)^{15}$ \\ \fontsize{8pt}{10pt}\selectfont $\bm{24.14(3)^{15}}$} & {$ 41.9(2)^{15}$ \\ \fontsize{8pt}{10pt}\selectfont $\bm{ 42.03(1)^{15}}$}& {$ 69.7(4)^{15}$ \\ \fontsize{8pt}{10pt}\selectfont $\bm{ 69.91(4)^{15}}$} \\
			16 & 0(2)^{12} & 0(8)^9 & {$8.34(5)^{15}$ \\ \fontsize{8pt}{10pt}\selectfont $\bm{ 8.39(2)^{15}}$}& {$26.2(2)^{15}$ \\ \fontsize{8pt}{10pt}\selectfont $\bm{ 26.44(3)^{15}}$} & {$67.0(6)^{15}$ \\ \fontsize{8pt}{10pt}\selectfont $\bm{ 67.5(2)^{15}}$} & {{ $150(2)^{15} $}\\$\fontsize{8pt}{10pt}\selectfont\bm{151.14(6)^{15} }$}& {$304(4)^{15}$ \\ \fontsize{8pt}{10pt}\selectfont $\bm{ 306.9(8)^{15}}$} & {$574(7)^{15}$\\ \fontsize{8pt}{10pt}\selectfont $\bm{ 583(1)^{15}}$}& {$1.03(2)^{18} $ \\\fontsize{8pt}{10pt}\selectfont $\bm{1.043(4)^{18}}$} & {$1.75(3)^{18} $ \\\fontsize{8pt}{10pt}\selectfont $\bm{1.781(3)^{18}}$} \\
			17 & 0(23)^{12} & 0(120)^9 & 178(2)^{15} & 580(6)^{15} & 1.53(2)^{18} & {$3.53(5)^{18} $ \\ \fontsize{8pt}{10pt}\selectfont $\bm{ 3.553(9)^{18}}$} & 7.4(2)^{18} & 14.4(3)^{18} & 26.3(5)^{18} & 46.0(2)^{18}\\
			18 & -1.1(9)^{15} & 0(3)^{12} & 3.97(8)^{18} & 13.4(4)^{18} & 37(2)^{18} & 87(3)^{18} & 187(8)^{18} & 370(20)^{18} & 700(40)^{18} & 1.25(8)^{21}
		\end{tblr}
		\caption{Evaluations of the $L$-loop coefficients $\beta^\text{prim}_L(N)$ of the primitive beta function at specific values of $N$. For $N \geq -2$, the coefficient $\beta^\text{prim}_L(N)$ grows monotonically with $N$. For $L>1$, $\beta^\text{prim}_L(N)$ has a zero close to $N=-4$, this is reflected by a comparably small function value at $N=-4$. The notation $0(a)^b$ indicates that the value is indistinguishable from zero with an uncertainty of $a\cdot 10^b$. \\
			For $L\geq 14$, the values arise from uniform non-complete samples of \cite{balduf_statistics_2023}, hence their uncertainty is notably larger than for $L \leq 13$. These values are written in small print. In particular, for $L\in \left \lbrace 17,18 \right \rbrace $, the samples are so small that the uncertainty estimate might be too optimistic. Bold font indicates that a value has been obtained with the importance sampling algorithm of \cite{balduf_predicting_2024}, these are still non-complete samples, but more accurate than the uniform sample (given above the bold value). }
		\label{tab:evaluations}
	\end{table}

	\begin{table}[p]
		\centering
		\begin{tblr}{ vlines, 
				vline{2}={1}{-}{solid},
				vline{2}={2}{-}{solid},
				vline{Y}={1}{-}{solid},
				vline{Y}={2}{-}{solid},
				hline{1}={solid},
				hline{2,Z}={solid},
				rowsep=0pt,
				cells={font=\fontsize{9pt}{12pt}\selectfont,mode=math },
				columns={halign=r},
				row{1}  = { rowsep=2pt, halign=c, valign=m,font = \fontsize{10pt}{12pt}\selectfont } 
			}
			L& N^0 &  N^1 & N^2 & N^3 & N^4 & N^5 & N^6 & N^7 & \left \langle k\right \rangle _L \\
			1 & \frac 12 & \frac 1{16} & & & & &  & &0.111 \\
			3 & \frac{11}{4} & \frac{5}{8} & & & & & & & 0.185 \\
			4 & \frac{93}{4} & \frac{55}{8} & \frac 14 & & & & & & 0.243 \\
			5 &  340 & \frac{973}{8} & \frac{147}{16} & & & & & & 0.297 \\
			6 & 5500 & \frac{18259}{8} & \frac{3949}{16} & 5 & & & & & 0.347 \\
			7 & \frac{405351}{4} & \frac{189811}{4} & 656 & \frac{8395}{32} & \frac 78 & & & & 0.394  \\
			8 & \frac{8297865}{4} & \frac{4297721}{4} & 178559 & \frac{336221}{32} & \frac{1211}{8} & & & & 0.438  \\
			9 & 46581140 & \frac{105141971}{4}& \frac{40380769}{8} & \frac{12386715}{32} & \frac{654445}{64} & 44 & & & 0.480 \\
			10 & 1136329932 & \frac{2761781725}{4} & \frac{1193546199}{8} & \frac{446298325}{32} & \frac{34315811}{64} & \frac{49435}{8} & 5 & &0.519 \\
			11 & \frac{119592596531}{4} & \frac{155018200259}{8} & \frac{9243683059}{2} & \frac{16159964821}{32} & \frac{405018107}{16} & \frac{63513385}{128} & \frac{72137 }{32} & & 0.556 \\
			12 & \frac{3374570327461}{4} & \frac{4630005232849}{8} & \frac{600700832091 }{4} & \frac{596312684627}{32} & \frac{18283668245}{16} & \frac{4077699355}{128} & \frac{19881095}{64} & \frac{3619}{8} & 0.590
		\end{tblr}
		\caption{For every loop order $L$, this table shows the coefficients of $N$ dependence of $p_L(N)$ according to \cref{primitive_generating_function}, multiplied by $\frac{3^{L+1}}{4!}$. This is the sum of the circuit partition polynomials $J(g,N)$ of primitive decompletions $g$, weighted by their automorphism symmetry factor (where external vertices are fixed, or equivalently a factor $4!$ is included). Although the degree of these polynomials grows with $L$ (see \cref{lem:TGN_bound}), the  average order $\left \langle k \right \rangle _L$ (\cref{def:average_order})   does not even reach unity for $L\leq 12$.}
		\label{tab:0dim_primitive_coefficients}
	\end{table}
	
		\begin{table}[p]
		\centering
		\begin{tblr}{ vlines, 
				vline{2}={1}{-}{solid},
				vline{2}={2}{-}{solid},
				vline{Y}={1}{-}{solid},
				vline{Y}={2}{-}{solid},
				hline{1}={solid},
				hline{2,Z}={solid},
				rowsep=0pt,
				cells={font=\fontsize{8pt}{10pt}\selectfont, mode=math },
				columns={halign=r},
				row{1}  = { rowsep=1pt, halign=c, valign=m,font = \fontsize{10pt}{12pt}\selectfont },
				stretch=.7
			}
			L & N^0 &  N^1 & N^2 & N^3 & N^4 & N^5 & N^6& N^7 & \left \langle k\right \rangle _L \\
			1 & 8 & 1 & & & & & & & 0.111 \\
			3 & 66 & 15 & & & & & & & 0.185 \\
			4 & 186 & 55& 2 & & & & & & 0.243\\
			5 & 1,054 &375 & 29 & & & & & & 0.297 \\
			6 & 7,564 &3,048 &317 &  6 & & & & &0.338\\
			7 & 59,986 & 27,839 &3,858 & 169 & 2 & & & &0.393 \\
			8 & 601,600 & 308,644 & 51,092 & 3081 & 50 & & & & 0.436 \\
			9 & 7,973,840 & 4,491,030 & 868,191 & 68,983 & 2,065 & 14 & & &0.481 \\
			10 & 136,783,626 & 83,167,949 & 18,098,205 & 1,734,427 & 71,781 & 1,046 & 4 & & 0.521 \\
			11 & 2,831,713,102 & 1,840,518,187 & 442,488,075 & 49,304,820 & 2,588,273 & 56,598 & 347 & & 0.558\\
			12 & 68,493,277,652 & 47,153,165,436 & 12,325,920,197 & 1,552,897,804 & 98,221,630 & 2,916,283 & 32,440 & 73 & 0.593
		\end{tblr}
		\caption{Analogous to \cref{tab:0dim_primitive_coefficients}, but disregarding automorphism symmetry factors.  Note the substantial numerical difference with respect to \cref{tab:0dim_primitive_coefficients} at small $N$, indicating a large influence of automorphisms. }
		\label{tab:leading_decompositions_count}
	\end{table}

	\begin{table}[htbp]
		\centering
		\begin{tblr}{ vlines, 
				vline{2}={1}{-}{solid},
				vline{2}={2}{-}{solid}, 
				hline{1}={solid},
				hline{2,Z}={solid},
				rowsep=4pt,
				cells={font=\fontsize{8pt}{9pt}\selectfont , mode=math},
				columns={halign=c},
				row{1}  = { rowsep=2pt, halign=c, valign=m,font = \fontsize{10pt}{12pt}\selectfont } 
			}
			\text{Quantity}&  \text{Prefactor} & u^0 & u^1 & u^2 & u^3 \\
			\mathcal A^{\frac 32}_\hbar \partial_j^0 W\big|_{j=0} & \frac{\hbar^{-\frac{N-3}{2}} 3 ^{\frac{N-1}{2}}}{\sqrt{2 \pi}\Gamma \!\left( \frac N2 \right) } & 1 & - \frac{ N^2-2N+4}{12} & \frac{96-104N+24N^2-8N^3+N^4}{288} & \frac{N^6 + 18 N^5 - 124 N^4 + 600 N^3 - 3568 N^2 + 6432 N - 5760}{10368} \\
			\mathcal A^{\frac 32}_\hbar \partial_j^2 W\big|_{j=0} & \frac{ \hbar^{-\frac{N+1}{2}} 3^{\frac{N-1}{2}}  }{2\sqrt{2 \pi}\Gamma \!\left( \frac {N+2}2 \right) } & 6 & - \frac{N(N+2)}{3} &  \frac{N(N+2) (N^2 -2N-12)}{48} & \frac{N(N+2)(N^4 - 8 N^3 - 16 N^2 + 200 N + 480) }{1728} \\ 
			\mathcal A^{\frac 32}_\hbar \partial_j^4 W\big|_{j=0} & \frac{ \hbar^{-\frac{N+5}{2}} 3^{\frac{N-1}{2}}  }{\frac 43 \sqrt{2 \pi}\Gamma \!\left( \frac {N+4}2 \right) } & 36 & \scalemath{.8}{-3(N+2)(N+4)} &  \frac{(N+2)(N^3+6N^2 -12N-32 )}{8} & \frac{(N+2)(N^5+4N^4-64N^3+8N^2 +960N+1152) }{288}\\
			A^{\frac 3 2}_\hbar \partial^2_\varphi G\big|_{ 0}  & \frac{\hbar^{-\frac{N+1}{2}} 3^{\frac{N-1}{2}} }{2 \sqrt{2\pi} ~\Gamma  ( \frac{N+2}{2} ) } &6 & \frac{(N+2)(N+4)} {2} &  \frac{ (N+2)(N+4)(N^2 +2N-12 )}{48} &   \frac{(N+2)(N+4)(N^4 -40N^2 +72N+ 768)}{1728}   \\
			A^{\frac 3 2}_\hbar \partial^4_\varphi G\big|_{ 0}  & \frac{ \hbar ^{-\frac{N+5}{2}} 3^{\frac{N-1}{2}} }{ \frac 4 3 \sqrt{2\pi} ~\Gamma  ( \frac{N+4}{2} ) }  & 36 & \scalemath{.8}{-3(N+2)(N+12)} & \frac{N(N+2)(100+22N+N^2)}{8} & \frac{-(N+2)(4992-192N+8N^2+176N^3+28N^4+N^5)}{288} \\
			\mathcal A^{\frac 32}_\hbar \hbar_\ren  & \frac{\hbar^{-\frac{N+3}{2}} 3^{\frac{N-1}{2}} }{ \frac 4 3 \sqrt{2\pi} ~\Gamma  ( \frac{N+4}{2} ) } &
			36 &  \scalemath{.8}{-3(N+2)(N+8) } & -\frac{(N+2)(N+4)(12-10N-N^2) }{8} & \frac{ (N+2)(N+4)(1056-120N-40N^2+12N^3+N^4)}{288} \\
			\mathcal A^{\frac 32}_{\hbar_\ren} \hbar  & \frac{ \hbar_\ren^{-\frac{N+3}{2} } 3^{\frac{N-1}{2} } }  {\frac 4 3 \sqrt{2\pi} \Gamma \! \left( \frac{N+4}{2} \right) } e^{-\frac{12+3N}{4}} & 
			-36 &  \frac{-9N^2+60N+336}{2} & -\frac{(N+4)(9N^3-172N^2+552N+2368)}{32} & -\frac{(N+4)(27N^5-792N^4+6088N^3+3872N^2-144256 N-363008)}{2304} \\
			\mathcal A^{\frac 32}_{\hbar_\ren} z^{(4)} & \frac{\hbar_\ren^{-\frac{N+5}{2} } 3^{\frac{N-1}{2} }}  {\frac 4 3 \sqrt{2\pi} \Gamma \! \left( \frac{N+4}{2} \right) } e^{-\frac{12+3N}{4}} &
			-36 & -\frac{9N^2-12N-240}{2} & -\frac{9N^4-40N^3-488N^2-416N+1664}{32}  & -\frac{27N^6-252N^5-1880N^4+960N^3+21376N^2-59904N-419840}{2304} 
		\end{tblr}
		\caption{Asymptotic growth coefficients of 1PI Green function, renormalized coupling and counterterms computed in \cref{sec:0dim_asymptotics}, see \cref{sec:4d_coefficients}. Here $u=\hbar$ or $u=\hbar_\ren$, depending on the argument as indicated by $\mathcal A^{\frac 32}_u$. Setting $N=1$, the series reproduce \cite{borinsky_renormalized_2017}. Notice that the prefactor and leading terms of $\partial^4_j W$, $\partial^4_\varphi G$, and $z^{(4)}$ coincide up to $e^{-\frac{15+3N}{4}}$. This means that to asymptotic leading order, the number of connected-, 1PI-, and primitive vertex type graphs is the same up to $e^{-\frac{15+3N}{4}}$. The distinction between these classes results in a subleading change in their number.}
		\label{tab:primitive_asymptotic_coefficients}
	\end{table}

\end{landscape}

\begin{multicols}{2}
	\printbibliography
\end{multicols}

\end{document}